\pdfoutput=1  % No later than line 5, for ArXiv

\documentclass[a4paper,11pt,twoside]{thesis-style}
\def\COMPLETE{YES}

\usepackage{etex}

\usepackage{xfrac} % This provides \sfrac, for making slanted functions

\usepackage{xspace}
\usepackage{amsmath}
\usepackage{amssymb}             % AMS Math
\usepackage{mathrsfs} % Used for \mathscr
\usepackage{mathabx}
\usepackage{slashed}

\usepackage{caption} % Useful for figures spanning multiple pages

\usepackage{stmaryrd}
\usepackage{url}
\usepackage{multicol}
\usepackage{ucs}
\usepackage[utf8x]{inputenc}
\usepackage{listings}
\usepackage{fancyvrb} % For the Verbatim (big V) environment.
\usepackage{color}
\usepackage{bussproofs}
\usepackage{semantic} % For T diagrams
\usepackage{wasysym} % for \frownie and \lightning
\usepackage{ifsym}
\usepackage[T1]{fontenc}
\usepackage{enumerate}
\usepackage{pgf}
\usepackage{tikz}
\usetikzlibrary{automata,calc,through,backgrounds,arrows,decorations.pathmorphing,decorations.shapes,fit,positioning,matrix}
\tikzset{terminal/.style=...}

\newcommand{\WHENCOMPLETE}[1]{\ifx\COMPLETE\THISNAMEISDEFINITELYNOTBOUND{}\else{#1}\fi}

\newcommand{\UNLESSCOMPLETE}[1]{\ifx\COMPLETE\THISNAMEISDEFINITELYNOTBOUND{#1}\else{}\fi}

\definecolor{pagecounter}{rgb}{0.5,0.5,0.5}
\definecolor{darkred}{rgb}{0.6,0,0}
\definecolor{darkgreen}{rgb}{0,0.6,0}
\definecolor{darkblue}{rgb}{0,0,0.6}
\definecolor{darkyellow}{rgb}{0.4,0.4,0}
\definecolor{purple}{rgb}{0.4,0,1}
\definecolor{brown}{rgb}{0.6,0.1,0}
\definecolor{white}{rgb}{1,1,1}
\definecolor{black}{rgb}{0,0,0}

\newcommand{\urlsmall}[1]{{\scriptsize\url{#1}}}

\newcommand{\PREMISEWHICHCOULDBEPROVEN}[1]{}%{{\textcolor{purple}{#1}}}

\newcommand{\IGNORE}[1]{}

\newcommand{\INVISIBLE}[1]{{\textcolor{white}{{#1}}}}

\newcommand{\BANANABRACKETS}[1]
  {\mbox{$\llparenthesis {#1} \rrparenthesis$}}

\newcommand{\LIFT}[1]
  {\lfloor{#1}\rfloor}
\newcommand{\EQD}[0]{\triangleq}

\newcommand{\FUTURE}[1]{\mathcal{T}({#1})} % token
\usepackage{aecompl}

\usepackage{color}
\definecolor{linkcol}{rgb}{0,0,0.4} 
\definecolor{citecol}{rgb}{0.5,0,0} 

\usepackage{rotating}                    % Sideways of figures & tables
\makeatletter

\def\cleardoublepage{\clearpage\if@twoside \ifodd\c@page\else%
  \hbox{}%
  \thispagestyle{empty}%              % Empty header styles
  \newpage%
  \if@twocolumn\hbox{}\newpage\fi\fi\fi}

\makeatother
 
{%

\hrulefill
\vspace*{0.5cm}%
\end{minipage}
}

\usepackage{multirow}
\usepackage{pdflscape}
\usepackage{amsfonts} % This interferes with mathbbol and they are not
{ \begin{list}%
	{$\bullet$}%
	{\setlength{\labelwidth}{25pt}%
	 \setlength{\leftmargin}{30pt}%
	 \setlength{\itemsep}{\parsep}}}%
{ \end{list} }

\newcommand{\SPACER}[0]{\\\INVISIBLE{\small{\tiny{hack!}}}}
\newcommand{\SPACERP}[0]{\SPACER$|\ $}
\newcommand{\SPACERF}[0]{\SPACER\INVISIBLE{$|\ $}}
\newcommand{\DISJOINTUNION}[0]{\uplus}

\newcommand{\BUNDLENAME}[0]
  {\CODE{bundle}\xspace}
\newcommand{\LETNAME}[0]
  {\CODE{let}\xspace}
\newcommand{\PRIMITIVENAME}[0]
  {\CODE{primitive}\xspace}
\newcommand{\CALLNAME}[0]
  {\CODE{call}\xspace}
\newcommand{\CALLINDIRECTNAME}[0]
  {\CODE{call-{\allowbreak}indirect}\xspace}
\newcommand{\IFNAME}[0]
  {\CODE{if}\xspace}
\newcommand{\THENNAME}[0]
  {\CODE{then}\xspace}
\newcommand{\ELSENAME}[0]
  {\CODE{else}\xspace}
\newcommand{\FORKNAME}[0]
  {\CODE{fork}\xspace}
\newcommand{\JOINNAME}[0]
  {\CODE{join}\xspace}

\newcommand{\DEFINEPROCEDURENAME}[0]
  {\CODE{procedure}\xspace}%{\CODE{define-procedure}}

\newcommand{\RED}[1]
  {\textcolor{red}{#1}}

\newcommand{\EPSILON}[0]{$\varepsilon$\xspace}
\newcommand{\EPSILONZERO}[0]{$\varepsilon_{0}$\xspace}
\newcommand{\EPSILONONE}[0]{$\varepsilon_{1}$\xspace}

\newcommand{\EPSILONZEROHOLE}[0]{$\varepsilon_{0}^{\HOLE}$\xspace}

\newcommand{\LAMBDACALCULUS}[0]{\mbox{$\lambda$-calculus}\xspace}
\newcommand{\PICALCULUS}[0]{\mbox{$\pi$-calculus}\xspace}

\newcommand{\HOLEDES}[0]{\ensuremath{\SET{E}_{\HOLE}}}

\newcommand{\iem}[1]{\index{#1}{\em{#1}}}
\newcommand{\textiti}[1]{\index{#1}{\textit{#1}}}
\newcommand{\ind}[1]{\index{#1}{#1}}
\newcommand{\idef}[1]{\index{#1}{\sl{#1}}}
\newcommand{\TDEF}[1]{\idef{#1}}

\newcommand{\QUOTATION}[3]
      {\begin{quotation}
        {\em #1}\\
        {\INVISIBLE{}\hfill---~{#2}, {#3}}
      \end{quotation}}

\newcommand{\TOENEW}[0]
  {\Longrightarrow_{\SET{E}}}
\newcommand{\TOE}[0]
  {\LINEBREAK\longrightarrow_{\SET{E}}\LINEBREAK}
\newcommand{\TOEP}[0]
  {\LINEBREAK\longrightarrow^{+}_{\SET{E}}\LINEBREAK}
\newcommand{\TOET}[0]
  {\LINEBREAK\longrightarrow_{\SET{E}}^{+}\LINEBREAK}
\newcommand{\TOEST}[0]
  {\LINEBREAK\longrightarrow_{\SET{E}}^{\slashed{\parallel}+}\LINEBREAK}

\newcommand{\TOERT}[0]
  {\LINEBREAK\longrightarrow_{\SET{E}}^{*}\LINEBREAK}
\newcommand{\TOES}[0]
  {\longrightarrow_{\SET{E}}^{\slashed{\parallel}}}

\newcommand{\EVENTUALLYFAILSBECAUSEOF}[1]
  {\Downarrow_{\SET{E}}\FAILURE_{#1}}
\newcommand{\DOESNTFAILBECAUSEOFENVIRONMENTS}[0]
  {\DOESNTFAILBECAUSEOF{\SET{X}}}

\newcommand{\FAILURE}[0]
  {\textrm{\rm \large \textreferencemark}}

\newcommand{\FAILSBECAUSEOF}[1]
  {\longrightarrow_{\SET{E}}\FAILURE_{#1}}
\newcommand{\DOESNTFAILBECAUSEOF}[1]
  {{\slashed{\longrightarrow}}_{\SET{E}}\FAILURE_{#1}}
\newcommand{\FAILSBECAUSEOFDIMENSION}[0]
  {\FAILSBECAUSEOF{\#}}
\newcommand{\FAILSBECAUSEOFPRIMITIVE}[0]
  {\FAILSBECAUSEOF{\SET{P}}}
\newcommand{\FAILSBECAUSEOFENVIRONMENTS}[0]
  {\FAILSBECAUSEOF{\SET{X}}}
\newcommand{\FAILS}[0]
  {\FAILSBECAUSEOF{}}

\newcommand{\DIVERGESE}[0]
  {\Uparrow_{\SET{E}}}
\newcommand{\CONVERGES}[0]
  {\Downarrow}

\newcommand{\CONVERGESE}[0]
  {\CONVERGES_{\SET{E}}}

\newcommand{\CONVERGESES}[0]
  {\CONVERGES_{\SET{E}}^{\slashed{\parallel}}}

\newcommand{\PARALLELCONFIGURATION}[2]
  {{#1} \rhd {#2}}
\newcommand{\PCONF}[2]
  {\PARALLELCONFIGURATION{#1}{#2}}
\newcommand{\REWRITEP}[5]
  {{#1}\ \vdash\ {\PCONF{#2}{#3}}\ \TOE\ {\PCONF{#4}{#5}}}

\newcommand{\SET}[1]
  {\mathbb{#1}}

\newcommand{\NATURALS}{\SET{N}}

\newcommand{\UNION}{\cup}

\newcommand{\SECTION}[0]
  {§}
\newcommand{\SECTIONS}[0]
  {§§}

\newcommand{\UNFILLEDUNPOSITIONEDQED}{\ensuremath{\square}}
\newcommand{\FILLEDUNPOSITIONEDQED}{\ensuremath{\blacksquare}}

\newcommand{\FILLEDQED}{\hfill \ensuremath{\FILLEDUNPOSITIONEDQED}}

\newcommand{\PROOFQED}{\FILLEDQED}

\newcommand{\CODE}[1]{{\rm \texttt{#1}}}
\newcommand{\FILE}[1]{\texttt{#1}}
\newcommand{\UNIT}[0]{$\bullet$}
\newcommand{\HASDIMENSION}{:_{\#}}

\newcommand{\LINEBREAK}[0]
  {\allowbreak}
\newcommand{\DHANDLE}[2]
  {\LINEBREAK{#2}_{#1}}
\newcommand{\DCONSTANT}[2]
  {\LINEBREAK{#2}_{#1}}
\newcommand{\DVARIABLE}[2]
  {\LINEBREAK{#2}_{#1}}
\newcommand{\DBUNDLE}[2]
  {\LINEBREAK[\BUNDLENAME\xspace\ \LINEBREAK{#2}]_{#1}}
\newcommand{\DPRIMITIVE}[3]
  {\LINEBREAK[\PRIMITIVENAME\ \LINEBREAK{{\tt #2}}\ \LINEBREAK{#3}]_{#1}}
\newcommand{\DCALL}[3]
  {\LINEBREAK[\CALLNAME\ \LINEBREAK{{{#2}}}\ \xspace\LINEBREAK{#3}]_{#1}}
\newcommand{\DCALLINDIRECT}[3]
  {\LINEBREAK[\CALLINDIRECTNAME\ \LINEBREAK{{{#2}}}\ \xspace\LINEBREAK{#3}]_{#1}}
\newcommand{\DIFIN}[5]
  {\LINEBREAK[\IFNAME\ \LINEBREAK{#2} \LINEBREAK\in \LINEBREAK\{{#3}\}\ \LINEBREAK\CODE{then}\ \LINEBREAK{#4}\ \LINEBREAK\CODE{else}\ \LINEBREAK{#5}]_{#1}}
\newcommand{\DLET}[4]
  {\LINEBREAK[\LETNAME\ \LINEBREAK{#2}\ \text{\LINEBREAK{\tt be}}\ \LINEBREAK{#3}\ \text{\LINEBREAK{\tt in}}\ \LINEBREAK{#4}]_{#1}}
\newcommand{\DFORK}[3]
  {\LINEBREAK[\FORKNAME\ \LINEBREAK{#2}\ \LINEBREAK{#3}]_{#1}}
\newcommand{\DJOIN}[2]
  {\LINEBREAK[\JOINNAME\ \LINEBREAK{#2}]_{#1}}

\newcommand{\DDEFINENONPROCEDURE}[3]
  {\LINEBREAK[\DEFINENONPROCEDURENAME\ \LINEBREAK{#2}\ \LINEBREAK{#3}]_{#1}}
\newcommand{\DDEFINEPROCEDURE}[3]
  {\LINEBREAK[\DEFINEPROCEDURENAME\CODE{\ \LINEBREAK(}{{{#1}}\ \xspace\LINEBREAK{#2}}\CODE{)\ \LINEBREAK}{#3}]}
\newcommand{\DCALLWITHHOLE}[2]
  {\LINEBREAK[\CODE{call}\ {#2}\ \LINEBREAK{\HOLE}]_{#1}}
\newcommand{\DFORKWITHHOLE}[2]
  {\LINEBREAK[\FORKNAME\ \LINEBREAK{#2}\ \LINEBREAK\HOLE]_{#1}}
\newcommand{\DBUNDLEWITHHOLE}[1]
  {\LINEBREAK[\BUNDLENAME\ \LINEBREAK\HOLE]_{#1}}
\newcommand{\DPRIMITIVEWITHHOLE}[2]
  {\LINEBREAK[\PRIMITIVENAME\ \LINEBREAK{#2}\ \LINEBREAK\HOLE]_{#1}}

\newcommand{\SEQUENCE}[1]
  {\langle {#1} \rangle}
\newcommand{\EMPTYSEQUENCE}[0]
  {\SEQUENCE{}}

\newcommand{\SQMEET}[0]
  {\sqcap}
\newcommand{\SQJOIN}[0]
  {\sqcup}

\newcommand{\SUCHTHAT}[0]
  {:}

\newcommand{\SQL}[0]
  {\sqsubset}
\newcommand{\SQLE}[0]
  {\sqsubseteq}
\newcommand{\SQGE}[0]
  {\sqsupseteq}

\newcommand{\NSQLE}[0]
  {\nsqsubseteq}

\newcommand{\NATURALSB}[0]
  {\NATURALS_{\bot}}
\newcommand{\NATURALSTB}[0]
  {\NATURALS_{\bot}^{\top}}
\newcommand{\NATURALSBT}[0]
  {\NATURALSTB}

\newcommand{\EXPRESSIONS}{\SET{E}}

\newcommand{\PROCEDURES}{\SET{F}}

\newcommand{\PARTIAL}{\rightharpoonup}

\newcommand{\HOLE}{\Box}

\newcommand*{\SUPERIMPOSE}[2]
            {\mathrel{\mathchoice{\hbox{\hbox to 0pt{$\displaystyle{#1}$\hss}$\displaystyle{#2}$}}
                {\hbox{\hbox to 0pt{$\textstyle{#1}$\hss}$\textstyle{#2}$}}
                {\hbox{\hbox to 0pt{$\scriptstyle{#1}$\hss}$\scriptstyle{#2}$}}
                {\hbox{\hbox to 0pt{$\scriptscriptstyle{#1}$\hss}$\scriptscriptstyle{#2}$}}}}

\newcommand{\PARAMETER}[0]
  {\mathunderscore}

\newcommand{\xqed}[1]{%
  \leavevmode\unskip\penalty9999 \hbox{}\nobreak\hfill
  \quad\hbox{\ensuremath{#1}}}
\def\QEDDEFINITION{\xqed{\UNFILLEDUNPOSITIONEDQED}}
\def\QEDNOPROOF{\xqed{\UNFILLEDUNPOSITIONEDQED}}
\def\QEDPROOF{\xqed{\FILLEDUNPOSITIONEDQED}}
\def\QEDAXIOM{\xqed{\UNFILLEDUNPOSITIONEDQED}}

\def\QEDIMPLEMENTATIONNOTE{\QEDDEFINITION}
\def\QEDSYNTACTICCONVENTION{\QEDIMPLEMENTATIONNOTE}

\newcommand{\EMPTYSET}[0]{\varnothing}
\newcommand{\EMPTYSTACK}[0]{\EMPTYSEQUENCE}

\newcommand{\VALUESEPARATOR}[0]{\text{\ensuremath{\wr}}} % This hack suppresses the spaces around \wr
\newcommand{\ACTIVATIONSEPARATOR}[0]{\ensuremath{\ddagger}}

\newcommand{\UNIFY}[0]
  {\equiv}
\newcommand{\UNIFIES}[0]
  {\UNIFY}
\newcommand{\DOESNTUNIFY}[0]
  {\ensuremath{\mathrel{\slashed{\equiv}}}}
\newcommand{\STATE}
  {state\xspace}
\newcommand{\STATES}
  {states\xspace}
\newcommand{\UPDATEENVIRONMENT}[3]
  {{#1}[{#2} \mapsto {#3}]}
\newcommand{\UPDATESTATE}[3]
  {{#1}[{}_{#2}^{#3}]}
\newcommand{\UPDATESTATEIN}[4]
  {\UPDATESTATE{#1}{#2}{{#3}~\mapsto~{#4}}}

\newcommand{\UNNUMBEREDCHAPTER}[1]{\chapter*{#1}\addstarredchapter{#1}}

\newcommand{\UNNUMBEREDSECTION}[1]{\section*{#1}\addcontentsline{toc}{section}{\protect\numberline{}#1}}

\newcommand{\EPSILONGC}[0]{\CODE{epsilongc}\xspace}

\newcommand{\FATBRACKETSWITH}[2]
  {E_{\SET{#1}}\BANANABRACKETS{#2}}
\newcommand{\LOTSA}[1]
  {{#1}s}
\newcommand{\EXPANDE}[1]{\FATBRACKETSWITH{E}{#1}}

\newcommand{\EXPANDCS}[1]{\FATBRACKETSWITH{\LOTSA{C}}{#1}}
\newcommand{\EXPANDXS}[1]{\FATBRACKETSWITH{\LOTSA{X}}{#1}}
\newcommand{\EXPANDES}[1]{\FATBRACKETSWITH{\LOTSA{E}}{#1}}

\newcommand{\BOOTSTRAPPHASE}[1]{{\em ({#1})}}
\newcommand{\WHATEVER}[1]{\CODE{#1}} % The definition above was a bad idea

\newcommand{\BETWEENTINYANDSMALL}{\fontsize{9}{11}\selectfont}

\usepackage{censor}

\def\TITLE{GNU epsilon}
\def\FRENCHTITLE{GNU epsilon --- un langage de programmation extensible}
\def\SUBTITLE{an extensible programming language}

\usepackage{fancyhdr}                    % Fancy Header and Footer
\let\headruleORIG\headrule
\renewcommand{\headrule}{\color{black} \headruleORIG}

\usepackage{colortbl}
\arrayrulecolor{black}

\fancypagestyle{plain}{
  \fancyhead{}
  \fancyfoot{}
  
}

\fancypagestyle{unobstrusive}{
  \fancyhead{}
  \fancyfoot[LE,RO]{\textcolor{black}{\bfseries\thepage}}
\fancyfoot[RE]{\textcolor{black}{\bfseries\nouppercase{\leftmark}}}
\fancyfoot[LO]{\textcolor{black}{\bfseries\nouppercase{\rightmark}}}
}

\usepackage[left=1.5in,right=1.3in,top=1.1in,bottom=1.1in,includefoot,includehead,headheight=13.6pt]{geometry}

\usepackage[nottoc, notlof, notlot]{tocbibind}
\usepackage{makeidx}

\makeindex
\usepackage{minitoc}
\let\minitocORIG\minitoc
\renewcommand{\minitoc}{\minitocORIG \vspace{1.5em}}
\usepackage{xcolor}
\pagecolor{white}

\usepackage{ifpdf}
\ifpdf
  \DeclareGraphicsExtensions{.ps,.eps,.jpg}
  \usepackage[colorlinks=true,linkcolor=darkblue,filecolor=darkblue,citecolor=darkblue,urlcolor=darkred,pagebackref,hyperindex=true]{hyperref}

\else
  \usepackage{graphicx}
  \DeclareGraphicsExtensions{.ps,.eps}
  \usepackage[colorlinks=true,linkcolor=blue,filecolor=blue,citecolor=blue,urlcolor=red,dvipdfm,pagebackref,hyperindex=true]{hyperref}
\fi
\graphicspath{{.}{images/}}

\hypersetup
{
pdftitle={\TITLE\ -\ \SUBTITLE\ - Luca Saiu's PhD thesis},
pdfauthor="Luca SAIU", 
pdfsubject={\TITLE\ -\ \SUBTITLE}, %subject of the document
pdfkeywords={programming} {language} {extensibility} {macro} {transformation} {reflection} {bootstrap} {interpretation} {compilation} {parallelism} {concurrency} {garbage collection},
pdfpagemode=None, %aucun mode de page
pdfpagelayout=SinglePage, %ouverture en simple page
pdffitwindow=true, %pages ouvertes entierement dans toute la fenetre
}

\usepackage[amsmath,standard,thmmarks,framed,thref]{ntheorem} % Shall I really use the amsmath option?
\renewtheorem{theorem}{Theorem}[chapter]
\renewtheorem{definition}[theorem]{Definition} 
\renewtheorem{lemma}[theorem]{Lemma}
\renewtheorem{corollary}[theorem]{Corollary}
\renewtheorem{proposition}[theorem]{Proposition}

\newtheorem{axiom}[theorem]{Axiom}

\newtheorem{implementationnote}[theorem]{Implementation Note}
\newtheorem{syntacticconvention}[theorem]{Syntactic Convention}

\usepackage{relsize} % Relative font size is handy
\fvset{numbers=left, fontsize=\relsize{-1}, frame=lines}

\usepackage[bbgreekl]{mathbbol} % This changes the appearance of \SET stuff,
\newcommand*\cleartoleftpage{%
  \clearpage
  \ifodd\value{page}\hbox{}\newpage\fi
}
\usepackage[multiple]{footmisc} % Place footnote markers before the punctuation following them

\begin{document}
\VerbatimFootnotes % This must be run *after* the preamble; it comes from fancyvrb.

% Stop Bibtex from bitching when there are no citations:
\nocite{growing-a-language}
\nocite{history-of-t}

%% %% %% %%%%%%%%%%%%%%%%%%%%%%%%%%%%%%%%%%%%%%%%%%%%%%%%%%%%%%%%%%%%%%%%%%%%
%% %% %% %%%%%%%%%%%%%%%%%%%%%%%%%%%%%%%%%%%%%%%%%%%%%%%% Title page
%% %% %% %%%%%%%%%%%%%%%%%%%%%%%%%%%%%%%%%%%%%%%%%%%%%%%%%%%%%%%%%%%%%%%%%%%%
\pagenumbering{roman}
\WHENCOMPLETE{% -*- mode: latex; fill-column: 79; mode: flyspell; buffer-file-coding-system: -*- utf-8 -*-
% No auto-fill mode on this file
\begin{titlepage}
\begin{center}
\noindent {\large {\em UNIVERSITÉ PARIS 13 --- INSTITUT GALILÉE}} \\
\vspace*{1cm}
\noindent {\em \large \underline{THESIS}} \\
%\vspace*{0.3cm}
\noindent {To obtain the title of} \\
%\vspace*{0.3cm}
\noindent {\Large DOCTOR OF SCIENCE OF UNIVERSITÉ PARIS 13} \\
%\noindent {\Large DOCTOR OF UNIVERSITÉ PARIS 13} \\
\noindent {\bf Specialty: Computer Science} \\
%%%%%%%%%%%%%%%%%%
\vspace*{3.5cm}
\noindent {\huge \bf \TITLE} \\
\vspace*{0.2cm}
\noindent {\Large \bf \SUBTITLE} \\
\vspace*{1cm}
%%%%%%%%%%%%%%%%%%
%% \vspace*{2.5cm}
%% \noindent {\Large \bf \TITLE} \\
%% \vspace*{0.2cm}
%% \noindent {\large \bf \SUBTITLE} \\
%% \vspace*{1cm}
\noindent {Presented by {\bf Luca \textsc{Saiu}} to be defended in public on November 19\textsuperscript{th} 2012} \\
\vspace*{3.5cm}
\end{center}
%\noindent \large \textbf{Jury :} \\
\begin{center}
%\noindent \large 
\begin{tabular}{lll}
{\large \bf Jury:}      &                                       & \\
\\
\textit{Reviewers:}
%\textit{Rapporteur\TODOQ{s}:}
                        & Emmanuel \textsc{Chailloux}           & Université Pierre et Marie Curie -- Paris 6\\
%                        & Michel \textsc{Mauny}                 & ENSTA ParisTech, INRIA\\
                        & Michel \textsc{Mauny}                 & ENSTA ParisTech\\
\textit{Advisors:}      & Christophe \textsc{Fouqueré}          & LIPN, Université Paris 13\\
                        & Jean-Vincent \textsc{Loddo}           & LIPN, Université Paris 13\\
%\textit{President:}     & \TODOQ{[Name \textsc{Surname}]}       & \TODOQ{[Affiliation]}\\
\textit{Examiners:}     & Roberto \textsc{Di Cosmo}             & PPS, Université Paris Diderot -- Paris 7\\
                        & Manuel \textsc{Serrano}               & INRIA Sophia-Antipolis\\
                        & Basile \textsc{Starynkevitch}         & CEA LIST/DILS \\
%                        & Peter \textsc{Van Roy}                & Université catholique de Louvain, Belgium
                        & Peter \textsc{Van Roy}                & Université catholique de Louvain, Belgium
                        
%\textit{Invited:}       & \TODOQ{Name \textsc{Surname}}         & \TODOQ{Affiliation}\\
%                        & \TODOQ{Name \textsc{Surname}}         & \TODOQ{Affiliation}\\
\end{tabular}
%% \TODO{Chailloux would like a proof of Proposition~\ref{dimension-monotonicity-proposition}
%% Lemma~\ref{lemma-1-jv} and Lemma~\ref{lemma-2-jv}
%% «en annexe».
%% \TODO{Shall I name failure rules?  It looks pretty useless to me: there is essentially one per form name, so it's impossible to mistake one for another.}
%% }
\\
\vspace*{1.5cm}
\begin{small}
\begin{tabular}{c}
\noindent{\em Titre français~: \FRENCHTITLE} \\
\hline
\noindent{\bf Laboratoire d’Informatique de Paris-Nord, UMR 7030 --- CNRS, Univ. Paris 13}\\
\end{tabular}
\end{small}
\end{center}
\end{titlepage}

%\sloppy
\titlepage
\dominitoc
\cleardoublepage
\setcounter{page}{3} % I want the title page to be numbered {i} 
}

%% %%%%%%%%%%%%%%%%%%%%%%%%%%%%%%%%%%%%%%%%%%%%%%%%%%%%%%%%%%%%%%%%%%%%
%% %%%%%%%%%%%%%%%%%%%%%%%%%%%%%%%%%%%%%%%%%%%%%%%% Abstract
%% %%%%%%%%%%%%%%%%%%%%%%%%%%%%%%%%%%%%%%%%%%%%%%%%%%%%%%%%%%%%%%%%%%%%
% -*- mode: latex; fill-column: 79; mode: auto-fill; mode: flyspell; buffer-file-coding-system: utf-8 -*-
%%%%%%%%%%%%%%%%%%%%%%%%%%%%%%%%%%%%%%%%%%%%%%%%%%%%
%\UNNUMBEREDCHAPTER{}

% Horrible kludge: I want the starred chapter to appear in the table of
% contents, but not on the actual page:
\color{white}
\UNNUMBEREDCHAPTER{Abstract}
\color{black}
\vskip-2.75cm

\UNNUMBEREDSECTION{Abstract}
%\UNNUMBEREDSECTION{Abstract}
%\subsection*{Abstract}
% -*- mode: latex; fill-column: 79; mode: auto-fill; mode: flyspell; buffer-file-coding-system: utf-8 -*-
%%%%%%%%%%%%%%%%%%%%%%%%%%%%%%%%%%%%%%%%%%%%%%%%%%%%

\index{abstract}
%\em % I like all the abstract text to be in italics. I can change that if needed

\textit{Reductionism} is a viable strategy for designing and implementing practical
programming languages, leading to solutions which are easier to extend,
experiment with and formally analyze.

We formally specify and implement an extensible programming language, based on
a \textit{minimalistic first-order imperative core language} plus strong
\textit{abstraction mechanisms}, \textit{reflection} and
\textit{self-modification} features.  The language can be extended to very high
levels: by using Lisp-style \textit{macros} and code-to-code
\textit{transforms} which automatically rewrite high-level expressions into
core forms, we define closures and first-class continuations on top of the core.

Non-self-modifying programs can be analyzed and formally reasoned upon, thanks
to the language simple semantics.  We formally develop a \textit{static
  analysis} and prove a \textit{soundness property} with respect to the
dynamic semantics.

We develop a \textit{parallel garbage collector} suitable to multi-core
machines to permit efficient execution of parallel programs.

%\addcontentsline{toc}{subsection}{\protect\numberline{}Keywords}

\vskip2cm

%% % Don't remove this, at the very end of the file
%% \rm % Revert to upright

%%%%%%%%%%%%%%%%%%%%%%%%%%%%%%%%%%%%%%%%%%%%%%%%%%%%
%\subsection*{Résumé}
\UNNUMBEREDSECTION{Résumé}
\index{french}
\index{résumé}
% -*- mode: latex; fill-column: 79; mode: auto-fill; mode: flyspell; buffer-file-coding-system: utf-8 -*-
%%%%%%%%%%%%%%%%%%%%%%%%%%%%%%%%%%%%%%%%%%%%%%%%%%%%

\index{résumé}

%\em % I like all the abstract text to be in italics. I can change that if needed
Le \TDEF{réductionnisme} est une technique réaliste de conception et
implantation de vrais langages de programmation, et conduit à des
solutions plus faciles à étendre, expérimenter et analyser.

Nous spécifions formellement et implantons un langage de programmation
extensible, basé sur un \textit{langage-noyau minimaliste impératif du
  premier ordre}, équipé de \textit{mécanismes d'abstraction} forts et
avec des possibilités de \textit{réflexion} et
\textit{auto-modification}.  Le langage peut être étendu à des niveaux
très hauts~: en utilisant des \textit{macros} à la Lisp et des
\textit{transformations de code à code} réécrivant les expressions
étendues en expressions-noyau, nous définissons les clôtures et les
continuations de première classe au dessus du noyau.

Les programmes qui ne s'auto-modifient pas peuvent être analysés
formellement, grâce à la simplicité de la sémantique.  Nous
développons formellement un exemple d'\textit{analyse statique} et
nous prouvons une \textit{propriété de soundness} par apport à la
sémantique dynamique.

Nous développons un \textit{ramasse-miettes parallèle} qui convient aux machines
multi-cœurs, pour permettre l'exécution efficace de programmes parallèles.

%\addcontentsline{toc}{subsection}{\protect\numberline{}Mots-clés}
\paragraph{Mots-clés~:}
% -*- mode: latex; fill-column: 79; mode: flyspell; buffer-file-coding-system: -*- utf-8 -*-
programmation, langage, extensibilité, macro, transformation, reflection, \textit{bootstrap}, interprétation, compilation, parallélisme, concurrence, ramasse-miettes

\rm % Revert to upright

\rm % Revert to non-italic, if needed

%% %% %% %%%%%%%%%%%%%%%%%%%%%%%%%%%%%%%%%%%%%%%%%%%%%%%%%%%%%%%%%%%%%%%%%%%%
%% %% %% %%%%%%%%%%%%%%%%%%%%%%%%%%%%%%%%%%%%%%%%%%%%%%%% Dedication
%% %% %% %%%%%%%%%%%%%%%%%%%%%%%%%%%%%%%%%%%%%%%%%%%%%%%%%%%%%%%%%%%%%%%%%%%%
%\INCLUDEWHENCOMPLETE{dedication}%
% -*- mode: latex; fill-column: 79; mode: auto-fill; mode: flyspell; buffer-file-coding-system: utf-8 -*-
%\chapter*{Dedication}

\color{white}
\UNNUMBEREDCHAPTER{Dedication}
\color{black}
%\vskip-2.75cm
%\UNNUMBEREDCHAPTER{Dedication}
%\cleardoublepage\UNNUMBEREDSECTION{Dedication}

\begin{quotation}
A large, crowded maze of a building that is just one part of one branch of the
local administration, in the Paris neighborhood.  Under the Summer heat I've
been standing there or somewhere very close since the early morning, awake
since before 6am just for the privilege of being near the front of the line.
It's finally my turn, after half a day spent waiting.  And now she tells me
that no, my \textit{avis d'imposition fiscal} is not a valid
\textit{justificatif de domicile}.  And who cares if they had told me the
opposite in her very office: she has no intention of listening to my
complaints.  I'll have to return another day, with a signed copy of my
landlord's identity card.

After I get back to the main hall near the entrance to arrange that next
appointment I must look as irritated as I am.  The woman at the desk asks me
what happened.  When I repeat to her what I've been told just a few minutes
before, she explodes.  \textit{---What!?  Come with me.}  Shouting that she'll
be right back to the people waiting behind me, she abandons her place and
angrily storms away to another office.  I follow her.

We sit. Between an half-muttered insult to her colleague and the next she asks
me for my papers one by one, and checks each of them.  She has to come back to
her own work: fueled by adrenaline she's thorough but efficient.
%% The certificate of conformity, yes, I've already given it to you, it's right
%% there behind the envelope.  Another quick look.  Oh, yes, that's fine.
Last, I hand her my \textit{avis d'imposition fiscal}.  She compares the
address, looks at the date, and skims the rest. \textit{---Yes,
  it's perfectly fine!  That---} labeling her colleague by a final one-word
definition.  She signs my dossier herself, overriding or simply ignoring the
other's authority.  I'll have to go pay the tax at the cash register, yes,
right there on the left and yes, then I'm done.  I barely have the time to
thank her before she runs back to her desk.  \vskip2.5cm
\INVISIBLE{[We need more rage in the world.  May her example inspire others.]}
\end{quotation}
%\\
%\\
To that blondish, forty-something woman who was working in a government
building near Paris during the Summer of 2010, whatever her name is, I dedicate
this work.
\\
\\
May her rage inspire others to do the right thing.
\begin{flushright}
Luca Saiu, December 2012
\end{flushright}

%% %% %% %%%%%%%%%%%%%%%%%%%%%%%%%%%%%%%%%%%%%%%%%%%%%%%%%%%%%%%%%%%%%%%%%%%%
%% %% %% %%%%%%%%%%%%%%%%%%%%%%%%%%%%%%%%%%%%%%%%%%%%%%%% Acknowledgments
%% %% %% %%%%%%%%%%%%%%%%%%%%%%%%%%%%%%%%%%%%%%%%%%%%%%%%%%%%%%%%%%%%%%%%%%%%
%\INCLUDEWHENCOMPLETE{acknowledgments}%
% -*- mode: latex; fill-column: 79; mode: auto-fill; mode: flyspell; buffer-file-coding-system: utf-8 -*-
%\chapter*{Acknowledgments}
%\section*{Acknowledgments}
%\index{acknowledgments}
%\UNNUMBEREDSECTION{Acknowledgments}
\UNNUMBEREDCHAPTER{Acknowledgments}
At the beginning of it all, one day of March, I left the place where I was
born, alone on my already old car loaded with everything I owned, headed for
Paris.  Not speaking the language at the time I arrived in France as an
awkward foreigner depending on
% the
others' kindness.  Luckily many people were
indeed kind to me.

And soon enough I fell in love with France: to my eyes it looked like
everything a good country should be.  Yet, every time I expressed this
enthusiasm (in English) to my first French friends, I was gently reminded that
my vision was partial, my experience too limited.  In other words I was a naïf
simpleton who hadn't understood anything yet --- they said it
% much
more softly,
\textit{ça va sans dire}, but the idea was clear.

My optimistic naïveté, the will of trusting strangers \textit{because it's the
  right thing} is more of a philosophical choice of mine --- which is to say,
isn't actually real.  But for the rest I have to admit that those people,
advising me to be more cautious in my judgements, were not wrong.

Years have since passed.  I met many people in France, all across the spectrum
from the constructively angry civil servant of the dedication
% (who symbolically deserves to be the first to be thanked once more)
to her
opponent
% adversary
\textit{the lady of many nicknames}.  Of everybody I met here
I'd say that more people were closer to the former.
% --- not too bad a compliment, when you think about it.
Everything considered, I'm grateful for the opportunities I got at
%LIPN,
Université Paris 13.
\\
\\
Thanks to Jean-Vincent Loddo, who trusted me and made all of this possible by
hiring me to work on
%Marionnet\footnote{See \url{http://www.marionnet.org}.}
Marionnet
during my first six
months in France, and then co-directed my doctorate.  Jean-Vincent has also
been a friend, very reliable and always patient, knowing when to insist
and when to let me be.  I've never been able to approach my other thesis
co-director, Christophe Fouqueré, quite as closely from a human point of
view.  This is entirely my fault, since Christophe has shown the same good
qualities and has been just as helpful and understanding as Jean-Vincent
whenever I've asked; he even spontaneously offered me concrete help when
I was dealing with practical problems, such as the first time my car got broken-in
and I had to go to the police, still not speaking French --- and then
there are the other episodes he remembers.  To both Christophe and
Jean-Vincent, thanks.  Thanks for real.
\\
\\
I particularly wish to thank some other people in the lab for their
friendliness and warmth: Laura Giambruno, Sébastien Guérif, Pierre Boudes
(who also offered to help in the same break-in case, and not only then), and Franck
Butelle.  Christian Codognet, Laure Petrucci and Adrian Tanasa also deserve a
high place in this list.  And Sophie Toulouse, and Hayat Cheballah.

I'll miss the many roommates I had in office A207 during these years:
Jean-Vincent, and then in particular Virgile Mogbil, Daniele Terreni and Giulio Manzonetto.
And Sylviane Schwer.

I had an especially good relationship with the other people working on the
Marionnet project: of course Jean-Vincent, and then especially Jonathan
Roudiere, Marco Stronati, Abdelouahed, and Franck Butelle.

Sébastien Guérif is a friend and possibly the best colleague I've ever had:
meticulous, very hardworking, interested in feedback.  I've learned a lot
about teaching while doing
class and lab exercises for his \textit{Programmation Impérative} course.  I
also enjoyed the company of the other colleagues doing exercises for
Sébastien's courses, particularly Daniel Kayser, Christophe Tollu, Antoine
Rozenknop, Hanane Allaoua, Manisha Pujari, and again Laura Giambruno.
%and Adrian Tanasa.

Thanks to Jean-Yves Moyen as well, for approving
% essentially
more or less
every idea I had
in organizing the \textit{Programmation Fonctionnelle Avancée} course for two
years,
implicitly
accepting to scale down his role to just lab exercises, despite his experience.
It was fun.

Some students of mine were actually interested in what I had to say.  I wish to
thank them for that sparkle in their eyes saying ``that's cool'' at the
same time as ``I get it''.  Thanks to my bad students as well, if they really
tried --- but no thanks at all to cheaters.

I fondly remember a couple idle afternoons spent with Christine Recanati
speaking about philosophy.
I had other agreeable exchanges,
geeky, intellectual or simply human, with
James Avery,
Roberto Wolfler Calvo,
Jean-Christophe Dubacq,
Jalila Sadki Fenzar,
Mario Valencia Pabon,
Marco Pedicini.
Thanks to Micaela Mayero, Patrick Baillot, Damiano Mazza, Érick Alphonse and
Yue Ma as well.
\\
\\
Thanks to Morena Olivieri, who's a good person but doesn't believe it yet.
Thanks for setting me in a good enough mood to finally decide to quit smoking.
I can see in retrospect how stressful quitting was; Morena and many others
encouraged and supported me during that time, particularly José Marchesi and
Matteo Golfarini.

I quit in 2010 but of course I still like many smokers, or people who were
smokers at the time.  Thanks to the friends and colleagues at the lab who
helped me socialize and learn the language when meeting outside for a cigarette
and some human companionship: first of all Sophie Toulouse and Hayat Cheballah,
but also
Jonathan Roudiere (again),
Jonathan Van Puymbrouck,
Pierre Boudes,
Amine Hemdane,
Vlady Ravelomanana,
Frédéric Toumazet,
Pascal Coupey,
Ferhan Pekergin,
Erwan Moreau,
Paolo Di Giamberardino,
Hichem Kenniche,
Haïfa Zargayouna,
and the funny guy with a white ponytail who introduced himself as ``a mathematician''.
\\
\\
Thanks to the old-time friends from the University of Pisa, for the good memories;
those who staid, those who escaped like I did, and those who still talk about
escaping someday.
Thanks in particular to Matteo Golfarini for being a good friend, and also for
having me so many times in his place near Madrid; and to my other good friends
Carlo Bertolli (who by the way also let me stay in his place in London) and Francesco
Nidito.

Some of the others I've not seen in years, but I'm sure they won't feel
offended for being included here: Riccardo Vagli, Dario Russo, Marco Righi,
Maria Cristina Favini Berti, Alessio Baldaccini, Antonio Mirarchi, Giandomenico
Napolitano, Di\-mi\-tri Dini, Robert Alfonsi, Massimo Cecchi, Federico Ruzzier,
Andrea Venturi, Erika Rossi, Marco Peccianti, Alessio Mazzanti, Sandra Zimei,
Marlis Valentini,
Eliana Anderlini.  And of course Massimiliano Brocchini, who first spoke about
me to Jean-Vincent during his time in France.

Then there are some new and found-again friends such as Marco Stronati, Enrico
Rubboli and Roberto Pasini --- Yet more people who ended up escaping, now
that I'm thinking of it.

And Gabriella, who still remembers me. Thanks.
\\
\\
Thanks to Richard Stallman for changing the world by starting the free software
movement, for setting an example with his integrity, and for GNU.  Several
fellow GNU hackers,
particularly José Marchesi and Alfred Szmidt, have also become close friends to
me; and I wish to remember the former ``rabbit'' from the UK, who might
or might not want to be mentioned by name here.

I remember with pleasure a few beautiful conversations stretching from the
evening nearly to the next morning, for example the debate on Algebra we had in
Spain with José, Alfred and Aleksander Morgado, going on long after the pub
closed and we were sent out at some crazy hour like 4am.  Or the long discussions
about GCC optimizations and LTO in a Bruxelles hotel hall with Dodji Seketeli,
Alfred, José and Laurent Guerby, over what was apparently the only bad beer to
be found in Belgium; I had to speak at FOSDEM the next morning --- and in the
end, after getting barely any sleep, the talk even came out nice.
%And a couple times when we shared some gloom, as well.

Of the other GNU people I have to remember at least Nacho Gonzáles,
Giuseppe Scrivano,
Ludovic Courtès, Andy Wingo,
Sylvain Beucler,
Bruno Haible (who noticed some surprising similarity between \EPSILONZERO and the M4 preprocessor),
Werner Koch,
Marcus Brinkmann,
Neal Walfield (who had graciously offered to proofread this thesis, even if in the
end I finished too late to let him),
Henrik Sandklef,
Simon Josefsson,
Juan-Pedro Bolívar Puente,
Juan Antonio Añel Cabanelas,
Reuben Thomas,
Ralf the Autotools guy,
Andreas Enge,
Daiki Ueno, Neil Jerram, Stepan Kasal,
Paolo Carlini.
Christian Grothoff, Nathan Evans and the Lilypond guy showed genuine interest
in epsilon and kept asking me about it in The Hague.
And of course I can't forget Fred and George, who deserve to be mentioned for
their crucial non-programming contribution to GNU recutils.
When I was hit by a small emergency René Mérou was kind enough
to lend me money after knowing me for less than one day, an act of trust
justified by our common ideas.
%Of course I repaid him, yet [no need to state this]
I still find this sense of commonality very moving.

Karl Berry was extremely patient in helping me with legal counseling about the
epsilon copyright status; for this I also have to thank Christophe Fouqueré
(again), Donald Robertson, the FSF lawyers, José Marchesi and Richard Stallman.

Thanks to the people at the FSF, FSFE, FFII, April and La Quadrature du Net.
% for fighting the good fight.  
Keep going, friends.
Even if we've not won all battles yet, we're writing history.
\\
\\
I'll be forever grateful to my masters, who changed my life.
I've recently rediscovered one of
% the
my
first ones, Nemo Galletti, from whose work I
learned procedural abstraction at the age of ten\footnote{The story at
  \url{https://blog.ageinghacker.net/posts/5/} has a nice second part that I've
  not found the time to write yet.}.  Some outstanding professors at the
University of Pisa, particularly Marco Vanneschi and Giorgio Levi, shaped
much of the good part of what I am now.
From Marco Bellia, my Master's advisor in Pisa,
% I've learned less than I could have, which I regret.
I regret having learned less than I could have.
I was inspired by the frighteningly talented OCaml people
Xavier Leroy, Damien Doligez, Jérôme Vouillon, Pierre Weis;
and most of all by
the great masters I've never met:
Abelson and Sussman,
Chuck Moore,
Paul Graham,
Richard Gabriel.
And I'll remember John McCarthy.
\\
A long time ago when I was young and inexperienced, after reading what some
people I admired were writing about Lisp, I decided to study it.  I started with
Steele's \textit{CLtL2} \cite{cltl2}, but without getting much at my first reading.  I
was to become a convinced Lisper only years later, after first understanding functional
programming from \textit{Functional Programming using Caml Light}
\cite{functional-programming-using-caml-light--mauny} by Michel Mauny; that's
the book from which I learned ``everything''.
I was very impressed following my first meeting with Mr.\@ Mauny in Paris, many
years later, casually discussing with him without knowing his identity
after somebody's seminar.  Jean-Vincent asked me if I knew who that man who'd just left was.  I
said no.  ``Michel Mauny''.  I turned my head, but he was already gone.
Now that this ordeal is over and I have no more conflicts of interest, I feel I
can finally let Mr.\@ Mauny know this.  Thanks.
\\
\\
Among the other jury people I've particularly enjoyed the voluminous, useful
and friendly feedback I got from Basile Starynkevitch --- an unofficial third
review.  Along with Manuel Serrano, Basile actually tried the software and
reported a couple minor bugs.  Basile, one of these days I'll actually e-mail
your own great master, as you suggested so many times.

And thanks to Emmanuel Chailloux, who even
accepted to meet me \textit{on August 17\textsuperscript{th}}
%--- at the UPMC University, with some students (!) ---
to have his paper copy of this document.
\\
\\
Juanma Díaz acquired and set up the server machine we share with Matteo and did
the initial installation of the virtual system images running on top of it;
later he also migrated the host a couple times.  As a command-line guy I
administer my virtual machine myself, but I have to recognize that it's also
thanks to Juanma if my personal host
%\url{ageinghacker.net}
\href{http://ageinghacker.net}{\nolinkurl{ageinghacker.net}}
has been more
reliable than many ``professional'' servers.
\\
\\
Thanks to my landlords Tina and Mohamed Serhane for the rare virtue of being
cheery and discreet at the same time.
To Nadège.
To Filippo Bellissima, still my favorite philosopher, for his good will.
To James Randi (who has something interesting to say about PhDs) and the JREF
people, for keeping up their mission against woo-woo.  To Randall Munroe,
Mohammed Jones and Scott Adams for their webcomics.  Scott Pakin taught me a
neat \TeX\ hack on \href{news:comp.text.tex}{\nolinkurl{comp.text.tex}}, even
if I didn't use it in this final version of the document.
\\
\\
Thanks to Richard Stallman for GNU (again) and particularly for GNU Emacs, to Donald
Knuth for \TeX, to Leslie Lamport for \LaTeX, to Lars Magne Ingebrigtsen for
Gnus (using such a great client for e-mail and Usenet has actually made me a
happier man), to Michael Stapelberg for \texttt{i3}; plus the thousands of
other people who contributed.

Thanks to Wikimedia Foundation, to the fellow contributors to the
English Wikipedia and to the contributors (I don't dare myself yet) to the
French Wiktionary.
\\
\\
Thanks to the people at INRIA Saclay where I'm a postdoc now:
Fabrice Le Fessant, who trusted me, and the colleagues
Çağdaş
%Çagdas
Bozman, Pierre Chambart and Michael Laporte.

I apologize once more to Christophe Cérin, Jean-Paul Smets and Camille Coti
for backing down from their postdoc offer at the last minute.  If I hadn't had
this unforeseen exciting opportunity at INRIA, I'd have been happy to work with
them instead.
\\
\\
There are a few more people I'd really want to name, but since I'm honestly not
sure if they would appreciate it I'll avoid mentioning them.  In any case they
will not mistake my respect of their privacy for an omission out of spite.
There are indeed many people I'd be supposed to thank here ``by custom'' and
yet are glaringly absent from this list
%\textit{'cause they don't deserve to be in it}
\textit{because they don't deserve to be in it}
 --- which sums up pretty well what I think of customs; but I wouldn't
imagine for one moment that people who are dear to me could doubt my sincere
affection for them.
\\
\\
\INVISIBLE{[The world is burning.  Run!]}
\\
I regard the design of \EPSILON as a quite personal issue, which I
think will be obvious from the document.  It's the expression of my philosophy of
what a programming language should be, for in the end the place in the
design space where one chooses to look ends up being more a matter of personal
preference than anything else.  I've developed my opinions in years of reading
and discussing, mostly influenced by the Scheme, Forth and Common Lisp
communities.  Interestingly, very few of the people working physically close to
me ever seemed to share my views.  When comparing opinions --- which has
been useful in any case --- I've resisted interferences and rejected many
suggestions; most strongly, the mistaken idea that static typing should always
come at the expense of everything else.  The final shape of the \EPSILONZERO
and \EPSILONONE languages as described here is a product of \textit{my} ideas,
something I'm willing to take responsibility for, pros and cons and all.

Jean-Vincent and Christophe, thanks for the freedom you gave me to explore and
to do what I believed.  All in all, it was fun.
\begin{flushright}
Luca Saiu, December 2012
\end{flushright}

%% %% %% %%%%%%%%%%%%%%%%%%%%%%%%%%%%%%%%%%%%%%%%%%%%%%%%%%%%%%%%%%%%%%%%%%%%
%% %% %% %%%%%%%%%%%%%%%%%%%%%%%%%%%%%%%%%%%%%%%%%%%%%%%% Contents and stuff
%% %% %% %%%%%%%%%%%%%%%%%%%%%%%%%%%%%%%%%%%%%%%%%%%%%%%%%%%%%%%%%%%%%%%%%%%%
%\WHENCOMPLETE{\tableofcontents}
\tableofcontents

%% %%%%%%%%%%%%%%%%%%%%%%%%%%%%%%%%%%%%%%%%%%%%%%%%%%%%%%%%%%%%%%%%%%%%
%% %%%%%%%%%%%%%%%%%%%%%%%%%%%%%%%%%%%%%%%%%%%%%%%% Chapters
%% %%%%%%%%%%%%%%%%%%%%%%%%%%%%%%%%%%%%%%%%%%%%%%%%%%%%%%%%%%%%%%%%%%%%
%\setcounter{page}{1} % Needed not to screw up minitocs, because of my previous
                     % page number forcing
\mainmatter
\pagenumbering{arabic}
% -*- mode: latex; fill-column: 79; mode: auto-fill; mode: flyspell; buffer-file-coding-system: utf-8 -*-
%\chapter{\ORTWO{Introduction}{Language design philosophy}}
\chapter{Introduction}
\label{chapter-introduction}

%Minimalistic design is a viable method of
%designing and implementing programming languages, leading to very expressive
%solutions amenable to formal analysis.

Reductionism is a viable
% style
% way
 strategy
 for designing and implementing practical
programming languages, leading to solutions which are easier to extend,
experiment with and formally analyze.
%\TODO{If I have good performance results at the end, add a claim here}

%% \RATIONALE{Shivers suggests to write the ``thesis'' here, in one sentence. His
%%   dissertation starts with the sentence ``Control-flow analysis is feasible and
%%   useful for higher-order languages.''.
%% % He recommends to read John Ellis'
%% %  dissertation, ``Bulldog: A Compiler for VLIW Architectures'' -- yes, but
%% %  it's not anywhere to be found on the web.}

\minitoc

%\minitoc
Programming languages have proliferated nearly since the beginning of Computer Science.
However, despite the sheer number of dialects with different syntaxes and details,
 there is
still comparatively little variety in programming models and
paradigms --- yet programming problems remain at least as hard as ever.

In order to really innovate in this field researchers need extensible languages which
are easy to modify and experiment with, but at the same time not limited to
simplified idealizations.
Bringing the same idea out of the lab and into practice, an expert
end-user should be able to bend and adapt the language to make it fit her
problem, rather than the opposite.
\\
\\
For this to be possible a language has to start out simple and open-ended:
able to express different paradigms, yet not hardwired for any;
easy to reason about in a rigorous way when needed, without
being unconditionally constrained.

%%%%%%%%%%%%%%%%%%%%%%%%%%%%%%%%%%%%%%%%%%%%%%%%%%%%%%%%%%%%%%%%%%%%%%%%%%%%%%%
\section{Programming language taxonomy}
\label{taxonomy}
Languages may be classified along at least three mostly orthogonal axes:
%the
 {\em
  paradigm}, {\em typing policy} and {\em concurrency model}. In the following
we limit ourselves to a quick overview; reviews articles such as %-- for example --
\cite{programming-paradigm-for-dummies} contain a much more 
detailed topology, with extensive examples.
\\
Furthermore, not all of the relevant concepts have satisfactory formal
definitions; but in this whirlwind tour we are going
to renounce most pretenses of being exact, accepting to speak of very
general concepts in terms somewhat vague.

%%%%%%%%%%%%%%%%%%%%%%%%%%%%%%%%%%%%%%%%%%%%%%%%%%%%
\subsection{Paradigm}
Many popular languages such as C are \TDEF{imperative}. Imperative
languages, based on destructive mutation of state and explicit control flow,
trace their origin to Turing machines, and in practice are easy to understand
in terms of the underlying machine language.

\TDEF{Functional} languages such as Haskell, ML and Lisp, shunning or at least
limiting the occurrences of assignment statements, are basically sugared versions of some
\LAMBDACALCULUS variant; their level of abstraction is much farther away
from the hardware than imperative languages. Most functional languages are
\TDEF{higher-order}, i.e.\@ they allow to pass functions as parameters to other
functions, and to return functions as results.

At an even higher level, \TDEF{relational} and \TDEF{constraint} programming
attempt to support a declarative, rather than algorithmic, style by dealing
with sequences of data in an extensional fashion and having the user exploit
data relations instead of building explicit data structures. Such languages
tend to be based on particularly clean and simple mathematical theories such as
relational algebra (the SQL query language) or some subset of the Predicate
Calculus (\ind{Prolog}\footnote{Despite not being meant as general-purpose
  languages, we argue that database query languages are actually much better
  examples of declarative non-algorithmic programming than logic languages:
  query languages allow to reason about objects and their relations, {\em
    completely abstracting away} from data structures and even more importantly
  search strategies, i.e.\ algorithms. By contrast programming in Prolog in
  practice requires to constantly keep in mind its operational semantics, for
  reasons of efficiency and even correctness: for example just reordering two
  Horn clauses, which from a logic point of view simply yields an uninteresting
  equivalent variation, can easily change complexity from linear to
  exponential or dramatically alter termination properties and the number of
  results.}).

\TDEF{Object-oriented} languages such as Smalltalk are more pragmatic:
they encourage modelling data structures upon real-world entities by making the
``behavior'' or a computational object a function of the object identity, and
making it easy to
define related classes of objects by only specifying their differences.

Other
%Some minority
 families including \TDEF{concatenative languages} such as \ind{Forth} and
 \ind{Postscript},
and
\TDEF{array languages} like \ind{APL} can be more or less directly traced back
to one of the four main groups.

%%%%%%%%%%%%%%%%%%%%%%%%%%%%%%%%%%%%%%%%%%%%%%%%%%%%
\subsection{\ind{Typing policy}}
Another orthogonal attribute of programming languages is their support for
\ind{typing}: programs written in \TDEF{statically-typed} languages are
mechanically analyzed prior to execution in order to check that some \ind{soundness
property} is satisfied, thus preventing certain errors from ever occurring at
runtime: the compiler will simply reject any ``suspicious'' program ---
invariably including some false positives. \ind{ML} and \ind{Haskell} are examples of statically
typed languages with \TDEF{strong} type systems; many popular languages such as \ind{C}, \ind{C++}
and \ind{Java} are also statically-typed, but their very complex semantics do not allow the
extensive static checks which are relatively easy to perform in
\ind{functional} languages\footnote{It could be argued that the idea of passing
  parameters to and receiving results from a function lends itself to reasoning
  about compatibility; sequential side effects, on the other hand, always ``compose
  well'' with one another in a superficial sense, but may lead to more subtle
  violations of implied invariants.},
and many more runtime errors remain possible: we speak of \TDEF{weak} static
type checking.

By contrast \TDEF{dynamically-typed} languages such as \ind{Lisp}, \ind{Perl} and \ind{Python}
perform checks {\em at runtime} before executing each operation subject to
failure, typically at some cost in performance but gaining expressivity in
comparison to the static-typing case; of course dynamic typing by itself cannot
statically guarantee any soundness property.

Some low-level languages such as \ind{Forth} and most assemblies are \TDEF{untyped}: each
datum is interpreted as-is without any check or conversion, assuming that it is always
a valid operand for its operator; such languages trading safety for
efficiency may not be suitable for all applications, yet they also definitely
have a place in programming.
\\
\\
%An interesting
A subset of statically and strongly-typed languages including ML and
Haskell employs \TDEF{type inference} to automatically reconstruct type
declarations from programs, rather than have the programmer provide them; this is convenient,
but since
% the undecidability  
 type inference is undecidable for the most powerful type systems
 \cite{basic-polymorphic-type-checking}, relying on it alone reduces the expressivity of the language.
Type inference is harder to employ in non-functional languages, and
possibly because of cultural biases it is not widely
used with
% combined with
% applied to
 weak type systems.

Even if mainstream languages seem resistant to adopt any such technique, the
idea of static checking can be extended from typing to other properties
computable via (necessarily partial for nontrivial languages) \ind{static analyses} such as \ind{termination},
time and space \ind{complexity}, or \ind{escaping}.

%%%%%%%%%%%%%%%%%%%%%%%%%%%%%%%%%%%%%%%%%%%%%%%%%%%%
\subsection{Concurrency model}
Another aspect which we must at least cite constitutes another whole
axis in the language space topology: the \TDEF{model of
concurrency} (synchronous versus asynchronous, message-passing versus shared
state) also has a deep impact on the language semantics, and not only on the
implementation of the runtime system.

The concurrency model of all the mainstream languages named above is
\TDEF{asynchronous shared-state}: concurrent ``threads'' read and mutate the same
global state, explicitly synchronizing accesses when needed.
The other important concurrency model is \TDEF{message-passing}: threads or
processes don't share state, but cooperate by exchanging messages.
Erlang supports message-passing only; in the other languages mentioned above message-passing
can be implemented on top of shared state, or is available as a thin ``wrapper'' over
the inter-process communication primitives provided by the operating system.
Languages with a \TDEF{synchronous} concurrency level (at the software level)
are a research topic \cite{esterel,sl,lucid-synchrone} but have yet to see major application.
\\
\\
Synchronous models are better suited to formal reasoning:
in his formal calculi for concurrency
CCS \cite{ccs--milner} and
\PICALCULUS\ \cite{pi-calculus--original--first-part--milner,pi-calculus--original--second-part--milner}
Milner considered synchronous communication as primitive, and represented asynchronous
processes by adding (synchronous) ``queue processes''.

On the other hand modern parallel hardware is strongly asynchronous, and
truly parallel synchronous implementations on top of it tend to be %prohibitively inefficient.
prohibitively inefficient.

Older-generation languages tend to support no notion of concurrency at all.

%%%%%%%%%%%%%%%%%%%%%%%%%%%%%%%%%%%%%%%%%%%%%%%%%%%%%%%%%%%%%%%%%%%%%%%%%%%%%%%
%\section{Searching for the chimera}
%%%%%%%%%%%%%%%%%%%%%%%%%%%%%%%%%%%%%%%%%%%%%%%%%%%%%%%%%%%%%%%%%%%%%%%%%%%%%%%
\section{Hybridization and complexity}
\label{naive-question}\textit{Why are there so many programming languages? Couldn't they just agree on one?}
%Naïve as it may be, this beginner question does not deserve to be just
%dismissed with a patronizing smile.
%Naïve as it may be, this beginner question deserves an answer rather than a
%patronizing smile.
We have all heard this naïve question.

The fact that such a question can only come from a beginner is evident from our
experience of how coding in even surprisingly similar languages ``feels''
different: for example about\footnote{The different \iem{funarg} problem
  \cite{funarg} workarounds bear a much more limited impact in practice.} %Here we can ignore the difference in approach in working around the \iem{funarg} problem \cite{funarg}.}
 the only real semantic difference between \ind{Pascal} and \ind{C} is
the different \textit{strength} of their type systems --- static in both cases; yet the
subjective ``experiences'' of writing in the two languages are {\em far} apart,
%\footnote{\NO{It is maybe
%unfortunate, but in the end understandable, how vocal user communities form
%around each single language, focusing on apparently shallow issues as
%concrete syntax. Thinking about the daily programming practice, such a regard of surface
%features is not always irrational.}}
as any programmer having used both can witness.
That said, we also have to recognize that many differences between languages
are in fact incidental, due to backward compatibility concerns or cultural inertia.
%;
%many surface syntax issues are not particularly deep.

Another, deeper, answer to the beginner question is that different problems
call for different languages. But then, why not merging the greatest
possible number of features from different styles into one ``perfect'' language?
In fact there exists such a trend:
languages inspire and influence one another, and
some recent ones such as Oz \cite{oz} make a point of offering support for
{\em as many different paradigms as possible}; but even without looking at such
extreme examples, a move towards \iem{hybridization} is evident in most recent
languages: contemporary languages such as \ind{C++}, \ind{Java} and \ind{C\#}, but also the
popular ``\ind{scripting language}s'' \ind{Python}, \ind{Perl}, \ind{JavaScript}, incorporate at least
{\em two} paradigms --- imperative and object-oriented, with some elements
of functional programming being more slowly accepted into the mainstream; in
fact it could be argued that object-oriented programming is itself firmly
rooted in the imperative paradigm, with only some restricted patterns taken from functional
languages\footnote{The idea of \iem{late binding} at the heart of object-orientation
would be easy to emulate with data structures containing functions;
in fact virtual method tables are typically implemented as chained
arrays of pointers to functions, where functions have access to a ``struct''
holding field values: in other words, chained \iem{closure} arrays: if
method lookup fails in one class, a link is followed and the next one is tried,
up the inheritance chain --- or occasionally \textit{sideways}, in case
of multiple inheritance.
As a different kind of hybridization, other non-imperative languages now include
object-oriented features: the ML dialect OCaml
\cite{RemyVouillon98,Chailloux:2000}
managed to also add objects in a mostly-functional language and still keep its
type system strong and static, at the cost of some complexity.
More pragmatically, most modern SQL systems include some kind of
object-oriented extension, more or less well-integrated in the relational
paradigm.
%% \JV{Ho riformulato la parte precedente, pur senza spiegare tutti i dettagli,
%%   che qui sono inutili. Mi sembra che l'idea di base sia comprensibile: non
%%   sarà una spiegazione, ma è almeno un'evocazione intuitiva dell'idea. N'est pas~?}
%\TODO{Shall I speak about multi-methods? Not here... It would be nice to introduce a powerful
%object system as a sample personality later, but I don't see why I should
%complicate things now.}
%\TODO{Despite what object-orientation advocates state, the converse
%  emulation requires too many extensive rewritings to feel practical.}
}.
Most object-oriented languages also have hybrid type systems, with some static
{\em and} some dynamic checks. Often {\em both} message-passing and shared-state are
available as concurrency models.

%%%%%%%%%%%%%%%%%%%%%%%%%%%%%%%%%%%%%%%%%%%%%%%%%%%%
\subsection{Hybridization limits}
%\subsection{\MAYBEQ{The} Hybrid vigor is not enough}
There are clear limits to hybridization: some features regarded as desirable in
different communities are mutually incompatible, if not opposite. For
example both having a static typing discipline and \textit{not} having one may be
reasonably argued to be useful features: one solution permits to prove run-time
properties of a program before running it, the other improves expressivity.
In the same spirit, the useful properties of purely
functional languages \cite{functional-programming,haskell98} would be
destroyed by adding an assignment operator.
\\
\\
But even if we forget for a moment that many possible sets of features are just
incompatible, designing a strongly hybrid language entails giving up on finding
%single
simple
 answers to programming problems, and just hoping that programmers will be
better prepared for the unknown with a bigger \ind{toolbox}: the bigger,
the better. This pragmatic approach hits
its limit when a language becomes too big to be \iem{intellectually manageable}
--- at which stage the language may
or may not
%\textit{or may not}
be adequate for most tasks.

As a further objection against hybridization, working with such large \ind{chimera}s
makes harder to experiment with language features by building prototypes
--- in fact it may not be by chance that most such experimentation has
historically taken place in Lisp dialects, which we will see to be the closest to
our model.
%% \STRONG{Indeed we feel that the sheer dearth of language innovation in
%% recent decades could mainly be explained by this particular difficulty.}

%%%%%%%%%%%%%%%%%%%%%%%%%%%%%%%%%%%%%%%%%%%%%%%%%%%%%%%%%%%%%%%%%%%%%%%%%%%%%%%
\section{\em Growing a language}

Guy L. Steele dealt with the issue of the ``size'' of programming languages in his
famous OOPSLA 1998 keynote talk ``\ind{Growing a language}'' \cite{growing-a-language}. In a
wonderfully
%masterfully
deconstructionist exploit, Steele constrained his own English
to follow the same rigid rules of formal languages in which \textit{every
  non-primitive ``word'' needs to be explicitly defined before use}.
By taking as primitives only English monosyllables he tried to communicate the
feel of using a very small (programming) language.
\\
\\
The main point of the speech was the idea of working with a language powerful
enough to be \textit{\ind{evolve}d} by the user community under the
coordination of a maintainer; and possibly even more important, \textit{the user herself}
would bend the language to her needs, as part of the daily
practice of programming.

Even after several polysyllable definitions Steele's original prose retains its
peculiar charm:
\label{growing-a-language-quotation}
\label{growing-a-language--quotation}
\QUOTATION
{[...]
a language design of the old school is a pattern for programs. But now we need
to `go meta.' We should now think of a language design as a pattern for language designs,
a tool for making more tools of the same kind.
[...]
My point is that a good programmer in these times does not just write programs. A
good programmer builds a working vocabulary. In other words, a good programmer
does language design, though not from scratch, but by building on the frame of a base
language.
}
{Guy L. Steele Jr.}
{\cite{growing-a-language}}
As the initial iteration of this process Steele
proposed Java with some minor changes --- thus at least a middle-sized language; in his
opinion intentionally small languages such as the Lisp dialect Scheme, originally his own
brainchild \cite{scheme,rrs},
would remain hopelessly inadequate for
modern tasks, as he tried to
suggest with the one-syllable metaphor.
\\
\\
More than a decade has since passed, and the envisaged extension of Java by the
community has not materialized\footnote{Ironically, something close to Steele's vision of
  mostly-decentralized extensions has materialized \textit{in Scheme} with SRFIs \cite{scheme-SRFI};
  development is active, at least in the comparatively small scale of
  the Scheme community.}.
\\
\\
The idea of ``growing a language'' remains a valid strategy, if not
even the only realistic one. Without
%disparaging such an important engineering insight,
%belittling
overlooking
 this important engineering insight, we find it worth to
spend some words on what we do \textit{not} agree with in Steele's
presentation.
Anyway, since much of the controversy will center around Java, to Steele's credit
we must at least cite his later contributions based on
\ind{Fortress} \cite{a-growable-language--abstract}, sharing the same idea of a
``growth plan'' but with a more suitable core language.  That is the point: what
makes Fortress a better match than Java for the task?  And what yet superior
alternative can we extrapolate from this trend?

%%%%%%%%%%%%%%%%%%%%%%%%%%%%%%%%%%%%%%%%%%%%%%%%%%%%
\subsection{Procedural and syntactic abstraction}
\label{procedural-and-syntactic-abstraction}
In our opinion it is not by chance that the crucial insight for finding the
missing ingredient was provided in \cite{wizard}, co-authored by
Gerald J. Sussman
---
the other father of Scheme.

Our sketch of a language topology is reasonable but fails to
really capture the actual expressive power of a language, as until this moment
we have ignored a fundamental and orthogonal class of language features:
\TDEF{means of abstraction}, called ``patterns'' in Steele's quotation
above.
\\
Starting from the very first chapter, \cite{wizard} speaks about means of
abstraction as ways of \TDEF{naming} patterns of code, possibly with
parameters, so that they may be re-used at will as if they were primitive.  In
other words a mean of abstraction allows to \TDEF{factor} away some code, so
that we can reason at a higher level and ignore irrelevant details unless
needed.  The idea is by necessity vague as it potentially extends to any point
of our language topology and beyond, yet we do not feel much danger of
ambiguity: any programmer will promptly recognize abstraction features. We are
speaking about \TDEF{procedural abstractions} (including all the obvious
generalizations to functions, predicates, modules and classes);
% meant as groups of methods);
and, as a separate group, \TDEF{syntactic abstractions} such as macros.
No other example of a syntactic abstraction feature comes to mind and in fact no other
case is in common use\footnote{Some historical Lisp dialects used \textiti{fexprs}
  \cite{against-fexprs--pitman}
  as an alternative to macros; so does a new dialect called Kernel
  \cite{fexprs--shutt}, resurrecting them thirty years later.
  Fexprs are also discussed at some length as an implementation device in \cite[pp.~25-26]{rrs}.
}
 --- but we are going to introduce another
kind in
\SECTION\ref{transform-informal-idea}, and more fully in
\SECTION\ref{transform-chapter}. 
Modern high-level
% \RATIONALE{don't bother me with C and Forth now; they are low-level, so it's reasonable that they lack closures}
languages support more or less adequate procedural
  abstractions, including higher order (C++ supports ``lambdas'' as per its latest Standard \cite{c++0x}
  and even Java should follow suit),
%% \IGNORE{\footnote{\OBSOLETE{Java advocates sometimes argue that closures can be
%%   simulated with inner classes, which is technically correct. Yet the idiom
%%   produces notoriously hard to follow code, and is not widely used in
%%   practice. The difficulty could be overcome by automating the simulation with
%%   a syntactic abstraction, if Java were not \textit{also} lacking in that area,
%%   as we will show in this section. \RATIONALE{I'd prefer to avoid an example. It's
%%     long and hard to read, and my public agrees anyway that higher-order is
%%     inconvenient to simulate in Java.}}}}),
 but most are still very weak in syntactic
abstraction; we are now proceeding to explain why this is important.

%%%%%%%%%%%%%%%%%%%%%%%%%%%%%%%%%%%%%%%%%%%%%%%%%%%%
\subsection{Syntactic abstraction and core-based languages: macros}
In order to clarify
% both the strengths and the limitations \TODOQ{[check that we say sth about limitations]} of syntactic
what we mean by syntactic
abstraction, we are now going to informally present a classic example.

\label{macro-examples--introduction}
%\ADVISORS{Le prochain exemple est-il trop banal pour mon public~?}
Let us assume an imperative language similar to
Pascal or C, with a \texttt{while}{\allowbreak}..{\allowbreak}\texttt{do}{\allowbreak}..{\allowbreak}\texttt{done} loop but no
\texttt{repeat}..\texttt{until}; we want to extend the language so that we
can write:
\begin{Verbatim}
procedure print_at_least_once (n):
  variable x = 1;
  repeat
    print("x is ");
    print(x);
    newline();
    x := x + 1;
  until x > n;
end.
\end{Verbatim}
With a suitable macro system we could define \texttt{repeat}..\texttt{until} as
syntactic sugar, so that for any sequence of statements
$s$ and any expression $e$, the loop ``\CODE{repeat $s$ until $e$}'' is rewritten
into ``\CODE{$s$;\ while not ($e$) do $s$;\ done}''.  Handwaving away some trivial
details which are not relevant here, we could say that the macro definition is:
\begin{Verbatim}
macro repeat <s> until <e>:
  <s>;
  while not (<e>) do
    <s>;
  done;
end.
\end{Verbatim}
Using the \texttt{repeat}..\texttt{until} macro, the macroexpander stage of the
compiler would automatically rewrite the above definition of
\texttt{print\_at\_least\_once} into:
\pagebreak
\begin{Verbatim}
procedure print_at_least_once (n):
  variable x = 1;
  print("x is ");
  print(x);
  newline();
  x := x + 1;
  while not (x > n) do
    print("x is ");
    print(x);
    newline();
    x := x + 1;
  done;
end.
\end{Verbatim}
Unsurprisingly enough, the rewritten code contains a repetition: the macro call
``factors away'' an undesired regularity which would make the code harder to
maintain if written directly in the extended form; even in such a trivial example as
this the code using the macro looks easier to read, and its purpose more
explicit.  Also notice how macroexpansion had only \textit{local} effect: the
macro call has been replaced, but not its surrounding code.

It is worth stressing how our \texttt{repeat}..\texttt{until} loop functionality
can \textit{not} be defined as a procedure, unless the language supports higher
order or some more exotic language feature such as passing statements as
 parameters; and even
where those features were available, notice how macroexpansion might still
produce a more efficient result and possibly be safer, as it takes place
entirely \textit{before} runtime.
\\
Our sample macro simply glues together pieces of code without
performing any substantial computation. This is not always the case:
several languages, mostly Lisp dialects, have \textit{Turing-complete}
  macro systems \cite{r6rs,common-lisp,emacs-lisp}. Other languages such as C
are limited to weak token-based preprocessors \cite{c99} whose power does not
% go
reach much further than defining \texttt{repeat}..\texttt{until} (with uglier
syntax). Other languages, Java included, do not support syntactic abstraction at
all.  We argue that precisely this weakness of Java
has prevented Steele's plan from materializing.
\\
Indeed the \TDEF{bottom-up}
programming style in which the language is extended to suit the problem, as
implicit in Steele's quotation,
%at \SECTION\ref{growing-a-language-quotation},
is very typical in the Lisp world --- and quite alien to most other communities.
\\
\\
Macros are so helpful in building syntactic sugar for existing languages
that some Lisp dialects such as Scheme
 are in
fact \TDEF{core-based}
\cite[\SECTION1.9, \SECTION B]{r6rs}: % R6RS speaks about ``derived forms''
%\cite[\SECTION4 and \SECTION7.3]{r5rs}: % R5RS speaks about ``derived expression types''
implementations may choose to develop at a low level some core forms
and define the rest of the language as a set of macros,
eventually rewriting programs into combinations of core forms only.
\\
\\
As an obvious such example, a block in a higher-order functional language can
always be rewritten into the immediate application of an anonymous function
with the bound variables as its formals and the block body as its body, and
with the bound expressions as actual parameters:
``\CODE{let a = 1 + 2, b = 3 + 4 in a + b}''
 has the same operational\footnote{Despite the vagueness entailed by not having specified any
  particular semantics we have to at least recognize the fact that what is
  equivalent \textit{operationally} might not be under a corresponding \textit{static}
    semantics; in particular we might have good reasons for using different
    type rules in the case of the expanded form. Such an observation is not new at all:
    Milner had already recognized in
    \cite[\SECTION3.5]{a-theory-of-type-polymorphism-in-programming--milner}
    what is now called
    \textiti{let-polymorphism} \cite[\SECTION22.7]{type-theory--pierce}.
    Here we intentionally disregard any static semantics, delaying the
    justification to %issue until
    \SECTION\ref{chapter-static-semantics}.
    \IGNORE{I'm not very fond of this phrase ``operational behavior'', but here's
    something I've learned from Dale Miller: \textit{saying ``equivalent''
      without an adverb is a symptom of sloppy thinking}. I think these were
    exactly his words, and I completely agree.  Now, something like
    ``operationally equivalent'' is what I really mean, but the context is so
    vague that I can't really refer to the operational semantics of
    anything. And what about that ``operational'' adjective? Saying just
    ``behavior'' without an adjective sounds just as sloppy.  \\ What I mean,
    of course, is that my rewritten \CODE{let} works as intended also with a
    call-by-value, left-to-right strategy (and right-to-left too, by the way).}
    This same remark also applies to the following examples in this subsection.}
behavior as
% is indistinguishable from
``\CODE{($\lambda$ a b .\ a + b) (1 + 2) (3 + 4)}''; since the language needs
$\lambda$ and function application anyway, we can define \CODE{let} as a
macro rather than having it as a primitive form, obtaining a simpler core language.
\\
As the number of bound variables is arbitrary, the sample macro language we
showed above can not express the rewrite quite adequately; yet the task is
considered very ordinary in Lisp dialects, which to ease metaprogramming \textit{use the same syntax for
programs and data structures}.  Deferring the explanation of details
to \SECTION\ref{macro-chapter}, we show here a possible definition of
\CODE{let} just to highlight its simplicity\footnote{We have omitted
  the check verifying that bindings are two elements long. Apart from
  that, the definition is perfectly realistic.}:
\begin{Verbatim}
(define-macro (let bindings . body)
  `((lambda ,(map car bindings)
      ,@body)
    ,@(map cadr bindings)))
\end{Verbatim}
This version of \CODE{let} binds variables ``in parallel'': defined
expressions have no visibility of bound variables.
The alternative ``sequential-binding'' block known as \CODE{let*} in Lisp is also easy
to obtain by macroexpansion, in this case by rewriting it into nested trivial
parallel-binding blocks in which the \textit{outermost} \CODE{let} binds the
\textit{first}
variable bound by \CODE{let*}.  Recursive macros can still be very simple:
\begin{Verbatim}
(define-macro (let* bindings . body)
  (if (null? bindings)
      `(let () ,@body)
      `(let ((,(caar bindings) ,(cadar bindings)))
         (let* ,(cdr bindings)
           ,@body))))
\end{Verbatim}
\CODE{let} follows the intended evaluation order under a
\textit{call-by-value} strategy: fully reducing all the operands before the
application forces to evaluate the body only after all bound expressions. Much
in the same way, \CODE{let*} constrains the evaluation order so that bound
expressions are reduced top-to-bottom.
\\
\\
There are many other similar examples: the short-circuiting left-to-right
versions of the \CODE{and} and \CODE{or} operators are easy to express with
conditionals: ``\CODE{a and b}'' has the same behavior as ``\CODE{if a then b else
  false}'', while ``\CODE{a or b}'' can be rewritten into ``\CODE{if a then true else
  b}''.
\\
\\
We can go even further: why having a conditional at all? By Church-encoding
booleans so that
``\CODE{true}'' is
``\CODE{$\lambda x y . x$}''
and
``\CODE{false}'' is ``\CODE{$\lambda x y . y$}'',
and taking ``\UNIT'' as the element of a unit type (or indeed as any object),
we can define
``\CODE{if a then b else c}'' as syntactic sugar for
``\CODE{((a ($\lambda z$.b) ($\lambda z$.c)) \UNIT)}'', for some variable $z$ not
occurring free in \CODE{b} and \CODE{c}. Again, $\lambda$ and the 
function application with \UNIT\ allow the conditional to reduce as expected also under
call-by-value.

%%%%%%%%%%%%%%%%%%%%%%%%%%%%%%%%%%%%%%%%%%%%%%%%%%%%%%%%%%%%%%%%%%%%%%%%%%%%%%%
\subsection{Transforms as syntactic abstraction}
\label{transform-informal-idea}
\label{transformation-informal-idea}
\label{transformations-informal-idea}
The \TDEF{continuation} of a subexpression of a program is a function or
procedure with side effects which, given the result of the subexpression,
returns the final result of the program \cite{scheme,compiling-with-continuations}.

Programs can be \textit{automatically} rewritten into \TDEF{Continuation-Passing Style}, a normal
form making all continuations explicit as $\lambda$-terms.  For example, let the continuation of 
$a + 2$ be $\kappa$; then by using one of the transformations in \cite{transforms--leroy} we obtain
$(\lambda x . (\lambda y . \kappa (x + y)) 2) a$ as its CPS version.  It is not too difficult to see how both versions
 yield the same result:
\begin{itemize}
\item First we evaluate $a$, passing it to its continuation $\lambda x . ((\lambda y . \kappa (x + y)) 2)$;
\item the continuation of $a$ binds it to $x$, and passes $2$ to its continuation \linebreak$\lambda y . (\kappa (x + y))$;
\item the continuation of $2$ binds it to $y$, evaluates the sum of $x$ and $y$, and passes the result on to $\kappa$;
\item $\kappa$ provides the sum of $x$ and $y$ to the rest of the computation, which will finally yield the result.
\end{itemize}
One desirable feature of the CPS form is its independence from the evaluation
strategy: in particular, the reduction sequence above holds for both
call-by-value and call-by-name.  Because of this and other useful properties,
CPS may be convenient to use in compilers as an intermediate form to perform some
semantic-preserving optimizations
\cite{rabbit--steele,compiling-with-continuations,orbit--kranz-phd-thesis}.  However, since our main
interest here is program expressivity, we are particularly interested in
\TDEF{first-class continuations}: a syntactic form such as \CODE{call/cc}
permits to access continuations \textit{as data} in an untransformed program, performing ``jumps''
into or out from expressions \cite{r6rs}.  The \CODE{call/cc} form can also be rewritten into
an ordinary $\lambda$-term along with the rest, by the same transformation which turns the program
into CPS.
\\
\\
First-class continuations are famously counter-intuitive and difficult to
employ directly, but they can simulate powerful control features such as
\textit{exceptions}, \textit{coroutines}, \textit{generators}, and even
\textit{backtracking}; by using macros we can \textit{syntactically abstract}
away the implementation of these control features, and provide a simple
high-level syntax.

A prerequisite for doing this is, of course, hiding the transformed program
from the user.  As already evident from the trivial example above, CPS-transformed
programs are long and tedious to read: the user should simply use
\CODE{call/cc} (or, better, syntactic forms reducing to \CODE{call/cc} uses) in
\textit{direct-style}, untransformed programs.  We already stated that the transformation can
be automatic, and that it supports \CODE{call/cc} as well --- hence such a kind
of syntactic abstraction is clearly possible.  But
\textit{can we define a CPS transformation using macros?}
\\
The answer is \textbf{no}: the result of CPS-transforming an expression depends
on its \textit{context}, while macros only have access to their parameters.  If
we are to support global program rewritings such as CPS, we need to introduce a
second syntactic abstraction feature: we call it the
\TDEF{transform}\footnote{Some controversy remains among English purists about
  the use of the term ``transform'' versus ``transformation''; according to
  some, a transform is \textit{the result of} a transformation.  We admit that
  even in Computer Science the use of the term ``transform'' for \textit{the
    function operating the transformation} is not universal, but we still
  prefer the shorter form.}.  Transforms map syntactic objects into other
syntactic objects, and can be employed either before defining a global object,
or retroactively on an existing program.

Transforms can be quite useful \cite{transforms--leroy}: just as it is the case for CPS, the
\textit{Closure Conversion} process rewriting $\lambda$-terms into explicit \textit{closures}
\cite{defunctionalization,modern-compiler-implementation-in-ml} may require
contextual information: unless we want to close over globals\footnote{ML
  dialects do in fact close over ``global'' variables as well, but we find that
  this choice complicates interactive programming, sometimes yielding ``unexpected'' results in case of variable redefinitions.  As a matter of personal taste we tend to prefer the solution of Common Lisp and Scheme which close over nonlocals only, not including globals.  In practice we suppose that the ML behavior is dictated by the need not to invalidate the results of previous type inferences whenever a global is redefined.}, when building a closure we need to know the set
of variables which are bound at that program point.% (\SECTION\ref{closure-conversion-transform}).

To reiterate, by composing two transforms it is possible \textit{to build a
  language supporting first-class continuations based on a core not even
containing anonymous $\lambda$-terms}.
\\
\\
Program transformations are commonly seen in formal mathematical presentations,
but to our knowledge they have not been available as an abstraction tool in any
general-purpose programming language up to this point.

%%%%%%%%%%%%%%%%%%%%%%%%%%%%%%%%%%%%%%%%%%%%%%%%
%% \subsubsection{\TODO{Remove most of this}}
%% \TODO{Problem to mention: macro definitions themselves are often hard to debug;
%%   macros may do error checking of their own}
%% \TODO{performance is a reasonable concern: there are compromises to be made here}
%% \TODO{yes, it may be hard to use. who cares. ``too powerful'' is not a defect.}
%% \TODO{introduce my new definition: ``reductionistic'', as an extension of
%%   ``kernel-based''. I like it better than
%% ``minimalistic'', is more explicit and not overloaded by other uses}.
%% % The user is free to build upon this foundation.
%% \MOVE{So far only Lisp dialects seem to have had any degree of success in this,
%% and in fact they are the languages closest to our proposed structure; again, we
%% believe that this is not by chance.}

%%  Any given language is based on a fundamental syntactic category,
%% enjoying some semantic properties: for example imperative languages are based
%% on statements, functional languages on expressions, logic languages on
%% predicate calls, and so on; 

%%%%%%%%%%%%%%%%%%%%%%%%%%%%%%%%%%%%%%%%%%%%%%%%%%%%
\subsection{Why reductionism}
It should be clear by now that taking the core-based approach to its
logical extreme yields a very simple core language.  Anyway before committing
to that route it may be worth to pause and consider our ultimate purpose, and
of course the tradeoffs involved.
\\
\\
First of all a small core language is easy to reason about, particularly if program
analysis is automated ---
which had better be,
% http://www.englishpage.com/modals/hadbetter.html
at a time when
%programs are already millions of lines long even as written by humans
%in their original, more compact form.
programs as written by humans can be millions of lines long.
 Moreover the core language will tend to be
easier to ``get right'', as a small number of features will reduce the chances
of unforeseen bad interactions. A naïve implementation will also be quick to
build.

\label{tracking-source-locations--introduction}On the other hand, as the core language will not be usable in practice without
several layers of extensions, the problem of tracking source locations becomes
relevant: a user will normally write in the extended language and expect the programming
system to refer the extended program \textit{as she wrote it} in terms of file names,
line numbers and syntactic forms: for example error messages referring the final
transformed program would prove very hard to understand\footnote{Actually,
  simple partial solutions to this problem have been known for a long
  time. For example the C language preprocessor generates a
  \CODE{\#line} directive in the output mentioning the original file name
  and line whenever the source of the output changes; this is enough information
  for the
  compiler to map each element of the single stream of code it receives back to its
  original location. Anyway debugging C code using macros remains notoriously
  hard: one reason is the lack of a similar output-to-input mapping at the
  level of syntactic forms.}.
Just a little more subtly, static analyses will often need to refer to the non-primitive
forms of some higher-level intermediate languages, instead of the
core\footnote{Again, an example is \CODE{let} under Hindley-Milner type inference,
  which could be defined with a macro such as the one above, but would benefit
  from being typed with \textit{let-polymorphism}, differently from generic
  function calls.  As another example,
% dealt with more fully in \SECTION\ref{transforms-chapter},
  a CPS transform can encode \textit{first-class
  continuations} into anonymous functions --- but in a static type system,
  continuations would need their own separate typing rules.}:
hiding details is the whole point of abstractions.
The obvious implementation will also be inefficient, as evident from the
examples in \SECTION\ref{macro-examples--introduction}
and \SECTION\ref{transformations-informal-idea}; however, since the code will start small
and manageable, it will be less hard to introduce optimizations where needed.
%\TODOF{Shall I
%  explicitly say why a trivial macroexpansion of \CODE{let} into \CODE{lambda}
%  and application is slow? I think it's quite obvious}.
\\
\\
We have mentioned pros and cons; in fact we argue that \textit{all} the
objections above can be answered,
but even this it not
% important in our view:
essential in our view:
we believe that any or all of
the problems above would still be
%eclipsed
%overshadowed
%outweighed
offset
%anyway
by the crucial advantage of obtaining an \TDEF{open-ended} language,
able to grow in unanticipated directions.

In order to achieve this ultimate goal we need \textit{both} a small core
language, and strong syntactic abstractions.

The choice above is a conscious restriction on the region of the design space
we are setting to explore.  Other choices are certainly possible, and a couple
have been tackled in the past, with interesting results.

\subsection{Related languages}
The most famous
example of this design as well as a source of inspiration for this work, again,
is Scheme \cite{r6rs}: anyway we have to remark how the core contains much more
than a functional language\footnote{Tom Lord's failed proposal for the new
Scheme Standard
--- before his alleged expulsion from Working Group 1 ---
would have been philosophically closer to our
vision, despite the very different core.  Lord's core language ``WG0 Scheme'' would have used \textit{fexprs} \cite{against-fexprs--pitman,fexprs--shutt} and reified
environments.  His own recount is at \url{http://lambda-the-ultimate.org/node/3861#comment-57967}.}, including complex features such as first-class
continuations, which could \textit{not} be re-implemented by using Scheme's
syntactic abstraction alone.  Despite its beauty Scheme's syntactic abstraction
system itself is also quite complex, based as it is on hygienic macros
\cite{hygienic-macros--kohlbecker}. The macro system does not look easy to factor,
as shown by the experience of \ind{\CODE{psyntax}}: \CODE{psyntax} is a
compiler translating Scheme with macros into pure Scheme, which is itself
written in Scheme \textit{with macros} and hence needs to be bootstrapped with
a pre-compiled version.  \CODE{Psyntax} is elegant but, by the author's
admission \cite[\SECTION18]{syntax-case}, not trivial.
\\
\\
As our second example of a reductionistic design, Forth
\cite{forth94,forth200x} is about as obvious: its basic mechanism is imperative
state mutation involving a fixed set of global stacks, with no need for actual
expressions: ``\CODE{42}'' is \textit{an imperative statement pushing a number
  on a stack}, and the ``\CODE{+}'' \TDEF{word} replaces the two topmost stack
elements with their sum.  All words, predefined or not, are zero-parameter
zero-result procedures with imperative stack effects.  Even control features
such as loops are defined as stack operations, involving the stack normally
used for procedure return addresses.

Defining a word involves temporarily switching to a ``compile state'', in which
each encountered word is taken from the program input stream and appended to
the current definition rather than being executed immediately; to implement
this, all word definitions begin and end with the state-switching words
``\CODE{:}''\ and ``\CODE{;}''. The words\footnote{\RED{2016 correction: I have
    since learned that ``\CODE{)}'' in its role shown above is not actually a
    word.  In actuality the definition of ``\CODE{(}'' discards all the input
    text up to and including the first ``\CODE{)}'' character, in a way
    analogous to, for example, ``\CODE{s"}'' and ``\CODE{"}''; this solution
    avoids the need for the comment closing character to be delimited with
    whitespace.  The comparison with
    ``\CODE{:}''\ and ``\CODE{;}'' remains valid, despite the slightly
    different nature of the comment solution.}} ``\CODE{(}'' and ``\CODE{)}'',
respectively opening and closing a comment, work in the same way.  No syntactic
structure exists at a scale larger than individual words.

%% \begin{Verbatim}[label={\rm Greater Common Divisor in Forth, using a variant of Euclid's algorithm}]
%% : gcd ( n1 n2 -- n )
%%    begin tuck mod dup 0= until drop ;
%% \end{Verbatim}
%% The comment above shows the word ``stack effect'': the two parameters are
%% replaced by one result on the main stack.  Most words consume their arguments.

Forth is an unusual language somewhat defying classification, and it is
debatable whether abstraction features such as changing ``state'' even count as
syntactic --- it might be the case that very strong procedural abstraction may
partially compensate for the absence of syntactic abstraction features, like
higher order does in functional programming\footnote{Our favorite example is
  \CODE{with-mutex} (for example as in \cite[\SECTION{Mutexes and Condition
      Variables}]{guile}) in a shared-state concurrent language with
  exceptions: \CODE{with-mutex} executes some given statements in a critical
  section so that a mutex is acquired at entry and released at exit,
  {\textit{including when an uncaught exception causes a jump out of the
    critical section}}.  It is easy to define \CODE{with-mutex} as a macro, and
  where macros are not available it can still be simulated using a higher-order
  procedure, at some loss of elegance.  Without either of these features, the
  user is forced to duplicate code.}.
Syntactic of not, Forth abstraction features have proved to be effective and
they do allow to abstract to a high level, despite the language being often
used on the bare metal without even an operating system.
Due to its simplicity, like Lisp,
Forth has been independently re-implemented many times, by building some core
primitive words in assembly and then writing the rest of the system in itself.
The language is so
small that \textit{hardware} implementations exist
\cite{stack-computers--koopman,greenarrays}, and its key proponents such as Chuck
Moore exhibit a cultural tendency to reject all conventional software as bloated
and hopelessly overcomplicated \cite{chuck}.
\\
\\
Finally some classic object-oriented systems such as Smalltalk are fairly
minimalistic as well, and provide relatively strong procedural abstraction.
Their example suggests \TDEF{reflection} as a further strategy to help build
complex programs: in object-oriented systems the state of each computational
entity in the running program is available to the program itself, which can
query objects for their interfaces at run time.  Such runtime type information
also provides the foundation of dynamic method dispatching, but we find this
style of late binding to be less interesting for our purposes, serving better
as a practical modelling tool than as a foundation for a core language.

%%%%%%%%%%%%%%%%%%%%%%%%%%%%%%%%%%%
\section{Our solution}
What kind of language do we want?  In such a vast design space there is no
clear answer, and committing a decision appears dangerous.  Even with no
breakthrough in sight at this particular time, our topology
of \SECTION\ref{taxonomy} might even get enriched by entirely new dimensions in
the future.

{Of course we \textit{do} have opinions about what kind of language would
  be better to solve the software crisis, and we will not even try too hard to hide
  our personal preferences as the description unfolds; but opinions are not science.}
%In the absence of a silver bullet
%???In such a situation,
Lacking a silver bullet,
the best course of action seems to leave our design
open-ended and follow in software the lesson of RISC, eschewing any particular
focus or specialization in exchange for wider applicability.
\\
\\
In order to achieve Steele's vision in \cite{growing-a-language}, our language:
\begin{itemize}
\item
will be built on a very small core, like Forth, but in a form more amenable to formal reasoning;
\item
will provide strong syntactic abstraction features, like Scheme does, plus transforms, in the interest of expressivity;
\item
will provide reflection, like object-oriented systems;
\item
\label{static-analyses-are-not-mandatory}
will not depend on either static or dynamic type checks in the core; such systems can be added as extensions;
\item
is meant to be practical, and efficiently implementable.%% \footnote{\NEW{The usual
%% warnings by software engineers about how ``performance doesn't matter'' have not become
%% any less wrong in recent years.  In 2010 the Paris metro was full of huge
%% billboards in which Google announced its new web browser, using performance as
%% its selling point to the public. \url{http://tunisiangeeks.files.wordpress.com/2010/01/google-chrome-dans-paris.jpg}.}}.
\end{itemize}
We call our language ``\EPSILON'', following the convention of naming small
variables in Mathematical Analysis.  When written  in the Latin alphabet
as ``epsilon'', the initial ``e'' in its name should always be lowercase.

A \TDEF{personality} is a language made of the \EPSILON core plus extensions,
in analogy with research operating systems implementing several different APIs
on top of the same microkernel \cite{personalities}.  Personalities may reach
very far from the core, as our transform examples above show.
\\
\\
We anticipate the development of complex\footnote{The composition of extensions
  is a difficult problem, for which no general solution is apparent:
  see \SECTION\ref{extension-composition}.} and widely diverging personalities, viewing the
emergence of incompatible dialects, so feared in some communities, rather as a
sign of health.

It is unfortunate but very possible that in a setting where a
strongly extensible language is adopted, a new separation of programmers and
meta-programmers (as personality developers) will follow the existing divide
between programmers and language implementors; we do not claim to be able to
cover such a cultural gap, but we mean to substantially ease the work of the
second group.

The Programming Languages discipline needs more experimentation and
prototyping.  Let people play, and the language will grow.
\section{Summary}

Programming languages have traditionally been diverse in paradigm, typing
policy and concurrency model.  The current trend of hybridization makes
languages more expressive, but also much harder to reason about; moreover most
current languages remain difficult to extend and
% \STRONG{severely}
% generally
 lacking in syntactic abstraction
features.

In keeping with the philosophy of Scheme but bringing it much further, we
propose the new programming language \EPSILON as an example of an alternative
style of language definition in which strong abstraction capabilities allow the
end user to express the needed linguistic features as translations into an
extremely simple core language which is easy to
%\MAYBEQ{automatically}
 reason about.
%% \EPSILONZERO\ is a functional language consisting in
%% \TDEF{call-by-value \LAMBDACALCULUS \TODO{\LAMBDACALCULUS my ass. I have to re-read all of this
%% very carefully, because something might still refer the old implementation!}}, \TDEF{conditional}, \TDEF{constants}, \TDEF{primitive
%%   function application} and \TDEF{futures}. \EPSILON\ allows the user to extend
%% \EPSILONZERO\ by defining primitives and constants in low-level languages, use
%% Lisp-like macros as local syntactic abstractions and define transforms to
%% specify non-local program transformations; \MAYBE{we speak of \EPSILON\ as of a
%% \TDEF{reductionistic} system as all the user code is translated into
%% \EPSILONZERO\ before execution.}
%% The presence of reflective primitives reifying the program state is enough for
%% the user to define static analyses, optimizations and even complete compilers
%% \textit{within the language itself}.
%% \\\\
A \TDEF{personality} is a library of language extensions, in fact defining
a new language in \EPSILON\ itself.
%%  By contrast a \TDEF{meta-library} may be
%% independent of any particular personality, and is used to build
%% infrastructure such as general-purpose analysis routines and compilers,
%% potentially usable with any personality.

We argue in favor of language experimentation, recognizing dialect
proliferation as beneficial in the long term.

\chapter{The core language\ \EPSILONZERO}
\label{chapter-epsilonzero}
\label{chapter-epsilon-zero}
\label{epsilon-zero-chapter}
\label{epsilonzero-chapter}
In this chapter we are going to formally describe the core language
\EPSILONZERO by giving a small-step operational semantics for it and stating under which
conditions an implementation is bound to behave according to the semantics.

As the foundation of much of the rest of this work, the specification will be used
in \SECTION\ref{reflection-chapter} for describing the language reflective features,
in \SECTION\ref{dimension-analysis-chapter} to prove correct a static analysis,
and in \SECTION\ref{syntactic-extension-chapter} for defining syntactic
extension semantics.

\minitoc

\section{Features and rationale}
Our core language \EPSILONZERO must be easy to reason about and efficiently
compilable, has to include reflective features providing access to the program
itself, and allow for parallelism.
On the other hand the language does \textit{not} need to be especially friendly to
human users, since programmers will normally access it by means of higher-level
syntactic extensions.

Satisfying such a
% \TODOQ{conflicting}
set of requirements yields an
idiosyncratic language whose extreme simplicity risks being overlooked at a
first glance, obscured by some slightly unusual design choices.
\\
\\
%% Despite being extremely simple by design, \EPSILONZERO is somewhat idiosyncratic,
%% and because of this we need some premises to precede the actual formalization.  Such
%% premises constitute \CHECKATTHEEND{the bulk} of the present chapter.
Before formally specifying \EPSILONZERO's syntax and semantics, it is worth to clarify the
rationale of some design decisions.

\subsection{First order}
\EPSILONZERO is a first-order call-by-value expression-based imperative
language with mutually-recursive procedures accepting zero or more parameters
and returning zero or more results, where procedures are globally defined in a
flat namespace.

The language is \TDEF{first-order}: no anonymous procedures exist, and
procedures can not be passed as parameters or returned as results\footnote{We
  are going to relax this restriction in the implementation for efficiency's
  sake; anyway, as shown in \SECTION\ref{call-indirect},
  it will be trivial to automatically transform any program using ``procedure
  pointers'' into an equivalent first-order program.  Closures
  will be remarkably more involved to define (see \SECTION\ref{closure-conversion-transform}).}.

Variable references are trivial to resolve: a block form is provided for
binding \TDEF{local} variables, which take precedence over procedure
\TDEF{parameters}, which in their turn take precedence over \TDEF{global
  variables}.  Since variables bound in other procedures are never accessible
\textit{no other scoping rule is necessary}.

Expressions return values and are allowed to have side effects; no looping form
exists, and since recursion is only permitted at the top level among global
procedures no explicit fixpoint operator is needed.  This sets apart
\EPSILONZERO from most functional languages, as the language of \EPSILONZERO
expressions not referring global procedures is \textit{not} Turing-complete.
\\
\\
The language exhibits relatively low-level features, making it easy to write a
simple compiler with a clear efficiency model for non-self-modifying programs;
Control Flow Analysis is also trivial, since all callees are explicitly
identified by name at call sites.  No escape mechanism such as exceptions,
\CODE{longjmp} or first-class continuations is provided at this level, so
evaluation strictly follows an intuitive stack discipline; in fact \EPSILONZERO
is \textit{stack-implementable}, and
after macroexpansion and transforms have run, the residual \EPSILONZERO program
does not necessarily require garbage collection.

\subsection{Reflection}

The current set of procedures is part of the global state of \EPSILONZERO, and
procedure definitions are accessible to the program for both reading and
writing; this allows a program to analyze and modify itself.
\\
\\
\textit{Compilation} in \EPSILONZERO consists in examining the current global
state in terms of data and procedures, and producing as output a low-level
program which, when executed, will reproduce the current state, in a style
reminiscent of some Smalltalk systems and the Emacs \CODE{unexec} hack
\cite[\SECTION{}Building Emacs]{emacs-lisp}.

An \EPSILONZERO compiler can be an ordinary set of \EPSILONZERO definitions,
running on top of the interpreter; in this sense we can say that the compiler,
if any, is part of the program being compiled, rather than an external tool;
and in the same way the user is free to build other \textit{meta-}level tools
such as code analyzers, transformers or optimizers.

\subsection{Handles}
\label{handle-introduction}
Since a program has to reason about itself, \EPSILONZERO needs some mechanism for
unambiguously referring to \TDEF{program points}, also distinguishing different
occurrences of otherwise identical syntactic forms.  For this reason each syntactic form
contains a unique identifier which we call \TDEF{handle}, the only requirement
being that each expression of a program
have % Correct: this is subjunctive
a different handle.
\\
At the implementation level it is reasonable to think of handles as unboxed
integers or pointers to unique objects, but the specific nature of handles as a
data type is immaterial: in practice the only relevant feature of a handle is
its identity.

Handles are contained in expressions at all nesting levels, so that
subexpressions at any depth may be referred by global names.
\\
\\
It is easy to associate information to handles, typically using global tables.

\subsection{Primitives}
The language specification should be complemented with a set of ``predefined''
primitive operators and data types for such operators to act upon, integer
arithmetics being the obvious example; other useful primitives include
memory allocation and side effects, and input/output.
Primitives may accept parameters, return results and affect the global state
but, the rationale being analyzability, they may \textit{not} alter program
control; this prevents ``jumping'' operations of the kind of exceptions, \CODE{longjmp} \cite{c99} and \CODE{call/cc}
\cite{r6rs} from being implemented \textit{as primitives}.

In the following we will not assume any particular set of primitives, limiting
ourselves to some reasonable constraints which the particular primitives have
to satisfy.

We do not dwell further on the specification of primitives which are to be
implemented at low level, in practice using C or assembly language; a formal
semantics of such low-level definitions is outside the scope of the present work.
{When working with any particular \EPSILONZERO or \EPSILON program, we will
always take the set of available primitives as fixed.}
%\CHECKATTHEEND{Check that I really can erase the previous sentence}

\subsection{Bundles}
We allow \EPSILONZERO procedures and expressions to \TDEF{return any number of
  results}, including zero; such as decision being more a concession to efficiency
\cite{compiling-standard-ml-to-java-bytecodes,types-are-calling-conventions},
than an attempt at restoring symmetry between
input and output.

A \TDEF{bundle} is an ordered sequence of values which may be the result of a
computation.
\label{bundles-are-expressible-but-not-denotable}
The only feature distinguishing a bundle from an ordinary $n$-uple
or list is the fact that bundles are not treated as data structures and in
particular are \textit{expressible but not denotable}: an \EPSILONZERO variable
can only refer to \textit{one} object, even if an expression is allowed to
return a bundle of any size; the rationale being, of course, that no
bundle data structures need to be expensively allocated and destroyed at
runtime: each separate bundle element will be simply assigned a stack cell or
register, possibly not even consecutively: no single ``value'' represents the
whole bundle.

In this sense bundles bear resemblance to the the Common Lisp and Scheme
\textit{multiple values} feature \cite{common-lisp,r6rs}, with the important
difference that in \EPSILONZERO callers do \emph{not} ignore all results
except the first one by default.

In order to work on bundle components, \EPSILONZERO's block form binds
up to
\textit{as many variables as the dimension of the bundle} that its bound expression
evaluates to; hence \EPSILONZERO blocks also serve the purpose of ``destructuring''
bundles, which in practice simply means locally naming their components.

\label{quotient-remainder-example}
For example if a ``\CODE{quotient-remainder}'' primitive returned both the quotient and the
remainder of two naturals
(and indeed many hardware architectures provide such a machine instruction) a
% \LETNAME
block could compute the quotient and remainder of some parameters, name the
results respectively $x$ and $y$, and evaluate a body in an 
environment where such variables are visible. It is worth to stress that at
runtime this naming does \textit{not} entail any moving, copying or -- worse -- memory allocation.

\label{let-blocks-can-ignore-some-values}
It is useful in practice not to always name \textit{all} the components of a
bundle, in particular for using nested blocks to simulate a statement sequence
where the results of the intermediate steps are irrelevant; more in general,
often one wants to ignore the result of a subexpression.
\\
\\
Bundles do complicate somewhat expression composition and are a possible cause of
errors, but the performance gain they offer seems hard to
%obtain by compiler optimization only;
%obtain, instead, by compiler optimization;
obtain automatically by compiler optimization only.
%recover by compiler optimization only;
Recursive procedures returning more than one result seem a
quite compelling example.

Of course personality implementors aiming for very simple extensions are always
free
\textit{not} to use bundles, or for that matter any other \EPSILONZERO feature.

\subsection{Parallel features}
The parallel features of \EPSILONZERO appear mundane
compared to some of the points above, limited as they are to creating
\TDEF{futures} associated to asynchronous \TDEF{threads}, and extracting the result of
a given future when waiting for its computation to terminate.

The system lends itself to both shared memory and message-passing, depending on
primitives. In complex personalities aiming at high efficiency on large parallel
machines or clusters, it is reasonable to expect that both styles will be used
at different levels.
% \CHECKATTHEEND{I'm not convinced that a message-passing
%  implementation can have acceptable performance, at least without local memories}
\\
\\
Again, parallel features introduce some complications into \EPSILONZERO but are
too ``fundamental'' to be left out and then meaningfully reintroduced as
language extensions.

\section{Syntax}
We are now ready to formally specify the syntax of \EPSILONZERO expressions,
and establish some terminology about subexpressions.

\label{where-we-speak-of-numerable-sets}
Let
the \TDEF{set of variables} $\SET{X}$,
the \TDEF{set of procedure names} $\SET{F}$,
the \TDEF{set of primitive names} $\SET{\Pi}$,
the \TDEF{set of handles} $\SET{H}$
and
the \TDEF{set of thread identifiers} $\SET{T}$
be any numerable sets.  By convention we will use the following metavariables,
possibly with decorations, to represent objects of the respective sets:
$x \in \SET{X}$ for \TDEF{variables},
$f \in \SET{F}$ for \TDEF{procedure names},
%$e \in \SET{E}$ for \TDEF{\EPSILONZERO expressions},
$\pi \in \SET{\Pi}$ for \TDEF{primitive names},
$h \in \SET{H}$ for \TDEF{handles},
and
$t \in \SET{T}$ for \TDEF{thread identifiers}.

%% In the following we will use the following metavariables, possibly with decorations, to represent objects
%% of the following respective sets: $c \in \SET{C}$ for \TDEF{unboxed constants}
%% or \TDEF{values},
%% $x \in \SET{X}$ for \TDEF{variables},
%% $f \in \SET{F}$ for \TDEF{procedure names},
%% $e \in \SET{E}$ for \TDEF{\EPSILONZERO expressions},
%% $\pi \in \SET{\Pi}$ for \TDEF{primitive names},
%% $h \in \SET{H}$ for \TDEF{handles},
%% $t \in \SET{T}$ for \TDEF{thread identifiers}.

%% All the sets mentioned above except expressions can be thought of as
%% \label{where-we-speak-of-numerable-sets}
%% numerable,
%% the elements being distinguishable from the other values \textit{of the same set}
%% but not necessarily with any internal structure of note.

\label{where-we-hint-at-the-value-notation}
Since the actual nature of ``values'' is irrelevant for
the purposes of this chapter but has some other deep ramifications, we postpone
its discussion until \SECTION\ref{the-nature-of-values}; as of now we
simply speak of a \TDEF{set of values} $\SET{C}$, using the metavariable
$c \in \SET{C}$
for representing its elements.

%$c \in \SET{C}$ for \TDEF{unboxed constants}
%.
For our examples in this chapter it will suffice to
just use natural numbers, writing ``$\mathcal{N}(n)$'' to represent
$n \in \NATURALS$,
booleans $b \in \{\CODE{\#t}, \CODE{\#f}\}$ written as ``$\mathcal{B}(b)$''
\textit{pointers} or memory addresses $a$ written as ``$\mathcal{A}(a)$''
and
\textit{thread identifiers} or \TDEF{futures} $t$ written as ``$\mathcal{T}(t)$''.
%\CHECKATTHEEND{check that we don't need other tags in this chapter}.
%It is worth to stress again that, as we are dealing with a first-order language,
%\textit{procedures are not values} in \EPSILONZERO.
%% \\
%% \\
%% We are now ready to formally define \EPSILONZERO expressions below. \TODO{Shall I warn that user syntax is very different from this?}

%% In the following we will use the metavariables $e$ for expressions,
%% $x$ for variables $\SET{X}$, $c$ for literal constants representing unboxed values, $f$
%% for procedure names, $\pi$ for primitives and $h$ for handles, possibly with
%% %subscript and prime
%% decorations.

\begin{definition}[\EPSILONZERO syntax]
\label{epsilonzero-syntax}
\label{epsilonzero-grammar}
\label{epsilonzero-syntax-definition}
\label{epsilonzero-grammar-definition}
We define an \TDEF{\EPSILONZERO expression}
%$e \ni \EXPRESSIONS$
% by the following grammar.
according to the following grammar:
\\
{{\rm
$e ::=$
\SPACERF$\DVARIABLE{h}{x}$
\SPACERP$\DCONSTANT{h}{c}$
\SPACERP$\DLET{h}{x^{*}}{e}{e}$
\SPACERP$\DCALL{h}{f}{e^{*}}$
\SPACERP$\DPRIMITIVE{h}{\pi}{e^{*}}$
\SPACERP$\DIFIN{h}{e}{c^{*}}{e}{e}$
\SPACERP$\DFORK{h}{f}{e^{*}}$
\SPACERP$\DJOIN{h}{e}$
\SPACERP$\DBUNDLE{h}{e^{*}}$}}

We call each separate production right-hand side an \TDEF{expression form} or \TDEF{form}.
%%\\
%% \begin{figure}[h!]
%% \centering
%% \begin{minipage}{0.4\linewidth}
%% {{\rm
%% $e ::=$
%% \SPACERF$\DVARIABLE{h}{x}$
%% \SPACERP$\DCONSTANT{h}{c}$
%% \SPACERP$\DLET{h}{x^{*}}{e}{e}$
%% \SPACERP$\DCALL{h}{f}{e^{*}}$
%% \SPACERP$\DPRIMITIVE{h}{\pi}{e^{*}}$
%% \SPACERP$\DIFIN{h}{e}{c^{*}}{e}{e}$
%% \SPACERP$\DFORK{h}{f}{e^{*}}$
%% \SPACERP$\DJOIN{h}{e}$
%% \SPACERP$\DBUNDLE{h}{e^{*}}$}}
%% \end{minipage}
%% \caption{\EPSILONZERO expression syntax.}
%% \end{figure}

We call the first two cases a \TDEF{variable} and a literal \TDEF{constant},
respectively. A \LETNAME \TDEF{block} contains zero or more distinct \TDEF{bound variables}, a
\TDEF{bound expression}, and a \TDEF{body}. A \TDEF{procedure call} \CALLNAME
expression mentions a \TDEF{procedure name} and zero or more \TDEF{actual
  parameters}. Very similarly, a \TDEF{primitive call} mentions a
\TDEF{primitive name}, and zero or more actual parameters. The
\TDEF{conditional} form \IFNAME comprises a
\TDEF{discriminand} expression, zero or more \TDEF{conditional cases}, and finally the
\TDEF{then branch} and \TDEF{else branch} expressions. A \TDEF{fork} expression
has the
same syntax as a procedure call, while a \TDEF{join} expression simply contains one
\TDEF{future expression}. A \TDEF{bundle expression} contains zero or more
\TDEF{bundle items}.

Each expression and its subexpressions, at all levels, contain unique
\TDEF{handles}.

We define $\SET{E}$ to be the set of all \EPSILONZERO expressions; we use the
metavariable $e$, possibly with decorations, to represent its elements.
\QEDDEFINITION
\end{definition}
The grammar in Definition~\ref{epsilonzero-syntax-definition} should be hardly
surprising at this point, except possibly for the shape of the
conditional and fork expressions.
\\
\\
The conditional expression shape is actually another small concession to
efficiency: operationally, the discriminand expression is evaluated and
compared to all the given conditional cases: if the discriminand evaluates to
one of the given constants then the conditional reduces to the then branch,
otherwise it reduces to the else branch. In many cases this kind of expression,
when nested, is easy to optimize into multi-way conditional branches
implemented as jump tables or balanced comparison trees.

Of course an \EPSILONZERO\ \IFNAME expression can also simulate an ordinary two-way
%McCarthy conditional \cite{lisp-mccarthy},
McCarthy conditional,
by using a boolean literal as its only conditional
case. Writing the \textit{false} literal as \CODE{\#f} in the style of Scheme, we can
simulate the two-way conditional ``$(e \to e', e'')$''
%with
by
% ----------------------- begin --------------------
${\DIFIN{{h_{0}}}{e}{\mathcal{B}(\CODE{\#f})}{e''}{e'}}$, for some handle $h_0$.
Of course reasonable personalities will define their own friendlier conditionals.
\\
\\
\label{fork-and-nonlocals}As for the \FORKNAME expression, at a first look the form given above might appear
gratuitously complicated compared to an alternative containing only one
``asynchronous expression''. Anyway such an
alternative would be difficult to evaluate, as the asynchronous expression
could then refer variables bound in the original thread, which would effectively
become nonlocals. For this reason, as elsewhere in \EPSILONZERO, we chose 
a more constrained syntax
%the form of Definition~\ref{epsilonzero-syntax-definition}
without too much fear of
inconveniencing the user: personalities will provide higher-level fork operators.
\\
\\
Facilities for \textit{defining} procedures will be dealt with
in \SECTION\ref{reflection-chapter}.

\subsection{Meta-syntactic conventions for expressions}
Since every syntactic object contains a handle, independently of the syntactic
category or the specific case, when identifying particular components of a
syntactic form instance we may explicitly specify a (meta-)handle in a
\textit{sub-}expression of a given expression, despite referring to the
sub-expression itself just with a meta-variable; for example $h_2$ represents
the handle of the \JOINNAME\ future expression in
${\DJOIN{h_{1}}{\DHANDLE{h_2}{e}}}$, regardless of the future expression specific
``shape''.
% Here I used join as an example because it has only one sub-expression. By
% doing that I don't need to use subscript indices, which I would simplify
% away in any case in the following paragraph
We will also omit subscripts in
meta-variables which already contain (meta-)handles unambiguously identifying
instances: for example we will simply write
${\DPRIMITIVE{h_0}{+}{\DCONSTANT{h_1}{e}\ \DCONSTANT{h_2}{e}}}$
instead of the heavier
${\DPRIMITIVE{h_0}{+}{\DCONSTANT{h_1}{e_1}\ \DCONSTANT{h_2}{e_2}}}$. When only
one identifier appears as a subscript, it should always be interpreted as a
handle rather than a metavariable decoration.

We will usually name (meta-)handle indices in a top-to-bottom left-to-right
order according to the expression syntax; we may let indices
start from either $0$ or $1$, according to which option provides more notational
convenience, such as avoiding the occasional ``$+\ 1$'' in
subscripts. For example we prefer writing an $n$-element multiple expression as
$\DBUNDLE{h_0}{\DHANDLE{h_1}{e}...\DHANDLE{h_{n}}{e}}$ rather than as
$\DBUNDLE{h_1}{\DHANDLE{h_2}{e}...\DHANDLE{h_{n+1}}{e}}$.
Since we use handles only to identify occurrences of syntactic
forms, their actual value is always immaterial:
% and we are only interested in checking them for equality; of course 
starting indices in meta-handle
sequences can be just as arbitrary.

\section{Semantics and the real world}
\label{epsilonzero-semantics-and-the-real-world}
Before finally specifying \EPSILONZERO's semantics it is worth to add one last
remark to prevent some misunderstandings, explicitly delimiting the cases in
which an implementation is compelled to respect the semantics.  This point is
crucial if we are to speak about actual programs running in actual machines,
rather than just another formal calculus whose terms unfold into other terms in
a Platonic universe where memory is infinite and checking for error conditions
is for free.

As \EPSILONZERO is the underlying common language all personalities ultimately
reduce to, it also represents
% an ideal efficiency limit:
\textit{an efficiency upper bound}:
a program written in
a higher-level personality will only run as fast as its translation into
\EPSILONZERO; hence the critical need for speed, also offsetting the cost of
making some implementations unfriendly and \textit{unforgiving} of mistakes.
But thankfully an unforgiving implementation does not need to be the only one,
and when developing an application a user will benefit from the feedback of a
\textit{slower interpreter} failing in a more descriptive way than a
``\CODE{Segmentation fault}'' message, and possibly allowing some form of
debugging.

The nature itself of failure needs to be carefully stated here: what we refer
to as ``failure'' 
(\SECTION\ref{epsilonzero-failure})
or ``resource overflow''
(\SECTION\ref{resource-limits})
in the semantics
does \textit{not} necessarily translate into a
dynamic check at the implementation level:

\begin{implementationnote}[implementation guarantees]
\label{failure-implementation-note}
\label{implementation-guarantees-implementation-note}
A conforming implementation will behave according to the semantics
provided that the execution never reaches an error configuration
and never exceeds any resource limit; otherwise, the implementation behavior is
unspecified.
\QEDIMPLEMENTATIONNOTE
\end{implementationnote}
Not failing and not overflowing resources is just a \textit{sufficient}
condition for an implementation to respect the semantics; in the implementation
a program violating one of these condition is allowed to crash or silently return
any result, possibly even the correct one: \textit{no guarantees at all}.
If a personality implementor wants to specify some behavior in one of these cases
then it's her responsibility to perform \textit{static checks on the input
  code} or to include \textit{dynamic checks in the generated \EPSILONZERO
  code}, in order to prevent the conditions for unspecified behavior from
occurring.
\\
\\
It may be worth to state explicitly that, as a trivial consequence of
Implementation Note~\ref{implementation-guarantees-implementation-note},
an execution consuming an unbounded quantity of \textit{any} resource has
unspecified implementation behavior.

\subsection{Resource limits}
\label{resource-limits}
\label{resource-limits-different-from-the-word-size}
As it is easy to imagine, our first mention of \textit{numerable sets}
in \SECTION\ref{where-we-speak-of-numerable-sets} already hid a caveat: only a finite
number of distinct values will be representable in an implementation, due to
the finite nature of address spaces and word sizes.

%% \SECTION\ref{word-size-resource-limits} will touch some implications of this
%% particular problem,
%% but here it is worth to stress the existence of resource limits in general.

There exist other remarkable cases of resource limits: for example
the amount of available virtual memory, often dramatically smaller than what
the address space allows for; operating systems also usually constrain the
number of concurrent threads; an implementation might limit the stack size
to a constant value. All of these cases will be covered by Implementation
Note~\ref{dynamic-resource-limit-implementation-note}.
\\
\\
In the following we will state which resources may be limited by
implementations in explicit \TDEF{Implementation Notes} such as the following
one; as already explained in
Implementation Note~\ref{implementation-guarantees-implementation-note},
an implementation is not forced to detect the failure at
runtime, and may proceed with undefined behavior in case of resource overflow.

\begin{implementationnote}[Syntactic resource limits]
\label{syntactic-resource-limit-implementation-note}
Each instance of the following items occupies some \TDEF{memory} in an
implementation; an implementation will pose a limit to the sum total of all the used
memory at any given time, possibly multiplied by a logarithmic factor.
\begin{itemize}
\item variable and procedure names, the cost being proportional to the name length;
\item the number of handles;
\item expression syntactic complexity.
%% \item \TODO{make sure that I say that defining new \textbf{primitives} consumes
%%   memory; write that when I hint at primitive definition (not here, but
%%   somewhere in \SECTION\ref{reflection-chapter} or \SECTION\ref{syntactic-extension-chapter})}
\QEDIMPLEMENTATIONNOTE
\end{itemize}
\end{implementationnote}
In Implementation Notes dealing with resource limits such as
Implementation~Note~\ref{syntactic-resource-limit-implementation-note} above,
we deliberately ignore constant terms: for example if the physical resource
occupation of $n$ items of a certain kind is $n \cdot a + k$ units, an
implementation may simply declare each item to take $a$ units and the total
resource availability to be $k$ units lower than its actual dimension.
\\
Moreover, since avoiding resource overflows is only a sufficient condition
for an implementation to respect the stated semantics, we allow Implementation
Notes to describe resource occupation \textit{as an upper bound}.

At the cost of sounding pedantic, we stress that the statement above
\textit{does not} limit our reasoning about resources to asymptotic
approximations; in fact, where a given implementation instantiates the precise
costs and resource availability, it is possible to reason about whether a
program ``fits'' the implementation on which it runs --- the rationale of course
being that a program overflowing resources is not any better than an incorrect one.

\section{Configurations}
We are now ready to formally define the mathematical structures used in
\EPSILONZERO's dynamic semantics.
 
%%%%%%%%%%%%%%%%%%%%%%%%%%%%%%%%%%%%%%%%%%%%%%%%%%%%%%%%%%%%%%%%%%%%%%%%%%%%%%%
\subsection{The global \STATE}
A \TDEF{global state} or simply \TDEF{state}, always represented with a
possibly decorated $\Gamma$ metavariable, represents
\textit{the instantaneous condition of an execution}; an execution may
access, reading and also destructively mutating parts of
a state.  We call ``$\SET{\Gamma}$'' the set of all possible states.

\label{state-environment-introduction}
A state is a inherently composite object made of several \TDEF{state
  environments};
%, or sometimes simply \TDEF{environments};
by ``environment'' we
simply mean a mathematical function mapping keys into values.  We do not list
all state environments here, since the need of some of them will not be apparent
until later.  In order to lighten our notation and to allow for
yet-unspecified components without depending on some arbitrary order, we also
avoid traditional projection and update operators, opting instead for a
notation referring components \textit{by name}, as if the state were a single
``record'' of environments.

\subsubsection{Notational conventions for states and environments}
\label{notaional-conventions-for-states-and-environments}
It is often convenient to exploit the \textit{set-of-pairs} nature of relations to
represent environments in an extensional style, as in 
``$\{x_1 \mapsto q_1, ..., x_n \mapsto q_n\}$''; an interesting particular
case is the empty environment, which is to say a nowhere-defined function,
which we write as the empty set ``$\EMPTYSET$''.
\\
\\
If $\vartheta$ is the state environment named $n$ in $\Gamma$ then we may write
$\Gamma_{n}$ to mean $\vartheta$; and of course we also employ the ordinary notation for
function application by writing, for example, $\Gamma_{n}(\CODE{x}) = q$
or equivalently $\Gamma_{n} : \CODE{x} \mapsto q$.

As per the standard update notation, we write
``$\UPDATEENVIRONMENT{\vartheta}{\CODE{x}}{q}$''
to represent an environment equal to $\vartheta$ everywhere on its domain
except on $\CODE{x}$, which the updated environment maps to $q$.

Extending the standard notation, we will also deal with \textit{updated
  environments in the state}: in other words we build a state identical to a
given one save for one environment, which has been updated in its turn; we will
write ``$\UPDATESTATEIN{\Gamma}{n}{\CODE{x}}{q}$'' to represent
 the updated
\STATE identical to $\Gamma$ except for the environment named $n$, which will
be $\UPDATEENVIRONMENT{\Gamma_n}{\CODE{x}}{q}$ instead of $\Gamma_{n}$.

{We also write $\UPDATESTATE{\Gamma}{n}{\vartheta}$ to represent
  a state identical to the state $\Gamma$ except for the state environment
  named $n$, entirely replaced by the environment $\vartheta$.}
\\
\\
Our use of brackets for updated \STATES is distinct from the usual
environment update notation, which we \textit{also} adopt: we write
``$\eta[\xi]$'' to mean an environment identical to $\eta$ everywhere except on
the domain of $\xi$, where instead it is identical to $\xi$.
\\
\\
Notice that, unless we are dealing with \textit{meta}-labels such as
``$n$'' here
in \SECTION\ref{notaional-conventions-for-states-and-environments}, we always
write state environment labels in typewriter font: this makes it clear that we
are establishing a label for some state environment at its first mention,
without the need of detailing every time how one label represents the
similarly-named state environment, when the association is always obvious from
the context anyway.

\subsection{Global and local environments}
\label{global-environments}\label{global-environment}
\label{local-environments}\label{local-environment}
The \TDEF{global environment} is a state environment mapping global variable
names into values, and can be thought of as a partial function $\SET{X} \PARTIAL \SET{C}$.

The global environment keeps track of globally-visible objects (\TDEF{globals}
or \TDEF{non-procedures}), which are
always accessible by a variable name unless shadowed by a procedure parameter
or a local variable, which instead are bound in \TDEF{local environments}: local
environments, also $\SET{X} \PARTIAL \SET{C}$ functions, take precedence over the
global environment, and or course they are \textit{not} state environments; we
use the $\rho$ metavariable for local environments, possibly with decorations.

For example, when evaluated in a state
$\UPDATESTATEIN{\Gamma}{\CODE{global-environment}}{x}{\mathcal{N}(42)}$
and the local environment $\EMPTYSET$,
the variable $x$ in the expression $\DHANDLE{h_0}{x}$ will refer
the value $\mathcal{N}(42)$; but if instead the local environment was
$\{x \mapsto \mathcal{N}(10)\}$, then $\mathcal{N}(10)$ would take
precedence over the global value in $x$. % Not an expression!

%%%%%%%%%%%%%%%%%%%%%%%%%%%%%%%%%%%%%%%%%%%%%%%%%%%%
\subsection{Memory}
\label{memory-introduction}
\label{memory-state-environment}
%\TODO{remove any reference to ``memory cells'': they don't exist any longer as individual objects}

\EPSILONZERO expressions are allowed to perform imperative operations
on mutable data structures: in particular expressions may
\textit{read} or \textit{update} cells of memory buffers, which can be \textit{allocated} and \textit{destroyed}.

Such operations rely on the \CODE{memory} state environment $\SET{A} \PARTIAL \SET{C}^{*}$ as
mapping \TDEF{addresses} into mutable word sequences or \TDEF{buffers};
we might occasionally refer to each buffer element as a \TDEF{memory cell}.
\\
\\
It is important to notice that the memory state environment models
\textit{the heap} in the implementation, of which each cell makes up
\textit{one word}: here we are dealing with cells which can be
allocated and destroyed with any strategy, rather than a simple LIFO policy.

No data structures such as conses, tuples and arrays are hardwired in
\EPSILONZERO, but memory makes it easy to define such objects in a
personality.  The fact that dynamically-created structures are ``made
of'' memory entails their mutability in a natural way.  Immutability,
if one chooses to enforce it for some class of data in a high-level
personality, can be realized with dynamic or static checks which
prevent updating\footnote{A more radical strategy could involve a
  syntactic ``extension'' un-defining or otherwise making inaccessible
  the operators needed for the update.} --- but in \EPSILONZERO all
memory cells are freely mutable, so as not to restrict the user in any
way.
\\
\\
At this point the reader may already be suspecting that the global environment
could be used for simulating memory; while ---assuming the availability of certain
primitives--- that intuition is
correct, \SECTION\ref{static-programs} will provide a strong
argument in favor of having a separate memory state environment.

\subsection{Procedures}
\label{procedure-state-environment}
A state also keeps track of the current set of procedure definitions.

The \TDEF{procedure state environment} is a
$\SET{F} \PARTIAL (\SET{X}^{*} \times \SET{E})$
partial function mapping each procedure name into a pair
holding its zero or more formal parameters and the procedure body; for example, if
a procedure named $f$ has formals $x_1...x_n$ and body $\DHANDLE{h}{e}$ in the
state $\Gamma$
we write
``$\Gamma_{\CODE{procedures}}:{f} \mapsto ({\SEQUENCE{x_1...x_n}}, {\DHANDLE{h}{e}})$'';
we may also omit angle brackets when no ambiguity can arise, in this case writing
``$\Gamma_{\CODE{procedures}}:{f} \mapsto (x_1...x_n, {\DHANDLE{h}{e}})$''.

We remind the reader that, since \EPSILONZERO is a first-order language and all
procedures are global \textit{no nonlocals can exist}; for this reason there is
no need for closure environments at this level.
%Procedure definitions will be dealt with later, in \SECTION\ref{reflection-chapter}.
\\
\\
Up to this point we have dealt with the syntax of \EPSILONZERO expressions
only, stating that global mutually-recursive procedures are also somehow
available but without specifying any way of \textit{defining} them.  Because of
interactions with the rest of the system the issue turns out to be more
delicate than one could imagine, and we defer its full treatment
to \SECTION\ref{reflection-chapter}; what we can hint at now is that procedures
can be defined with \textit{primitives} and other procedures, as explained below.

\subsection{Primitives}
\label{epsilonzero-primitive-introduction}
We call \TDEF{primitives} a set of low-level routines accessible from
\EPSILONZERO
expressions, used for computation, program reflection or side effects.
Primitives range from simple arithmetic operations such as \CODE{+}
to reflection and procedure definition operations, potentially also involving
destructive state updates.
\\
\\
In an implementation primitives are routines implemented in a low-level
language such as C or directly in Assembly.  This does not mean however that a
primitive is allowed to ``do anything'': primitives
must not disrupt the program control flow by performing jumps or non-local
exits or reentries \textit{à la} \CODE{longjmp} or \CODE{call/cc}: primitives
may affect the global state but have to behave in a procedural fashion,
always giving control back to their caller;
primitive behavior can actually be modeled by \textit{partial functions} taking a fixed
number of parameters and returning a fixed number of results --- including an
input and output state.  Such
higher-order functional specification is a consequence of the fact that
primitives, unlike procedures, are not directly implemented in \EPSILONZERO and
hence lack high-level bodies or any treatable ``source''.
\\
\\
%% A \TDEF{primitive function} is a partial function
%% $(\SET{C}^{*} \times \SET{\Gamma}) \PARTIAL (\SET{C}^{*} \times \SET{\Gamma})$,
%% mapping a tuple of zero or more actual parameters and a state
%% into another tuple of zero or more results and another state;
%% a \TDEF{primitive} is a triple comprising a primitive function and
%% two natural numbers, respectively the \TDEF{primitive in-dimension} and \TDEF{primitive out-dimension},
%% which respects Axiom~\ref{reasonable-primitive-axiom}
%% and Axiom~\ref{primitive-totality-axiom}.
%% We call $\SET{P}$ the set of all primitives, with
%% $\SET{P} \subset ((\SET{C}^{*} \times \SET{\Gamma}) \PARTIAL (\SET{C}^{*} \times \SET{\Gamma})) \times \NATURALS \times \NATURALS$.
%%
%% The \TDEF{primitive environment} state environment
%% $\SET{\Pi} \PARTIAL \SET{P}$
%% maps each primitive name into a primitive.
%%
%% \begin{axiom}[primitive dimension]\label{reasonable-primitive-axiom}
%% Let a triple
%% $(p, n, m) \in ((\SET{C}^{*} \times \SET{\Gamma}) \PARTIAL (\SET{C}^{*} \times \SET{\Gamma})) \times \NATURALS \times \NATURALS$ be given. Then,
%% if there exist $c_1,..., c_a, c'_1,..., c'_b, \Gamma$ and $\Gamma'$ such that $p(\SEQUENCE{c_1,..., c_a}, \Gamma) = \SEQUENCE{c'_1,..., c'_b, \Gamma'})$, then we have that $a = n$ and $b = m$.
%% \QEDAXIOM
%% \end{axiom}
A \TDEF{primitive function with in-dimension $n$ and out-dimension $m$}
($n, m \in \NATURALS$ and $n, m \geq 0$) is a partial function
$(\SET{C}^{n} \times \SET{\Gamma}) \PARTIAL (\SET{C}^{m} \times \SET{\Gamma})$,
mapping an $n$-uple of values and a state
into an $m$-uple of values and another state;
a \TDEF{primitive} is a triple comprising the primitive function, its
in-dimension and out-dimension,
which respects Axiom~\ref{primitive-totality-axiom}.
We call $\SET{P}$ the set of all primitives, with
$\SET{P} \subset \bigcup_{n,m \in \NATURALS}\{\SEQUENCE{p, n, m } \ |\ p \in (\SET{C}^{n} \times \SET{\Gamma}) \PARTIAL (\SET{C}^{m} \times \SET{\Gamma})\}$.

The \TDEF{primitive environment} state environment
$\SET{\Pi} \PARTIAL \SET{P}$
maps each primitive name into a primitive.

%% Axiom~\ref{reasonable-primitive-axiom} simply states that each primitive always
%% takes the same number of parameters and returns the same number of results, as
%% specified in the primitive environment; that, along with the functional nature
%% of primitive functions, suffices to restrict our attention to the low-level
%% routines which can actually implement primitives: \CHECKATTHEEND{non-diverging}
%% routines not jumping away, taking the appropriate number of parameters and
%% returning the appropriate number of results, allowed to have side effects.
Axiom~\ref{primitive-totality-axiom}, defined in
\SECTION\ref{where-primitive-totality-axiom-is-defined}
and only needed for technical reasons, will just
affirm that primitive success and failure are mutually exclusive.
\\
\\
From now on we will bend our notation a little further by  writing
``$\Gamma_{\CODE{primitives}}(\pi)$\linebreak$(c_{1}, ..., c_{n}, \Gamma) = \SEQUENCE{c'_{1}, ..., c'_{m}, \Gamma'}$''
or
``$\Gamma_{\CODE{primitives}}(\pi) \HASDIMENSION n \to m$''
as needed, %in order
 to avoid useless pedantries such as 
``$p(\SEQUENCE{c_{1}, ..., c_{n}}, \Gamma) = (\SEQUENCE{c'_{1}, ..., c'_{m}}, \Gamma')$
where
$\Gamma_{\CODE{primitives}}(\pi) = (p, n, m)$''.

%% \ADVISORS{Je voulais montrer au moins un exemple de primitive, pas nécessairement
%%   ici~; est-ce raisonnable à votre avis~?  On peut le couper aussi, je n'aime
%%   pas beaucoup cet exemple en particulier.}
As an example, considering the \CODE{quotient-remainder} primitive of \SECTION\ref{quotient-remainder-example}
in some state $\Gamma$
we could write
``$\Gamma_{\CODE{primitives}}(\CODE{quotient-remainder}) (\mathcal{N}(13),\LINEBREAK \mathcal{N}(3), \Gamma) =$\linebreak$\SEQUENCE{\mathcal{N}(4), \mathcal{N}(1), \Gamma}$''
meaning that the quotient and remainder of the naturals $13$ and $3$ are (respectively) the naturals $4$ and $1$,
and that the primitive does not affect the global state; we could also write ``$\Gamma_{\CODE{primitives}}(\CODE{quotient-remainder}) \HASDIMENSION 2 \to 2$'', by which we would mean that \CODE{quotient-remainder} has two parameters and two results --- which does not prevent the primitive function from being partial, as indeed it is.
\label{primitive-dimension-christophe}
Where the particular state is obvious from the context or irrelevant, we even write
``$\pi \HASDIMENSION n \to m$''
to mean
$\Gamma_{\CODE{primitives}}(\pi) \HASDIMENSION n \to m$, for the appropriate $\Gamma$.

As a further and possibly more interesting case, we just hint at the fact
that \textit{memory operations} such as allocating buffers and loading
and storing words are performed by appropriate \textit{primitives}:
this will let us keep the semantics simple, ignoring the details of memory, and
treating memory operations as just another instance of effects on the global state.
\\
\\
Specifying a complete set of ``default'' primitives is out of the scope of this
work, but
\SECTION\ref{syntactic-extension-chapter} will informally introduce most
primitives currently used in the implementation, while
\SECTION\ref{reflection-chapter} will
deal with reflection and program-updating in relation with primitives.
\\
\\
{We may informally speak of \TDEF{applicables}, when abstracting away the distinction
between procedures and primitives.}

%%%%%%%%%%%%%%%%%%%%%%%%%%%%%%%%%%%
\subsection{Holed expressions}
\label{epsilonzero-holed-expressions}
In our dynamic semantics we need to capture intermediate computation snapshots
in which an expression is in the middle of being evaluated.
\\
\\
We define below an extended \EPSILONZERO expression grammar, where the \TDEF{hole} ``$\HOLE$'' stands for a
subexpression which is yet to be fully evaluated.

\begin{definition}[\EPSILONZEROHOLE syntax]
\label{epsilonzerohole-syntax}
\label{epsilonzerohole-grammar}
\label{epsilonzerohole-syntax-definition}
\label{epsilonzerohole-grammar-definition}
We define the set $\HOLEDES$ of \TDEF{possibly-holed expressions} or ``\TDEF{\EPSILONZEROHOLE
  expressions}'' by the following grammar:
\\
{\rm
$e_\HOLE ::=$
\SPACERF$e$
\SPACERP$\DLET{h}{x^{*}}{\HOLE}{e}$
\SPACERP$\DCALLWITHHOLE{h}{f}$
\SPACERP$\DPRIMITIVEWITHHOLE{h}{\pi}$
\SPACERP$\DIFIN{h}{\HOLE}{c^{*}}{e}{e}$
\SPACERP$\DBUNDLEWITHHOLE{h}$
\SPACERP$\DFORKWITHHOLE{h}{f}$
\SPACERP$\DJOIN{h}{\HOLE}$
}
\\
Syntactic cases are respectively named: \TDEF{non-holed expression},
\TDEF{holed block} or \TDEF{holed \LETNAME}, \TDEF{holed call} or \TDEF{holed
  procedure call}, \TDEF{holed primitive} or \TDEF{holed primitive call},
\TDEF{holed conditional}, \TDEF{holed \BUNDLENAME}, \TDEF{holed \FORKNAME} and
\TDEF{holed \JOINNAME}.

All cases save the first represent \TDEF{properly holed expressions}.
\QEDDEFINITION
\end{definition}
Notice that holes do not occur in all possible expression contexts: this issue is
related to \textit{tail contexts}.

Moreover, as the nonterminal $e_\HOLE$ never occurs in a production right-hand
side, no holed expression can contain other properly holed expressions: this
``single hole'' property reflects \EPSILONZERO's deterministic sequential
evaluation strategy.

%%%%%%%%%%%%%%%%%%%%%%%%%%%%%%%%%%%%%%%%%%%%%%%%%%%%
\subsection{Stacks}
\label{epsilonzero-stacks}
Rather than resorting to the traditional small-step semantics style
\cite[\SECTION2.6]{semantics--winskel} in which the computed parts of an
expression are replaced with values, here we adopt a more realistic and
lower-level model using explicit stacks and keeping track of ``return points'';
this should already be clear at this point from
\SECTION\ref{epsilonzero-holed-expressions}.
\\
\\
We keep two separate aligned stacks per thread for describing evaluation, one stack
representing the dynamic nesting of partially-evaluated expression forms
and the other representing the dynamic nesting of values; we respectively
call them
the \TDEF{main stack} or even simply \TDEF{the stack}, and the \TDEF{value
  stack}:
\begin{itemize}
\item
The main stack is a sequence of pairs, each pair containing a holed expression
and its associated local environment (\SECTION\ref{local-environments}):
the set of all possible main stacks is $\SET{S} \triangleq (\HOLEDES \times (\SET{X} \PARTIAL \SET{C}))^{*}$;
\item
The value stack is a sequence of objects, each of which being one of a value,
the \TDEF{value separator} ``$\VALUESEPARATOR$'', or the \TDEF{activation separator} ``$\ACTIVATIONSEPARATOR$''.
Value stacks belong to $\SET{V} \triangleq (\SET{C} \DISJOINTUNION \{\VALUESEPARATOR, \ACTIVATIONSEPARATOR\})^{*}$.
\end{itemize}
A two-stack solution is particularly appropriate because of bundles and is
visually intuitive, but of course efficient implementations for
conventional machines will reasonably use a single stack per thread.
\\
\\
We write stacks \textit{horizontally, with the top on the left}: this is
analogous to list syntax in Lisp and functional languages, and the
opposite of Forth conventions.

We usually represent main stacks with the metavariables $S$ and value
stack with the metavariables $V$, possibly decorated.

%%%%%%%%%%%%%%%%%%%%%%%%%
\subsection{Futures}
\label{epsilonzero-futures}
As we have already hinted at in \SECTION\ref{epsilonzero-stacks} and will be
made more clear in \SECTION\ref{epsilonzero-dynamic-semantics-section},
evaluation in \EPSILONZERO needs \textit{two stacks per thread}, along with the
global state.
\\
\\
The ``main'' thread of a computation is called the \TDEF{foreground} thread;
the global state holds information about \textit{all the others}.

We call \TDEF{future state environment} the state environment \CODE{futures} holding thread
information. Such an environment simply maps each thread identifier into its
stack and value stack, and belongs to
$\SET{T} \PARTIAL (\SET{S} \times \SET{V})$.

\begin{implementationnote}[global state resource limits]
\label{global-state-resource-limit-implementation-note}
In an implementation the following resource limits hold:
\begin{itemize}
\item
each global environment binding occupies a constant amount of the \TDEF{memory} resource;
\item
each memory cell occupies a constant amount of the \TDEF{memory} resource;
\item
each defined procedure occupies a constant amount of the \TDEF{memory} resource;
\item
each defined primitive occupies a constant amount of the \TDEF{memory} resource;
\item
each thread which is either running or being waited by a \JOINNAME expression
in a background or foreground thread (\SECTION\ref{semantics-commentary})
occupies a constant amount of \TDEF{memory}.

An implementation may also limit the \TDEF{number of threads} running or
being waited for (as above) existing at any given moment, independently from
memory usage.
\end{itemize}
In all the cases above, some implementations may also scale the total amount of
occupied resource by a logarithmic factor.
\QEDIMPLEMENTATIONNOTE
\end{implementationnote}

%%%%%%%%%%%%%%%%%%%%%%%%%
\subsection{Configurations}
A \TDEF{configuration} contains information about the \TDEF{foreground thread},
and a global state; of course the global state, among the rest, holds information
about the other ``\TDEF{background}'' threads.

The set of all configurations is
$\SET{S} \times \SET{V} \times \SET{\Gamma}$; we usually represent
configurations with the letter $\chi$, possibly decorated; since configurations
are potentially complex when we show their three components we always omit 
commas to reduce the visual clutter.
\\
\\
Evaluating a given expression $\DHANDLE{h}{e}$ in a given state $\Gamma$ entails
building an \TDEF{initial configuration}
$(\DHANDLE{h}{e}, \EMPTYSET)\ \VALUESEPARATOR\ \Gamma$:
An initial configuration always has a main stack made by just a non-holed expression
coupled with an empty local environment, and a value stack
made of just one ``$\VALUESEPARATOR$'' separator.

\TDEF{Final success configurations} contain an empty main stack and a value stack
$\VALUESEPARATOR c_{n}c_{n-1}...c_{2}c_{1} \VALUESEPARATOR$
holding the zero or more elements of the result bundle in a reversed sequence, preceded and followed
by a ``$\VALUESEPARATOR$'' separator --- the bundle inversion phenomenon being
a consequence of the LIFO evaluation style.
\\
\\
For example,
assuming a ``reasonable'' \CODE{+} primitive and some state $\Gamma$,
we expect that by evaluating starting from the initial configuration
$(\DPRIMITIVE{h_0}{\CODE{+}}{{\DHANDLE{h_1}{\mathcal{N}(2)}}\ {\DHANDLE{h_2}{\mathcal{N}(3)}}},$\linebreak$\EMPTYSET)\ \VALUESEPARATOR\ \Gamma$
we eventually reach a final success configuration
$\EMPTYSEQUENCE\ \VALUESEPARATOR \mathcal{N}(5) \VALUESEPARATOR\ \Gamma'$;
it is possible to have $\Gamma \neq \Gamma'$ because of background
threads already started in $\Gamma$.

%%%%%%%%%%%%%%%%%%%%%%%%%%%%%%%%%%%%%%%%%%%%%%%%%%%%%%%%%%%%%%%%%%%%%%%%%%%%%%%
\section{Small-step dynamic semantics}
\label{epsilonzero-dynamic-semantics-section}

We are now finally ready to specify \EPSILONZERO's dynamic semantics.

%%%%%%%%%%%%%%%%%%%%%%%%%%%%%%%%%%%%%%%%%%%%%%%%%%%%
\subsection{Small-step reduction}
We need to formalize the intuitive notion of reduction.
Given two configurations $\chi$ and $\chi'$, we say that 
\TDEF{$\chi$ reduces to $\chi'$} and we write
``$\chi \TOE \chi'$'' according to the following definition:

\begin{definition}[small-step reduction]
\label{epsilonzero-semantics-definition}
\label{epsilonzero-semantics}
We define the \TDEF{small-step evaluation}
relation\\
$\PARAMETER \PARAMETER \PARAMETER \TOE \PARAMETER \PARAMETER \PARAMETER
\subseteq
(\SET{S} \times \SET{V} \times \SET{\Gamma}) \times (\SET{S} \times \SET{V} \times \SET{\Gamma})$
according to the rules at
pp.~\pageref{epsilonzero-semantics-rules--beginning}-\pageref{epsilonzero-semantics-rules--end}.
In the rules we always assume $n \geq 0$, with the convention that
an indexed sequence with left index 1 and right index 0 is empty.

Each rule has associated a name, written on the left in brackets.
\QEDDEFINITION
\end{definition}

\begin{landscape}
\thispagestyle{unobstrusive} % We overflow into header here
\label{epsilonzero-semantics-rules--beginning}
\begin{prooftree}
      \LeftLabel{$[{constant}]$}
      \RightLabel{}
      \AxiomC{}
      \UnaryInfC{$
        (\DCONSTANT{h}{c},\ \rho).S
        \ \ \VALUESEPARATOR V\ \ \Gamma
        \TOE
        S
        \ \ \VALUESEPARATOR c \VALUESEPARATOR V\ \ \Gamma
        $}
\end{prooftree}

\begin{prooftree}
      \LeftLabel{$[{variable}]$}
      \RightLabel{$\Gamma_{\CODE{global-environment}}[\rho] : x \mapsto c$}
      \AxiomC{}
      \UnaryInfC{$
        (\DVARIABLE{h}{x},\ \rho).S
        \ \ \VALUESEPARATOR V\ \ \Gamma
        \TOE
        S
        \ \ \VALUESEPARATOR c \VALUESEPARATOR V\ \ \Gamma
        $}
\end{prooftree}

\begin{prooftree}
      \LeftLabel{$[\CODE{let}_e]$}
      \RightLabel{}
      \AxiomC{}
      \UnaryInfC{$
        (\DLET{h_0}{x_1...x_n}{\DHANDLE{h_1}{e}}{\DHANDLE{h_2}{e}},\ \rho).S
        \ \ \VALUESEPARATOR V\ \ \Gamma
        \TOE
        (\DHANDLE{h_1}{e},\ \rho).(\DLET{h_0}{x_1...x_n}{\HOLE}{\DHANDLE{h_2}{e}},\ \rho).S
        \ \ \VALUESEPARATOR V\ \ \Gamma
        $}
\end{prooftree}

\begin{prooftree}
      \LeftLabel{$[\CODE{let}_c]$}
      \RightLabel{$m \geq n$}
      \AxiomC{}
      \UnaryInfC{$
        (\DLET{h_0}{x_1...x_n}{\HOLE}{\DHANDLE{h_2}{e}},\ \rho).S
        \ \ \VALUESEPARATOR c_{m}c_{m-1}...c_{2}c_{1} \VALUESEPARATOR V\ \ \Gamma
        \TOE
        (\DHANDLE{h_2}{e},\ \rho[x_1 \mapsto c_1, x_2 \mapsto c_{2}, ...,x_n \mapsto c_n]).S
        \ \ \VALUESEPARATOR V\ \ \Gamma
        $}
\end{prooftree}

\begin{prooftree}
      \LeftLabel{$[\CODE{call}_e]$}
      \RightLabel{}
      \AxiomC{}
      \UnaryInfC{$
        (\DCALL{h_0}{f}{{\DHANDLE{h_1}{e}}...{\DHANDLE{h_n}{e}}},\ \rho).S
        \ \ \VALUESEPARATOR V\ \ \Gamma
        \TOE
        ({\DHANDLE{h_1}{e}},\ \rho)...({\DHANDLE{h_n}{e}},\ \rho).({\DCALLWITHHOLE{h_0}{f}},\ \EMPTYSET).S
        \ \ \VALUESEPARATOR \ACTIVATIONSEPARATOR V\ \ \Gamma
        $}
\end{prooftree}

\begin{prooftree}
      \LeftLabel{$[\CODE{call}_c]$}
      \RightLabel{$\Gamma_{\CODE{procedures}}:{f} \mapsto ({x_1...x_n}, {\DHANDLE{h}{e}})$}
      \AxiomC{}%\PREMISEWHICHCOULDBEPROVEN{$\rho = \EMPTYSET$ in reachable configurations}}
      \UnaryInfC{$
        ({\DCALLWITHHOLE{h_0}{f}},\ \rho).S
        \ \ \VALUESEPARATOR c_{n} \VALUESEPARATOR c_{n-1} \VALUESEPARATOR ... \VALUESEPARATOR c_2 \VALUESEPARATOR c_1 \VALUESEPARATOR \ACTIVATIONSEPARATOR V\ \ \Gamma
        \TOE
        ({\DHANDLE{h}{e}},\ \rho[x_1 \mapsto c_1, x_2 \mapsto c_2, ..., x_{n-1} \mapsto c_{n-1}, x_n \mapsto c_n]).S
        \ \ \VALUESEPARATOR V\ \Gamma
        $}
\end{prooftree}

\begin{prooftree}
      \LeftLabel{$[\CODE{primitive}_e]$}
      \RightLabel{}
      \AxiomC{}
      \UnaryInfC{$
        (\DPRIMITIVE{h_0}{\pi}{{\DHANDLE{h_1}{e}}...{\DHANDLE{h_n}{e}}},\ \rho).S
        \ \ \VALUESEPARATOR V\ \ \Gamma
        \TOE
        ({\DHANDLE{h_1}{e}},\ \rho)...({\DHANDLE{h_n}{e}},\ \rho).(\DPRIMITIVEWITHHOLE{h_0}{\pi},\ \EMPTYSET).S
        \ \ \VALUESEPARATOR \ACTIVATIONSEPARATOR V\ \ \Gamma
        $}
\end{prooftree}

\begin{prooftree}
      \LeftLabel{$[\CODE{primitive}_c]$}
      \RightLabel{%$n,m \geq 0$,
        $\Gamma_{\CODE{primitives}}(\pi) (c_{1}, ..., c_{n}, \Gamma) = \SEQUENCE{c'_{1}, ..., c'_{m}, \Gamma'}$}
      \AxiomC{}
      \UnaryInfC{$
        (\DPRIMITIVEWITHHOLE{h_0}{\pi},\ \rho).S
        \ \ \VALUESEPARATOR c_{n} \VALUESEPARATOR c_{n-1} \VALUESEPARATOR ... \VALUESEPARATOR c_2 \VALUESEPARATOR c_1 \VALUESEPARATOR \ACTIVATIONSEPARATOR V\ \ \Gamma
        \TOE
        S
        \ \ \VALUESEPARATOR c'_{m} c'_{m-1} ... c'_2 c'_{1} \VALUESEPARATOR V\ \ \Gamma'
        $}
\end{prooftree}

\begin{prooftree}
      \LeftLabel{$[\CODE{if}_e]$}
      \RightLabel{}
      \AxiomC{}
      \UnaryInfC{$
        (\DIFIN{h_0}{\DHANDLE{h_1}{e}}{c_1...c_n}{\DHANDLE{h_2}{e}}{\DHANDLE{h_3}{e}},\ \rho).S
        \ \ \VALUESEPARATOR V\ \ \Gamma
        \TOE
        (\DHANDLE{h_1}{e},\ \rho).(\DIFIN{h_0}{\HOLE}{c_1...c_n}{\DHANDLE{h_2}{e}}{\DHANDLE{h_3}{e}},\ \rho).S
        \ \ \VALUESEPARATOR V\ \ \Gamma
        $}
\end{prooftree}

\begin{prooftree}
      \LeftLabel{$[\CODE{if}_{c}^{\in}]$}
      \RightLabel{$c \in \{c_1...c_n\}$}
      \AxiomC{}
      \UnaryInfC{$
        (\DIFIN{h_0}{\HOLE}{c_1...c_n}{\DHANDLE{h_2}{e}}{\DHANDLE{h_3}{e}},\ \rho).S
        \ \ \VALUESEPARATOR c \VALUESEPARATOR V\ \ \Gamma
        \TOE
        ({\DHANDLE{h_2}{e}},\ \rho).S
        \ \ \VALUESEPARATOR V\ \ \Gamma
        $}
\end{prooftree}
\begin{prooftree}
      \LeftLabel{$[\CODE{if}_{c}^{\notin}]$}
      \RightLabel{$c \notin \{c_1...c_n\}$}
      \AxiomC{}
      \UnaryInfC{$
        (\DIFIN{h_0}{\HOLE}{c_1...c_n}{\DHANDLE{h_2}{e}}{\DHANDLE{h_3}{e}},\ \rho).S
        \ \ \VALUESEPARATOR c \VALUESEPARATOR V\ \ \Gamma
        \TOE
        ({\DHANDLE{h_3}{e}},\ \rho).S
        \ \ \VALUESEPARATOR V\ \ \Gamma
        $}
\end{prooftree}

\begin{prooftree}
      \LeftLabel{$[\CODE{bundle}_e]$}
      \RightLabel{}
      \AxiomC{}
      \UnaryInfC{$
        (\DBUNDLE{h_0}{{\DHANDLE{h_1}{e}}...{\DHANDLE{h_n}{e}}},\ \rho).S
        \ \ \VALUESEPARATOR V\ \ \Gamma
        \TOE
        ({\DHANDLE{h_1}{e}},\ \rho)...({\DHANDLE{h_n}{e}},\ \rho).(\DBUNDLEWITHHOLE{h_0},\ \EMPTYSET).S
        \ \ \VALUESEPARATOR \ACTIVATIONSEPARATOR V\ \ \Gamma
        $}
\end{prooftree}

\begin{prooftree}
      \LeftLabel{$[\CODE{bundle}_c]$}
      \RightLabel{}
      \AxiomC{\PREMISEWHICHCOULDBEPROVEN{{$\rho = \EMPTYSET$ in reachable configurations}}}
      \UnaryInfC{$
        (\DBUNDLEWITHHOLE{h_0},\ \rho).S
        \ \ \VALUESEPARATOR c_{n} \VALUESEPARATOR c_{n-1} \VALUESEPARATOR ... \VALUESEPARATOR c_{2} \VALUESEPARATOR c_{1} \VALUESEPARATOR \ACTIVATIONSEPARATOR V\ \ \Gamma
        \TOE
        S
        \ \ \VALUESEPARATOR c_{n} c_{n-1} ... c_{2} c_{1} \VALUESEPARATOR V\ \ \Gamma
        $}
\end{prooftree}

\begin{prooftree}
      \LeftLabel{$[\CODE{fork}_e]$}
      \RightLabel{}
      \AxiomC{}
      \UnaryInfC{$
        (\DFORK{h_0}{f}{{\DHANDLE{h_1}{e}}...{\DHANDLE{h_n}{e}}},\ \rho).S
        \ \ \VALUESEPARATOR V\ \ \Gamma
        \TOE
        ({\DHANDLE{h_1}{e}},\ \rho)...({\DHANDLE{h_n}{e}},\ \rho).(\DFORKWITHHOLE{h_0}{f},\ \EMPTYSET).S
        \ \ \VALUESEPARATOR \ACTIVATIONSEPARATOR V\ \ \Gamma
        $}
\end{prooftree}

\begin{prooftree}
      \LeftLabel{$[\CODE{fork}_c]$}
      \RightLabel{\small ${\text{fresh }t\text{, }}{\Gamma_{\CODE{procedures}}:f \mapsto ({x_0...x_n}, {\DHANDLE{h}{e}})}$}
%      \AxiomC{\PREMISEWHICHCOULDBEPROVEN{$\rho = \EMPTYSET$ in reachable configurations}}
      \AxiomC{}
      \UnaryInfC{$
        (\DFORKWITHHOLE{h_0}{f},\ \rho).S
        \ \ \VALUESEPARATOR c_{n} \VALUESEPARATOR c_{n-1} \VALUESEPARATOR ... \VALUESEPARATOR c_{2} \VALUESEPARATOR c_{1} \VALUESEPARATOR \ACTIVATIONSEPARATOR V\ \ \Gamma
        \TOE
        S
        \ \ \VALUESEPARATOR \FUTURE{t} \VALUESEPARATOR V
        \ \ \Gamma[{}_{{\CODE{futures}}}^{t \mapsto (({\DHANDLE{h}{e}},\ \rho[x_0 \mapsto \FUTURE{t}, x_1 \mapsto c_1, ..., x_n \mapsto c_n])\ \VALUESEPARATOR)}])
        $}
\end{prooftree}

\begin{prooftree}
      \LeftLabel{$[\CODE{join}_e]$}
      \RightLabel{}
      \AxiomC{}
      \UnaryInfC{$
        (\DJOIN{h_0}{\DHANDLE{h_1}{e}},\ \rho).S
        \ \ \VALUESEPARATOR V\ \ \Gamma
        \TOE
        (\DHANDLE{h_1}{e},\ \rho).(\DJOIN{h_0}{\HOLE},\ \rho).S
        \ \ \VALUESEPARATOR V\ \ \Gamma
        $}
\end{prooftree}
\begin{prooftree}
      \LeftLabel{$[\CODE{join}_c]$}
      \RightLabel{$\Gamma_{\CODE{futures}} : t \mapsto (\EMPTYSEQUENCE,\ \VALUESEPARATOR c_t \VALUESEPARATOR)$}
      \AxiomC{}
      \UnaryInfC{$
        (\DJOIN{h_0}{\HOLE},\ \rho).S
        \ \ \VALUESEPARATOR \FUTURE{t} \VALUESEPARATOR V\ \ \Gamma
        \TOE
        S
        \ \ \VALUESEPARATOR c_t \VALUESEPARATOR V\ \ \Gamma
        $}
\end{prooftree}
\begin{prooftree}
      \LeftLabel{$[\parallel]$}
      \RightLabel{$\Gamma_{\CODE{futures}} : t \mapsto (S_t,\ V_t)$}
      \AxiomC{$
        S_t\ \ V_t\ \ \Gamma
        \TOE
        S_t'\ \ V_t'\ \ \Gamma'
        $}
      \UnaryInfC{$
        S\ \ V\ \ \Gamma
        \TOE
        S\ \ V\ \ \Gamma'[{}_{\CODE{futures}}^{t \mapsto (S_t',\ V_t')}]
        $}
\end{prooftree}
\label{epsilonzero-semantics-rules--end}\INVISIBLE{[hack! (label-target)]}
\thispagestyle{unobstrusive} % We overflow into header here
%\fancyfoot{\bfseries\thepage}
%\fancyhead{}
\end{landscape}
%\fancyfoot{} % Empty footer from now on

It is easy to classify rules into four sets according to the holed expression
case in the top stack pair, if any.  We have:
\begin{itemize}
\item the \TDEF{basic rules} $[constant]$ and $[variable]$;
\item \TDEF{expansive rules}, one per non-holed expression case, named after the form with an
  ``$e$'' subscript;
\item \TDEF{contractive rules}, one per holed expression case (except for the
  conditional, which needs two contractive rules), named after the form with a
  ``$c$'' subscript;
\item the \TDEF{parallel rule} $[{\parallel}]$ in the end, standing apart from all the others.
\end{itemize}
The core ideas of the evaluation are simple, and strongly rooted on the inductive
nature of the expression syntax.  Rule groups help to highlight the quite
pleasant symmetry of the system:
\begin{itemize}
\item
\emph{base}: we evaluate a ``basic'' expression found on the top of the stack by
popping it and pushing a corresponding value onto the value stack;
\item
\emph{expansion}: before we can evaluate a non-basic expression on the top of the stack we need to
evaluate its sub-expressions: so we replace the expression with its holed
counterpart, and push its subexpressions on the stack on top of it, in an
order such that the \textit{first} one to be evaluated end %subjunctive
up \textit{on the top};
\item
\emph{contraction}: if a holed expression is on the top of the stack, this
means that we have just finished evaluating its subexpressions: pop their
values from the value stack, pop the holed expression from the stack, and
proceed: according to the case this can mean pushing a further subexpression
onto the stack or pushing results onto the value stack;
\item
\emph{parallelism}: the parallel rule lets us concurrently perform a reduction in a background thread,
whenever possible.
\end{itemize}
The \textit{LIFO} policy outlined above enforces a rigid \textit{call-by-value, depth-first
left-to-right evaluation strategy}.  We find that having such a simple and
predictable evaluation order is very useful for both programming and reasoning
about programs.
\\
\\
\label{semantics-commentary}$[constant]$ is trivial.

In $[variable]$ it should be noted how the (topmost) local environment
prevails over the current global environment in the variable rules.  Of course
the rule cannot fire if the variable is unbound.

$[\CODE{let}_e]$ is simple enough: a \LETNAME block is
evaluated by first pushing the \LETNAME-bound expression $\DHANDLE{h_1}{e}$;
when such evaluation eventually ends producing a bundle in the value stack the
\textit{\LETNAME contractive rule} can fire, assuming the bundle dimension is
sufficient: then the holed \LETNAME expression is replaced by the \LETNAME body
on the stack, with an updated environment in which the first $n$ bundle
components are named, and all $m$ of them are popped off the value stack, 
implementing the behavior described in \SECTION\ref{let-blocks-can-ignore-some-values}.
It should
be remarked that the holed \LETNAME expression ``disappears from
the stack'' as soon as its body is pushed.  This behavior is useful for
potentially \textit{tail-position subexpressions}:
%(see \SECTION\ref{tail-contexts}):
after we reduce a \LETNAME block to its
body, the \LETNAME block itself can be disposed of, saving stack space.

\label{calling-a-nonexisting-procedure}
$[\CODE{call}_e]$\footnote{\RED{Trivial error fixed in 2014: the original text mistakenly
  said ``$[\CODE{let}_c]$'' instead of ``$[\CODE{call}_e]$''.}} just consists in replacing the \CALLNAME
expression on the top with its holed counterpart (with an immaterial local
environment), and pushing actuals on top of it, so that they will be evaluated
starting from the leftmost one, all in the same local environment of the
call. When actuals are evaluated the \textit{\CALLNAME contractive rule} has
the opportunity to fire, provided that the value stack contains a topmost
activation with exactly as many 1-dimension bundles as the (current!) number of
parameters of the called procedure, and of course provided that a procedure
with the appropriate name exists.  If that is the case the holed \CALLNAME
expression is replaced with the procedure body, and the local environment with
an environment containing \textit{only} the parameter bindings. It is crucial
here \textit{not} to extend the call-time local environment, since we want to
prevent nonlocal visibility, for efficiency reasons.  In a similar vein to the
\LETNAME case, a tail-position holed \CALLNAME is \textit{replaced} by the
called procedure body.

$[\CODE{primitive}_e]$ and $[\CODE{primitive}_c]$ are very similar to
their \CALLNAME counterparts; the contractive rule cannot fire if the primitive
name is not bound, or the primitive function is undefined. Notice that the
primitive function is allowed to return a new global state, and the contractive
rule effectively establishes it for the resulting configuration.

$[\CODE{if}_e]$ simply replaces the topmost expression with
a holed conditional, pushing the discriminand subexpression on top of it; when
the discriminand is completely evaluated, either of the two \textit{\IFNAME
  contractive rules} $[\CODE{if}_c^{\in}]$ and $[\CODE{if}_c^{\notin}]$ may fire, provided the discriminand yielded a 1-dimension
bundle: if the value belongs to the conditional case set, the \THENNAME
subexpression replaces the holed \IFNAME; otherwise, the \ELSENAME
subexpression does.  The conditional expression is replaced by one of the
branch subexpressions without consuming stack space, which is useful in tail
contexts.

$[\CODE{bundle}_e]$ resembles $[\CODE{call}_e]$ and
$[\CODE{primitive}_e]$; again the empty environment associated to the
\BUNDLENAME holed expression is immaterial.  $[\CODE{bundle}_c]$, if
the correct number of 1-dimension bundles is on the top of the value stack,
replaces them all with a
single bundle holding all the values.

$[\CODE{fork}_e]$ is essentially identical to
$[\CODE{call}_e]$, $[\CODE{primitive}_e]$ and $[\CODE{bundle}_e]$.
$[\CODE{fork}_c]$ is more interesting: if the actual parameter result bundles
are 1-dimensioned and correct in number, they are simply replaced by one
\textit{future} on the value stack, and the \FORKNAME evaluation terminates
immediately: the concurrent evaluation will take place asynchronously in a new
thread created for the purpose, and associated to the future identifier in the
future state environment.  Notice that the thread identifier is also visible to
the new thread as the zeroth parameter, to be used in personalities for
``\CODE{self-thread-name}'' forms or thread-local variables.

$[\CODE{join}_e]$ replaces the topmost expression with its
holed counterpart, pushing its future expression over it;
$[\CODE{join}_c]$, provided that a 1-dimension bundle is on the top of the
stack, that the bundle contains a future value \textit{and the thread corresponding
  to the future terminated}, returns the result from the thread.
$[\CODE{join}_c]$ cannot fire until the asynchronous thread has terminated.

$[\parallel]$, provided that that a configuration obtained by
making a background thread the foreground thread \textit{could} reduce, allows
to perform the reduction ``concurrently'', in the future state environment.
\\
\\
It is easy to see that in the contractive rules
$[\CODE{call}_c]$, $[\CODE{primitive}_c]$, $[\CODE{bundle}_c]$ and
$[\CODE{fork}_c]$ the local environment associated to the holed
expression will in practice always be empty, for reachable configurations.
\\
\\
The role of ``$\VALUESEPARATOR$'' value separators should be clear at this
point: the values of the same bundle are stored on the value stack
sequentially, \textit{without} any separators in between --- and again in
reverse order, because of the LIFO strategy.  Separators help establish the
correct conditions for a rule to fire, so that \textit{no bundle of the wrong
  dimension can be used}.

The motivation for activation separators ``$\ACTIVATIONSEPARATOR$'' is similar but
slightly more subtle: the problem is being able to distinguish a local, temporary
bundle from a surrounding bundle which is being built on the value stack.
Without such explicit markers it would be possible to pop the ``wrong'' number of
values from the value stack.

\label{changing-arity-on-the-go}
Moreover a procedure can be \textit{redefined}, or even defined for the first
time, by one of its actual parameters.  We only define semantics if the number
of the passed parameters is correct, but their good number cannot be determined
until\footnote{\RED{2014 correction: the original text said ``after'' instead
of ``until''.}}
all of them have been evaluated: hence, before letting a contractive rule
fire, we have to check that the topmost objects in the value stack be
\textit{all and only} the actual values.
\\
\\
At this point it may be worth to remind the reader of Implementation
Note~\ref{implementation-guarantees-implementation-note}: markers do not necessarily
need to be represented and checked for at run time in an efficient
implementation; quite the opposite, by specifying that some case yields an
error we free ourselves from any implementation constraint.

For this reason we intentionally let, for example, ``wrong arity'' be an error
condition
(\SECTION\ref{epsilonzero-failure})
instead of specifying some ``fallback behaviour'' such as ignoring extra
arguments or providing defaults for missing ones: in practice an efficient
implementation will need to reserve stack frame slots or registers for return
addresses, garbage collection structures or for some other implementation
bookkeeping purpose: of course passing the wrong number of parameters will
likely interfere with these conventions.  We do not want this to be made more
difficult or less efficient just because of the need of implementing a specific
behavior, whose utility was dubious in the first place.
\\
\\
Unfortunately an implementation cannot let configurations grow to an arbitrary complexity:

\begin{implementationnote}[dynamic execution resource limits]
\label{dynamic-resource-limit-implementation-note}
Each instance of the following items occupies some \TDEF{memory} in an
implementation (See Implementation~Note~\ref{syntactic-resource-limit-implementation-note}):
\begin{itemize}
\item a stack item;
\item a value stack item;
\item a local environment binding.
\end{itemize}
Some implementations may further limit the stack item and value stack item
number to another smaller constant, independently from memory usage.
\QEDIMPLEMENTATIONNOTE
\end{implementationnote}

\subsection{Sequential reduction}
\label{sequential-reduction}
As we have just remarked in \SECTION\ref{semantics-commentary}, we
value the predictability of \EPSILONZERO semantics, with its
well-specified evaluation strategy.  In the same vein
\textit{determinism} in an evaluation relation is a desirable
property.

It is easy to observe that, save for the parallel rule, the reduction
relation \textit{is} in fact trivially deterministic, up to the (immaterial) choice of
thread identifiers.

\begin{definition}[sequential small-step reduction]
\label{epsilonzero-sequential-semantics-definition}
\label{epsilonzero-sequential-semantics}
We define the \TDEF{sequential small-step evaluation}
relation
$\PARAMETER \PARAMETER \PARAMETER \TOES \PARAMETER \PARAMETER \PARAMETER
\subseteq
(\SET{S} \times \SET{V} \times \SET{\Gamma}) \times (\SET{S} \times \SET{V} \times \SET{\Gamma})$
according to the rules on
pp.~\pageref{epsilonzero-semantics-rules--beginning}-\pageref{epsilonzero-semantics-rules--end},
minus the parallel rule.
\QEDDEFINITION
\end{definition}
Interestingly, a sequential reduction can still work with \FORKNAME and
\JOINNAME, and futures can be passed around and even created anew or joined if
their result is ready: since the only source of non-determinism is the actual
concurrent reduction, as long as no background thread ``advances'' it is
possible to work with futures using only $\PARAMETER \PARAMETER \PARAMETER
\TOES \PARAMETER \PARAMETER \PARAMETER$.

\subsection{Failure}
\label{epsilonzero-failure}
In \SECTION\ref{semantics-commentary} we have explicitly shown that there are
cases in which the small-step semantics is undefined because rule premises
cannot be satisfied.
After having formalized the notion of ``correct reduction'', here we are
going to exactly specify and classify failure conditions.

%\TODO{Formally define $a \UNIFIES b$ and $a \DOESNTUNIFY b$, and explain when, in practice, we have that $a \DOESNTUNIFY b$.}

\begin{definition}[error configurations]\label{error-configuration-definition}\label{failure-definition}
We define the error configuration relations
\TDEF{fails because of environments}, written as ``$\PARAMETER\ \PARAMETER\ \PARAMETER \FAILSBECAUSEOFENVIRONMENTS$''
and
\TDEF{fails because of dimension}, written as
``$\PARAMETER\ \PARAMETER\ \PARAMETER \FAILSBECAUSEOFDIMENSION$'',
%\TDEF{fails because of another thread}, written as
%``$\PARAMETER\ \PARAMETER\ \PARAMETER \FAILSBECAUSEOFANOTHERTHREAD$'',
%and the generic
%\TDEF{fails}, written as ``$\PARAMETER\ \PARAMETER\ \PARAMETER \FAILS$'',
all subsets of the set of configurations
$\SET{S} \times \SET{V} \times \SET{\Gamma}$,
by the following rules:

{\rm % upright: begin
\begin{prooftree}
      \LeftLabel{}
      \RightLabel{\small $x \notin dom(\Gamma_{\CODE{global-environment}}[\rho])$}
      \AxiomC{}
      \UnaryInfC{$
        (\DVARIABLE{h}{x},\ \rho).S
        \ \VALUESEPARATOR V\ \Gamma
        \FAILSBECAUSEOFENVIRONMENTS
        $}
\end{prooftree}

\begin{prooftree}
      \LeftLabel{}
      \RightLabel{\small $\nexists m \SUCHTHAT
        (m \geq{}n
        \wedge
        V \UNIFIES \ \VALUESEPARATOR c_{m}c_{m-1}...c_{2}c_{1}\VALUESEPARATOR V')$}
      \AxiomC{}
      \UnaryInfC{$
        (\DLET{h_0}{x_1...x_n}{\HOLE}{\DHANDLE{h_2}{e}},\ \rho).S
        \ V\ \Gamma
        \FAILSBECAUSEOFDIMENSION
        $}
\end{prooftree}

\begin{prooftree}
      \LeftLabel{}
      \RightLabel{\small
        $\neg
        (\Gamma_{\CODE{procedures}}:{f} \mapsto ({x_1...x_n}, {\DHANDLE{h}{e}})
        \wedge
        V \UNIFIES \ \VALUESEPARATOR c_{n} \VALUESEPARATOR c_{n-1} \VALUESEPARATOR ... \VALUESEPARATOR c_2 \VALUESEPARATOR c_1 \VALUESEPARATOR \ACTIVATIONSEPARATOR V')$}
      \AxiomC{\PREMISEWHICHCOULDBEPROVEN{$\rho = \EMPTYSET$ in reachable configurations}}
      \UnaryInfC{$
        ({\DCALLWITHHOLE{h_0}{f}},\ \rho).S
        \ V\ \Gamma
        \FAILSBECAUSEOFDIMENSION
        $}
\end{prooftree}

\begin{prooftree}
      \LeftLabel{}
      \RightLabel{\small $\Gamma_{\CODE{primitives}}(\pi) \HASDIMENSION n \to m
        \wedge
        V \DOESNTUNIFY \VALUESEPARATOR c_{n} \VALUESEPARATOR c_{n-1} \VALUESEPARATOR ... \VALUESEPARATOR c_2 \VALUESEPARATOR c_1 \VALUESEPARATOR \ACTIVATIONSEPARATOR V'$}
      \AxiomC{}
      \UnaryInfC{$
        (\DPRIMITIVE{h_0}{\pi}{\HOLE},\ \rho).S
        \ V\ \Gamma
        \FAILSBECAUSEOFDIMENSION
        $}
\end{prooftree}

\begin{prooftree}
      \LeftLabel{}
      \RightLabel{\small $V \DOESNTUNIFY \ \VALUESEPARATOR c \VALUESEPARATOR V'$}
      \AxiomC{}
      \UnaryInfC{$
        (\DIFIN{h_0}{\HOLE}{c_1...c_n}{\DHANDLE{h_2}{e}}{\DHANDLE{h_3}{e}},\ \rho).S
        \ V\ \Gamma
        \FAILSBECAUSEOFDIMENSION
        $}
\end{prooftree}

\begin{prooftree}
      \LeftLabel{}
      \RightLabel{\small $V \DOESNTUNIFY \VALUESEPARATOR c_1 \VALUESEPARATOR c_2 \VALUESEPARATOR ... \VALUESEPARATOR c_{n-1} \VALUESEPARATOR c_n \VALUESEPARATOR \ACTIVATIONSEPARATOR V'$}
      \AxiomC{\PREMISEWHICHCOULDBEPROVEN{$\rho = \EMPTYSET$ in reachable configurations}}
      \UnaryInfC{$
        (\DBUNDLEWITHHOLE{h_0},\ \rho).S
        \ V\ \Gamma
        \FAILSBECAUSEOFDIMENSION
        $}
\end{prooftree}

\begin{prooftree}
      \LeftLabel{}
      \RightLabel{\small $\neg
        (\Gamma_{\CODE{procedures}}:{f} \mapsto ({x_0...x_n}, {\DHANDLE{h}{e}})
        \wedge
        V \UNIFIES \ \VALUESEPARATOR c_{n} \VALUESEPARATOR c_{n-1} \VALUESEPARATOR ... \VALUESEPARATOR c_2 \VALUESEPARATOR c_1 \VALUESEPARATOR \ACTIVATIONSEPARATOR V')$}
      \AxiomC{\PREMISEWHICHCOULDBEPROVEN{$\rho = \EMPTYSET$ in reachable configurations}}
      \UnaryInfC{$
        (\DFORKWITHHOLE{h_0}{f},\ \rho).S
        \ V\ \Gamma
        \FAILSBECAUSEOFDIMENSION
        $}
\end{prooftree}

\begin{prooftree}
      \LeftLabel{}
      \RightLabel{\small $V \DOESNTUNIFY \VALUESEPARATOR c \VALUESEPARATOR V'$}
      \AxiomC{}
      \UnaryInfC{$
        (\DJOIN{h_0}{\HOLE},\ \rho).S
        \ V\ \Gamma
        \FAILSBECAUSEOFDIMENSION
        $}
\end{prooftree}} % upright: end
The \TDEF{fails because of a primitive} relation
$\PARAMETER\ \PARAMETER\ \PARAMETER \FAILSBECAUSEOFPRIMITIVE \subseteq \SET{S} \times \SET{V} \times \SET{\Gamma}$
is a superset of the relation defined by the following rule:
{\rm % upright: begin
\begin{prooftree}
      \LeftLabel{}
      \RightLabel{\small $c \DOESNTUNIFY \FUTURE{t}$}
      \AxiomC{}
      \UnaryInfC{$
        (\DJOIN{h_0}{\HOLE},\ \rho).S
        \ \VALUESEPARATOR c \VALUESEPARATOR V\ \Gamma
        \FAILSBECAUSEOFPRIMITIVE
        $}
\end{prooftree}} % upright: end
The exact definition of $\PARAMETER\ \PARAMETER\ \PARAMETER \FAILSBECAUSEOFPRIMITIVE$
relies on the specific set of available primitives, which we intentionally leave open.
\\
\\
We define \TDEF{the generic ``fails'' relation}, written as ``$\PARAMETER\ \PARAMETER\ \PARAMETER \FAILS$'',
as the union of the specific failure relations:
$(\PARAMETER \PARAMETER \PARAMETER \FAILS) =
(\PARAMETER \PARAMETER \PARAMETER \FAILSBECAUSEOFENVIRONMENTS) \cup
(\PARAMETER \PARAMETER \PARAMETER \FAILSBECAUSEOFPRIMITIVE) \cup
(\PARAMETER \PARAMETER \PARAMETER \FAILSBECAUSEOFDIMENSION)
%\cup (\PARAMETER \PARAMETER \PARAMETER \FAILSBECAUSEOFANOTHERTHREAD)
$.  We also call \TDEF{final failure configuration} a configuration
that fails.  A \TDEF{final configuration} is either a final success
configuration or a final failure configuration.
\QEDDEFINITION
\end{definition}

Since the definition above does not mention background threads at all
we have that failure in a background thread does \textit{not}
propagate to any other thread.  We chose this solution in the interest
of simplicity and realism for a core language such as \EPSILONZERO,
which should reflect the behavior of system-level facilities.  Of
course higher-level personalities are free to implement more complex
policies, as hinted at in \SECTION\ref{failure-recovery}.

Since failure never propagates to other threads in \EPSILONZERO, there
is no need for alternate ``sequential'' relations for failure.
\\
\\
As a further point of note, above we have chosen to classify the failure of \JOINNAME{}-ing an object
different from a future as a primitive error, because of the strong analogy of
the condition with a primitive ``wrong parameter'' error\footnote{It is easy at
this point to mistake such an error for a type error. The difference is
actually subtle, and will be dealt with in \SECTION\ref{the-nature-of-values}}.

Many actual primitives also fail for some values of their parameters, even when
they receive the correct number of them.
For example a division primitive ``$\div$'' might fail on a zero divisor;
writing ``$\mathcal{N}(0)$'' for zero as in \SECTION\ref{where-we-hint-at-the-value-notation}, we get:
\begin{prooftree}
      \LeftLabel{}
      \RightLabel{}
      \AxiomC{}
      \UnaryInfC{$
        (\DPRIMITIVE{h_0}{\div}{\HOLE},\ \rho).S
        \ \VALUESEPARATOR \mathcal{N}(0) \VALUESEPARATOR c \VALUESEPARATOR \ACTIVATIONSEPARATOR V\ \Gamma
        \FAILSBECAUSEOFPRIMITIVE
        $}
\end{prooftree}
We intentionally omit a list of all the specific cases of primitive failure, a
complete specification belonging in the primitive definition --- with the only
% Christophe objected to this, but I don't see the problem
constraint of having failure rules covering \textit{all} possible failure
cases; in other words, given a set of parameters a primitive either fails or
% About divergence, Christophe wrote (\Omega \Omega), meaning self-application
% if I interpret correctly.  But we have no lambdas, and we're speaking of primitives.
returns a result {but no other behavior such as divergence is possible}, as
specified by Axiom~\ref{primitive-totality-axiom}, which we are now ready to state:

\label{where-primitive-totality-axiom-is-defined}
\begin{axiom}[primitive ``totality'']\label{primitive-totality-axiom}
For any primitive $\pi$ such that $\Gamma_{\CODE{primitives}}(\pi) \HASDIMENSION n \to m$ in some \STATE $\Gamma$ and for each sequence
$\SEQUENCE{c_1,...,c_n}$,
exactly one of the following holds:
\begin{itemize}
\item
  there exist $c_1',...,c_m',\Gamma'$ such that
$\Gamma_{\CODE{primitives}}(\pi) (c_{1}, ..., c_{n}, \Gamma) = \SEQUENCE{c'_{1}, ..., c'_{m}, \Gamma'}$;
%  $\pi(c_1,...,c_n, \Gamma) = \SEQUENCE{c_1',...,c_m', \Gamma'}$;
\item
  for any $h_0, S, V$ we have 
  {\rm
    $(\DPRIMITIVE{h_0}{\pi}{\HOLE},\ \rho).S
    \ \VALUESEPARATOR c_1 \VALUESEPARATOR ... \VALUESEPARATOR c_n \VALUESEPARATOR \ACTIVATIONSEPARATOR V\ \Gamma
    \FAILSBECAUSEOFPRIMITIVE$}. % \rm
  \QEDAXIOM
\end{itemize}
\end{axiom}

We can prove a result in the same spirit for general expressions: any given
configuration either can be reduced for at least one more step, or it
immediately fails; but it is not possible that a non-failing configuration does
not allow reductions (unless joining a future), or that a configuration
simultaneously fails and allows a sequential reduction:

%% %% Of course Implementation
%% %% Note~\ref{implementation-guarantees-implementation-note} does not \textit{prevent} an
%% %% implementation from failing gracefully and showing readable error
%% %% messages, at some cost in efficiency; and again, it is reasonable to expect
%% %% such a ``forgiving'' implementation to also exist, for aiding debugging.

%%%%%%%%%%%%%%%%%

\begin{proposition}[reduce xor fail xor wait]\label{reduce-xor-fail-xor-wait}
Given any reachable configuration $\chi = (\DHANDLE{h}{e}, \rho).S\ V\ \Gamma$
%%  In a sequential program for any \EPSILONZEROHOLE\ expression $\DHANDLE{h}{e}$, any stack $S$, any value stack $V$,
%% any environment $\rho$ and any \STATE $\Gamma$,
we have exactly one of the following:
\begin{itemize}
\item
there exist $S'$, $V'$ and $\Gamma'$ such that
$(\DHANDLE{h}{e},\ \rho).S\ V\ \Gamma  \TOES  S'\ V'\ \Gamma'$;
\item
$(\DHANDLE{h}{e},\ \rho).S\ V\ \Gamma  \FAILS$;
\item
there exist $t \in \SET{T}$, $V' \in \SET{V}$
such that
{\rm
$
\chi =
(\DJOIN{h}{\HOLE},\ \rho).S\ \VALUESEPARATOR \FUTURE{t} \VALUESEPARATOR V'\ \Gamma
$}.
\end{itemize}
\end{proposition}
\begin{proof}[sketch]
Since the main stack is not empty, $\chi$ is not terminal.

We are dealing with
$\PARAMETER \PARAMETER \PARAMETER \TOES \PARAMETER \PARAMETER \PARAMETER$
rather than
$\PARAMETER \PARAMETER \PARAMETER \TOE \PARAMETER \PARAMETER \PARAMETER$,
hence $[\parallel]$ cannot fire, by Definition~\ref{epsilonzero-semantics-definition}.

In any configuration where the top expression is an \EPSILONZERO\ non-holed expression
except for the variable, we may apply an expansive rule or $[constant]$ for
$\PARAMETER \PARAMETER \PARAMETER \TOES \PARAMETER \PARAMETER \PARAMETER$,
leading to an evaluation step: all such rules can always fire independently
of subexpressions, the state of the environments or the stacks
(Definition~\ref{epsilonzero-semantics-definition}).

If $\chi = (\DJOIN{h}{\HOLE},\ \rho).S\ \VALUESEPARATOR \FUTURE{t} \VALUESEPARATOR V'\ \Gamma$
for some $t, V'$ then the thesis follows trivially.

In the remaining cases the top expression is a variable or a holed expression,
and another disjoint set
of rules applies.  
Then one of the contractive rules for
$\PARAMETER \PARAMETER \PARAMETER \TOES \PARAMETER \PARAMETER \PARAMETER$,
or an error rule for
$\PARAMETER \PARAMETER \PARAMETER \FAILSBECAUSEOFENVIRONMENTS$,
$\PARAMETER \PARAMETER \PARAMETER \FAILSBECAUSEOFDIMENSION$ and
$\PARAMETER \PARAMETER \PARAMETER \FAILSBECAUSEOFPRIMITIVE$
in Definition~\ref{failure-definition} must apply.

In all the cases above it is easy to see that each configuration matches the
premise of exactly one rule (in particular, since $\chi$ is reachable, the value stack must
have $\VALUESEPARATOR$ on top); the only nontrivial case is
$\DPRIMITIVE{h}{\pi}{\HOLE}$, in which either
$[\CODE{primitive}_c]$ or a $\PARAMETER \PARAMETER \PARAMETER \FAILSBECAUSEOFPRIMITIVE$
rule applies because of Axiom~\ref{primitive-totality-axiom}.
\QEDPROOF
\end{proof}

\subsection{Error recovery and personalities}
\label{error-recovery}
\label{failure-recovery}
At the level of \EPSILONZERO all errors are fatal.  In the interest of
simplicity and efficiency, no mechanism is provided for handling a case of
failure by recovering or retrying.  Any such machinery can be defined in
high-level personalities by checking for failure conditions at run time with
explicit conditional expressions to be automatically generated; in this way it
is possible to completely prevent \EPSILONZERO failures from ever occurring, if
so desired.

As usual and in the same spirit of typing, the personality implementor
has the freedom of choosing an efficient model where failures are
always fatal, or a friendlier alternative where the personality
presents an ordinary \EPSILONZERO configuration as an ``error'' state
from which the conventionally ``normal'' execution can resume.
\\
\\
Like for static typing it is also possible to check for the possibility of
failures ``statically'' at code generation time, and generate fast code under
the assumption that some kind of failure is impossible.  The next chapters
hint at how one can define such analyses.

%%%%%%%%%%%%%%%%%%%%%%%%%%%%%%%%%%%%%%%%%%%%%%%%%%%%%%%%%%%%%%%%%%%%%%%%%%%%%%%
\section{One-step dynamic semantics}
When dealing with ``toplevel'' \EPSILONZERO expressions, often we are
less interested in the small-step evaluation relation $\PARAMETER
\PARAMETER \PARAMETER \TOE \PARAMETER \PARAMETER \PARAMETER$ than in
its \textit{iteration}: where only the final configurations (if any)
are of interest it is convenient to completely ignore stacks and value
stacks, restricting our attention to an expression, an initial state,
its result bundle and the terminal state.

\begin{definition}[one-step convergence]\label{one-step-semantics-definition}
We define the \TDEF{one-step operational semantics} relation for expressions
$\PARAMETER \PARAMETER \CONVERGESE \PARAMETER \PARAMETER \subseteq (\SET{E} \times \SET{\Gamma}) \times (\SET{C}^{*} \times \SET{\Gamma})$ by the rule:
\begin{prooftree}
      \LeftLabel{}
      \RightLabel{}
      \AxiomC{$
        ({\DHANDLE{h}{e}}, \EMPTYSET)\ \VALUESEPARATOR\ \Gamma
        \TOET
        \EMPTYSEQUENCE\ \VALUESEPARATOR c_n ... c_1 \VALUESEPARATOR\ \Gamma'
        $}
      \UnaryInfC{$
        \DHANDLE{h}{e}\ \Gamma
        \CONVERGESE
        \SEQUENCE{c_1 ... c_n}\ \Gamma'
        $}
\end{prooftree}
Similarly, we define the \TDEF{one-step sequential operational semantics} relation for expressions
$\PARAMETER \PARAMETER \CONVERGESES \PARAMETER \PARAMETER \subseteq (\SET{E} \times \SET{\Gamma}) \times (\SET{C}^{*} \times \SET{\Gamma})$
by the rule:
\begin{prooftree}
      \LeftLabel{}
      \RightLabel{}
      \AxiomC{$
        ({\DHANDLE{h}{e}}, \EMPTYSET)\ \VALUESEPARATOR\ \Gamma
        \TOEST
        \EMPTYSEQUENCE\ \VALUESEPARATOR c_n ... c_1 \VALUESEPARATOR\ \Gamma'
        $}
      \UnaryInfC{$
        \DHANDLE{h}{e}\ \Gamma
        \CONVERGESES
        \SEQUENCE{c_1 ... c_n}\ \Gamma'
        $}
\end{prooftree}
When we have that ``$\DHANDLE{h}{e}\ \Gamma\CONVERGESE \SEQUENCE{c_1...c_n}\ \Gamma'$''
we say that
$\DHANDLE{h}{e}$ in $\Gamma$ \TDEF{converges to} $\SEQUENCE{c_1...c_n}$ in $\Gamma'$.
In the same way when we have ``$\DHANDLE{h}{e}\ \Gamma\CONVERGESES \SEQUENCE{c_1...c_n}\ \Gamma'$'' we say that
$\DHANDLE{h}{e}$ in $\Gamma$ \TDEF{sequentially converges} to $\SEQUENCE{c_1...c_n}$ in $\Gamma'$.

We may omit state names or results when irrelevant in context.
\QEDDEFINITION
\end{definition}
It trivially follows from the determinism of
$\PARAMETER \PARAMETER \PARAMETER \TOES \PARAMETER \PARAMETER \PARAMETER$
that
$\PARAMETER \PARAMETER \CONVERGESES \PARAMETER \PARAMETER$
is also a (partial) \textit{function}.

Notice that according to our definition
a reduction chain of $\PARAMETER \PARAMETER \PARAMETER \TOE \PARAMETER \PARAMETER \PARAMETER$
may converge even if some background thread potentially runs forever, when a finite reduction
chain \textit{exists} for $\PARAMETER \PARAMETER \PARAMETER \TOES \PARAMETER \PARAMETER \PARAMETER$,
and hence also for its super-relation\footnote{\RED{Error fixed in 2014: the original
  text mistakenly said ``sub-relation'' instead of ``super-relation''.}}
\mbox{$\PARAMETER \PARAMETER \PARAMETER \TOE \PARAMETER \PARAMETER \PARAMETER$.}
\\
\\
It is also useful to speak of the \TDEF{eventual} failure of an expression
in a state, ignoring the zero or more reduction steps leading to the
failure configuration, and the specific failure configuration as well:
\begin{definition}[eventual failure]\label{eventual-failure-definition}
For each error-configuration relation
$\PARAMETER \PARAMETER \PARAMETER \FAILSBECAUSEOF{f}$,
with 
$f \in \{\text{``{}''}, \text{``$\SET{X}$''}, \text{``$\SET{P}$''}, \text{``$\#$''}\}$,
we define a corresponding \TDEF{eventual failure relation}
$\PARAMETER \PARAMETER \EVENTUALLYFAILSBECAUSEOF{f}
\subseteq
\SET{E} \times \SET{\Gamma}
$ by the meta-rule:
\begin{prooftree}
      \LeftLabel{}
      \RightLabel{}
      \AxiomC{$
        ({\DHANDLE{h}{e}}, \EMPTYSET)\ \VALUESEPARATOR\ \Gamma
        \TOERT
        S\ V\ \Gamma'
        $}
      \AxiomC{$
        S\ V\ \Gamma'
        \FAILSBECAUSEOF{f}
        $}
      \BinaryInfC{$
        \DHANDLE{h}{e}\ \Gamma
        \EVENTUALLYFAILSBECAUSEOF{f}
        $}
\end{prooftree}
If we have that ``$\DHANDLE{h}{e}\ \Gamma \EVENTUALLYFAILSBECAUSEOF{f}$''
with
$f \in \{\text{``{}''}, \text{``$\SET{X}$''}, \text{``$\SET{P}$''}, \text{``$\#$''}\}$
we respectively say that $\DHANDLE{h}{e}$ in $\Gamma$
\TDEF{eventually fails},
\TDEF{eventually fails because of environments},
\TDEF{eventually fails because of primitives},
or
\TDEF{eventually fails because of dimension}.
\QEDDEFINITION
\end{definition}
Finally, we characterize \TDEF{looping} expressions and states:
\begin{definition}[divergence]\label{divergence-definition}
We define the \TDEF{divergence relation} 
$\PARAMETER \PARAMETER \DIVERGESE\ \subseteq\ \SET{E} \times \SET{\Gamma}$ the following way:

let $\DHANDLE{h}{e}$ be an expression and $\Gamma$ be a state; then
we say that \TDEF{$\DHANDLE{h}{e}$ diverges in $\Gamma$} and we write
``$\DHANDLE{h}{e}\ \Gamma \DIVERGESE$''
if for any configuration $\chi$ such that
$(\DHANDLE{h}{e}, \EMPTYSET)\ \VALUESEPARATOR\ \Gamma \TOET \chi$
there exists another configuration $\chi'$ such that $\chi \TOE \chi'$.
%\begin{prooftree}
%\end{prooftree}
%\TODOQ{[This definition is about looping; shall I include waiting on a future forever as well?]}
\QEDDEFINITION
\end{definition}
We defined divergence with the parallel reduction relation
$\PARAMETER \PARAMETER \PARAMETER \TOE \PARAMETER \PARAMETER \PARAMETER$ rather
than its sequential restriction; hence our notion of divergence covers both
``busy looping'' in the foreground and waiting forever for a background thread.

It may be worth to stress how, for example, the sentence ``$\DHANDLE{h}{e}$ in
$\Gamma$ does not converge'' has a different meaning from ``$\DHANDLE{h}{e}$ in
$\Gamma$ diverges'', since it is possible that $\DHANDLE{h}{e}$ in
$\Gamma$ eventually fails.  The same problem holds for the phrases
``converges'' and ``eventually fails''.

We will
% usually
avoid such wording in the negative.

%%%%%%%%%%%%%%%%%%%%%%%%%%%%%%%%%%%%%%%%%%%%%%%%%%%%%%%%%%%%%%%%%%%%%%%%%%%%%%%
\section{Summary}
We conceived the \EPSILONZERO core language for expressivity,
efficiency and ease of formal manipulation: \EPSILONZERO is powerful
but idiosyncratic and unsuitable for direct use by human users, who
are expected to only access it through extensions.
\\
\\
After dealing at length with design issues and providing a rationale for
\EPSILONZERO language we proceeded to formally specify its syntax, semantics
and error conditions.

We also described a sufficient set of conditions under which an implementation
is compelled to respect the specified behavior, allowing for both \textit{inefficient
but friendly} and \textit{efficient but unsafe} implementations.

The small-step operational semantics is relatively simple and has a deterministic
sub-relation obtained by simply ignoring one rule.
\\
\\
We defined the ``one-step'' semantics, hiding the complexity of stacks, by
iterating the small-step reduction relation.

%% \SOMEWHERE{Bruno Haible told me \EPSILONZERO reminded him of \CODE{M4}, in
%%   terms of the language forms.  I should to cite it somewhere, even very
%%   rapidly.}

% -*- mode: latex; fill-column: 79; mode: auto-fill; mode: flyspell; buffer-file-coding-system: utf-8 -*-
\chapter{Reflection and self-modification}
\label{meta-chapter}
\label{reflection-chapter}
\label{self-modification-chapter}

The presentation of \EPSILONZERO in \SECTION\ref{epsilonzero-chapter} showed
the \textit{state} component as already containing procedure and global
bindings, but did not illustrate any explicit way of updating either.

In this chapter we will start by discussing global definitions, and then
proceed to clarify what we mean by a ``program''; our somewhat unusual solution
has important implications on how programs are loaded, saved, and compiled.

%building the foundation for \SECTION\ref{syntactic-extensions-chapter}.

\minitoc

%%%%%%%%%%%%%%%%%%%%%%%%%%%%%%%%%%%%%%%%%%%%%%%%%%%%%%%%%%%%%%%%%%%%%%%%%%%%%%%
%\section{Global definitions and self-modification}
\section{Global definitions}
\label{program-and-repl}
\label{toplevel-and-repl}
The expression semantics given
in \SECTION\ref{epsilonzero-semantics-definition} does not explicitly mention any
functionality to alter the set of procedure or global bindings
that we have always considered as \textit{already} defined, as part of
the state; anyway such functionality is clearly needed: if not for anything
else, at least for defining new recursive procedures, as expressions can not
express recursion without referring global procedures; and, of course, global
definitions are useful for reasons of modularity.
\\
\\
A ``traditional'' solution to this problem would consist in adding
\textit{toplevel forms} to \EPSILONZERO as a new syntactic category: toplevel
forms would comprise a procedure definition form and a non-procedure definition
form; the program would then become a sequence of toplevel definitions,
possibly followed by a main expression returning a final ``result''.

We have several good reasons to reject this simplistic notion of a fixed
program to be written from start to finish and then executed:
\begin{itemize}
\item
  in the spirit of Forth\footnote{It is not by coincidence that we mention Forth
    first in this case.  One important lesson of Forth is
  how complex programs can be written even with no support whatsoever for typing,
  \textit{provided that} each small component is individually testable.
  We extrapolate the following motto from our experience of working with ML, Lisp and
  Forth: \textit{the less typing, the more important a Read-Eval-Print Loop}.} 
  and Scheme, we want to also support
  \textit{interactive} systems interleaving user input with evaluation and answers;
\item
  having \textit{exactly two} toplevel definition forms may be adequate for
  \EPSILONZERO, but certainly not for its extensions: a personality may need global
  definitions for new entities such as classes, exceptions, or types.  Even
  syntax is not fixed: we want to make new entities and their associated
  syntactic extensions definable at any point during the user interaction, to be
  immediately available for use;
%% , and syntax for them, might be defined at some point in the program
%% \TODOQ{it is perfectly reasonable, for example, to introduce
%% at some point of the execution an input syntax for some new user-defined data
%% type; the updated syntax would become available after the definition, to be
%% used for the rest of the input.  More
%% radically, \SECTION\ref{syntactic-extension-chapter} \TODOQ{[refer a specific
%%     section, when I write it]} will show a realistic example in which
%% \textit{all the procedures defined up to a point need to be destructively
%%   replaced with others}.}

\item
  adding toplevel forms to \EPSILONZERO is not necessary, since expressions are
  already powerful enough to express state updates using primitives or
  procedures;
\item
  a powerful language should let the user update global definitions \textit{from
    any program point}, not just at the top level.
\end{itemize}
For all these reasons we will simply assume the presence of the
\TDEF{global-definition} or \TDEF{self-modification} procedures \CODE{state:global-set!} and
\CODE{state:procedure-set!}, to be defined later in
\SECTION\ref{state:procedure-set!-definition},
p.~\pageref{state:procedure-set!-definition}.
\label{reflective-procedure-names}
We also assume the presence of
their companion \TDEF{reflective procedures} \CODE{state:global-get},
\CODE{state:procedure-get-formals}, \CODE{state:procedure-get-body},
\CODE{state:global-names} and \CODE{state:procedure-names}.

No
toplevel forms are needed\footnote{Global-procedure-definition procedures are not
  particularly friendly to use directly, since the user has to pass
  \textit{expressions} as parameters, and expressions are relatively complex
  data structures which need to be built.  However a friendly definition form
  is not hard to define on top of global-definition procedures, by using a macro
  (\SECTION\ref{syntactic-extension-chapter}).}: definitions are expressions like any others, and can be
executed at the top level or just as well within other expressions.

%%%%%%%%%%%%%%%%%%%%%%%%%%%%%%%%%%%%
%\section{The \textit{dynamic vs.\ static} debate}
\section{Programs and self-modification}
The last point in the dotted list above, easily the most controversial, illustrates
well the tension between our will of providing an expressive system, and the
desire of also keeping the language easy to reason about and efficient --- in
mainstream terms, the
% ``dynamic vs.\ static''
\textit{dynamic vs.\ static}
debate.

At a first look \EPSILONZERO with global-definition procedures appears flatly sided with the
``dynamic'' party, allowing any program to capriciously modify itself at run
time; for example in
\SECTION\ref{changing-arity-on-the-go} at
p.~\pageref{changing-arity-on-the-go} we even considered the case of
calling a not-yet-existing procedure which is \textit{created by one of its
  actual parameters}.  Indeed, we will find use for creating procedures from
arbitrary expressions
(\SECTION\ref{closure-conversion-transform}); but
whenever possible we would still prefer not to renounce to better intellectual
manageability, and efficiency.
\\
\\
As a reasonable compromise and a way out of the dilemma, it is possible
\textit{to freely use global-definition procedures to let the program reach a
  final ``static'' form}; after that point, the program may be analyzed to
check for properties and compiled efficiently, under the assumption that no
more self-modifications will occur.
%This point of view drives the specific notions of programs and unexecing that we will adopt below.

%%%%%%%%%%%%%%%%%%%%%%%%%
\subsection{Programs}
The most convenient notion of a program, for our purposes, is slightly unusual.
Given a state and an expression, we can imagine to somehow generate a snapshot
containing the ``frozen'' state, plus the expression.

``Executing'' a program then means to fire up evaluation on the saved
expression from the resumed state:
\begin{definition}[program]\label{program-definition}
Let $\Gamma$ be an \EPSILONZERO state and $\DHANDLE{h}{e}$ be an \EPSILONZERO
expression; then we define their corresponding \TDEF{program}
%$p(\Gamma, \DHANDLE{h}{e}) \in \SET{\Gamma} \times \EXPRESSIONS$
%as the pair $(\UPDATESTATE{\Gamma}{\rm\text{\CODE{futures}}}{\EMPTYSET}, \DHANDLE{h}{e})$.
as the pair $(\UPDATESTATE{\Gamma}{\rm\text{\CODE{futures}}}{\EMPTYSET}, \DHANDLE{h}{e}) \in \SET{\Gamma} \times \EXPRESSIONS$.
%% quadruple
%% {\rm
%% $$
%% \Gamma_{\CODE{global-environment}},
%% \Gamma_{\CODE{procedures}},
%% \Gamma_{\CODE{memory}},
%% \DHANDLE{h}{e}
%% $$}.
\QEDDEFINITION
\end{definition}
We intentionally disregard the background threads in $\Gamma_{\CODE{futures}}$,
in order not to have to deal with execution stacks or partial expression
evaluation.  Background threads do not look particularly useful anyway in this
context, since the main idea is simply to use global-definition procedures to
have the program self-modify into something which contains every needed
auxiliary procedure and data, before the ``main expression'' can be finally
evaluated; clearly, the main expression itself will be free to create
background threads.

Of course there is no guarantee that a program, when executed, will not start
self-modifying again.

%%%%%%%%%%%%%%%%%%%%%%
\subsection{Static programs}
\label{static-programs}
A static program is a program which, when executed, never self-modifies:
\begin{definition}[static program]\label{static-program-definition}
Let $(\Gamma, \DHANDLE{h}{e})$ be a program; then we say that it is
\TDEF{semantically static}
or \TDEF{static}
if in no configurations reachable
from evaluating $\DHANDLE{h}{e}$ in $\Gamma$, the global environment or the
procedure environment are different from the ones in $\Gamma$.

More formally, a program $(\Gamma, \DHANDLE{h}{e})$ is static if for all
%$(S'\ V'\ \Gamma') \in \{\chi'\ |\ (\DHANDLE{h}{e}, \EMPTYSET)\ \VALUESEPARATOR\ \Gamma \TOEP \chi'\}$
$(S'\ V'\ \Gamma')$ such that
%\linebreak
$(\DHANDLE{h}{e}, \EMPTYSET)\ \VALUESEPARATOR\ \Gamma \TOEP (S'\ V'\ \Gamma')$
we have that
$\Gamma_{\rm\text{\CODE{global}}} = \Gamma'_{\rm\text{\CODE{global}}}$ and
$\Gamma_{\rm\text{\CODE{procedures}}} = \Gamma'_{\rm\text{\CODE{procedures}}}$.
\QEDDEFINITION
\end{definition}
We consider re-defining a global to be self-modification: anyway the user can
still define mutable variables in the style of ML and Forth in a static
program, by adding one indirection level so that a global maps to a memory
cell, whose content can be updated any number of times without affecting its
identity.  The value of imperative variables would then be held in the
\textit{memory} state environment (\SECTION\ref{memory-state-environment}), which
does not affect staticity.

Interestingly, the use of \textit{reflective} procedures is not problematic
even for a static program: a static program can safely \textit{read} its own
global and procedure state environments, which are constant by definition.
\\
\\
Our \textit{semantic} versus \textit{syntactic} naming convention is important
and deserves some comments.  The convention comes from standard garbage
collection jargon \cite{survey--wilson}, and is used to distinguish between a
datum which will not be accessed in the rest of the computation from a datum
which can not be reached by traversing pointers from the roots.  It is
undecidable whether a heap object is ``semantic garbage'', so garbage
collection works by recycling ``syntactic garbage'', a conservative decidable
approximation (any piece of semantic garbage is also syntactic garbage).

Like the property of being semantic garbage our semantic staticity property is
trivially undecidable, so it would be tempting to define a notion of
%% Like in the case of semantic garbage, our semantic staticity property is
%% trivially undecidable, so it would be tempting to define a notion of
``syntactic staticity'' involving the use of global-definition procedures in
reachable code.  Such attempts are doomed to fail in \EPSILONZERO, because of the way
global-definition procedures are defined
(\SECTION\ref{state:procedure-set!-definition}): in practice, at least with our
current set of primitives\footnote{It would be possible to bootstrap the
  language with a different set of primitives
  (\SECTION\ref{bootstrap-phase-3}), so that \CODE{state:global-set!}  and
  \CODE{state:procedure-set!} are primitives themselves and do not depend on
  others.  However such a solution would be unrealistic for a practical
  implementation, where identifiers and expressions are data structures like
  any others.
  
  As a slightly more subtle point, a syntactic staticity guarantee would also have
  to prevent \textit{destructive modification of expressions}, which is
  possible as well in our implementation
  of \SECTION\ref{syntactic-extensions-chapter}.}, it
is always possible to modify procedures or globals with ordinary memory stores,
bypassing the ``high-level'' procedures for program self-modification.

Syntactic staticity properties \textit{are} definable in typed personalities,
where dynamic or static checks prevent the user from writing at arbitrary
memory addresses.
\\
\\
In accordance with our open-ended design principles we do not consider an
``error'' for a program to be self-modifying; yet staticity remains desirable
since self-modification makes most analyses impossible, prevents many compiler
optimizations including inlining, and indeed challenges the very idea of
compilation\footnote{It \textit{is} possible to compile only parts of the code,
  as several Lisp systems do, but the interaction between interpreted code and
  compiled code complicates design, also requiring dynamic invalidation and
  substitution of compiled code.
  It is also possible to have a JIT, or more simply a compiler to be executed
  at run time which translates every expression as soon as it is generated,
  like in the SBCL Common Lisp system \cite{sbcl}.}.

A static program can instead be compiled and optimized in a
traditional way and as a consequence of our design a \textit{whole-program}
approach, lending itself to global optimizations \cite{mlton}, feels
particularly natural.  Since the expressions occurring in a fixed program are
all known, it is easy to build global tables with handles as keys
(\SECTION\ref{handle-introduction}), to perform any kind of analysis\footnote{A
  useful notion for which we cannot claim novelty: the idea of attaching
  user-defined data to syntactic objects, now mostly popular because of
  Java ``annotations''
  (\url{http://download.oracle.com/javase/1.5.0/docs/guide/language/annotations.html}),
  was already quite explicit in McCarthy's 1959 LISP \cite{lisp-mccarthy}.}.

Particularly in an untyped context, where users are supposed to be competent,
it is reasonable to consider a certain program as semantically static
\textit{when users demand so} by requesting to analyze or compile a program in
a modality which takes advantage of staticity.  We stress once more how all
such functionality, including compilers, can be written in the language itself
as part of a ``library'', and is not specially hardwired in the system in any
way\footnote{We did not implement a complete compiler yet
  (\SECTION\ref{implementation-status})
  but we have a custom bytecode virtual machine, and the beginnings of native bindings.  No fundamental
  obstacle in implementing an \EPSILONZERO native compiler is apparent, and we
  plan to write one in the coming months.}.
\\
\\
%\subsubsection{High-level notation for static programs}
Since the set of procedures in a static program is fixed and so is its main
expression, it makes sense to define a notation to show a program in a
``linear'' (and leaner) form.  It is also convenient to speak about individual
program components using the ``$\PARAMETER \in \PARAMETER$'' operator, without
making state environments explicit.

This will be useful in \SECTION\ref{dimension-analysis-chapter}, when we describe a static
analysis in detail.
%\pagebreak
\begin{definition}[static program linear syntax]\label{lean-belonging-sytnax-for-static-programs}
Let $p = (\Gamma, \DHANDLE{h}{e})$ be a static program.
If {\rm$\Gamma_{\CODE{procedures}} = \{
  f_1 \mapsto (x_{1_1}...x_{1_{n_1}}, \DHANDLE{h_1}{e}),
  f_2 \mapsto (x_{2_1}...x_{2_{n_2}}, \DHANDLE{h_2}{e}),
  ...,
  f_m \mapsto (x_{m_1}...x_{m_{n_m}}, \DHANDLE{h_m}{e})
  \}$},
then we can write the whole program as:\\
``{\rm$\DDEFINEPROCEDURE{f_1}{x_{1_1} ... x_{1_{n_1}}}{\DHANDLE{h_1}{e}}$\\
\INVISIBLE{``}$\DDEFINEPROCEDURE{f_2}{x_{2_1} ... x_{2_{n_2}}}{\DHANDLE{h_2}{e}}$\\
\INVISIBLE{``}$...$\\
\INVISIBLE{``}$\DDEFINEPROCEDURE{f_m}{x_{m_1} ... x_{m_{n_m}}}{\DHANDLE{h_m}{e}}$\\
\INVISIBLE{``}$\DHANDLE{h}{e}$}''.
\\
We also write:
\begin{itemize}
\item
``{\rm $\DDEFINEPROCEDURE{f}{x_1 ... x_n}{\DHANDLE{h_1}{e}} \in p$}''
to mean that
{\rm $\Gamma_{\CODE{procedures}}:{f} \mapsto (x_1...x_n, {\DHANDLE{h_1}{e}})$};

\item
``{\rm $\DHANDLE{h}{e} \in p$}'',
to mean that $\DHANDLE{h}{e}$ is the main expression of $p$.
\QEDDEFINITION
\end{itemize}
\end{definition}

\subsection{When to run analyses}
\label{when-to-run-analyses}
In traditional languages the act of % to make the English explanation easier to parse
performing a ``static'' analysis means running some procedure over the
syntax trees from a compilation unit \textit{before} the unit is compiled or
executed; but with our program notion above blurring phases and units, the very
idea of ``static analysis'' in the context of \EPSILON becomes fuzzy.
\\
\\
Activities closely analogous to static analysis remain meaningful: for example
in a statically-typed personality the procedures and global variables which are
part of a program at a given point can still be usefully checked for type
safety, \textit{independently from the way each entity was defined} in the past
evaluation history.

Some analyses may be attempted even for non-static programs.  The problem becomes
rather \textit{the point in time} at which to run analyses: since no ``end
point'' is apparent, no obvious solution comes to mind.  A personality might
run some or even all the analyses right after evaluating \textit{each} toplevel
expression; as a more radical hypothesis an advanced editor such as Emacs can certainly be
programmed to communicate with an interpreter after each \textit{character}
modification, demanding to run analyses\footnote{The difficulty of this
  approach is due less to program analysis than to the difficulty of defining
  the semantics of incremental modifications to non-contiguous points of a
  program.  This seems hard to accomplish for \EPSILON without rebuilding the
  entire state from scratch at every change, at a prohibitive cost in
  performance; caching mechanisms can be conceived for some
  personalities to make such operations more efficient.} and visualizing their
continuously updated results.

Of course any similar
solution needs to cope with the possibility of yet unresolved forward
references, which cannot be prevented in general due to the mutual recursion
% already
inherent in \EPSILONZERO procedures: it is to be expected and regarded as normal
that analyses fail at some points where the current set of global definitions is ``open''.
\\
\\
One likely appropriate time for running analyses is right before compilation, since 
no unresolved forward references would be present at that point; but in
\EPSILON compilation does not necessarily mark any ``terminal'' point of the
evaluation history, either.
It seems reasonable to also allow analyses to be run at any point, \textit{on demand}.
\\
\\
We will see in \SECTION\ref{a-transform-may-be-a-good-place-to-run-an-analysis}
how transforms may be conveniently used to automatically associate analyses to
global definitions.

%%%%%%%%%%%%%%%%%%%%%%%%%%%%%%%%%%%%%%%%
\section{Unexec}
\label{unexec}

Our programs, be they static or not, are in fact ``system images'', which would
be convenient to write to disk for later execution, or even to be transferred
to different computers; the main expression to be saved as part of the program
might simply be \textit{a call to the REPL procedure, itself calling the
  interpreter}: that way restoring the system image would open an interactive
session in the saved state.

One could even envisage ``snapshots'' as a way of saving the current system
state before performing an experimental and potentially destructive
modification in an interactive way: if the modification fails, the user can
revert to the old state by loading the last snapshot, presumably a much faster
operation than re-building the previous state by repeating the same
self-modifications which generated it in the first place from the initial
state.
\\
\\
The functionality cursorily described above resembles the Emacs
``unexec'' hack \cite[\SECTION{}Building Emacs]{emacs-lisp}.  Emacs consists of
a relatively small Lisp interpreter written in C which contains the core
primitives, plus the bulk of the system implemented in Lisp; in order to avoid
loading hundreds of Lisp files at every startup, the native Emacs executable is
built so that it fires up \textit{in the state which would be produced by
  loading the initial Lisp files}.  Building such a functionality in C with
native processes is tricky and requires system-specific low-level code.

Despite the similarity of intent our implementation will be much simpler, and
largely machine-independent.
\\
\\
The general idea of our unexecing strategy is to simply \textit{marshal} data
structures into a linear representation which can later be read back in an
\TDEF{exec} phase, based on unmarshalling.

The composition of unexec and exec yields a state identical\footnote{{We
    are assuming that memory encodes the complete global state, but this
    assumption breaks if the state refers system structures such as
    open files or sockets: unexec can not reproduce any object out of its
    process address space.}} to the original one up to buffer addresses.

{We cannot claim novelty for this idea, considering for example Hoare's
  early intuition in
  \cite[\SECTION3.3(2-3)]{hints-on-programming-language-design--hoare}; in a
  couple recent systems unexecing exists, but plays a less central role than in
  ours: SML/NJ for example supports a ``\CODE{heap2exec}''
  utility \cite{unexec-in-smlnj}; it generates native code,\footnote{Keeping
    the two functionalities separate has the advantage of providing a working
    unexec feature also on platforms where a native compiler is not
    implemented.} yet it can only run on a couple of platforms ---
\CODE{heap2exec} is not by any means ``the compiler'' for SML/NJ, but rather just one
tool among others.  Unexec support has also been discussed or experimented with
for Perl, Python and Guile \cite{guile}.}

%%%%%%%%%%%%%%%%%%%%%%%%%%%%%%%%%%%%%%%%
\subsection{The stuff values are made of}%Memory model}
\label{memory-model}
\label{the-nature-of-values}
Since any realistic implementation must work on general-purpose Von Neumann
machines, it is clear that the implementations of all state environments and
expressions share in practice \textit{the same machine memory}; and that memory
holds the data structures we have to marshal.

\label{few-in-number-vague-reference}
Encoding details will be made clear in \SECTION\ref{bootstrap-phase-3}; but
even without specifying here how each kind of data is represented, we need to
describe the \textit{memory model} followed by all our in-memory objects.  As a
consequence of other design decisions, the actual data structures to be
marshalled for unexecing will be surprisingly few in number
(\SECTION\ref{unexec-the-symbol-table-and-the-main-expression}).
\\
\\
We can ignore background threads, which are not involved in unexecing as they
are excluded from programs as per Definition~\ref{program-definition}.  The
remaining values are of only two kinds:
\begin{itemize}
\item
\TDEF{unboxed} values;
\item
heap buffer \TDEF{pointers}, also called \TDEF{boxed} values.
\end{itemize}
A machine word, in practice not wider than a general register (32 or 64 bits on
modern machines), can hold either an unboxed value, or a pointer to a buffer; a
buffer is a contiguous array of other machine words in heap memory.  Pointers
are always \textit{initial}: we preferred to simply avoid \textit{interior}
pointers as they may exhibit bad interactions with some garbage collection
algorithms which we may want to adopt in the future \cite{survey--wilson}, even
if ours has no such restriction
(\SECTION\ref{interior-pointers-are-ok}).
\\
\\
Unboxed values are often used for \TDEF{fixnums}, which is to say fixed-range
integers, represented in two's complement on modern hardware.
Booleans, characters and enumerates also fit comfortably in the range of
unboxed values.  No provision is made for objects
smaller than a word: at this level, one word is the smallest representable datum.
Complex data containing multiple ``fields'' usually need to be boxed, but if
all fields put together fit within the width of an unboxed datum, they can also be
packed into a single word: from the point of view of the memory model there is
no difference between a single-field and a multi-field unboxed object.
Efficient implementations of dynamically-typed personalities are free to
reserve some bits as tags in unboxed objects and pointers
\cite{representing-type-information-in-dynamically-typed-languages--gudeman},
in the case of pointers
exploiting the fact that allocation alignment will free at least two or three bits for
all buffer addresses, on current byte-addressed machines (\SECTION\ref{gc-alignment}).

Some values are in practice necessarily boxed, notably reified
\textit{expressions}; as a consequence in some cases what we informally called
a ``value'' in \SECTION\ref{epsilonzero-chapter} is actually a function of the
pointer, which we use as a reference to the entire object to pass around, plus
all the memory which it refers, closed under the ``points-to'' relation.

Values can then be visualized as a graph, possibly containing converging edges
and cycles.
\\
\\
It is in practice possible to alter memory to change a boxed datum component
even when such value does not appear as ``mutable'' in the semantics, for example
by using a store primitive on an expression datum.  Such practices would entail
primitive failure\footnote{We did not specify explicit rules for our chosen set
  of primitives, including preconditions to be satisfied to avoid failure;
  however we can quickly hint at a solution: we can imagine that each buffer
  contains an initial boolean tag word, recording its mutability or lack
  thereof; the store primitive will only permit to write mutable buffers, and
  another primitive will be available to change a mutability tag from mutable
  to immutable, but never the converse.  Load and store primitives would
  implicitly skip the tag word in the offset they receive (\SECTION\ref{buffer-set!-primitive}).
  
  Of course the implementation does not need to actually represent the
  mutability tag word: Implementation Note~\ref{failure-implementation-note}
  permits us to \textit{assume that no failure occurs}, and still respect our
  specification.} and hence not respect the hypotheses for implementation
guarantees (see Implementation Note~\ref{failure-implementation-note}); the
fact that they are possible does \textit{not} constitute a violation of \EPSILONZERO
semantics.
\\
\\
\label{memory-dumps}The following linearized textual format for memory data structure including
addresses is very convenient for debugging the implementation and also as a generic fallback ``untyped printer'':
\begin{syntacticconvention}[memory dump]\label{memory-dump-syntactic-convention}
We \TDEF{dump} a given datum into a string of colored characters, according to
its shape.  There are two cases:
%  By induction: [NO, THIS IS NOT (directly) INDUCTIVE AT ALL]
\begin{itemize}
\item an \textit{unboxed datum} is written in \textit{green} in decimal, as a two's complement signed integer;
\item a \textit{pointer} is written as a hexadecimal number, prefixed by the string ``\CODE{0x}'';
  \begin{itemize}
  \item if the referred buffer occurs for the \textit{first time} in the data structure
    (depth-first left-to-right), its address is written in \textit{red} followed by a dump of its
    \textit{buffer elements} between brackets, with consecutive content words separated
    by a space;
  \item if the referred buffer has \textit{already occurred}, its address is written in
    \textit{yellow} and the buffer content is not repeated.
    \QEDSYNTACTICCONVENTION
  \end{itemize}
\end{itemize}
\end{syntacticconvention}
With some practice it is not difficult to make sense of quite complicated data
structures by reading memory dumps.  Despite not being strictly needed
to parse textual dumps, color makes it easier for humans to recognize
a structure's shape at a glance.
\\
\\
In most cases (but not all: see the discussion of hashes
in \SECTION\ref{reflection--unexecing-and-hashes}) the actual numeric address held by
pointers is not relevant for algorithms, and a pointer simply ``identifies''
a certain buffer, independently from its specific placement in memory.  It is hence
reasonable to represent data graphically, ignoring addresses and simply using
arrows for pointers, multi-slot boxes for buffers, and numbers for unboxed data.

Such ``address invariance'' is fortunate, since usually we do not have control
over buffer addresses at allocation time\footnote{System libraries ultimately
  choose data structure addresses, providing very few guarantees.  Sometimes
  the problem is made even worse by deliberate address space boundary
  randomizations performed for security reasons
  \cite{address-space-randomization}.}, hence we cannot reliably re-create a
buffer at a specified memory address.  What the composition of marshalling and
unmarshalling will accomplish, then, is the reproduction of the \TDEF{data
  structure graph}
(Figures~\ref{circular-list-figure} and \ref{shared-box-figure}).
For example, after marshalling and unmarshalling, the data structure dumped in
Figure~\ref{shared-box-figure} might be ``cloned'' into
\CODE{\textcolor{red}{0x26aaaf0[0x2899220[\textcolor{darkgreen}{42}] 0x3078920[\textcolor{darkyellow}{0x2899220} \textcolor{darkgreen}{0}]]}}.
\begin{figure}[h!]
\centering
\begin{tikzpicture}[every node/.style=draw,minimum width=24mm, minimum height=5mm]
\pgfsetmatrixcolumnsep{10mm}
\pgfmatrix{rectangle}{center}{mymatrix}
{\pgfusepath{}}{\pgfpointorigin}{\let\&=\pgfmatrixnextcell}
{
%\node(a){8}[style=white]; \&[20mm] \node{1}; \&[-1mm] \node{6}; \\
                          \& \node(p3-b){\CODE{57}}; \&                          \& \\
 \node(c1){}[style=white]; \& \node(p1-a){};\fill(0,0) circle (1.7pt);          \&  \node(p1-b){\CODE{3}};  \& \\
 \node(cl){}[style=white]; \&                         \&  \node(p2-a){};\fill(0,0) circle (1.7pt);         \& \\
 \node(c3){}[style=white]; \& \node(p2-b){\CODE{-2}}; \&                          \& \\
                          \& \node(p3-a){};\fill(0,0) circle (1.7pt);          \&  \node(cr){}[style=white]; \& \\
}
%\draw [*->,red,thick] (a.center) -- (a) node [above,midway,transparent] {};
\draw [-latex] (p1-a.center) -- (p1-b);
\draw [-latex] (p2-a.base) .. controls (cr.north) .. (p2-b.east);
\draw [-latex] (p3-a.center) .. controls (cl) .. (p3-b.west);
\end{tikzpicture}
\caption{\label{circular-list-figure}
A circular list holding the fixnums \CODE{57}, \CODE{3} and \CODE{-2}, whose dump could be
\CODE{\textcolor{red}{0x27032d0[\textcolor{darkgreen}{57} 0x279ead0[\textcolor{darkgreen}{3} 0x28a66e0[\textcolor{darkgreen}{-2} \textcolor{darkyellow}{0x27032d0}]]]}}.}
\end{figure}
\begin{figure}[h!]
\centering
\begin{tikzpicture}[every node/.style=draw,minimum width=24mm,minimum height=5mm]
\pgfsetmatrixcolumnsep{0mm}
\pgfmatrix{rectangle}{center}{mymatrix}
{\pgfusepath{}}{\pgfpointorigin}{\let\&=\pgfmatrixnextcell}
{
%\node(a){8}[style=white]; \&[20mm] \node{1}; \&[-1mm] \node{6}; \\
    \node(p1-a){};\fill(0,0) circle (1.7pt);         \&                            \&                     \& \\
    \node(p3-a){};\fill(0,0) circle (1.7pt);         \&  \node{}[style=white];     \&  \node(p2-a){};\fill(0,0) circle (1.7pt);     \& \\
                           \&  \node(c){}[style=white];  \&  \node(p3-b){\CODE{0}};    \& \\
                           \&  \node{}[style=white];     \&                     \& \\
 \node(c2){}[style=white]; \&  \node(p2){\CODE{42}};            \&                     \& \\
}
%\draw [*->,red,thick] (a.center) -- (a) node [above,midway,transparent] {};
\draw [-latex] (p1-a.center) .. controls (c2.west) .. (p2.west);
\draw [-latex] (p2-a.center) .. controls (c.north) .. (p2.north);
\draw [-latex] (p3-a.center) -- (p2-a);
\end{tikzpicture}
\caption{\label{shared-box-figure}
An example of sharing: a two-element list using \CODE{0} as a terminator whose elements both point to the same one-element buffer, holding the fixnum \CODE{42}.  One possible dump is
\CODE{\textcolor{red}{0x29ecd90[0x2714220[\textcolor{darkgreen}{42}] 0x29549f0[\textcolor{darkyellow}{0x2714220} \textcolor{darkgreen}{0}]]}}.}
\end{figure}

%%%%%%%%%%%%%%%%%%%%%%%
\subsection{Marshalling}
\label{marshalling}
Textual dumps as per Syntactic Convention~\ref{memory-dump-syntactic-convention}
could serve as a marshalling format; however our implementation marshals data
structures into binary files, for efficiency reasons.  Since specific pointer
values are immaterial, in marshalling we replace them with sequential 0-based
identifiers, which enables some minor optimizations.  The logic of marshalling
and unmarshalling algorithms resembles moving garbage collecting algorithms
such as \textit{semispace} \cite{survey--wilson}, which have to recursively
``clone'' data structure graphs.

Similarly to textual dumps, when \textit{marshalling}, the idea is to
recursively \textit{trace} a data structure keeping into account which buffers
we already visited; marshalling produces a sequence of zero or more ``buffer
definitions'' followed by the single main object, be it a pointer or an
unboxed value.  Pointers are encoded as buffer indices, following the
definition order.

Conversely, the \textit{unmarshal} procedure will allocate and fill buffers,
and then resolve buffer identifiers into pointers in a second pass.
\\
\\
In our current implementation each file field is a 32-bit big-endian word\footnote{Using 32-bit rather than 64-bit words helps to avoid relying on non-portable behavior
by mistake when dumping very large unboxed data.
We could reduce tag words to bytes
  or even single bits, but particularly in the latter case it is not clear
  whether the denser format would compensate for the additional required
  shuffling in terms of time efficiency.  Space consumption is not a problem:
typical unexec binary dumps have sizes ranging in the order of one to a few megabytes.
}.  A \TDEF{binary dump} (see Figure~\ref{binary-dump-file-format-figure})
begins with a word holding the \textit{number of buffers}, followed by the same
number of ``buffer definitions'', and finally by the ``main object''; each
buffer definition contains a word encoding the \textit{buffer size} in words,
following by as many ``elements'', each element containing two words: either a \textit{0
  tag} for an unboxed object followed by the \textit{content}, or a \textit{1
  tag} for a boxed object followed by the \textit{buffer index} (in the order
of buffer definitions, 0-based); the main object is one further element.
\begin{figure}[h!]
\centering
\[%$%\[
 \text{\em buffer-no}
 \ 
 {\overbrace{(\text{\em element-no}
     \ {\underbrace{((0|1)\ \text{\em fixnum})}_{\text{{\em element-no} times}}}{}^{*})}^{\text{{\em buffer-no} times}}}{}^{*}
 \ 
 (0|1)\text{\em fixnum}
\]%$%\]
\caption{\label{binary-dump-file-format-figure}Binary dump file format. Each of
\textit{buffer-no}, \textit{element-no}, $0|1$ and \textit{fixnum} is
encoded as a 32-bit big-endian word.}
\end{figure}

%%%%%%%%%%%%%%%%%%%
\subsubsection{Boxedness tags}
\label{boxedness-tags}\label{reflection--tagging-for-unexecing}Up to this point
we have assumed that marshalling, and textual dumping as well, can discriminate
between pointers and unboxed objects; but this it not possible at the hardware
level.
\\
\\
On the physical machine pointers are memory addresses, which is to say
\textit{numbers}, and as such in principle indistinguishable from unboxed
objects such as fixnums.

Modern hardware and operating systems tend to guarantee that objects will not
be allocated at very low addresses, so in practice it may be safe to assume that
all pointers have a numeric value larger than some constant such as $2^{16}$
\cite{compiling-with-continuations}; alignment causes all pointers to be multiples
of the word size, possibly times some other small factor; however many large fixnums
remain effectively impossible to discriminate from pointers.
\\
\\
The solution is providing a one-bit \TDEF{boxedness tag} associated to each datum, plus
a \TDEF{dimension field} per buffer --- dimensions tending not to be overly problematic
in practice\footnote{A memory system organized in the BiBOP style
  (\SECTION\ref{gc-chapter}) --- \textit{not necessarily a garbage collector}
  --- would permit not to represent them at all in the most common cases.}.

The bit can be stored within each word itself, if reducing its payload width is acceptable;
otherwise a \textit{much} less efficient but more flexible solution
%, not requiring to reserve a field in unboxed words,
consists in representing all
objects as boxed two-word buffers, using one element as the boxedness tag and
another for the payload.  We implemented this latter strategy, as it was slightly easier
to integrate with Guile (\SECTION\ref{bootstrapping}).

In either case, both \EPSILONZERO primitives and the memory management system
should keep boxedness tags into account: this will slow down arithmetic
operations and possibly complicate garbage collection.  On the other hand, some
garbage collectors (for example OCaml's) already require the same tagging
strategy for their own purposes: where such a collector is used anyway boxedness
tags cause no additional overhead.
\\
\\
Since boxedness tags are expensive
% and somewhat inconvenient,
it is conceivable
to provide two different runtime libraries, a \TDEF{tagged runtime} associating
a boxedness tag to every word and a length field to every buffer, and an \TDEF{untagged runtime} directly using
the machine representation: dumping, marshalling and unexecing will only be
possible on the ``tagged'' runtime, but the ``untagged'' runtime will be more
efficient.  One interesting feature of this solution is that, since
\textit{unmarshalling does not rely on tags}, the untagged runtime can always
be used as a last stage for a program which has been developed on the tagged
runtime, before being unexeced for the final time.  In case of compiled code, a
(presumably static) compiled program should probably always use an untagged
runtime.
\\
\\
\label{more-than-one-runtime}Boxedness tags, when present, can also be used by primitives to perform some
dynamic checks and prevent out-of-bounds errors in a very crude form of dynamic
``typing'', which
% however
has value when debugging.  In this view it might make
sense to provide \textit{three} different runtimes: ``untagged'', ``tagged
checked'' and ``tagged unchecked''.

Our implementation currently contains only a tagged checked runtime.
Implementing the other runtimes is not
% particularly
hard, being mostly a matter of using C preprocessor macros to wrap object
accesses; we will provide the two missing runtimes as soon as we eliminate the
dependency on
Guile.

\subsubsection{Marshalling properties}
\label{marshalling-properties}
Since we have not formally specified marshalling and unmarshalling algorithms,
here we simply assert their properties without proof, as guarantees to be
provided by an implementation.
\\
Again, the
strong
resemblance to moving garbage collectors
%\cite{survey--wilson}
is not coincidental.
\\
\\
We first need to specify exactly what we mean as the \textit{corresponding
substructures} of an object, ``before and after'' marshalling:
\begin{definition}[marshalling correspondence]\label{marshalling-correspondence-definition}
Let $a_0$ be an object which is marshalled into a binary dump, itself unmarshalled into the object $b_0$.  Then, by induction:
\begin{itemize}
\item $a_0$ \TDEF{corresponds to} $b_0$;
\item if $a$ corresponds to $b$ and both $a$ and $b$ are pointers to $n$-element buffers, then
the $i$-th component of the buffer pointed by $a$
\TDEF{corresponds to}
the $i$-th component of the buffer pointed by $b$
for any $0 \leq i < n$.
\QEDDEFINITION
\end{itemize}
\end{definition}
Marshalling has to ``preserve structure'', which is to say has to reproduce
the original pointer graph, mapping buffers into buffers and unboxed objects
into unboxed objects:
\begin{axiom}
Let $a$ correspond to $b$. Then we have that $a$ is unboxed if and only if $b$ is unboxed.
For every $n \in \NATURALS$ we have that
$a$ is a pointer to an $n$-element buffer if and only if $b$ is also a pointer to an $n$-element buffer.
\QEDAXIOM
\end{axiom}
\begin{axiom}
Corresponding unboxed objects are equal \textit{provided that} they both fit into a dump word payload.
\QEDAXIOM
\end{axiom}
Corresponding pointers are \textit{not} guaranteed to be equal\footnote{We do
  not want to assert that they are \textit{necessarily} different, because in
  practice garbage collection might intervene destroying the original object
  before its corresponding version is built, and it is conceivable that under
  unusual circumstances the unmarshalled object may reside at the same address
  as its corresponding original version.}, but marshalling ``preserves
equality'' without introducing or eliminating sharing, in the following sense:
\begin{axiom}
Let $a_1$ be a pointer corresponding to $b_1$ and $a_2$ be a pointer corresponding to $b_2$; then we have that
$a_1 = a_2$ if and only if $b_1 = b_2$.
\QEDAXIOM
\end{axiom}
As a consequence most operations over pointers continue to work with their
intended semantics after unmarshalling, including checking pointer equality ---
but checking whether an address is numerically \textit{smaller or bigger} then
another may yield a different result.

\label{reflection--unexecing-and-hashes}
Interpreting pointers as fixnums and doing arithmetic over them, for example to
compute a hash function, in general will yield different results before and
after unmarshalling.  But using only the \textit{unboxed elements} of boxed
structures yields the same results after unmarshalling, except in case of
overflow.

\section{Summary}
Instead of hardwiring definition forms into the language syntax, we can keep
the language simpler by providing procedures to update procedures and global
variables.  These procedures may be used anywhere, and allow for powerful
self-modifying code.

``Static'' code, on the other hand, has the advantage of allowing analyses and
being efficiently compilable.  A state where no more self-modification takes
place can be reached incrementally, by successive self-modifications.

Having access to the current global state permits to save a snapshot of the
system as a data structure, in a way similar to the Emacs unexec hack in terms
of functionality, but implemented much more simply by data structure
marshalling.

Marshalling relies on boxedness tags, which can be made optional for higher
performance.

%\SOMEWHERE{The runtime is written in C.  Is it clear?}

% -*- mode: latex; fill-column: 79; mode: auto-fill; mode: flyspell; buffer-file-coding-system: utf-8 -*-
%%%%%%%%%%%%%%%%%%%%%%%%%%%%%%%%%%%%%%%%%%%%%%%%%%%%%%%%%%%%%%%%%%%%%%%%%%%%%%%
\chapter{A static semantics for \EPSILONZERO: dimension analysis}
\label{chapter-static-semantics}
\label{chapter-dimension-analysis}
\label{dimension-analysis-chapter}

The core language \EPSILONZERO as described in
\SECTION\ref{chapter-epsilonzero}
is much simpler than other formally-specified languages such as SML,
whose description \cite{sml90,commentary-on-sml,sml97} looks strikingly
complex for a ``small'' language; Scheme Standards include non-normative
semantics for some language subset in appendices \cite{r3rs,r4rs,r5rs,r6rs}; mainstream
languages have no formal specification at all.

To make a realistic argument for the practicality of our \EPSILONZERO
semantics we are going to show an example of 
its application by formally describing a static analysis of bundle dimensions
for static programs (\SECTION\ref{static-programs}),
and then proving it sound with respect to the semantics.
\\
\\
We chose to deal with bundle dimensions in this sample analysis because bundles
are interesting as a slightly unusual
feature, but of course dimension analysis has no privileged status: just like any other static
analysis in \EPSILON, dimension analysis can be used as in ML for preventing runtime
errors at the cost of also rejecting some correct programs,
or just to obtain
warnings, or not at all; and of course any number of analyses (or ``type
systems'') can run side by side on the same program; it is up to the
personality implementor to decide what to do with the results.

\minitoc

%%%%%%%%%%%%%%%%%%%%%%%%%%%%%%%%%%%%%%%%%%%%%%%%%%%%
\section{Dimension inference}
In analogy with Hindley-Milner type inference \cite{damas-milner} we would like
to define a procedure automatically assigning\footnote{An alternative approach
  based on \textit{checking} user-supplied annotations would have been
  possible, since in practice only few expressions will have a dimension other
  than $\LIFT{1}$, which could be assumed as the default case.  Inference is however even less obstrusive, and
  does not seem to require a substantially different formalization.}
a \TDEF{dimension} to every expression in a static program, where the dimension
represents a conservative approximation of the size of
the bundle the expression may evaluate to at run time.

Intuitively, we want to associate dimension ``one'' to constants such as
${\DCONSTANT{{h_{0}}}{42}}$,
and also to all variables such as 
${\DVARIABLE{{h_{1}}}{x}}$ since non-singleton bundles are not
denotable. In the same spirit, a two-object bundle such as 
${\DBUNDLE{{h_{2}}}{{\DCONSTANT{{h_{3}}}{10}}\ \LINEBREAK{\DPRIMITIVE{{h_{4}}}{+}{{\DCONSTANT{{h_{5}}}{1}}\ \LINEBREAK{\DCONSTANT{{h_{6}}}{2}}}}}}$
would have dimension
``two'', and of course the zero-element bundle 
${\DBUNDLE{{h_{7}}}{}}$
would have dimension ``zero''.

Anyway by following this line of reasoning alone we get stuck very soon: for example, what
dimension should we assign to a call to the procedure $f1$?
\\
%% ----------------------- begin --------------------
%% (define (f1 x) (f2 x))
${\DDEFINEPROCEDURE{f1}{x}{{\DCALL{{h_{1}}}{f2}{{\DVARIABLE{{h_{2}}}{x}}}}}}$
\\%% (define (f2 x) x)
${\DDEFINEPROCEDURE{f2}{x}{{\DVARIABLE{{h_{3}}}{x}}}}$
\\%% (f1 42)
${\DCALL{{h_{4}}}{f1}{{\DCONSTANT{{h_{5}}}{42}}}}$
%% ------------------------ end ---------------------
\\
Of course the answer relies on $f1$'s definition, and in particular on the
dimension of its body. But $f1$'s body consists of a call to $f2$... It is
already clear that dimension inference 
has to work \textit{on an entire program}, using a fix point
construction of some sort:
in the fashion of type inference, the analysis will deduce a set of
constraints from a program (for example: $f1$ returns a
result with the same dimension as the result of $f2$;
$f2$ returns a singleton bundle; the main expression has the same
dimension as the result of $f1$), and attempt to resolve them.

\subsection{The dimension lattice $(\NATURALSBT, \SQMEET, \SQJOIN)$}
It is easy to see how our dimension domain needs to be at least slightly richer
than the set of natural numbers $\NATURALS$, for example by looking at the main
expression of the following program:
%% ----------------------- begin --------------------
\\%% (define (loop) (loop))
${\DDEFINEPROCEDURE{f}{}{{\DCALL{{h_{1}}}{f}{}}}}$
\\%% (loop)
${\DCALL{{h_{2}}}{f}{}}$
%% ------------------------ end ---------------------
\\
Since $f$ never returns anything
the analysis cannot discover any constraint on the dimension of its
result, other than a trivial one according to which such dimension
is equal to itself.  We call ``$\bot$'' \textit{the dimension of an expression on which
we have no constraints}, such as the main expression of the program above.

As it will be made clear below, in
practice only some trivially looping expressions have dimension $\bot$.
From the dimension point of view such expressions are particularly
unproblematic and easy to combine with others,
since they can never cause failures thanks to \EPSILONZERO's call-by-value
strategy: for example passing a parameter with 
dimension $\bot$ to any unary procedure will cause an infinite loop before the
body has a chance of ever being evaluated, and maybe failing.
\\
\\
At the opposite end of the spectrum, some expressions are clearly troublesome;
for example a procedure call with a wrong number of parameters will definitely
yield a dimension failure at run time, if the expression is reached and
parameters converge; we assign
the dimension ``$\top$'' to such \textit{trivially failing} expressions.

As a slightly more subtle case, and very similarly to Hindley-Milner type inference, we need
to give \IFNAME expressions a dimension which is the ``synthesis'' of
its branch dimensions: when the \THENNAME and \ELSENAME branches have incompatible
dimensions, such synthesis will be $\top$. For example we assign
the dimension $\top$
to the
\textit{inconsistently-dimensioned}
 expression
%% (if-in x (1 2 3) 10 (bundle))
${\DIFIN{{h_{0}}}{{\DVARIABLE{{h_{1}}}{x}}}{1,2,3}{{\DCONSTANT{{h_{2}}}{10\LINEBREAK}}}{{\DBUNDLE{{h_{3}}}{}}}}$;
such an expression is problematic to 
compose, because the dimension of the result bundle varies according
to which branch is taken at run time.
\\
\\
Our dimension domain is hence made of the natural numbers $\NATURALS$ extended with the
two elements $\bot$ and $\top$: we call this set $\NATURALSBT$.
We can easily define a partial order $\PARAMETER \sqsubseteq \PARAMETER$ as 
the reflexive closure of the relation $\PARAMETER \SQL \PARAMETER$,
where $\PARAMETER \SQL \PARAMETER = \bigcup_{i \in \NATURALS}{\{(\bot, \LIFT{i}), (\LIFT{i}, \top)\}}$.
\begin{figure}[h!]
  \centering
  \begin{tikzpicture}
    \matrix (dimension)
            [matrix of nodes,
              nodes in empty cells,
              nodes={outer sep=0pt,circle,minimum size=4pt},
              column sep={1cm,between origins},
              row sep={1cm,between origins}]
            {                  &   &   &   & $\top$&   &   &   &    \\
              \INVISIBLE{hack} & $\LIFT{0}$ & $\LIFT{1}$ & $\LIFT{2}$ & $\LIFT{3}$     & $\LIFT{4}$ & $\LIFT{5}$ & $\LIFT{6}$ & ...\\
                               &   &   &   & $\bot$&   &   &   &    \\};
            \foreach \a in {2,...,8}{
              \draw (dimension-1-5) -- (dimension-2-\a);
              \draw (dimension-3-5) -- (dimension-2-\a);
            }
  \end{tikzpicture}
  \caption{\label{dimension-lattice-figure}The flat lattice
    $(\NATURALSTB, \SQJOIN, \SQMEET)$.}
\end{figure}
\\
\\
The set $\NATURALSBT$ with the order $\PARAMETER \sqsubseteq \PARAMETER$ forms
a flat lattice: for any $a, b \in \NATURALSTB$, we call 
$a \SQJOIN b$ their
least upper bound or ``join'',
and
$a \SQMEET b$ their greatest lower bound or
``meet''.

In the lattice higher values correspond to more constrained dimensions,
with $\bot$ representing the absence of any constraint,
$\LIFT{n}$ with $n \in \NATURALS$ representing a bundle of exactly $n$ elements,
and $\top$
expressing several conflicting constraints; the
join operation $\PARAMETER \SQJOIN \PARAMETER$ is but the ``synthesis''
mentioned above, yielding \textit{the least constrained dimension which is compatible
with both parameters}:
joining $\bot$ with another element yields the other element, joining
$\LIFT{n}$ with itself yields $\LIFT{n}$ for every $n \in \NATURALS$, and
joining $\LIFT{n}$ with $\LIFT{m}$ for $n \neq m$ yields $\top$; joining $\top$
with any element yields $\top$.
\\
\\
Occasionally we may also use the set $\NATURALSB$, defined as
$\NATURALSTB \setminus \{\top\}$.

\subsection{Definition and properties}

We are now ready to formally enunciate dimension analysis, computing a
dimension for each expression occurring anywhere in the program, and for each
procedure.
% No: Christophe remarked that this is fundamentally wrong: the fixpoint is
% *inside* the definition.
%%; the following definition, being fundamentally recursive, encompasses
%%the minimum fix point construction.
%% \CH{À propos de dimension des
%%   primitives~: \SECTION\ref{primitive-dimension-christophe},
%%   p.~\pageref{primitive-dimension-christophe}.  Je n'ajouterais pas un axiome
%%   ici.} 
\begin{definition}[Dimension]\label{dimension-definition}
%\TODO{Shall I explicitly say ``minimum fix point''?}
Let a program $p$ be given.  We define in a mutually-recursive fashion:
  \begin{itemize}
  \item
    The dimension function for expressions,
    a partial function with signature $\EXPRESSIONS \PARTIAL \NATURALSB$, that we represent as
    the relation $\PARAMETER \HASDIMENSION \PARAMETER \ \subseteq \ \EXPRESSIONS \times \NATURALSB$:
  %% the \TDEF{dimension} of each expression occurring in it as an element of
  %% $\NATURALSTB$, by the following rules:
  
{\rm % All of this must be upright
\begin{minipage}{0.5\linewidth}
  \begin{prooftree}
    \AxiomC{}
    \LeftLabel{}
    \RightLabel{}
    \UnaryInfC{$\DVARIABLE{h}{x} \HASDIMENSION \LIFT{1}$}
  \end{prooftree}
\end{minipage}
\begin{minipage}{0.5\linewidth}
  \begin{prooftree}
    \AxiomC{}
    \LeftLabel{}
    \RightLabel{}
    \UnaryInfC{$\DCONSTANT{h}{c} \HASDIMENSION \LIFT{1}$}
  \end{prooftree}
\end{minipage}

\begin{prooftree}
  \AxiomC{$\DHANDLE{h_1}{e} \HASDIMENSION d_1\ \ ...\ \ \DHANDLE{h_{n}}{e} \HASDIMENSION d_n$}
  \LeftLabel{}
  \RightLabel{$d_i \SQLE \LIFT{1}$, for all $1 \leq i \leq n$}
  \UnaryInfC{$\DBUNDLE{h_0}{\DHANDLE{h_1}{e}...\DHANDLE{h_{n}}{e}} \HASDIMENSION \LIFT{n}$}
\end{prooftree}

\begin{prooftree}
  \AxiomC{$\DHANDLE{h_1}{e} \HASDIMENSION d_1$}
  \AxiomC{$\DHANDLE{h_2}{e} \HASDIMENSION d_2$}
  \LeftLabel{}
  \RightLabel{$d_1 \SQLE \LIFT{m}$ for some $m \geq n$, $d_2 \SQL \top$
    %$d_1 \SQLE \LIFT{n}$, $d_2 \SQL \top$
  }
  \BinaryInfC{$\DLET{h_0}{x_1...x_n}{\DHANDLE{h_1}{e}}{\DHANDLE{h_2}{e}} \HASDIMENSION d_2$}
\end{prooftree}

\begin{prooftree}
  \AxiomC{$\pi \HASDIMENSION n \to m$}% in $p$}
  \AxiomC{$\DHANDLE{h_1}{e} \HASDIMENSION d_1\ \ ...\ \ \DHANDLE{h_n}{e} \HASDIMENSION d_n$}
  \LeftLabel{}
  \RightLabel{$d_i \SQLE \LIFT{1}$, for all $1 \leq i \leq n$}
  \BinaryInfC{$\DPRIMITIVE{h_0}{\pi}{\DHANDLE{h_1}{e}...\DHANDLE{h_n}{e}} \HASDIMENSION \LIFT{m}$}
\end{prooftree}

\begin{prooftree}
  \AxiomC{$f \HASDIMENSION n \to d$}% in $p$}
  \AxiomC{$\DHANDLE{h_1}{e} \HASDIMENSION d_1\ \ ...\ \ \DHANDLE{h_n}{e} \HASDIMENSION d_n$}
  \LeftLabel{}
  \RightLabel{$d_i \SQLE \LIFT{1}$, for all $1 \leq i \leq n$, $d \SQL \top$}
  \BinaryInfC{$\DCALL{h_0}{f}{\DHANDLE{h_1}{e}...\DHANDLE{h_n}{e}} \HASDIMENSION d$}
\end{prooftree}

\begin{prooftree}
  \AxiomC{$\DHANDLE{h_1}{e} \HASDIMENSION d_1$}
  \AxiomC{$\DHANDLE{h_2}{e} \HASDIMENSION d_2$}
  \AxiomC{$\DHANDLE{h_3}{e} \HASDIMENSION d_3$}
  \LeftLabel{}
  \RightLabel{$d_1 \SQLE \LIFT{1}$, $d = d_2 \SQJOIN d_3$, $d \SQL \top$}
  \TrinaryInfC{$\DIFIN{h_0}{\DHANDLE{h_1}{e}}{c_1...c_n}{\DHANDLE{h_2}{e}}{\DHANDLE{h_3}{e}} \HASDIMENSION d$}
\end{prooftree}

\begin{prooftree}
  \AxiomC{$f \HASDIMENSION n \to d$}% in $p$}
  \AxiomC{$\DHANDLE{h_1}{e} \HASDIMENSION d_1\ \ ...\ \ \DHANDLE{h_n}{e} \HASDIMENSION d_n$}
  %\LeftLabel{\TODOQ{\tiny $\LIFT{1}$ (below) used to be $d$}}
  \RightLabel{$d \SQLE \LIFT{1}$, $d_i \SQLE \LIFT{1}$, for all $1 \leq i \leq n$}
%  \BinaryInfC{$\DFORK{h_0}{f}{\DHANDLE{h_1}{e}...\DHANDLE{h_n}{e}} \HASDIMENSION d$}
  \BinaryInfC{$\DFORK{h_0}{f}{\DHANDLE{h_1}{e}...\DHANDLE{h_n}{e}} \HASDIMENSION {\LIFT{1}}$}
\end{prooftree}

\begin{prooftree}
  \AxiomC{$\DHANDLE{h_1}{e} \HASDIMENSION d_1$}
  %\LeftLabel{\TODOQ{\tiny $\LIFT{1}$ (below) used to be $d_1$}}
  \RightLabel{$d_1 \SQLE \LIFT{1}$}
  \UnaryInfC{$\DJOIN{h_0}{\DHANDLE{h_1}{e}} \HASDIMENSION {\LIFT{1}}$}
\end{prooftree}
} % End of the upright part

  \item
    the dimension function for procedures (written in relational notation)
    $\PARAMETER \HASDIMENSION \PARAMETER \to \PARAMETER$, with
    signature
    $\PROCEDURES \to (\NATURALS \times \NATURALSTB)$,
    associating a procedure name with the number of its parameters and the
    dimension of its result.\\
    %
    % Changed to satisfy Mauny's objection:
%    \BLUE{
%      \textbf{[Partie changée pour répondre à l'objection de Mauny]}
    For each procedure
    {\rm $\DDEFINEPROCEDURE{f}{x_1 ... x_n}{\DHANDLE{h_1}{e}} \in p$} we say
    $f$ has \TDEF{in-dimension} $n$ and \TDEF{out-dimension} $d$,
    and we write ``$f \HASDIMENSION n \to d$'' where $d$ is the minimum
    fixpoint such that $\#(\DHANDLE{h_1}{e}) = d$.
%    }
    %% when we have that $\#(\DHANDLE{h_1}{e}) = d$ we say that $f$ has
    %% \TDEF{in-dimension} $n$ and \TDEF{out-dimension} $d$,
    %% and we write ``$f \HASDIMENSION n \to d$''.
    % old version, before the changed to answer Mauny
    %
    %% For each procedure
    %% {\rm $\DDEFINEPROCEDURE{f}{x_1 ... x_n}{\DHANDLE{h_1}{e}} \in p$},
    %% when we have that $\#(\DHANDLE{h_1}{e}) = d$ we say that $f$ has
    %% \TDEF{in-dimension} $n$ and \TDEF{out-dimension} $d$,
    %% and we write ``$f \HASDIMENSION n \to d$''.

    %% % Older, much more complex, alternative definition. I don't see any advantages
    %% % now and I don't remember why I originally proposed that:
    %% let
    %% $T_f = \{\DHANDLE{h_i}{e} \SUCHTHAT \ISTAILFOR{\DHANDLE{h_i}{e}}{f}\}$
    %% be the set of the tail expressions of $f$, and let
    %% $D_f = \{d_i \SUCHTHAT \exists \DHANDLE{h_i}{e} \in T_f \SUCHTHAT \#(\DHANDLE{h_i}{e}) = d_i \}$
    %% be the set of the dimensions of the elements in $T_f$; if we call
    %% $d = \BIGSQJOIN_{D_f}$ the join of all the elements of $D_f$,
    %% then we say that
    %% \TDEF{$f$ has in-dimension $n \in \NATURALS$ and out-dimension $d \in \NATURALSTB$}, and we write
    %% ``$f \HASDIMENSION n \to d$''.
    %% \TODO{Why not just say that if
    %%   $\DDEFINEPROCEDURE{h_0}{f}{x_1 ... x_n}{\DHANDLE{h_1}{e}} \in p$, then
    %%   $f \HASDIMENSION n \to \#(\DHANDLE{h_1}{e})$ instead? It's much simpler
    %%   and it should always give the same result (think of the dimension of an
    %%   \IFNAME\ as the most complex case); I changed \FILE{dimension.ml} to use
    %%   this simpler definition, and nothing obvious has broken}
  \end{itemize}
  Then we define $\#(\PARAMETER) : \EXPRESSIONS \to \NATURALSTB$ as the total
  extension of $\PARAMETER \HASDIMENSION \PARAMETER$, so that
  $\#(\PARAMETER)$ returns $\top$ where $\PARAMETER \HASDIMENSION \PARAMETER$
  is not defined and its same result elsewhere.
    %% its total extension $\#(\PARAMETER) : \EXPRESSIONS \to \NATURALSTB$, which
    %% returns $\top$ where the function  $\PARAMETER \HASDIMENSION \PARAMETER$ is not defined,
    %% and the same result elsewhere;
  \QEDDEFINITION
\end{definition}
Definition~\ref{dimension-definition} depends on the fact that the relation
$\PARAMETER \HASDIMENSION \PARAMETER$ be a function, which is clearly true because rule 
  premises are pairwise disjoint.
  %% % Not needed any longer with the new, simpler definition for procedure dimension:
  %% \CHECK{$\BIGSQJOIN_{D_f}$ is well-defined even
  %%   without an explicit neuter element for $\PARAMETER \SQJOIN \PARAMETER$, since the set of tail
  %%   expressions for any given procedure is non-empty}.

Notice that all the side constraints of the form ``$d \SQL \top$''
(rules for \LETNAME, \CALLNAME and \IFNAME) are only
included for aesthetic symmetry, so that in case of any dimension inconsistency
$\PARAMETER \HASDIMENSION \PARAMETER$ remains undefined just as in
the other syntactic cases, rather than returning $\top$:
%the restriction is \textit{not} strictly necessary:
of course the total extension
$\#(\PARAMETER)$ would remain the same even if we erased such side constraints from
$\PARAMETER \HASDIMENSION \PARAMETER$.
\begin{definition}\label{singular-definition}\label{plural-definition}\label{singular-plural-definition}
We call an expression $\DHANDLE{h}{e}$ \TDEF{plural} if
$\#(\DHANDLE{h}{e}) = \LIFT{n}$ for some $n \neq 1$,
\TDEF{consistently-dimensioned} if $\#(\DHANDLE{h}{e}) \SQL \top$
and
\TDEF{inconsistently-dimensioned} if $\#(\DHANDLE{h}{e}) = \top$.
\QEDDEFINITION
\end{definition}
We intentionally wrote Definition~\ref{plural-definition} so that it makes empty bundles
plural, and trivially-looping expressions not plural.
The reason for this choice is bound to the implementation: only
what we call plural expressions requires some non-conventional implementation
technique such as placing a value in a number of registers or
%bundle stack slots (see \SECTION\ref{bundle-stack} \TODOQ{[check at the end that bundle stacks are well explained there]})
stack slots
different from one.
Expressions which never return anything do not pose particular problems --- and we
stress again that we do \textit{not} consider non-termination an
error; anyway the existence of expressions $\DHANDLE{h}{e} \HASDIMENSION \bot$
is the reason why we resist the temptation of defining ``singular''
expressions;
should $\DHANDLE{h}{e}$ be both singular and plural, singular and not plural, plural and not singular, or
neither? No solution seems intuitive, or particularly useful.
\\
\\
It is not hard to see how inconsistent dimensioning ``propagates outwards'',
from a \textit{contained} expression out to its \textit{containing} expression:
\begin{proposition}\label{top-dimension-propagates-outwards-proposition}
Any expression containing an inconsistently-dimensioned subexpression is
inconsistently-dimensioned itself.
\end{proposition}
\begin{proof}
Assuming $\DHANDLE{h}{e} \HASDIMENSION \top$, 
we have to prove that
$C[\DHANDLE{h}{e}] \HASDIMENSION \top$
for all contexts $C[\PARAMETER]$.
\\
A straightforward structural induction over contexts:
\begin{itemize}
\item
  $C[\DHANDLE{h}{e}] = \DHANDLE{h}{e}$ (base case):
  we trivially have that
  $\#(C[\DHANDLE{h}{e}]) = \#(\DHANDLE{h}{e}) = \top$;

\item
$C[{\DHANDLE{h}{e}}] = \DLET{h_0}{x_1...x_n}{C'[\DHANDLE{h}{e}]}{\DHANDLE{h_2}{e}}$:
since by hypothesis $\#(\DHANDLE{h}{e}) = \top$, by induction
hypothesis we also have $\#(C'[\DHANDLE{h}{e}]) = \top \NSQLE \LIFT{n}$; this makes impossible to satisfy the conditions of the
\LETNAME\ rule in
Definition~\ref{dimension-definition}; so the relation
$\PARAMETER \HASDIMENSION \PARAMETER$ is undefined on $C[\DHANDLE{h}{e}]$,
and again by Definition~\ref{dimension-definition} we have that
$\#(C[\DHANDLE{h}{e}]) = \top$;

\item
$C[{\DHANDLE{h}{e}}] = \DLET{h_0}{x_1...x_n}{\DHANDLE{h_1}{e}}{C'[\DHANDLE{h}{e}]}$:
again we have that $\#(C'[\DHANDLE{h}{e}]) = \top$,
and the \LETNAME rule in Definition~\ref{dimension-definition} cannot
fire because the \LETNAME body $C'[\DHANDLE{h}{e}]$ has dimension $\top$; if
the rule does not fire then $C[{\DHANDLE{h}{e}}] = \top$ because 
$\PARAMETER \HASDIMENSION \PARAMETER$ is undefined on the parameter, as in the
previous case;

\item
$C[{\DHANDLE{h}{e}}] = \DCALL{h_0}{f}{{\DHANDLE{h_1}{e}}...{\DHANDLE{h_n}{e}}{C'[{\DHANDLE{h}{e}}]}{\DHANDLE{h_{n+1}}{e}}...{\DHANDLE{h_{n+m}}{e}}}$:
again $\#(C'[\DHANDLE{h}{e}]) = \top$ by induction hypothesis; but then there
exist a procedure actual whose dimension is not lower than or equal to
$\LIFT{1}$, and the \CALLNAME rule in Definition~\ref{dimension-definition} cannot
fire; $\PARAMETER \HASDIMENSION \PARAMETER$ is undefined on $C[{\DHANDLE{h}{e}}]$,
hence $\#(C[{\DHANDLE{h}{e}}]) = \top$;

\item
$C[{\DHANDLE{h}{e}}] =
  \DPRIMITIVE{h_0}{\pi}{{\DHANDLE{h_1}{e}}...{\DHANDLE{h_n}{e}}{C'[{\DHANDLE{h}{e}}]}{\DHANDLE{h_{n+1}}{e}}...{\DHANDLE{h_{n+m}}{e}}}$:
same reasoning as the \CALLNAME case;

\item
$C[{\DHANDLE{h}{e}}] = \DIFIN{h_0}{C'[{\DHANDLE{h}{e}}]}{c_1...c_n}{\DHANDLE{h_2}{e}}{\DHANDLE{h_3}{e}}$:
again $\#(C'[\DHANDLE{h}{e}]) = \top$ by induction hypothesis, which means
that the dimension of the discriminand is not lower than or equal to
$\LIFT{1}$, the \IFNAME rule in Definition~\ref{dimension-definition} cannot
fire, hence $\PARAMETER \HASDIMENSION \PARAMETER$ is undefined on $C[{\DHANDLE{h}{e}}]$,
and $\#(C[{\DHANDLE{h}{e}}]) = \top$;

\item
$C[{\DHANDLE{h}{e}}] = \DIFIN{h_0}{\DHANDLE{h_1}{e}}{c_1...c_n}{C'[{\DHANDLE{h}{e}}]}{\DHANDLE{h_3}{e}}$:
similar to the \LETNAME body case: $\#(C'[\DHANDLE{h}{e}]) = \top$, which
prevents the \IFNAME rule in Definition~\ref{dimension-definition} from firing;

\item
$C[{\DHANDLE{h}{e}}] = \DIFIN{h_0}{\DHANDLE{h_1}{e}}{c_1...c_n}{\DHANDLE{h_2}{e}}{C'[{\DHANDLE{h}{e}}]}$:
same reasoning as the previous case;

\item
$C[{\DHANDLE{h}{e}}] = \DFORK{h_0}{f}{{\DHANDLE{h_1}{e}}...{\DHANDLE{h_n}{e}}{C'[{\DHANDLE{h}{e}}]}{\DHANDLE{h_{n+1}}{e}}...{\DHANDLE{h_{n+m}}{e}}}$:
same reasoning as the \CALLNAME case;

\item
$C[{\DHANDLE{h}{e}}] = \DJOIN{h_0}{C'[\DHANDLE{h}{e}]}$:
same reasoning as the \CALLNAME case;

\item
$C[{\DHANDLE{h}{e}}] = \DBUNDLE{h_0}{{\DHANDLE{h_1}{e}}...{\DHANDLE{h_n}{e}}{C'[{\DHANDLE{h}{e}}]}{\DHANDLE{h_{n+1}}{e}}...{\DHANDLE{h_{n+m}}{e}}}$:
same reasoning as the \CALLNAME case.
\end{itemize}
\end{proof}

It is also intuitive that replacing a subexpression with another whose
dimension is \textit{lower} or equal will not raise the dimension of the containing
expression, which makes $\#(\PARAMETER)$ a monotonic function:

\begin{proposition}[$\#$-monotonicity]\label{dimension-monotonicity-proposition}
Replacing a subexpression with another of lower or equal dimension cannot
raise the dimension of the containing expression.
\\
More formally, for any expression context $C[\PARAMETER]$, expression $e$ and expression $e'$,
if we have that $\#(e) \SQLE \#(e')$ then we also have that $\#(C[e]) \SQLE \#(C[e'])$.
\begin{proof}
Another straightforward structural induction over contexts.
\end{proof}
\end{proposition}
The following definition identifies programs where no expression of dimension
$\top$ occurs anywhere. As the reader will have anticipated, we are going to
prove the condition \textit{sufficient} to guarantee a desirable property with
respect to the dynamic semantics.

\begin{definition}[well-dimensioned]\label{well-dimensioned-definition}
We call the static program $p$ \TDEF{well-dimensioned} if both the following
conditions hold:
%% \TODO{Adapt after writing \SECTION\ref{reflection-chapter}.  The idea remains
%%   perfectly valid; I just have to make notation and language consistent.}
\begin{itemize}
\item
  for all procedure definitions
  \textup{$\DDEFINEPROCEDURE{f}{x_1...x_n}{\DHANDLE{h_1}{e}} \in p$}
  such that $f \HASDIMENSION n \to d$, we have
  $d \SQL \top$;

%% \item
%%   \NO{
%%   for all non-procedure definitions
%%   \textup{$\DDEFINENONPROCEDURE{h_3}{x}{\DHANDLE{h_4}{e}} \in p$} 
%%   we have that
%%   $\#(\DHANDLE{h_4}{e}) \SQLE \LIFT{1}$};

\item for the main expression \textup{$\DHANDLE{h_{2}}{e} \in p$}
%  we have that $\#(\DHANDLE{h_{2}}{e}) \SQLE \LIFT{1}$.
  we have that $\#(\DHANDLE{h_{2}}{e}) \SQL \top$.

\end{itemize}
We call \TDEF{ill-dimensioned} all programs which are not well-dimensioned.
\QEDDEFINITION
\end{definition}
%% Of course it is worth to stress that at this stage no connection has been yet drawn
%% between dimension analysis and \EPSILONZERO's dynamic semantics.

\subsubsection{There cannot be a \textit{most precise} dimension analysis}
It would be nice to be able to characterize our definition of an expression
dimension as the ``most precise'', but unfortunately our definition is
not the best one, and in fact no such definition can exist.

In order to see why at least at an intuitive level we consider the program:
%% ----------------------- begin --------------------
%% (define (loop) (loop))
\\
\\
${\DDEFINEPROCEDURE{loop}{}{{\DCALL{{h_{1}}}{loop}{}}}} \in p$
\\
%\\%% (if-in 1 (2) 42 (loop))
${\DIFIN{{h_{2}}}{{\DCONSTANT{{h_{3}}}{\mathcal{N}(1)}}}{\mathcal{N}(2)}{{\DCONSTANT{{h_{4}}}{\mathcal{N}(42)}}}{{\DCALL{{h_{5}}}{loop}{}}}} \in p$
%% ------------------------ end ---------------------
\\\\
It is obvious that the main expression loops, but the analysis assigns the main
expression the dimension $\LIFT{1}$ instead of $\bot$: hence our
definition of dimension does not correspond to ``the best possible'' analysis,
as it is possible to change it to account for more particular cases, yielding a
more precise result: in fact we can always improve the
analysis by recognizing particular patterns in programs ---
for example by simplifying statically-determined conditionals at compile time
as a first refinement; but because of the Halting Theorem we cannot hope to cover \textit{all}
possible cases.

This trivial fact prevents us from finding a result similar to the
Most-General-Type theorem in \cite{damas-milner}.

\section{Semantic soundness}
Before proving the result connecting dimension analysis with \EPSILONZERO's
dynamic semantics we need to define some machinery.

\subsection{Resynthesization}

The idea of \TDEF{resynthesization} consists in taking any reachable configuration $\chi$
and reconstructing from it an expression $e$ that, if evaluated at the
top level in the \STATE of $\chi$, would yield the same result and the same effects as $\chi$.
Actually we do not need to specify this equivalence any further, and in fact we will not
prove any result such as reduction-equivalence on resynthesization, since our use of it
here is very well-delimited, due to the technical need of assigning a
dimension to all reachable \EPSILONZERO configurations.
\\\\
We can easily view the content of any value stack $V$ in a reachable
configuration as \TDEF{a list of non-holed expressions}
 $E_V$, by remarking the intuitive role of ``$\VALUESEPARATOR$'' as a bundle delimiter; bundles within $V$
can be represented in $E_V$ as explicit \BUNDLENAME
expressions\footnote{Actually we would need to introduce explicit \BUNDLENAME
  expressions only for \textit{plural} bundles in $V$; the definition given below
  avoids this complication at the price of producing
  some trivial \BUNDLENAME expressions with only one item.}.
For the purposes of resynthesization it is also safe to ignore
``$\ACTIVATIONSEPARATOR$'' delimiters, since the particular arity mismatches they
were conceived to prevent (see \SECTION\ref{changing-arity-on-the-go})
cannot occur in static programs, our only programs of interest in this chapter.

More formally, we define the translation as follows:\\\\
$E_{\VALUESEPARATOR} \EQD \EMPTYSEQUENCE$\\
$E_{{\VALUESEPARATOR}{\ACTIVATIONSEPARATOR}V} \EQD E_{{\VALUESEPARATOR}V}$,\\
$E_{{\VALUESEPARATOR} {c_1}...{c_n}{\VALUESEPARATOR}V} \EQD \DBUNDLE{h'}{c_1...c_n}.E_{\VALUESEPARATOR V}$\ \ with some fresh handle $h'$.
% when $\VALUESEPARATOR \notin \{c_1...c_n\}$, 
\\\\
For example
$V = \VALUESEPARATOR 1\ 2 \VALUESEPARATOR 3 \VALUESEPARATOR$ would be transformed
into $E_V = \SEQUENCE{\DBUNDLE{h_0'}{\DCONSTANT{h_1'}{1} \DCONSTANT{h_2'}{2}}, \DBUNDLE{h_3'}{\DCONSTANT{h_4'}{3}}}$,
with fresh handles $h_0', h_1', h_2', h_3'$ and $h_4'$.
\\
\\
We define resynthesization as a relation $r$, written in functional notation as
$r(\PARAMETER\ \PARAMETER)$; for readability's sake we omit the comma between the
two parameters of $r$, since both tend to be syntactically complex.

Given \textit{a stack} and \textit{a list of non-holed expressions} as obtained
from the translation above, resynthesization produces \textit{a non-holed expression list}:

\begin{definition}[resynthesization]\label{resynthesization-definition}
  We define the \TDEF{resynthesization relation} $r(\PARAMETER\ \PARAMETER)$ as follows, with the
  convention that all the prime-decorated handles only appearing on the right
  sides be \textit{fresh}:
  \begin{itemize}
    \item
      $r(\EMPTYSTACK
        \ %
        E)
      =
      E$
      
    \item
      {\rm
      $r((\DHANDLE{h}{e},\ \rho).S
        \ %
        E)
      =
      r(S
        \ %
        \DHANDLE{h}{e}.E)$}, for any non-holed $\DHANDLE{h}{e}$;
      
      \item
        {\rm
          $r((\DLET{h_0}{x_1...x_n}{\HOLE}{\DHANDLE{h_2}{e}},\ \rho).S
          \ %
          \DHANDLE{a}{e} . E)
          \\=
          r((\DLET{h_0'}{x_1...x_n}{\DHANDLE{a}{e}}{\DHANDLE{h_2}{e}},\ \rho).S
          \ %
          E)
          $
        }

      \item
        {\rm
          $r(({\DCALLWITHHOLE{h_0}{f}},\ \rho).S
          \ %
          e_{a_{n}}e_{a_{n-1}}...e_{a_{2}}e_{a_{1}}.E)
          \\=
          r(({\DCALL{h_0'}{f}{e_{a_{1}}...e_{a_{n}}}},\ \rho).S
          \ %
          E)
          $
        }
      \item
        {\rm
          $r((\DPRIMITIVEWITHHOLE{h_0}{\pi},\ \rho).S
          \ %
          e_{a_{n}}e_{a_{n-1}}...e_{a_{2}}e_{a_{1}}.E)
          \\=
          r(({\DPRIMITIVE{h_0'}{\pi}{e_{a_{1}}...e_{a_{n}}}},\ \rho).S
          \ %
          E)
          $
        }
      \item
        {\rm
          $r((\DIFIN{h_0}{\HOLE}{c_1...c_n}{\DHANDLE{h_2}{e}}{\DHANDLE{h_3}{e}},\ \rho).S
          \ %
          \DHANDLE{a}{e}.E)
          \\=
          r((\DIFIN{h_0'}{\DHANDLE{a}{e}}{c_1...c_n}{\DHANDLE{h_2}{e}}{\DHANDLE{h_3}{e}},\ \rho).S
          \ %
          E)
          $
        }
      \item
        {\rm
          $r((\DBUNDLEWITHHOLE{h_0},\ \rho).S
          \ %
          e_{a_{n}} e_{a_{n-1}}  ... e_{a_{2}} e_{a_{1}} . E)
          \\=
          r((\DBUNDLE{h_0'}{e_{a_{1}}...e_{a_{n}}},\ \rho).S
          \ %
          E)
          $
        }
      \item
        {\rm
          $r((\DFORKWITHHOLE{h_0}{f},\ \rho).S
          \ %
          e_{a_{n}} e_{a_{n-1}}  ... e_{a_{2}} e_{a_{1}} . E)
          \\=
          r(({\DFORK{h_0'}{f}{e_{a_1}...e_{a_n}}},\ \rho).S
          \ %
          E)
          $
        }
      \item
        {\rm
          $r((\DJOIN{h_0}{\HOLE},\ \rho).S
          \ %
          \DHANDLE{a}{e} . E)
          \\=
          r((\DJOIN{h_0'}{\DHANDLE{a}{e}},\ \rho).S
          \ %
          E)
          $
        }
        \QEDDEFINITION
  \end{itemize}
\end{definition}

It is obvious from Definition~\ref{resynthesization-definition} that
resynthesization is deterministic up to handle identity, and
the same can be said about the value stack conversion defined above.
Since the specific
choice of handles is immaterial with respect to dimension, and
accounting for the specific choice of handles would make resynthesization much
harder to work with without any particular benefit, from now on we will commit a slight
abuse of language and speak about resynthesization \textit{as a function}.
\\\\
In the following we are also going to need a couple of simple properties of resynthesization:

%% As a further relaxation of our language, we gloss over the difference
%% between the two parameters of $r$ and a configuration $\chi$, since the couple
%% of stack and value stack taken by $r$ is just a configuration without the
%% context component, not needed for resynthesization.

%\subsection{\TODOQ{Alternative suggested by Jean-Vincent: don't say anything about the semantics of resynthesized expressions}}

\begin{lemma}[``$r$ does not delete expressions'']\label{lemma-1-jv}
If $r(S\ E) = E'$ for some $S, E$ and $E'$, then all expressions occurring in $E$
also occur somewhere in $E'$.
\\
More formally, for all $e$, if $r(S\ E_1.C[e].E_2) = E'$, then there exist $E_1',C'[\PARAMETER],E_2'$ such that
$E' = E_1' . C'[e] . E_2'$.
\begin{proof}
By induction on the number of recursive calls to $r$.
\end{proof}
\end{lemma}

\begin{lemma}[resynthesization shape-independence]\label{lemma-2-jv}
The shape of the expressions contained in $E$ does not affect the result of $r(S\ E)$.
\\
More formally, for all expression sequences $E_1, E_2, E_1', E_2'$, contexts $C[\PARAMETER], C'[\PARAMETER]$
and expression $e$, if we have that
$r(S\ E_1.C[e].E_2) = E_1'.C'[e].E_2'$
then we also have that
$r(S\ E_1.C[e'].E_2) = E_1'.C'[e'].E_2'$ for any other expression $e'$.
\begin{proof}
By induction on the number of recursive calls to $r$.
\end{proof}
\end{lemma}

\subsection{Weak dimension preservation}
Resynthesization allows us to gloss over the difference between
an \EPSILONZERO \textit{expression} and any \textit{configuration} reached by evaluating an
\EPSILONZERO expression, so that we may speak about the dimension of either;
  so, in order to further simplify our presentation, we extend 
Definition~\ref{dimension-definition} by assigning a dimension also to
reachable configurations: a reachable configuration will have the dimension of
its resynthesization; or, slightly more formally, if $\chi = (S\ V\ \Gamma)$ is a reachable
configuration, then we write ``$\#(\chi)$'' to mean ``$\#(e)$, where
$r(S\ E_V) = \SEQUENCE{e}$''.

It is not yet clear at this point why $r$ is always defined and always yields
a singleton expression sequence on reachable configurations; we defer the proof
to Corollary~\ref{resynthesization-is-total-and-yields-one-expression}.
\\
\\
The \TDEF{Weak Dimension Preservation} property, below, is the central result
bridging \EPSILONZERO's dynamic semantics to dimension analysis by showing that
``evaluation preserves dimension''; in some circles such properties are known
as ``subject reductions'' --- or more intuitively in French as
\textit{auto-réductions}.

\begin{figure}[h!]
  \label{weak-dimension-preservation-figure}
  \centering
  \begin{tikzpicture}
%    \draw [help lines] grid (3,2);
    nodes={outer sep=0pt,circle,minimum size=4pt},
    \node[] (chi) at (0,2) {$\chi$};
    \node[] (chiprime) at (3,2) {$\chi'$};
    \node[] (d) at (0,0) {$d$};
    \node[] (dprime) at (3,0) {$d'$};
%    \draw[->,decorate] (chi) -- (chiprime);
%    \draw[->,decorate] (chi) -- node[auto]{$\TOE$} (chiprime);
%    \draw[->,decorate] (chi) -- node[below]{$\SET{E}$} (chiprime);
    \draw[->,decorate] (chi) -- node[below]{$\TOE$} (chiprime);
    \draw[white,decorate] (d) -- node[black]{$\SQGE$} (dprime);
    \draw[->,decorate] (chi) -- node[left]{$\HASDIMENSION$} (d);
    \draw[->, decorate] (chiprime) -- node[left]{$\HASDIMENSION$} (dprime);
  \end{tikzpicture}
  \caption{The \textit{Weak Dimension Preservation property}
    (Lemma~\ref{weak-dimension-preservation-lemma}): when a configuration $\chi$ of
    dimension $d$ can reduce to another configuration $\chi'$ of dimension
    $d'$ we have that $d' \SQLE d$.}
\end{figure}

The property is ``weak'' in the sense that an expression is allowed to reduce
to another expression of \textit{lower} dimension: as it can be seen from the
proof, this may happen with conditionals: choosing one branch or the other
entails replacing the expression on the top of the stack by a subexpression
of it, hence by an expression with fewer dimension constraints. In
particular it is possible that an inconsistently-dimensioned expression
reduces to a consistently-dimensioned one
(``from $\top$ to $d \SQL \top$''); anyway the vice-versa
(``from $d \SQL \top$ to $\top$'') \textit{cannot} happen, which is all we need for our soundness property.

\begin{lemma}[Weak Dimension Preservation]\label{weak-dimension-preservation-lemma}
Let reachable configurations $\chi, \chi'$ be given, such that
$\chi = (S_{\chi}\ V_{\chi}\ \Gamma_{\chi}) \TOE \chi' = (S_{\chi'}\ V_{\chi'}\ \Gamma_{\chi'})$; now, if
$r(S_{\chi}\ E_{V_{\chi}}) = \SEQUENCE{e}$ then there exists $e'$ such that
$r(S_{\chi'}\ E_{V_{\chi'}}) = \SEQUENCE{e'}$
and
$\#(e') \SQLE \#(e)$.
\begin{proof}
%A long but mundane proof.
Induction is not needed: we just directly prove that, for all cases in which
$(S_{\chi}\ V_{\chi}\ \Gamma_{\chi}) = \chi \TOE \chi' = (S_{\chi'}\ V_{\chi'}\ \Gamma_{\chi'})$, we
have that $\#(r(S_{\chi'}\ E_{V_{\chi'}})) \SQLE \#(r(S_{\chi}\ E_{V_{\chi}}))$.
We avoid writing handles for converted value stack expressions, since they are immaterial anyway;
in this proof we also freely abuse the notation by saying that $r$ returns results which are
``equal to'' something else, writing ``$=$'' without explicitly stating that the equality
is up to handle choice.
\begin{itemize}
\item
$[constant]$\\
$
        (\DCONSTANT{h}{c},\ \rho).S
        \ \VALUESEPARATOR V\ \Gamma
        \TOE
        S
        \ \VALUESEPARATOR c \VALUESEPARATOR V\ \Gamma
        $:\\
$r(S_{\chi}\ E_{V_{\chi}}) = $
\{definition of $r$\}
$r(S\ c.E_{V_{\chi}}) =$
\{substitution\}
$r(S_{\chi'}\ E_{V_{\chi'}})$;
hence in this case we have that
$\#(r(S_{\chi'}\ E_{V_{\chi'}})) = \#(r(S_{\chi}\ E_{V_{\chi}}))$;

\item
$[variable]$\\
$
        (\DVARIABLE{h}{x},\ \rho).S
        \ \VALUESEPARATOR V\ \Gamma
        \TOE
        S
        \ \VALUESEPARATOR c \VALUESEPARATOR V\ \Gamma
        $:
%(when $\Gamma(\CODE{global\_environment})[\rho] : x \mapsto c$)
\\
looking at $\chi$, we have
$r(S_{\chi}\ E_{V_{\chi}}) = $
\{by definition of $r$\}
$r(S\ x.E_{V_{\chi}}) = $
\{hypothesis\}
$\SEQUENCE{e} =$
\{Lemma~\ref{lemma-1-jv}, for some $C[\PARAMETER]$\}
$\SEQUENCE{C[x]}$;
\\
looking at $\chi'$, we have
$r(S_{\chi'}\ E_{V_{\chi'}}) =$
\{substitution\}
$r(S\ c.E_{V_{\chi}}) =$
\{Lemma~\ref{lemma-2-jv}\}
$\SEQUENCE{C[c]}$, which we call
$\SEQUENCE{e'}$; since $\#(c) = \LIFT{1} \SQLE \#(x) = \LIFT{1}$,
by Proposition~\ref{dimension-monotonicity-proposition} we have that $\#(e') \SQLE \#(e)$;

\item
$[\CODE{let}_e]$\\
$       (\DLET{h_0}{x_1...x_n}{\DHANDLE{h_1}{e}}{\DHANDLE{h_2}{e}},\ \rho).S
        \ \VALUESEPARATOR V\ \Gamma
        \TOE
        (\DHANDLE{h_1}{e},\ \rho).(\DLET{h_0}{x_1...x_n}{\HOLE}{\DHANDLE{h_2}{e}},\ \rho).S
        \ \VALUESEPARATOR V\ \Gamma
        $:\\
%% %% \MAYBE{
%% %% looking at $\chi$ we find that
%% %% $r(S_{\chi}\ E_{V_{\chi}}) =$
%% %% \{definition of $r$\}
%% %% $r(S\ \DLET{h_0}{x_1...x_n}{\DHANDLE{h_1}{e}}{\DHANDLE{h_2}{e}}.E_{V_{\chi}}) =$
%% %% \{hypothesis\}
%% %% $\SEQUENCE{e}$;
%% %% \\}
looking at $\chi$ we have that
$r(S_{\chi}\ E_{V_{\chi}}) =$
\{definition of $r$\}\\
$r(S\ \DLET{h_0}{x_1...x_n}{\DHANDLE{h_1}{e}}{\DHANDLE{h_2}{e}}.E_{V_{\chi}})$;
looking at $\chi'$ we have that
$r(S_{\chi'}\ E_{V_{\chi'}}) =$
\{definition of $r$\}
$r((\DLET{h_0}{x_1...x_n}{\HOLE}{\DHANDLE{h_2}{e}}, \rho).S\ {\DHANDLE{h_1}{e}}.E_{V_{\chi}}) =$
\{definition of $r$\}
$r((\DLET{h_0'}{x_1...x_n}{\DHANDLE{h_1}{e}}{\DHANDLE{h_2}{e}}, \rho).S\ E_{V_{\chi}}) =$
\{substitution\}\\
$r(S\ \DLET{h_0'}{x_1...x_n}{\DHANDLE{h_1}{e}}{\DHANDLE{h_2}{e}}.E_{V_{\chi}}) =$
\{substitution\}
$r(S_{\chi}\ E_{V_{\chi}}) =$
\{hypothesis\}
$\SEQUENCE{e}$:
again we have that $r(S_{\chi}\ E_{V_{\chi}})$ and $r(S_{\chi'}\ E_{V_{\chi'}})$ are equal to the
same singleton sequence $\SEQUENCE{e} = \SEQUENCE{e'}$, hence $\#(e') = \#(e)$;
this proof case is essentially identical to the proof cases of the other
expansive rules;

\item
$[\CODE{let}_c]$\\
$
        (\DLET{h_0}{x_1...x_n}{\HOLE}{\DHANDLE{h_2}{e}},\ \rho).S
        \ \VALUESEPARATOR c_{m}c_{m-1}...c_{2}c_{1} \VALUESEPARATOR V\ \Gamma
        \TOE
        (\DHANDLE{h_2}{e},\ \rho[x_1 \mapsto c_1, x_2 \mapsto c_{2}, ...,x_n \mapsto c_n]).S
        \ \VALUESEPARATOR V\ \Gamma
        $:\\
looking at $\chi$ we find that
$r(S_{\chi}\ E_{V_{\chi}}) =$
\{definition of $r$, twice\}\\
$r(S\ \DLET{h_0'}{x_1...x_n}{\DBUNDLE{h_1'}{c_1...c_m}}{\DHANDLE{h_2}{e}}.E_{V_{\chi'}}) =$
\{hypothesis\}
$\SEQUENCE{e} =$
\{Lemma~\ref{lemma-1-jv}, for some context $C[\PARAMETER]$\}
$\SEQUENCE{C[\DLET{h_0'}{x_1...x_n}{\DBUNDLE{h_1'}{c_1...c_m}}{\DHANDLE{h_2}{e}}]}$;
\\
looking at $\chi'$ we have that
$r(S_{\chi'}\ E_{V_{\chi'}}) =$
\{definition of $r$, twice\}
$r(S\ {\DHANDLE{h_2}{e}}.E_{V_{\chi'}}) =$
\{Lemma~\ref{lemma-2-jv}\}
$\SEQUENCE{C[\DHANDLE{h_2}{e}]}$, which we call $\SEQUENCE{e'}$;
since by Definition~\ref{dimension-definition} a \LETNAME expression
has the same dimension as its body,
it follows that $\#(e') \SQLE \#(e)$ by Proposition~\ref{dimension-monotonicity-proposition};

\item
$[\CODE{call}_e]$\\
$
        (\DCALL{h_0}{f}{{\DHANDLE{h_1}{e}}...{\DHANDLE{h_n}{e}}},\ \rho).S
        \ \VALUESEPARATOR V\ \Gamma
        \TOE
        ({\DHANDLE{h_1}{e}},\ \rho)...({\DHANDLE{h_n}{e}},\ \rho).({\DCALLWITHHOLE{h_0}{f}},\ \EMPTYSET).S
        $%\ 
        \linebreak
        $\VALUESEPARATOR \ACTIVATIONSEPARATOR V\ \Gamma
        $:\\
identical to the other expansive rule cases;

\item
$[\CODE{call}_c]$\\
$
        ({\DCALLWITHHOLE{h_0}{f}},\ \rho).S
        \ \VALUESEPARATOR c_{n} \VALUESEPARATOR c_{n-1} \VALUESEPARATOR ... \VALUESEPARATOR c_2 \VALUESEPARATOR c_1 \VALUESEPARATOR \ACTIVATIONSEPARATOR V\ \Gamma
        \TOE
        ({\DHANDLE{h}{e}},\ \rho[x_1 \mapsto c_1, x_2 \mapsto c_2, ...,$\linebreak$x_{n-1} \mapsto c_{n-1}, x_n \mapsto c_n]).S
        \ \VALUESEPARATOR V\ \Gamma
        $:\\
by the rule side condition $f$ takes exactly $n$ parameters and has dimension $n \to d$ for some $d$.
Looking at $\chi$ we find that
$r(S_{\chi}\ E_{V_{\chi}}) =$
\{definition of $r$\}
$r((\DCALL{h_0'}{f}{c_1...c_n}, \rho).S\ E_{V_{\chi'}}) =$
\{hypothesis and Lemma~\ref{lemma-1-jv}, for some context $C[\PARAMETER]$\}
$\SEQUENCE{C[\DCALL{h_0'}{f}{c_1...c_n}]}$.\\
Starting at $\chi'$, 
$r(S_{\chi'}\ E_{V_{\chi'}}) =$
\{substitution\}
$r(({\DHANDLE{h}{e}}, \rho).S\ E_{V_{\chi'}}) =$
\{Lemma~\ref{lemma-2-jv}\}
$\SEQUENCE{C[{\DHANDLE{h}{e}}]}$.
But 
${\DHANDLE{h}{e}}$ and
$\DCALL{h_0'}{f}{c_1...c_n}$ have the same dimension $d$ by
Definition~\ref{dimension-definition}, hence by
Proposition~\ref{dimension-monotonicity-proposition} we have that
$\#(C[{\DHANDLE{h}{e}}]) \SQLE \#(C[\DCALL{h_0'}{f}{c_1...c_n}])$;

\item
$[\CODE{primitive}_e]$\\
$
        (\DPRIMITIVE{h_0}{\pi}{{\DHANDLE{h_1}{e}}...{\DHANDLE{h_n}{e}}},\ \rho).S
        \ \VALUESEPARATOR V\ \Gamma
        \TOE$\\$
        ({\DHANDLE{h_1}{e}},\ \rho)...({\DHANDLE{h_n}{e}},\ \rho).(\DPRIMITIVEWITHHOLE{h_0}{\pi},\ \EMPTYSET).S
        \ \VALUESEPARATOR \ACTIVATIONSEPARATOR V\ \Gamma
        $:\\
%, when $\pi \HASDIMENSION n \to m$, for some $m$
identical to the other expansive rule cases;

\item
$[\CODE{primitive}_c]$\\
$
        (\DPRIMITIVEWITHHOLE{h_0}{\pi},\ \rho).S
        \ \VALUESEPARATOR c_{n} \VALUESEPARATOR c_{n-1} \VALUESEPARATOR ... \VALUESEPARATOR c_2 \VALUESEPARATOR c_1 \VALUESEPARATOR \ACTIVATIONSEPARATOR V\ \Gamma
        \TOE
        S
        \ \VALUESEPARATOR c'_{m} c'_{m-1} ... c'_2 c'_{1} \VALUESEPARATOR V\ \Gamma'
        $, %:\\
%when $\pi(c_{1}, ..., c_{n}, \Gamma) = \SEQUENCE{c'_{1}, ..., c'_{m}, \Gamma'}$:
when $\Gamma_{\CODE{primitives}}(\pi) (c_{1}, ..., c_{n},  \Gamma) = \SEQUENCE{c'_{1}, ..., c'_{m}, \Gamma'}$:\\
%since $\pi$ must respect Axiom~\ref{reasonable-primitive-axiom} and the rule side condition applies, we have that
since the rule side condition applies, we have that
$\pi \HASDIMENSION n \to m$;\\
$r(S_{\chi}\ E_{V_{\chi}}) = $
\{definition, twice\}
$r(S\ \SEQUENCE{{\DPRIMITIVE{h_0'}{\pi}{c_1...c_n}}}.E_{V_{\chi}}) =$
\{hypothesis, for some context $C[\PARAMETER]$\}
$\SEQUENCE{C[\DPRIMITIVE{h_0'}{\pi}{c_1...c_n}]}$;\\
$r(S_{\chi'}\ E_{V_{\chi'}}) = $
\{substitution\}
$r(S\ \DBUNDLE{h'}{c'_{1}...c'_{m}}.E_{V_{\chi}}) = $
\{Lemma~\ref{lemma-2-jv}\}
$\SEQUENCE{C[\DBUNDLE{h'}{c'_{1}...c'_{m}}]}$;
since by Definition~\ref{dimension-definition} $\DBUNDLE{h'}{c'_{1}...c'_{m}}$
and\\$\DPRIMITIVE{h_0'}{\pi}{c_1...c_n}$ have the same dimension, we conclude by Proposition~\ref{dimension-monotonicity-proposition};

\item
$[\CODE{if}_e]$\\
$
        (\DIFIN{h_0}{\DHANDLE{h_1}{e}}{c_1...c_n}{\DHANDLE{h_2}{e}}{\DHANDLE{h_3}{e}},\ \rho).S
        \ \VALUESEPARATOR V\ \Gamma
        \TOE
        (\DHANDLE{h_1}{e},\ \rho).(\DIFIN{h_0}{\HOLE}{c_1...c_n}{\DHANDLE{h_2}{e}}{\DHANDLE{h_3}{e}},\ \rho).S
        \ \VALUESEPARATOR V\ \Gamma
        $:\\
identical to the other expansive rule cases;

\item
$[\CODE{if}_c^{\in}]$\\
$
        (\DIFIN{h_0}{\HOLE}{c_1...c_n}{\DHANDLE{h_2}{e}}{\DHANDLE{h_3}{e}},\ \rho).S
        \ \VALUESEPARATOR c \VALUESEPARATOR V\ \Gamma
        \TOE
        ({\DHANDLE{h_2}{e}},\ \rho).S
        \ \VALUESEPARATOR V\ \Gamma
        $, when $c \in \{c_1...c_n\}$:\\
$r(S_{\chi}\ E_{V_{\chi}}) =$
\{definition, twice\}
$r(S\ {\DIFIN{h_0'}{c}{c_1...c_n}{\DHANDLE{h_2}{e}}{\DHANDLE{h_3}{e}}}.$\linebreak$E_{V_{\chi}}) = $
\{hypothesis and Lemma~\ref{lemma-1-jv}, for some context $C[\PARAMETER]$\}\\
$\SEQUENCE{C[{\DIFIN{h_0'}{c}{c_1...c_n}{\DHANDLE{h_2}{e}}{\DHANDLE{h_3}{e}}}]}$;\\
$r(S_{\chi'}\ E_{V_{\chi'}}) =$
\{definition\}
$r(S\ {\DHANDLE{h_2}{e}}.E_{V_{\chi}}) = $
\{Lemma~\ref{lemma-2-jv}\}
$\SEQUENCE{C[{\DHANDLE{h_2}{e}}]}$;\\
by Definition~\ref{dimension-definition} we
have that
$\#({\DHANDLE{h_2}{e}}) \SQLE
\#(\DIFIN{h_0'}{c}{c_1...c_n}{\DHANDLE{h_2}{e}}{\DHANDLE{h_3}{e}})$, and we
conclude by Proposition~\ref{dimension-monotonicity-proposition};

\item
$[\CODE{if}_c^{\notin}]$\\
$
        (\DIFIN{h_0}{\HOLE}{c_1...c_n}{\DHANDLE{h_2}{e}}{\DHANDLE{h_3}{e}},\ \rho).S
        \ \VALUESEPARATOR c \VALUESEPARATOR V\ \Gamma
        \TOE
        ({\DHANDLE{h_3}{e}},\ \rho).S
        \ \VALUESEPARATOR V\ \Gamma
        $, when $c \notin \{c_1...c_n\}$:\\
identical to the previous case;

\item
$[\CODE{bundle}_e]$\\
$(\DBUNDLE{h_0}{{\DHANDLE{h_1}{e}}...{\DHANDLE{h_n}{e}}},\ \rho).S 
  \ \VALUESEPARATOR V\ \Gamma
        \TOE
        ({\DHANDLE{h_1}{e}},\ \rho)...({\DHANDLE{h_n}{e}},\ \rho).(\DBUNDLEWITHHOLE{h_0},\ \EMPTYSET).S$%\ 
        \linebreak
        $\VALUESEPARATOR \ACTIVATIONSEPARATOR V\ \Gamma
        $:\\
identical to the other expansive rule cases;

\item
$[\CODE{bundle}_c]$\\
$(\DBUNDLEWITHHOLE{h_0},\ \rho).S
        \ \VALUESEPARATOR c_{n} \VALUESEPARATOR c_{n-1} \VALUESEPARATOR ... \VALUESEPARATOR c_{2} \VALUESEPARATOR c_{1} \VALUESEPARATOR \ACTIVATIONSEPARATOR V\ \Gamma
        \TOE
        S
        \ \VALUESEPARATOR c_{n} c_{n-1} ... c_{2} c_{1} \VALUESEPARATOR V\ \Gamma
        $:\\
$r(S_{\chi}\ E_{V_{\chi}}) = $
\{definition, twice\}
$r(S\ \DBUNDLE{h'}{c_1...c_n}.E) = $
\{substitution\}
$r(S_{\chi'}\ E_{V_{\chi'}})$, and again we have that  $e = e'$;

\item
$[\CODE{fork}_e]$\\
$
        (\DFORK{h_0}{f}{{\DHANDLE{h_1}{e}}...{\DHANDLE{h_n}{e}}},\ \rho).S
        \ \VALUESEPARATOR V\ \Gamma
        \TOE
        ({\DHANDLE{h_1}{e}},\ \rho)...({\DHANDLE{h_n}{e}},\ \rho).(\DFORKWITHHOLE{h_0}{f},\ \EMPTYSET).S
        \ \VALUESEPARATOR \ACTIVATIONSEPARATOR V\ \Gamma
        $:\\
identical to the other expansive rule cases;

\item
$[\CODE{fork}_c]$\\
$
        (\DFORKWITHHOLE{h_0}{f},\ \rho).S
        \ \VALUESEPARATOR c_{n} \VALUESEPARATOR c_{n-1} \VALUESEPARATOR ... \VALUESEPARATOR c_{2} \VALUESEPARATOR c_{1} \VALUESEPARATOR \ACTIVATIONSEPARATOR V\ \Gamma
        \TOE
        S
        \ \VALUESEPARATOR \FUTURE{t} \VALUESEPARATOR V
        \ \Gamma[{}_{\CODE{futures}}^{t \mapsto (({\DHANDLE{h}{e}},\ \rho[x_0 \mapsto \FUTURE{t}, x_1 \mapsto c_1, ..., x_n \mapsto c_n])\ \VALUESEPARATOR)}]
        $:\\
%since $\chi$ is reachable $f$ is called with the right number of arguments, hence $f \HASDIMENSION n \to d$, for some $d$;\\
$r(S_{\chi}\ E_{V_{\chi}}) = $
\{definition, twice\}
$r(S\ \DFORK{h_0'}{f}{c_1...c_n}.E) = $
\{hypothesis, for some context $C[\PARAMETER]$\}
$\SEQUENCE{C[\DFORK{h_0'}{f}{c_1...c_n}]}$;
\\
$r(S_{\chi'}\ E_{V_{\chi'}}) = $
\{definition\}
$r(S\ \FUTURE{t}.E) = $
\{Lemma~\ref{lemma-2-jv}\}
$\SEQUENCE{C[\FUTURE{t}]}$. Since by Definition~\ref{dimension-definition} we
have that
$\#(\FUTURE{t}) = \LIFT{1} \SQLE \#(\DFORK{h_0'}{f}{c_1...c_n})$, we conclude
with Proposition~\ref{dimension-monotonicity-proposition};

\item
$[\CODE{join}_e]$\\
$
        (\DJOIN{h_0}{\DHANDLE{h_1}{e}},\ \rho).S
        \ \VALUESEPARATOR V\ \Gamma
        \TOE
        (\DHANDLE{h_1}{e},\ \rho).(\DJOIN{h_0}{\HOLE},\ \rho).S
        \ \VALUESEPARATOR V\ \Gamma
        $:\\
identical to the other expansive rule cases;

\item
$[\CODE{join}_c]$\\
$
        (\DJOIN{h_0}{\HOLE},\ \rho).S
        \ \VALUESEPARATOR \FUTURE{t} \VALUESEPARATOR V\ \Gamma
        \TOE
        S
        \ \VALUESEPARATOR c_t \VALUESEPARATOR V\ \Gamma
        $, when $\Gamma_{\CODE{futures}} : t \mapsto (\EMPTYSEQUENCE,\ \VALUESEPARATOR c_t)$:\\
$r(S_{\chi}\ E_{V_{\chi}}) = $
\{definition, twice\}
$r(S\ \DJOIN{h_0'}{\FUTURE{t}}.E) = $
\{hypothesis and Lemma~\ref{lemma-1-jv}, for some context $C[\PARAMETER]$\}
$\SEQUENCE{C[\DJOIN{h_0'}{\FUTURE{t}}]}$;\\
$r(S_{\chi'}\ E_{V_{\chi'}}) = $
\{substitution\}
$r(S\ c_t.E) = $
\{Lemma~\ref{lemma-2-jv}\}
$\SEQUENCE{C[c_t]}$;\\
since by Definition~\ref{dimension-definition} we have that
$\#(c_t) = \LIFT{1} \SQLE \#(\DJOIN{h_0'}{\FUTURE{t}})$,
we conclude by Proposition~\ref{dimension-monotonicity-proposition};

\item
$[\parallel]$\\
$
        S\ V\ \Gamma
        \TOE
        S\ V\ \Gamma'[{}_{\CODE{futures}}^{t \mapsto (S_t',\ V_t')}]
        $,
when
$\Gamma_{\CODE{futures}} : t \mapsto (S_t,\ V_t)$
and
      $
        S_t\ V_t\ \Gamma
        \TOE
        S_t'\ V_t'\ \Gamma'
        $:\\
here
$r(S_{\chi}\ E_{V_{\chi}}) = r(S_{\chi'}\ E_{V_{\chi'}})$, hence $e'$ exists and is
equal to $e$.\PROOFQED
\end{itemize}
\end{proof}
\end{lemma}

The following trivial consequence of Lemma~\ref{weak-dimension-preservation-lemma}
allows us to think of $r$ as always returning \textit{a single} non-holed
expression, when applied on a stack and a value stack (re-encoded as a list of non-holed
expressions) from a reachable configuration:

%% \begin{corollary}
%% Resynthesization is total on reachable configurations.
%% \QEDNOPROOF
%% \end{corollary}

\begin{corollary}\label{resynthesization-is-total-and-yields-one-expression}
Reachable configurations resynthesize into exactly one expression.
\end{corollary}
\begin{proof}
The property is obvious for initial configurations, which make up the induction
base; Lemma~\ref{weak-dimension-preservation-lemma} proves the inductive case.
\end{proof}

\subsection{Semantic soundness properties}
%Following the example
In the style of the \textit{Semantic Soundness} Theorem of
\cite[\SECTION3.7]{a-theory-of-type-polymorphism-in-programming--milner}, we
can now finally prove that ``well-dimensioned programs do not go wrong'':

\begin{theorem}[Dimension Semantic Soundness]\label{dimension-semantic-soundness-theorem}
No consistently-dimensioned expression fails because of dimension: more
formally, for all $\DHANDLE{h}{e}$ and $\Gamma$,
if $\#(\DHANDLE{h}{e}) \SQL \top$ then for each $\chi$ such that
$((\DHANDLE{h}{e}, \EMPTYSET)\  \VALUESEPARATOR\  \Gamma) \TOERT \chi$
we cannot have that $\chi~\FAILSBECAUSEOFDIMENSION$.
%we have that $\chi~\DOESNTFAILBECAUSEOFDIMENSION$.
\end{theorem}
\begin{proof}
%% We can restrict ourselves to the case in which $\chi$ is terminal, since when
%% $\chi$ is not terminal
%% we already know that  $\chi$ cannot be an error configuration
%% from Proposition~\ref{either-fail-or-not-fail}.
By contradiction, let us assume that a reachable consistently-dimensioned
expression $\DHANDLE{h}{e}$ fails because of dimension: then we have that
$\#(\DHANDLE{h}{e}) \SQL \top$ and
\linebreak
$((\DHANDLE{h}{e}, \EMPTYSET)\ \VALUESEPARATOR\ \Gamma) \TOERT \chi \FAILSBECAUSEOFDIMENSION$;
but because of Lemma~\ref{weak-dimension-preservation-lemma} each reduction
starting from the initial configuration either leaves the dimension unchanged
or lowers it, hence
$\#(\chi) \SQLE \#((\DHANDLE{h}{e},\ \EMPTYSET)\ \EMPTYSEQUENCE) \SQL \top$,
which means that $r(\chi)$ is also consistently-dimensioned.
\\
We examine all the possible cases where $\chi \FAILSBECAUSEOFDIMENSION$:
\begin{itemize}
\item
$(\DLET{h_0}{x_1...x_n}{\HOLE}{\DHANDLE{h_2}{e}},\ \rho).S
\ V\ \Gamma \FAILSBECAUSEOFDIMENSION$ when the top bundle on $V$ has less than
$n$ elements, let us say $c'_1 ... c'_k$ with $k < n$: then by
Definition~\ref{resynthesization-definition} applied twice and
Lemma~\ref{lemma-1-jv} for some context $C[\PARAMETER]$ we have that
$r(\chi) = \linebreak
C[\DLET{h_0}{x_1...x_n}{c'_1...c'_k}{\DHANDLE{h_2}{e}}]$; but then by
Definition~\ref{dimension-definition} the \LETNAME expression cannot be
consistently-dimensioned, and neither can $r(\chi)$ by
Proposition~\ref{top-dimension-propagates-outwards-proposition}: contradiction;

\item
$
        ({\DCALLWITHHOLE{h_0}{f}},\ \rho).S
        \ V\ \Gamma
        \FAILSBECAUSEOFDIMENSION
        $ when the top frame on the value stack has a wrong number of $\VALUESEPARATOR$-separated
        constants: then by Definition~\ref{resynthesization-definition} and Definition~\ref{dimension-definition}
        we have that $r(\chi) = \top$: contradiction;

\item
$
        (\DPRIMITIVE{h_0}{\pi}{\HOLE},\ \rho).S
        \ V\ \Gamma
        \FAILSBECAUSEOFDIMENSION
        $ when $\pi \HASDIMENSION n \to m$, $V \DOESNTUNIFY\ \VALUESEPARATOR c_{n} \VALUESEPARATOR c_{n-1} \VALUESEPARATOR ...$\linebreak$ \VALUESEPARATOR c_2 \VALUESEPARATOR c_1 \VALUESEPARATOR \ACTIVATIONSEPARATOR V'$:\\
        identical to the ${\DCALLWITHHOLE{h_0}{f}}$ case;
\item
$
        (\DIFIN{h_0}{\HOLE}{c_1...c_n}{\DHANDLE{h_2}{e}}{\DHANDLE{h_3}{e}},\ \rho).S
        \ V\ \Gamma
        \FAILSBECAUSEOFDIMENSION
        $ when $V \DOESNTUNIFY\ \VALUESEPARATOR c \VALUESEPARATOR V'$:\\
        by
        Definition~\ref{resynthesization-definition} and
        Definition~\ref{dimension-definition} we find immediately that $r(\chi) = \top$: contradiction;

\item
       $(\DBUNDLEWITHHOLE{h_0},\ \rho).S
        \ V\ \Gamma
        \FAILSBECAUSEOFDIMENSION$
        when $V \DOESNTUNIFY\ \VALUESEPARATOR c_1 \VALUESEPARATOR c_2 \VALUESEPARATOR ... \VALUESEPARATOR c_{n-1} \VALUESEPARATOR c_n \VALUESEPARATOR V'$:\\
        similar to the ${\DCALLWITHHOLE{h_0}{f}}$ case:
        since $\chi$ is reachable the original \BUNDLENAME expression contained exactly
        $n$ parameters, but some of them are plural: however by 
        Definition~\ref{resynthesization-definition} and
        Definition~\ref{dimension-definition} we have $r(\chi) = \top$: contradiction;
       
\item
$
        (\DFORKWITHHOLE{h_0}{f},\ \rho).S
        \ V\ \Gamma
        \FAILSBECAUSEOFDIMENSION
        $
when the top frame on the value stack has a wrong number of $\VALUESEPARATOR$-separated
constants:
identical to the ${\DCALLWITHHOLE{h_0}{f}}$ case;

\item
$
        (\DJOIN{h_0}{\HOLE},\ \rho).S
        \ V\ \Gamma
        \FAILSBECAUSEOFDIMENSION
        $ when $V \DOESNTUNIFY\ \VALUESEPARATOR c \VALUESEPARATOR V'$:\\
        identical to the $\DIFIN{h_0}{\HOLE}{c_1...c_n}{\DHANDLE{h_2}{e}}{\DHANDLE{h_3}{e}}$ case.
\QEDPROOF
\end{itemize}
\end{proof}

Corollary~\ref{well-dimensioned-programs-do-not-go-wrong-theorem}, a simple
consequence of Theorem~\ref{dimension-semantic-soundness-theorem}, extends
the semantic soundness result to whole programs by providing a sufficient
condition for avoiding dimension errors.

\begin{corollary}\label{well-dimensioned-programs-do-not-go-wrong-theorem}
Well-dimensioned-programs cannot fail because of dimension.
\QEDNOPROOF
\end{corollary}
Of course well-dimensioning is only a sufficient condition for the absence of
dimension failures: an expression containing unreachable code such as
$\DIFIN{h_0}{\DCONSTANT{h_1}{\mathcal{N}(1)}}{\mathcal{N}(2)}{\DBUNDLE{h_2}{\DCONSTANT{h_3}{3}\ \DCONSTANT{h_4}{4}}}{\DCONSTANT{h_5}{5}}$
may have dimension $\top$, without ever failing because of dimension.

\section{Reminder: why we accept ill-dimensioned programs}
Even when only speaking about static programs, we prefer \textbf{not} to
restrict ourselves to well-dimensioned programs for reasons of philosophical
coherency (\SECTION\ref{static-analyses-are-not-mandatory}), despite the
difficulty of finding believable examples of ill-dimensioned programs that we
would like to accept as ``correct''.
\\\\
We could argue that it is at least conceivable that syntactic extensions
automatically produce static but ill-dimensioned \EPSILONZERO\ programs --- which maybe
could be proved not to fail because of dimension, due to some property enjoyed by the
extension.  On the other hand the extension might actually be unsafe in the
general case, and still useful.

But independently of any extension, at the level of the core language it is
reasonable to accept any program which \textit{could} possibly yield a useful
output: since we want to respect the programmers' intelligence
\EPSILONZERO\ will not constrain the expressivity of the upper layers;
therefore we want to accept ill-dimensioned programs
\textit{and run them} until an error condition is reached, if ever. The compiler
should generate code which runs until possible, and
compilation itself should \textbf{not} fail because of ill-dimensioning.

Of course a personality implementer is always free to add static
checks generating warning messages or even fatal errors at compile time,
yielding a very safe --- if restrictive --- language.  Such languages do have a
place in the world, as shown by the experience of Ada, ML and Haskell; anyway we still hold that
%{running away and screaming}
refusing to proceed at any cost in a hysterical paralysis
is not the most useful reaction to the
discovery that a program might, or even \textit{will}, fail.

%%%%%%%%%%%%%%%%%%%%%%%%%%%%%%%%%%%%%%%%%%%%%%%%%%%%
\section{Summary}
We have shown a static semantics of \EPSILONZERO, permitting to statically infer
the dimension of the bundle each expression in a static program may evaluate to.
We have then proceeded to prove our static analysis to be sound with respect to \EPSILONZERO's
dynamic semantics, providing a sufficient condition which guarantees certain
failures not to happen at run time.

Such formal work is practical and not overly complicated, thanks to
the minimalistic nature of \EPSILONZERO.
\\
\\
Dimension analysis can be used for rejecting programs not respecting the sufficient
condition; anyway we advocate against such practice, in the interest of extensibility.

% -*- mode: latex; fill-column: 79; mode: auto-fill; mode: flyspell; buffer-file-coding-system: utf-8 -*-
%\chapter{Syntactic abstraction: macros and transformations}
%\chapter{Extension facilities{\tiny \TODOQ{[was: ``Syntactic abstraction'']}}}
\chapter{Syntactic extension}
%\chapter{Syntax and macros}
\label{transforms-chapter}
\label{transform-chapter}
\label{transformations-chapter}
\label{transformation-chapter}
\label{macros-chapter}
\label{macro-chapter}
\label{syntax-chapter}
\label{syntactic-extension-chapter}
\label{syntactic-extensions-chapter}
The core language \EPSILONZERO specified in \SECTION\ref{epsilonzero-chapter}
is useful but inconvenient for humans to write directly.  In this chapter we
are going to specify syntactic abstraction mechanisms allowing users to easily
extend the language by adding high-level syntactic forms to be automatically
rewritten into \EPSILONZERO.

Since the extension facility is defined in \EPSILONZERO itself and tightly
intertwined with the problem of expressing language syntax as a data structure,
we also need to deal with a bootstrapping problem in the process.

\minitoc

Despite some fundamental differences, the syntactic layer of \EPSILON is
strongly inspired by Lisp and indeed adopts many conventions taken from Scheme
and Common Lisp.
We are now proceeding to quickly review Lisp dialects, in order to
establish a coherent foundation for our critique.

%%%%%%%%%%%%%%%%%%%%%%%%%%%%%%%%%%%%%%
%\section{A review of Lisp and s-expressions}
\section{Preliminaries}
Lisp is a family of dynamically-typed higher-order call-by-value imperative
programming languages, suitable to be used in a functional style and
particularly convenient for symbolic processing.

The original ``LISP'' language described by John McCarthy back in
1959\footnote{McCarthy
  specified in a 1995 footnote that \cite{lisp-mccarthy}, published in April
  1960, ``was written in early 1959'': see footnote 4 at page 16 in
  \url{http://www-formal.stanford.edu/jmc/recursive.pdf}.
%  When dealing with McCarthy's original proposal from 1959 rather than modern
%  dialects, we write ``LISP'' in all caps.
}  \cite{lisp-mccarthy} has been
% variously
extended
and
independently re-implemented many times
%by different people
throughout the years
% with variations
giving birth to a wealth of \TDEF{dialects}, the most important being the large
and complex \TDEF{Common Lisp} \cite{common-lisp} and the elegant,
minimalistic
\TDEF{Scheme} \cite{scheme,r5rs,r6rs}.
%\TDEF{Scheme} \cite{scheme,rrs,r2rs,r3rs,r4rs,r5rs,r6rs}.
All dialects share the same core ideas.

Contrary to persistent misinformation, most modern Lisps are statically
scoped (``\TDEF{lexically} scoped'' in Lisp jargon); Scheme and Common Lisp in
particular have been using static scoping since their original inception in
1975 and 1984.
\\
\\
Lisp introduced several striking innovations most of which eventually found
their way into the mainstream, including
the interactive Read-Eval-Print Loop,
conditional expressions,
higher order
and
garbage collection.
Recursion has been
supported since the very beginning, and in the 1960s the possibility of
expressing a program as a collection of recursive procedures might have felt
like the most radical feature.  But what still sets apart Lisps from the
other languages after fifty years is their \TDEF{homoiconicity}: programs are
encoded using the same data structure they manipulate, which is in fact the
only existing ``data type'' in the language; such data structure, the
\TDEF{s-expression}, is simple and convenient for meta-programming and for
representing symbolic information in general.

Just to be explicit from the beginning and to prevent misunderstandings, we
already make it clear that \EPSILON's syntax will use s-expressions but
will \textit{not} be homoiconic: \EPSILON is not a Lisp.  Yet we find it best
to illustrate our solution in an incremental way, starting from a description
and critique of Lisp and then re-tracking the mind path by which we arrived at
our design.
\\
\\
In the following we give our definition of s-expressions and then proceed to
quickly review the main ideas of Lisp, without exactly following any particular
dialect.  Our lexical and syntactic conventions will mostly come from Scheme,
but our macro system will be closer to Common Lisp's.

Our \textit{meta}-linguistic conventions by contrast will be non-standard,
particularly with regard to s-expressions, in order to establish a new common
framework encompassing \EPSILON as well.  Experienced Lisp users interested in
a comparison with the traditional jargon are referred to the footnotes for
some discussion of the rationale for our changes.

%What follows should be accessible to non-Lispers, but experienced Lisp
%users should at least skim it as well, since our presentation is non-standard.

%%%%%%%%%%%%%%%%%%%%%%%%%%%%%%%%%%%%%%
%\subsection{S-expressions}
\section{S-expressions}

The s-expression is an inductive data structure: it can be seen a disjoint
union containing at least several fixed \TDEF{atomic types} and an
\TDEF{s-cons} type (pronounced ``ess-cons''); an \TDEF{s-cons} is an ordered pair of
s-expressions\footnote{S-conses are called ``conses'' in Common Lisp and
  ``pairs'' in Scheme.}.

The specific collection of atoms depends on the Lisp dialect, but at least some
types are always provided: a unique object called \TDEF{the empty list},
\TDEF{fixnums}, and \TDEF{symbols}.  Symbols are objects identified by a unique
name, which can be compared for equality with one another.  All dialects also
allow \TDEF{procedures} (zero or more s-expressions as parameters, one
s-expression as result, possibly with side effects) as s-expressions.

All practical Lisp dialects also support other atom types, including booleans and
other numeric types; other non-atomic types such as vectors are available as
well, despite not being required for our presentation.
\\
\\
It is worth to stress the \textit{disjoint-union} nature of s-expressions;
however in this slightly non-standard presentation we prefer to explicitly
specify an encoding for an s-expression as a pair made of a natural type
identifier and an element of the corresponding type.
%% It is worth to stress the \textit{disjoint-union} nature of s-expressions; and
%% in this spirit, at a slightly more abstract level, we can characterize an
%% s-expression as a pair made of a natural type identifier and an element of the
%% corresponding type.  We adopt a slightly non-standard jargon.

The following ``open-ended'' definition is slightly involved due to the nature
of s-expressions as a disjoint union whose cases, despite not being all specified,
are potentially recursive\footnote{The s-cons is not necessarily the only recursive
case. We have already hinted at the ``vectors'' (\TDEF{s-vectors}
for us) supported by all practical Lisps, whose elements are other arbitrary
s-expressions; but the idea of course is to enable the user to provide more recursive
addends herself.}:

\begin{definition}[s-expression]\label{s-expression}\label{s-expression-definition}
Let $\SET{A}_0,...,\SET{A}_{n-1}$ be an ordered collection of sets called
\TDEF{addend types} including at least the set of fixnums, the set of symbols,
the empty list singleton, the set of conses, and some set of s-expressions-to-s-expression
procedures.

We define the
\TDEF{set of s-expressions $\SET{S}_{\SET{A}_0,...,\SET{A}_{n-1}}$}
(or simply $\SET{S}$ without subscripts when the addends are clear from the context)
as
$(\{0\} \times \SET{A}_0) \UNION % \DISJOINTUNION
(\{1\} \times \SET{A}_1) \UNION % \DISJOINTUNION
...
\UNION %\DISJOINTUNION
(\{n-1\} \times \SET{A}_{n-1})$.

For each addend type $\SET{A}_i$ we also define:
\begin{itemize}
\item
the \TDEF{injected type s-$\SET{A}_i$} (pronounced like $\SET{A}_i$ preceded
by the syllable ``ess'') as $\{i\}\times\SET{A}_i$;
\item
the \TDEF{$\SET{A}_i$-injection} function
$in_{\SET{A}_i} : {\SET{A}_i} \to \SET{S}$ as $\{x \mapsto (i, x)\ |\ x \in \SET{A}_i\}$;
\item
the \TDEF{$\SET{A}_i$-ejection} partial function
$ej_{\SET{A}_i} : \SET{S} \PARTIAL {\SET{A}_i}$ as $\{(i, x) \mapsto x\ |\ x \in \SET{A}_i\}$.
\end{itemize}
%Finally we define
And
the \TDEF{untyped ejection} function $ej : \SET{S} \to \bigcup_{i}{\SET{A}_i}$ as
%$\{(i, x) \mapsto x\ |\ 0 \leq i < n,\  x \in \SET{A}_i\}$.
$\{(i, x) \mapsto x\ |\ (i, x) \in \SET{S}\}$.
\QEDDEFINITION
\end{definition}
As per Definition~\ref{s-expression-definition} we call \TDEF{s-fixnums},
\TDEF{s-symbols} and \TDEF{the empty s-list singleton} the injections of
fixnums, symbols, and the empty list singleton into s-expressions.  We call
\TDEF{s-conses} the injection of conses --- which, it is worth to stress once
more, means the set of \textit{s-expression} pairs, rather than \textit{any}
pairs.
The specific nature of \TDEF{S-procedures} depends on the Lisp dialect, but in general
we can think of them as the injection of procedures with effects accepting zero
or more s-expressions and returning one s-expression\footnote{In our presentation we use the ``s-'' prefix for
  explicitly highlighting the difference between a type and its injection into
  s-expressions, but such distinction is not needed in Lisp where
  \textit{every} object is an s-expression: s-fixnums, s-symbols and the empty
  s-list in Lisp are just ``fixnums'', ``symbols'' and ``the empty list''.
  Here we speak of an s-cons as an (injected) pair \textit{of two
    s-expressions} rather than two arbitrary objects; since here we do not have much
  use for non-injected \textit{conses} we can also avoid the issue of conses of
  non-s-expressions.}.
\\
\\
Up to this point we have defined s-expressions as a mathematical structure; but
since s-expressions are used for input and output, we also need to specify
their \textit{written notation} as a reasonably formal syntax.  However, to
avoid making our notation too heavy, we will not explicitly distinguish between
s-expression literals and their corresponding non-injected literals.
\begin{definition}[s-expression syntax]\label{s-expression-syntax}
Comments start with a semicolon and extend up to the end of the line; all
whitespace is otherwise ignored.
\begin{itemize}
\item
We write s-fixnums as strings of one or more digits in radix 10, preceded
by an optional sign;

\item
we write the empty s-list as ``\CODE{()}'', possibly with whitespace or
comments between the open and closed parentheses;

\item
we accept as an s-symbol any sequence of characters not containing spaces,
dots, semicolons, quotes, backquotes, commas or parentheses which is not
well-formed as an s-fixnum or an s-expression prefix as per Syntactic Convention~\ref{lisp-sexpression-prefix-syntactic-convention}\footnote{This description is simplified and idealized
compared to what realistic Lisps allow: practical dialects provide
escaping mechanisms to even embed \textit{whitespace} within a symbol name.};
%% compared to what any realistic Lisp allows.  Most dialects provide some
%% escaping mechanism to embed any character within a symbol name, including whitespace.};

\item
if $s_1$ and $s_2$ are s-expressions, then we write their s-cons as
``\CODE{($s_1$ .\ $s_2$)}''.
%For historical reasons
%The left element of an s-cons is called the \TDEF{car}, and the right element the \TDEF{cdr}.
\QEDDEFINITION
\end{itemize}
\end{definition}
It should be noticed that s-procedures have no syntax in
Definition~\ref{s-expression-syntax}: this means that they cannot be directly
expressed as literal constants.

Some sample s-fixnums are \CODE{1234}, \CODE{0}, \CODE{+12}, and \CODE{-42};
\CODE{()} is the empty s-list;
all of the following are s-symbols: \CODE{a}, \CODE{b}, \CODE{+},
\CODE{this-is-an-s-symbol}, \CODE{incr!}, \CODE{even?}, \CODE{pi/4}, \CODE{1-}.
Some s-conses are
\CODE{(1 .\ 2)}, \CODE{(a .\ ())},
\CODE{((a .\ 3) .\ 1)}, \CODE{((a .\ b) .\ (57 .\ d))}.
\\
\\
Since s-conses allow s-expressions to be nested at any depth, it is convenient
to unambiguously name specific substructures:
\begin{definition}[s-cons selectors]\label{s-cons-selectors}
Let $s_1$ and $s_2$ be s-expressions; we say that
the \TDEF{s-car} of \CODE{($s_1$ .\ $s_2$)} is $s_1$ and
the \TDEF{s-cdr} of \CODE{($s_1$ .\ $s_2$)} is $s_2$.

By definition, let \textit{s-car} and \textit{s-cdr} be \TDEF{s-cons selectors}; now, if
\textit{s-c$P$r} is an s-cons selector for some ``path'' $P \in \{a,d\}^{+}$, we define
\TDEF{the s-cons selector s-ca$P$r} of $s$ to be the s-car of the \textit{s-c$P$r} of $s$,
and
\TDEF{the s-cons selector s-cd$P$r} of $s$ to be the s-cdr of the \textit{s-c$P$r} of $s$\footnote{Again, we prepended the ``s-'' prefix to the traditional cons
  selector names.}.
%When pronounced, each ``a'' and ``d'' in an s-cons selector name is pronounced as a separate syllable.
When pronounced, each ``a'' and ``d'' in s-cons selector names belongs to a different syllable.
\QEDDEFINITION
\end{definition}
So for example, the s-caddr (pronounced ``ess-ca-dh-dr'') of an s-expression is the left element of the right
element of its right element, and the s-caddr of \CODE{(a .\ (b .\ (c .\ (d .\ e))))} is
\CODE{c}.

Apart for their ``s-'' prefix, introduced by us to distinguish s-expressions
from addends, s-cons selector names are well-established in Lisp.  They trace back their
% particularly
alien-sounding names to details
of the IBM 704, the machine on which McCarthy's LISP was originally implemented
\cite{lisp-mccarthy}; we retain the names for their easy composability, and as a
homage to Lisp culture.
\\
\\
Since s-conses may be nested arbitrarily, they can encode linear sequences of
any length.  Such sequences are conventionally nested on the right:
\begin{definition}[s-list]\label{s-list-definition}
We call an s-expression an \TDEF{s-list}\footnote{What we call s-list here is
  traditionally known as ``list'' or ``proper list'' in Lisp.  What we would
  refer to here as a \textit{non-s-list s-cons} is known in Lisp as a
  ``dotted list'' (since Syntactic
  Convention~\ref{compact-sexpression-notation-syntactic-convention} does not
  apply to the structure spine); somewhat confusingly, Lisp ``dotted lists'' are not ``lists''.}
(pronounced ``ess-list'') if it is either the empty s-list or an s-cons whose
s-cdr is an s-list.

If an s-list $s$ is empty then we say it has \TDEF{no elements}; otherwise we call its
\TDEF{elements} the s-car of $s$ followed by the elements of the s-cdr of $s$.
\QEDDEFINITION
\end{definition}
The following three s-expressions are
s-lists:
\CODE{()},
\CODE{(a .\ ())},
\CODE{(a .\ ((1 .\ 2) .\ ()))};
the following three s-expressions are \textit{not} s-lists:
\CODE{foo},
\CODE{(a .\ b)},
\CODE{(a .\ (b .\ c))}.

It may be worth stressing that an s-list is allowed to have other s-conses as
some or all of its elements, which are not restricted by homogeneity or indeed
any constraint on their shape.
\\
\\
S-cons syntax becomes clumsy to use when s-expressions are nested too deeply,
hence the need for the following syntactic convention:
\begin{syntacticconvention}[compact s-expression notation]
\label{compact-sexpression-notation-syntactic-convention}
An s-cons whose s-cdr is either another s-cons or the empty s-list may optionally be
written by \textit{both}:
\begin{itemize}
\item omitting the dot;
\item omitting the parentheses around the s-cdr.
\QEDSYNTACTICCONVENTION
\end{itemize}
\end{syntacticconvention}
For example the last two sample s-lists above may also be written as
\CODE{(a)}
and
\CODE{(a (1 .\ 2))};
\CODE{(a b .\ c)}
is another way of writing the last sample non-s-list above.  

Syntactic Convention~\ref{compact-sexpression-notation-syntactic-convention}
always applies to the ``spine'' of s-lists, making them more convenient to write
than alternative specular structures nested on the left.
\\
\\
It is easy to convince oneself that, even with Syntactic
Convention~\ref{compact-sexpression-notation-syntactic-convention},
s-expression notation remains unambiguous; in particular we do not need any
precedence or associativity rule to parse an s-expression in written form, nor
grouping brackets --- Far from it, parentheses are a fundamental part of the
syntax, and can never be added or removed without changing the denoted data
structure.
\\
\\
The following shorthand syntax for s-expressions will be useful later:
\begin{syntacticconvention}[Lisp s-expression prefixes]
\label{lisp-sexpression-prefix-syntactic-convention}\label{lisp-s-expression-prefix-syntactic-convention}
For any s-expression~$s$:
\begin{itemize}
\item {\rm\CODE{(quote $s$)}} may optionally be written as {\rm\CODE{'$s$}};
\item {\rm\CODE{(quasiquote $s$)}} may optionally be written as {\rm\CODE{`$s$}};
\item {\rm\CODE{(unquote $s$)}} may optionally be written as {\rm\CODE{,$s$}};
\item {\rm\CODE{(unquote-splicing $s$)}} may optionally be written as {\rm\CODE{,@$s$}}.
\end{itemize}
We say that ``\CODE{'}'', ``\CODE{`}'', ``\CODE{,}'' and ``\CODE{,@}'' are \TDEF{s-expression prefixes}.
\QEDSYNTACTICCONVENTION
\end{syntacticconvention}

\begin{figure}
\centering
{{\rm\small
\begin{tabular}{ll}
\textit{s-expression} $::=$ & \SPACERF\textit{atom} & \{ \textit{atom} \}
\SPACERP\CODE{(}\ \textit{s-expression}\ \textit{rest} & \{ s-cons(\textit{s-expression}, \textit{rest}) \}
\SPACERP\textit{prefix}\ \textit{s-expression} & \{
s-cons(lookup(\textit{prefix}), s-cons(\textit{s-expression}, \CODE{()})) \}
\\\\ \textit{rest} $::=$ & \SPACERF\CODE{)} & \{ \CODE{()} \}
\SPACERP\CODE{.\ }\textit{s-expression}\ \CODE{)} & \{ \textit{s-expression} \}
\SPACERP\textit{s-expression}\ \textit{rest} & \{ s-cons(\textit{s-expression},
\textit{rest}) \}
\end{tabular}
}}
%\end{minipage}
\caption{\label{s-expression-grammar-figure}
S-expressions can be parsed with the
  attributed LL(1) grammar \cite[\SECTIONS4-5]{dragon} above, %of course
  also supporting
  Syntactic Conventions~\ref{compact-sexpression-notation-syntactic-convention} and
  \ref{lisp-s-expression-prefix-syntactic-convention}.  The grammar is
  simple enough to allow for a hand-coded recursive-descent parser, with no
  need for generators.
  \\
  Replacing the second alternative for \textit{s-expression} with
  ``{\small $|$ \CODE{(}\ \textit{rest} \{ \textit{rest} \}}''
  yields an even simpler grammar which recognizes \CODE{(.\ $s$)} as an alternate
  degenerate form of any s-expression $s$, as in fact several Lisp
  implementations do.
  %% The grammar as shown permits the alternative syntax \CODE{(.\ $s$)} for any
  %% s-expression $s$, which is in fact permitted by many Lisp implementations.
  %% If one wishes to fail instead of accepting this degenerate case, it can be
  %% done at the cost of a making the second semantic action slightly more
  %% complicated.
}
\end{figure}

%%%%%%%%%%%%%%%%%%%%%%%%%%%%%%%%%%%%%%
%\subsection{Lisp informal syntax}
\section{Lisp syntax}
\label{lisp-syntax-section}
Up to this point we have described the syntax of the s-expression language,
without providing any corresponding semantics other than the
% simple
disjoint-union data structure; even Syntactic
Convention~\ref{lisp-s-expression-prefix-syntactic-convention} simply describes
a more compact way of writing down some inductive data structures, with no meaning
deeper
than % other than % A little poetic
their shape.
But of course the entire point of studying
s-expressions is encoding programming language syntax into them; the
``s-'' prefix indeed stands for ``symbolic'' in
\cite{lisp-mccarthy}, and s-expressions make up the syntax of (a
superset of) Lisp forms.%\TODOF{Shall I hint at m-expressions, and say that the success of s-expressions wasn't intentional according to McCarthy?}

\subsection{Lisp informal syntax}
Here we resist the temptation of formally specifying a mapping from
s-expressions to terms of a call-by-value \LAMBDACALCULUS with conditionals and
literals; such a definition would depend on the Lisp dialect details and would
be either idealized or overcomplicated, without adding much to comprehension in
any case.

%% \begin{definition}[idealized lisp syntax]\label{lisp-syntax}
%% Let $s$ be an s-expression.  Then:
%% \begin{itemize}
%% \item $\SSYNTAX{s} = x$ where $x$ is a variable named after the symbol
%%   $ej_{\SYMBOLS}(s)$, if $s$ is a symbol;
%% \item $\SSYNTAX{s} = s$ (a literal constant) if $s$ is an fixnum or the empty s-list
%% \item if it is symbol, then it is interpreted as a variable named after $ej_{\SYMBOLS}(s)$;
%% \item if it is an fixnum or the empty s-list, then it is interpreted as a
%%   literal constant with the value of $s$;
%% %\item if $s$ is a cons:
%% %\begin{itemize}
%% \item $\SSYNTAX{$\CODE{(lambda ($s_1$) $s_2$)}$} = \lambda x_1...x_n . \SSYNTAX{s_2}$
%%   where $x_1...x_n$ is the sequence of names obtained from the ejection of the
%%   symbols in the list of symbols $s_1$
%% \item $\SSYNTAX{$\CODE{(if $s_1$ $s_2$ $s_3$)}$} = \CODE{if }\SSYNTAX{s_1}\CODE{ then }\SSYNTAX{s_2}\CODE{ else }\SSYNTAX{s_3}$
%% \item $\SSYNTAX{$\CODE{($s_1$ .\ $s_2$)}$} = $
%% \end{itemize}
%% %\end{itemize}
%% \QEDDEFINITION
%% \end{definition}

The following high-level description and the examples below will suffice
to provide an intuitive idea:
\begin{syntacticconvention}[Lisp informal syntax]\label{lisp-syntax-syntactic-convention}
Let $s_1$, $s_2$
%, $s_3 and $s_4$
and $s_3$
be\\
s-expressions.
\begin{itemize}
\item
An s-symbol represents a variable with the same name as its ejection;
\item
an s-fixnum or empty s-list is \TDEF{self-evaluating}, which is to say represents itself as a literal
constant\footnote{In Scheme the ``empty list'' \CODE{()} is not considered a valid
  expression nor interpreted as a literal constant, which forces the user to
  needlessly quote literal \CODE{()} objects.  We consider this an unfortunate
  design mistake (within an otherwise quite beautiful construction), from
  which to intentionally deviate.};
\item
an s-cons whose s-car is an s-symbol in a specific set and whose s-cdr has the
right shape represents the corresponding syntactic form:
\begin{itemize}
\item \CODE{(if $s_1$ $s_2$ $s_3$)} represents a conditional, of which $s_1$
  represents the condition, $s_2$ the ``then'' branch and $s_3$ the ``else'' branch;
\item \CODE{(lambda $s_1$ .\ $s_2$)} represents an anonymous procedure with the
  parameter names encoded by $s_1$; $s_2$ represents the sequence of forms in
  the body; if $s_1$ is an s-list of s-symbols, then the parameters have the same
  names of its element ejections\footnote{It is not worth the trouble to introduce variadic
    procedures here, but this wording permits us at least not to arbitrarily exclude them.};
\item \CODE{(quote $s_1$)} represents $s_1$ as a literal constant;
%% \item \CODE{(define ($s_1$ .\ $s_2$) .\ $s_3$)} when $s_1$ is an s-symbol represents
%%   a global \textit{procedure} definition: it is an abbreviation of
%%   \CODE{(define $s_1$ (lambda $s_2$ .\ $s_3$))}

\item \CODE{(quasiquote $s_1$)} represents $s_1$ as a \TDEF{quasiquoted}
  ``mostly-literal'' structure: the result is a literal structure equal to
  $s_1$ except for substructures of the form \CODE{(unquote $s$)} or
  \CODE{(unquote-splicing $s$)}, which represent ordinary non-literal
  % I mean ``expressions'', not ``s-expressions''
  expressions:
  \begin{itemize}
  \item for an \CODE{(unquote $s$)} substructure, the result of evaluating $s$
    will replace the substructure in the quasiquoted structure;
  \item for an \CODE{(unquote-splicing $s$)} substructure the result of
    evaluating $s$, which must yield an s-list, will be spliced element by
    element within the containing s-list in the quasiquoted structure;
  \end{itemize}

\item \CODE{(define $s_1$ $s_2$)} when $s_1$ is an s-symbol represents a global
  definition; $s_2$ represents the expression to be evaluated and whose result
  will be named;

\item ...

\item
  an s-cons whose s-car is an s-symbol whose ejection is a macro name represents
  a user-defined syntactic form;
\end{itemize}

\item
an s-list whose s-car is not a syntactic form name represents a procedure
application: the s-car of the cons represents its operator, and the s-cdr contains
an s-list with its zero or more operands.
%% an s-cons whose s-car is not a syntactic form name represents a procedure
%% application: the s-car of the cons represents the operator, and the s-cdr contains
%% an s-list with the zero or more operands; it is a syntax error to have a
%% non-s-list s-cdr in this case.
\QEDSYNTACTICCONVENTION
\end{itemize}
\end{syntacticconvention}

So, for example:
\begin{itemize}
\item \CODE{57} represents the literal constant \CODE{57} (an s-fixnum);
\item \CODE{a} represents the variable named \CODE{a};
\item \CODE{'a} represents the literal constant \CODE{a} (an s-symbol);
\item \CODE{(a)} represents the application of a procedure named ``\CODE{a}'',
  with zero arguments;
\item \CODE{((a 43))} represents the application of a procedure named
  ``\CODE{a}'' with one argument, the literal constant \CODE{43}; the result,
  presumably another procedure, is in its turn applied with zero arguments;
% \item \CODE{(a .\ 43)} is ill-formed;
\item \CODE{(if (a b) c d)} represents a two-way conditional expression; if the result
  of the application of the procedure named ``\CODE{a}'' to the value of the variable
  named ``\CODE{b}'' is true, then the result is the value of the variable ``\CODE{c}'';
  otherwise the result is the value of the variable ``\CODE{d}'';
\item \CODE{(+ 1 2)} represents the application of a procedure named
  ``\CODE{+}'' to two arguments, the literal s-fixnums \CODE{1} and \CODE{2}; no special
  syntax is needed or available for arithmetic operators, which are considered
  ordinary procedures: as a consequence of Syntactic
  Convention~\ref{lisp-syntax-syntactic-convention} procedure application
  syntax is
  rigidly prefix;
\item \CODE{'(+ 1 2)}
%, which is another way of writing \CODE{(quote (+ 1 2))},
  represents the literal constant \CODE{(+ 1 2)}, which is
  an s-cons and in particular an s-list, and also happens to be a valid Lisp
  expression itself;
%% \item \CODE{''(+ 1 2)} represents the literal constant \CODE{'(+ 1 2)};
\item \CODE{'(if)} represents the literal constant \CODE{(if)}, an ordinary
  data structure which would not be valid as a Lisp expression;
\item \CODE{`(if)} also represents the literal constant \CODE{(if)};
\item \CODE{`(a ,b c)} represents an s-list of three elements: the s-symbol
  \CODE{a}, the value of the variable \CODE{b}, and the s-symbol \CODE{c};
\item \CODE{`(a ,@b c)} represents an s-list of two or more elements: the
  literal s-symbol
  \CODE{a}, all the elements of the s-list which is the value of the variable
  \CODE{b} (assumed to be an s-list), and finally the literal s-symbol \CODE{c};
\item
What follows is a reasonable definition of a recursive procedure:
\begin{Verbatim}[label={\rm \em Lisp}]
(define factorial
  (lambda (n)
    (if (= n 0)
      1
      (* n (factorial (- n 1))))))
\end{Verbatim}
% \CODE{(define fact (lambda (n) (if (= n 0) 1 (* n (fact (- n 1))))))}
The anonymous procedure is
evaluated and then globally named ``\CODE{factorial}'': the procedure has one
parameter called ``\CODE{n}'', and its body is a simple conditional: if the
result of calling the procedure ``\CODE{=}'' with the parameters \CODE{n} and
zero is true, then the result is one; otherwise the result is
the result of calling ``\CODE{*}'' with two parameters: \CODE{n}, and the
result of calling \CODE{factorial} with the result of calling ``\CODE{-}'' with
\CODE{n} and one.
\end{itemize}
Of course a small set of predefined procedures must be provided if we want to
perform arbitrary computation on s-expression data: in particular we will need
to check whether a given s-expression belongs to an addend type (for example,
the \CODE{symbol?}  procedure returns a true s-expression iff its parameter is
an s-symbol), plus constructors and selectors (for example, \CODE{cons} returns
a new s-cons containing its two parameters; \CODE{car} returns the s-car of its
parameter, which must be an s-cons); we also need a procedure \CODE{eq?} to
check whether two given s-symbols are equal.

Given such predefined procedures, it becomes conceptually easy to work on
symbolic information, including language transformers and interpreters. \cite{lisp-mccarthy}
contained the first Lisp interpreter written in itself \textit{as an ordinary
procedure}, in the space of a couple pages of code.
\\
\\
All realistic Lisps also include some macro facility, usually Turing-complete:
macros allow the user to define an s-expression-to-s-expression mapping for
rewriting a syntactic form into a combination of already available forms; a
macro may be thought of as a Lisp procedure to be automatically applied to all
instances of a user-defined form, in some phase prior to execution.

As a
simple but not unrealistic example, since global procedure definitions and
tests for zero are presumably very common, a user might prefer to be able to
write the factorial definition above in a more compact way, as:
%it could be convenient to be able
%to write the factorial definition above as
%\begin{lstlisting}
\begin{Verbatim}[label={\rm \em Lisp}]
(define-procedure (factorial n)
  (if-zero n
    1
    (* n (factorial (- n 1)))))
\end{Verbatim}
%\end{lstlisting}
User-defined forms still follow Lisp syntactic
conventions\footnote{Of course because of their different role Common Lisp
  ``reader macros'' \cite[\SECTION2.2]{common-lisp}, a form of extension for the s-expression parser, do not fit our
  classification; Common Lisp ``macros'' do.}: each use of the new forms
\CODE{define-procedure} and \CODE{if-zero} is encoded as an s-cons whose s-car
is the s-symbol uniquely identifying them.
\\
\\
Macros are a form of syntactic abstraction
(\SECTION\ref{procedural-and-syntactic-abstraction}) allowing to factorize
recurring code patterns; it should be obvious that procedural abstraction alone
as provided by \CODE{lambda} and \CODE{define} does \textit{not} suffice to express
\CODE{define-procedure} and \CODE{if-zero}, since their s-expression
subcomponents are not necessarily valid to be interpreted as expressions,
and in any case they do not follow the call-by-value evaluation
strategy of procedures.
\\
\\
As builders of syntax from other pieces of syntax, Lisp macros are a prime
example of symbolic computation, and a particularly good use case for quasiquoting.

For example, assuming the three parameters of the macro \CODE{if-zero} above to
be bound to the formals \CODE{discriminand},
\CODE{then-branch} and \CODE{else-branch}, the macro body might be as simple as
\CODE{`(if (= ,discriminand 0) ,then-branch \linebreak,else-branch)}.

%\subsection{A critique of Lisp}
\subsection{Critique}
\label{lisp-critique}
The peculiar syntax of Lisp has always been a polarizing issue for users,
either loved or despised with a violent fervor.  Without trying to pass our
personal opinions on the matter as science, we simply emphasize how powerful
macro systems of the kind hinted at above
%go hand in hand with s-expressions and homoiconicity.
are made possible by s-expressions and homoiconicity.

Syntax aside, some circles also perceive as a problem the apparent lack of
efficiency and the strongly dynamic nature of the language, including the
glaring absence of static checks.
\\
\\
As controversial topics do, Lisp has generated valid criticism and also
plenty of noise with popular slogans, myths and half-truths.

\begin{itemize}
\item
Lisp has always been used for symbolic processing, its very name standing for
``List processing''; many users consider it inherently inefficient out of the
field of symbolic computation, because of its very high abstraction level.

Of course Lisp is far from limited to ``lists'' (s-lists for us); in fact
s-lists are but an s-expression subset, useful in practice but not any more
``primitive'' than others.
More importantly, all practical dialects have also included addend types such
as random-access vectors and strings for decades; we avoided them in our
presentation of s-expressions simply because such addends are not needed for
encoding syntax, and in fact this lack of a homoiconic role might actually contribute to
make them less visible --- yet, they exist.

Lisp can be compiled with reasonable efficiency, but some overhead due to its
strongly dynamic nature is indeed hard to overcome.

\item
In particular Lisp is dynamically-typed at its core: there is only one data
type, the s-expression.  Apart from some runtime tagging and checking cost, the
main perceived problem is the difficulty of proving any useful static
properties on realistic programs.  It is not clear whether the language can be
made safer without seriously compromising its expressivity.

We consider this criticism to be valid.

\item
Popular claims according to which ``Lisp programs are abstract syntax trees''
or ``Lisp has no syntax'' (intended as a positive, negative or neutral remark
according to the speaker) can be taken as poetic exaggerations at best.

Equating valid expressions to ASTs is an oversimplification: in fact
most s-expressions do \textit{not} map into valid expressions, and the
difference between s-expressions and abstract syntax is relevant in practice.
The slogan would be slightly more believable if syntax were encoded as an ML-style
sum-of-products type, with its rigid constraints on arity and typing --- but
that would come with a high cost in extensibility.

Lisp syntax looks uniform when compared to traditional solutions,
but it is not nearly as regular as it could be; for example the two atomic
s-expressions \CODE{1} and \CODE{a} are interpreted in radically
different ways, the first as a literal s-fixnum and the second as a variable.
A literal s-symbol needs to be quoted as in \CODE{'a}, while a literal
s-fixnum may be indifferently quoted (once) or not: the s-expressions \CODE{1}
and \CODE{'1} are mapped into the exact same expression.  Procedure application
syntax is also problematic: an s-expression such as \CODE{(a b c)} is regarded
as a procedure call only ``as a fallback case'', % which is to say
%if
when
the s-symbol \CODE{a} does not happen to be the name of some syntactic form.

\label{lisp-syntax-could-be-made-more-regular}
Could Lisp syntax be made more regular?  Of course yes: as an alternative we could
require form names as explicit s-symbols in the first position of s-lists also
for variables and calls, and require quoting for all literals.  Then instead of
\CODE{(* n (f (- n 1)))}
we would have something like
``\CODE{(call (variable *) (variable n) (call (variable f) (call (variable -) (variable n) \\'1)))}'',
more uniform but hardly more convenient.  Notation would remain clumsy
even after introducing new s-expression prefix syntax for ``\CODE{variable}''
and ``\CODE{call}'' in the style of Syntactic
Convention~\ref{lisp-s-expression-prefix-syntactic-convention}: for example,
the very cluttered s-expression ``\CODE{@(\$* \$n @(\$f @(\$- \$n '1)))}'' is
a representation of the expression above, \textit{assuming an s-expression syntax
amendment} disallowing ``\CODE{\$}'' characters in symbol names ---
without the syntax change we would need more
whitespace, as in ``\CODE{@(\$ * \$ n @(\$ f @(\$ - \$ n '1)))}''.

Lisp syntax is a compromise and a consequence of conscious design decisions rather
than historical accidents, and these issues have been known for decades:
\cite[``\{\CODE{FUNCALL} is a pain\}'', pp.~26-27]{rrs} already deals with the
problem of using ``lists'' both for procedure calls and for other forms.

We have to recognize that Lisp notation in practice is useful and justified
as it stands, despite its relative asymmetry.
\end{itemize}

\section{Syntactic extensions: the \EPSILONONE personality}
In the following we are going to build upon the experience of Lisp and address
all three points in \SECTION\ref{lisp-critique}, so that:
\begin{itemize}
\item
  \EPSILON be efficiently implementable, and not especially tied to symbolic processing;
\item
  personalities remain open to any typing policy: strong, weak, static,
  dynamic, hybrid, or none at all;
\item
  \EPSILON syntax be at least as convenient as Lisp's while remaining simple to
  describe and extend.
\end{itemize}
For extensibility's stake, we use s-expressions to encode language syntax, as
Lisp dialects do; but differently from Lisp we choose to decouple syntax and generic
data structures, so that s-expressions are available as objects to compute just
\textit{as one data type among a wealth of others}: in practice data of each
addend type are available either injected into s-expressions (for example
s-fixnums, s-symbols), or untagged (fixnums, symbols): thus s-expressions
become a way of \textit{selectively} employing dynamic typing in a world where
untyped objects are also available, with injection and ejection operators to
provide a link between the two representations.  S-expressions are always used to
represent syntax before macroexpansion, but a user is free to employ them at
run time as well if she chooses to, where dynamic typing feels
more convenient.  For generality's sake, \textit{we want s-expressions to be
extensible} so that the user may provide more addends.

Expressions are just one addend type, \textit{distinct from s-expressions};
\EPSILONZERO expressions may be built and analyzed with constructor and selector
operators, injected to and ejected from s-expressions.  Said even more explicitly,
in our solution we have that s-expressions are distinct from
\textit{injected expressions}; and macros act like procedures turning
s-expressions into untagged expressions.  Moving farther from Lisp we will also
define \TDEF{transforms} (\SECTION\ref{transforms}), as a way of systematically
turning (possibly extended) untagged expressions into other untagged
expressions.

\subsection*{The personality stack}
The language roughly outlined above constitutes a personality we call \EPSILONONE:
%, which consitutes a qualitative step forward from \EPSILONZERO in terms of expressivity:
%\begin{quotation}
%\begin{definition*}
\begin{itemize}
\item
The \TDEF{\EPSILONONE personality} corresponds to \EPSILONZERO augmented with
forms to define \textit{globals}, \textit{procedures}, \textit{macros} and
\textit{transforms}; plus some utility library.
\item
Thanks to macros and transforms \EPSILONONE is suitable to further extend into
higher-level personalities.
\item
We call \TDEF{\EPSILON} \textit{the whole system}, including \EPSILONZERO, \EPSILONONE
and other (at this point still hypothetical) higher-level personalities built on top of
\EPSILONONE.
\end{itemize}
%\end{definition*}
%\end{quotation}
Higher-level personalities will contain macro and transform definitions in the style of the
ones of \SECTION\ref{sample-extensions}, later in this chapter.
\\
\\
As a language, \EPSILONONE has an abstraction level between
\EPSILONZERO and Lisp, closer to the former.
\label{garbage-collection-is-not-necessary}
Not necessarily aimed at the final user, \EPSILONONE has a low-level feel and
is by design unsafe and unforgiving: operators can be applied to the wrong
operands with no type checking at all, and pointers are explicit.  It lends itself to efficient execution,
and is portable if used correctly.  \EPSILONONE is compatible with garbage
collection but does not require it: the \TDEF{residual program} resulting when
all syntactic abstractions are transformed away might very well use manual memory
management only.
\\
\\
The implementation language of \EPSILONONE is \EPSILONZERO,
taking advantage of Scheme for bootstrapping only.  The implementation forces
us to
%finally
commit some decisions which we had left open in the description of \EPSILONZERO
in \SECTION\ref{epsilonzero-chapter}, such as the actual definition of names
and handles in terms of data structures.  All of this has a bearing on
\EPSILONONE, and in our solution the implementations of \EPSILONZERO and
\EPSILONONE are intimately bound: an implementation of \EPSILONZERO alone
directly parsing the syntax of Definition~\ref{epsilonzero-syntax}, despite
being certainly possible, in practice would be little more than an idle
exercise without the syntactic extension mechanisms of \EPSILONONE.

%% .\RATIONALEF{I don't want to speak about
%%   the lack of an \EPSILONZERO parser here; we will deal with that very soon}

%%%%%%%%%%%%%%%%%%
\subsection{Definition via bootstrapping}
One central idea of \EPSILON is to keep the core language as simple as
possible, and have more complex linguistic features defined \textit{as code}.
As a consequence of this strategy, a formal specification of \EPSILONZERO
automatically constitutes a formal\footnote{Of course up to the
details we did not describe, such as primitives.} specification of \EPSILONONE as well, if we
keep into account the source code to bootstrap it from \EPSILONZERO.  Our
implementation thus also serves as a specification of \EPSILONONE: code, rather
than much less flexible mathematics.
\\
\\
\label{bootstrapping}The bootstrapping process is nontrivial, and relying as it does on
alternative implementations of the same data types, macros, side effects on a
global state and \textit{unexec} it provides a particularly poor fit for the
graphical notation of T diagrams \cite{t-diagrams},
% \cite{t-diagrams,partial-evaluation}
\cite[\SECTION3]{partial-evaluation}; here we will resort to plain English
to describe the bootstrap phases, and present the source code following by
necessity a \textit{bottom-up} style.

The general plan, developed in greater detail throughout the rest of this section, consists of
four phases:
\begin{enumerate}[\em (i)]
\item extend Scheme by adding untyped data (\SECTION\ref{bootstrap-phase-1});
\item implement \EPSILONZERO with s-expression syntax plus definition forms
  using Scheme macros (\SECTION\ref{bootstrap-phase-2});
\item in this temporary \EPSILONZERO implementation, build the core data structures we
  need upon untyped data, an \EPSILONZERO self-interpreter relying on reflective global
  structures,
% \EPSILONONE
  macros and
% \EPSILONONE
  transforms (\SECTION\ref{bootstrap-phase-3});
\item fill reflective global structures by re-interpreting the core
      definitions above, so that the interpreter becomes usable (\SECTION\ref{bootstrap-phase-4});
\end{enumerate}
%% After the fourth phase we can use our syntactic extension mechanisms to make
%% \EPSILONONE more convenient to use, and start experimenting.
%% \\
%% \\
%% \color{red}
As it should be clear now, developing \EPSILONONE from \EPSILONZERO up to the
point where we can define
%an \EPSILONZERO self-interpreter,
s-expressions, macroexpansion and transforms requires a certain amount of code
(about 2000 lines) in which we have to use \EPSILONZERO to build some machinery,
much of which is useful as part of a generic utility library as well and hence
deserves to be considered as ``belonging'' to \EPSILONONE.  Part of the ``library'' in
\EPSILONONE exists because of this necessity, while most of the rest relies on
syntactic abstraction and is defined \textit{after} the fourth phase, the aim
being simply to make \EPSILONONE more convenient to use
(\SECTION\ref{sample-extensions}).

The fourth phase, after which the global state can be queried, also makes
it possible to \textit{unexec}
% (\SECTION\ref{reflection-chapter})
away from Guile
into a different runtime
(\SECTION\ref{more-than-one-runtime}).
%% \color{black}
\\
\\
In the following we are going to show code snippets from the implementation,
which is available in a public bzr repository on GNU Savannah:
\url{https://savannah.gnu.org/bzr/?group=epsilon} \RED{[2015 note: the
repository switched to git in late 2013: see \SECTION\ref{new-repo}]}.  We will usually omit
or condense comments and may change indentation for reasons of space, but we will
\textit{not} simplify the code for this presentation.

This discussion deals with the state of the implementation as of Summer 2012.

%%%%%%%%%%%%%%%%%%%%%%%%%%
\subsubsection{Phase~\BOOTSTRAPPHASE{i}: extend Scheme with untyped data}
\label{bootstrap-phase-1}

In order to eventually free ourselves from the dependency on Scheme, we need to
define our own data structures which are not based on the predefined version of
s-expressions.  To simplify debugging and avoid reusing Scheme
features \textit{by mistake}, it is also useful to make our data structure
incompatible with predefined s-expressions addends; and since we want to unexec
in the end, our ``untyped'' data structure will actually need \textit{boxedness tags}
(\SECTION\ref{reflection--tagging-for-unexecing}).

We use Guile \cite{guile} as our Scheme implementation for bootstrapping.  One
of the intended applications of Guile is as an embeddable Scheme system to make
C applications extensible in the style of Emacs \cite{emacs}, and in view of
this use case Guile's C interface was made particularly convenient and
flexible; we used it to define in C our new ``type'' that we call
\TDEF{whatever}, and operations over it.  Boxedness tags serve only for the internal
Guile garbage-collection machinery, at unexec time, and for debugging memory dumps
(\SECTION\ref{memory-dumps}); but data structures built with whatevers should be
\textit{thought of} as untyped most of the time, as in fact they are conceived
for being eventually unexeced into untyped objects, dropping any tagging
information.

Whatever operations help to prevent possible mistakes during the bootstrap
process by actually performing dynamic checks on tags, in particular to
prevent non-whatever objects from being written into whatever buffer slots: whatever
data structures must remain \textit{closed over the ``points-to'' relation}, so that no
dependency on Guile s-expressions can remain at unexec time, and instead
whatevers only refer other whatevers.
\\
\\
Since Guile is a Scheme implementation its only data type is the s-expression,
of which whatevers are seen as just one more addend type: in our extended Guile
it is possible to dynamically check whether an s-expression is a \textit{whatever
  injection}.
\\
\\
The implementation of this phase, mostly in
\FILE{bootstrap/{\allowbreak}whatever-{\allowbreak}guile/{\allowbreak}whatever-{\allowbreak}guile.c},
is dirty and not especially interesting in itself.  We defined the whatever ``type'' in C as a
\textit{Smob} \cite[\SECTION{}Defining New Types (Smobs)]{guile}.  Whatevers have the
printed syntax of
Syntactic Convention~\ref{memory-dump-syntactic-convention}, also using
ANSI terminal color escape sequences to help the user to recognize boxedness
at a glance.

The same C source file also defines operations over whatevers, making them
accessible to Scheme: there are trivial \textit{conversion operators} (for example from
Scheme (injected) fixnum or (injected) threads to whatever and vice-versa),
plus what \EPSILONZERO sees as primitives:
\begin{itemize}
\item
  \textit{arithmetic} and \textit{bitwise-logic} operators;
\item
  \textit{memory} allocation, disposing, lookup and update;
\item
  very simple \textit{input/output};
\item
  \textit{unexecing} primitives, for checking the boxedness tags and buffer sizes;
\item
  the single primitive \CODE{state:update-{\allowbreak}globals-{\allowbreak}and-{\allowbreak}procedures!},
  needed for transforms (\SECTION\ref{atomic-global-update-primitive}).
\end{itemize}
%simple \textit{arithmetic} and \textit{bitwise-logic} operators,
% \textit{memory} allocation, disposing, lookup and update, and very simple
% \textit{input/output}.
Primitives number around 30.
\\
\\
The result of this phase is \FILE{guile+whatever}, an extended Guile which can
also be used interactively, supporting our whatever objects while remaining completely
compatible with Scheme.  We will not show examples of its use, because some
counter-intuitive choices were dictated by efficiency concerns; the details of
the \FILE{guile+whatever} system become irrelevant anyway after phase
\BOOTSTRAPPHASE{ii}.

%%%%%%%%%%%%%%%%%%%%%%%%%%
\subsubsection{Phase~\BOOTSTRAPPHASE{ii}: implement \EPSILONZERO in extended Scheme}
\label{bootstrap-phase-2}
\label{call-indirect}\label{callindirect}
Our implementation uses a variant of \EPSILONZERO
in which the grammar of Definition~\ref{epsilonzero-syntax} is augmented by one
more production, for an \TDEF{indirect call} form:
$$e ::= \DCALLINDIRECT{h}{e}{e^{*}}$$
We avoid a formal specification of semantics
  for \CALLINDIRECTNAME: the idea is simply calling a procedure
  whose name is computed at run time as the result of an expression; parameters
  are evaluated call-by-value left-to-right as always in \EPSILONZERO, first the operator and then the operands.

It is easy to convince oneself that adding \CALLINDIRECTNAME  is a quite harmless
optimization, as its effect can be easily simulated by automatically
generating an ``\textit{apply} function'' dispatching over one of its
parameters, as in Reynolds' defunctionalization \cite{defunctionalization}.  In fact we
do that \textit{as well}, as a proof of concept in \FILE{bootstrap/{\allowbreak}scheme/{\allowbreak}core.e}
(\SECTION\ref{applier-procedures}).
\\
\\
Before we can use \EPSILONZERO for implementing \EPSILONONE, we need of course
a syntax for \EPSILONZERO.  In typical bootstrapping fashion, we would like to
define it using the language itself, \EPSILONZERO (or maybe
\EPSILONONE, for maintainability's stake) --- but no parser is available.  Our
solution is mapping an s-expression encoding of \EPSILONZERO syntax into
Scheme, by using Guile macros\footnote{We used Guile's non-standard
  Common Lisp-style macros instead of the standard R5RS hygienic macros
  \cite{r5rs} which Guile also provides.  This choice has no particularly deep reason
  except esthetic consistency with \EPSILONONE's macro system; an
  implementation based on hygienic macros would have worked just as well.}.

Later we will provide another cleaner frontend implementation\footnote{This is not implemented yet: see \SECTION\ref{implementation-status}.} in \EPSILONONE,
to break the bootstrap dependency from Guile; that second frontend will be
backward-compatible with this bootstrap implementation of \EPSILONZERO.

As a consequence of this decision, it is natural for our implementation to use
\textit{symbols} for names, encompassing \textit{all} the sets of variable,
procedure and primitive names $\SET{X}$, $\SET{F}$, $\SET{\Pi}$.
\\
\\
The following definition, very simple despite its length, follows the spirit of Syntactic
Convention~\ref{lisp-syntax-syntactic-convention} but is more rigorous due
to its importance: reading \EPSILONZERO syntax as encoded into s-expressions is
key to understand most of the details in the bootstrap process.

Since macros are not supported in this phase but the process is akin to
macroexpansion, we name this rewriting of an s-expression into an expression
\TDEF{non-macro expansion}.  The choice of generated fresh handles is
immaterial in practice, so we speak of non-macro expansion as of a function.

\begin{definition}[non-macro expansion]\label{nonmacro-expansion-definition}
Let $s, s_1, s_2, s_3, s_4$ be s-expressions.  Then we define,
up to the choice of a fresh $h' \in \SET{H}$,
the \TDEF{non-macro expansion} function
{\rm $\EXPANDE{\PARAMETER} : \SET{S} \PARTIAL \SET{E}$} as:
\begin{itemize}
\item
{\rm $\EXPANDE{\CODE{(e0:variable }s\CODE{)}}
  =
  \DVARIABLE{h'}{x}$}
where {\rm $x = ej_{symbol}(s)$};
\item
{\rm $\EXPANDE{\CODE{(e0:value }s\CODE{)}}
  =
  \DCONSTANT{h'}{c}$}
where {\rm $c = ej(s)$};
\item
{\rm $\EXPANDE{\CODE{(e0:let }s_1\CODE{ }s_2\CODE{ }s_3\CODE{)}}
  =
  \DLET{h'}{x_1...x_n}{\DHANDLE{h_1}{e}}{\DHANDLE{h_2}{e}}$}
%\\%\linebreak
where {\rm $\SEQUENCE{x_1...x_n} = \EXPANDXS{s_1}$},
$\DHANDLE{h_1}{e} = {\EXPANDE{s_2}}$
and
{\rm $\DHANDLE{h_2}{e} = {\EXPANDE{s_3}}$};
\item
{\rm $\EXPANDE{\CODE{(e0:call }s_1{\ .\ }s_2\CODE{)}}
  =
  \DCALL{h'}{f}{\DHANDLE{h_1}{e}...\DHANDLE{h_n}{e}}$}
where $f = ej_{symbol}(s_1)$
and
$\SEQUENCE{\DHANDLE{h_1}{e}...\DHANDLE{h_n}{e}} = \EXPANDES{s_2}$;
\item
{\rm $\EXPANDE{\CODE{(e0:call-indirect }s_1{\ .\ }s_2\CODE{)}}
  =
  \DCALLINDIRECT{h'}{\DHANDLE{h_0}{e}}{\DHANDLE{h_1}{e}...\DHANDLE{h_n}{e}}$}
where {\rm $\DHANDLE{h_0}{e} = {\EXPANDE{s_1}}$}
and
$\SEQUENCE{\DHANDLE{h_1}{e}...\DHANDLE{h_n}{e}} = \EXPANDES{s_2}$;
\item
{\rm
  $\EXPANDE{\CODE{(e0:primitive }s_1{\ .\ }s_2\CODE{)}}
  =
  \DPRIMITIVE{h'}{\pi}{\DHANDLE{h_1}{e}...\DHANDLE{h_n}{e}}$}
where $\pi = ej_{symbol}(s_1)$
and
$\SEQUENCE{\DHANDLE{h_1}{e}...\DHANDLE{h_n}{e}} = \EXPANDES{s_2}$;
\item
{\rm $\EXPANDE{\CODE{(e0:if-in }s_1\CODE{ }s_2\CODE{ }s_3\CODE{ }s_4\CODE{)}}
  =
  \DIFIN{h'}{\DHANDLE{h_1}{e}}{c_1...c_n}{\DHANDLE{h_2}{e}}{\DHANDLE{h_3}{e}}$}
where
{\rm $\DHANDLE{h_1}{e} = {\EXPANDE{s_1}}$},
{\rm $\SEQUENCE{c_1...c_n} = {\EXPANDCS{s_2}}$},
{\rm $\DHANDLE{h_2}{e} = {\EXPANDE{s_3}}$}
and
{\rm $\DHANDLE{h_3}{e} = {\EXPANDE{s_4}}$};
\item
{\rm $\EXPANDE{\CODE{(e0:fork }s_1{\ .\ }s_2\CODE{)}}
  =
  \DFORK{h'}{f}{\DHANDLE{h_1}{e}...\DHANDLE{h_n}{e}}$}
where $f = ej_{symbol}(s_1)$
and
$\SEQUENCE{\DHANDLE{h_1}{e}...\DHANDLE{h_n}{e}} = \EXPANDES{s_2}$;
\item
{\rm $\EXPANDE{\CODE{(e0:join }s\CODE{)}}
  =
  \DJOIN{h'}{\DHANDLE{h_1}{e}}$}
where
{\rm $\DHANDLE{h_1}{e} = {\EXPANDE{s}}$};
\item
we do not explicitly specify {\rm $\EXPANDE{\CODE{(e1:define .\ }s\CODE{)}}$};
\item
{\rm $\EXPANDE{\CODE{(e0:bundle .\ }s\CODE{)}}
  =
  \DBUNDLE{h'}{\DHANDLE{h_1}{e}...\DHANDLE{h_n}{e}}$}
where
$\SEQUENCE{\DHANDLE{h_1}{e}...\DHANDLE{h_n}{e}} = \EXPANDES{s}$;
\item
{\rm $\EXPANDE{s}
  =
  \EXPANDE{\CODE{(e0:variable }s\CODE{)}}$}
where $s$ is an s-symbol;
\item
{\rm $\EXPANDE{\CODE{(}s_1{\ .\ }s_2\CODE{)}}
  =
  \EXPANDE{\CODE{(e0:call }s_1{\ .\ }s_2\CODE{)}}$}
where $s_1$ is an s-symbol not in
{\rm
$\{
\CODE{e0:variable},
\linebreak
\CODE{e0:value},
\CODE{e0:let},
\CODE{e0:call},
\CODE{e0:call-indirect},
\CODE{e0:primitive},
\CODE{e0:if-in},
\linebreak
\CODE{e0:fork},
\CODE{e0:join},
\CODE{e0:bundle},
\CODE{e1:define}\}$};
%and $s_2$ is an s-list;
\end{itemize}
where the non-macro sequence expander {\rm $\EXPANDES{\PARAMETER} : \SET{S} \PARTIAL \SET{E}^{*}$} is:
\begin{itemize}
\item
{\rm $\EXPANDES{\CODE{()}}
  =
  \SEQUENCE{}$};
\item
{\rm $\EXPANDES{\CODE{(}s_1\CODE{ .\ }s_2\CODE{)}}
  =
  \DHANDLE{h_1}{e}.\EXPANDES{s_2}
$} where
$\DHANDLE{h_1}{e} = {\EXPANDE{s_1}}$;
\end{itemize}
the symbol sequence expander
{\rm $\EXPANDXS{\PARAMETER} : \SET{S} \PARTIAL \SET{X}^{*}$} is:
\begin{itemize}
\item
{\rm $\EXPANDXS{\CODE{()}}
  =
  \SEQUENCE{}$};
\item
{\rm $\EXPANDXS{\CODE{(}s_1\CODE{ .\ }s_2\CODE{)}}
  =
  x.\EXPANDXS{s_2}
$} where $x = ej_{symbol}(s_1)$;
\end{itemize}
and
the value sequence expander
{\rm $\EXPANDCS{\PARAMETER} : \SET{S} \PARTIAL \SET{C}^{*}$} is:
\begin{itemize}
\item
{\rm $\EXPANDCS{\CODE{()}}
  =
  \SEQUENCE{}$};
\item
{\rm $\EXPANDCS{\CODE{(}s_1\CODE{ .\ }s_2\CODE{)}}
  =
  c.\EXPANDCS{s_2}
$} where $c = ej(s_1)$.
\QEDDEFINITION
\end{itemize}
\end{definition}
The file \FILE{bootstrap/scheme/epsilon0-in-scheme.scm} implements non-macro
expansion with Scheme macros.  After loading it from \FILE{guile+whatever}, Scheme and
\EPSILONZERO can be used together:
\begin{itemize}
\item
\CODE{(+ 1 2)} yields \CODE{3} as a Guile s-expression.
\item
\CODE{(e0:primitive fixnum:+ (e0:value 1) (e0:value 2))} yields
the injected whatever
\WHATEVER{3}, written in green as ``\textcolor{darkgreen}{\WHATEVER{3}}'' (\SECTION\ref{memory-dumps}).
\end{itemize}
It is worth remarking how Definition~\ref{nonmacro-expansion-definition} does
not define any self-evaluating atom, since doing so would create ambiguity with
Scheme's predefined self-evaluating atoms: using \CODE{e0:value} in cases such
as the example above is hence necessary, at this stage: for example \CODE{(e0:value 2)}
generates \WHATEVER{2} as an injected whatever literal constant, which is
different from Guile's \CODE{2}.

By contrast it is not necessary to use \CODE{e0:variable} and \CODE{e0:call}
for implementing variables and procedure calls, as a consequence of the fact
that Scheme and \EPSILONZERO share the same namespace for identifiers --- at
least at this stage.
\\
\\
At this point \EPSILONZERO would be usable as an implementation language, if
it provided some way of defining procedures and updating the global
environment.  A correct implementation of such facilities relies on reflective
data structures and therefore belongs in Phase~\BOOTSTRAPPHASE{iii} or even later; but once
more we can use Guile to solve the bootstrap problem and provide a temporary
implementation of an \CODE{e1:define} form.

As for Scheme's \CODE{define}\footnote{An important difference with
  respect to Scheme is how our definition facility always works on
  \textit{state environments} (\SECTION\ref{state-environment-introduction}), therefore
  at the \textit{top} level, and can be invoked anywhere an expression can
  occur in the code, at any
  nesting level.
  By contrast Scheme's definition facility updates the ``current'' environment,
  which happens to be the global one only if the form is used at the top level.
  Implementing \EPSILON's definition form over Guile required a relatively advanced and
  non-portable hack relying on Guile's module system.  See the definition of
  \CODE{define-object-from-anywhere} in
  \FILE{bootstrap/scheme/epsilon0-in-scheme.scm} for the gory details.}, we use
the same form for defining either a non-procedure or a procedure, according to the
shape of its s-cadr --- respectively an s-symbol, or an s-list of one or more s-symbols.

For a non-procedure definition, the second parameter is non-macro expanded,
evaluated and the result bound to the symbol-ejection of the first parameter;
for a procedure definition, the second parameter is non-macro expanded and
bound as the body of a procedure whose name is the symbol-ejection of the s-car
of the first parameter; the s-cdr of the first parameter contains an s-list of
s-symbols whose ejections make up the procedure formals.
\\
\\
Again, \CODE{e1:define} is important for understanding the bootstrapping code
and deserves a more precise description.  Without explicitly specifying a
non-macro expansion of an \CODE{e1:define} form into an \EPSILONZERO expression,
we describe the behavior we require from such an expression:

\begin{axiom}[definition forms]\label{definition-forms-bootstrap-axiom}
Let $s_1$ and $s_2$ be s-expressions.  Then, if
{\rm $\DHANDLE{h}{e} = \EXPANDE{\CODE{(e1:define }s_1\CODE{ }s_2\CODE{)}}$}:
\begin{itemize}
\item if
{\rm$x = ej_{symbol}(s_1)$, $\DHANDLE{h_2}{e} = \EXPANDE{s_2}$} and
{\rm$\DHANDLE{h_2}{e}\ \Gamma \ \CONVERGESE\ \SEQUENCE{c}\ \Gamma'$}
then we have that\\
{\rm$\DHANDLE{h}{e}\ \Gamma \ \CONVERGESE \ \SEQUENCE{}\ \UPDATESTATEIN{\Gamma'}{\CODE{global-environment}}{x}{c}$};
\item if
{\rm$\SEQUENCE{f,x_1...x_n} = \EXPANDXS{s_1}$}
and {\rm$\DHANDLE{h_2}{e} = \EXPANDE{s_2}$}
then we have that\\
{\rm$\DHANDLE{h}{e}\ \Gamma \ \CONVERGESE \ \SEQUENCE{}\ \UPDATESTATEIN{\Gamma}{\CODE{procedures}}{f}{(\SEQUENCE{x_1...x_n}, \DHANDLE{h_2}{e})}$}.
\QEDAXIOM
\end{itemize}
%% {\rm
%% \begin{prooftree}
%%   \LeftLabel{}
%%   \RightLabel{}
%%   \AxiomC{$\DHANDLE{h}{e} = \EXPANDE{\CODE{(e1:define }s_1\CODE{ }s_2\CODE{)}}$}
%%   \AxiomC{$x = ej_{symbol}(s_1)$}
%%   \AxiomC{$\DHANDLE{h_2}{e} = \EXPANDE{s_2}$}
%%   \AxiomC{$\DHANDLE{h_2}{e}\ \Gamma
%%     \ \CONVERGESE\ 
%%     \SEQUENCE{c}\ \Gamma'$}
%%   \QuaternaryInfC{$\DHANDLE{h}{e}\ \Gamma
%%     \ \CONVERGESE \ 
%%     \SEQUENCE{}\ \UPDATESTATEIN{\Gamma'}{\CODE{global-environment}}{x}{c}
%%     $}
%% \end{prooftree}
%% \begin{prooftree}
%%   \LeftLabel{}
%%   \RightLabel{\ \ \QEDAXIOM}
%%   \AxiomC{$\DHANDLE{h}{e} = \EXPANDE{\CODE{(e1:define }s_1\CODE{ }s_2\CODE{)}}$}
%%   \AxiomC{$\SEQUENCE{f,x_1...x_n} = \EXPANDXS{s_1}$}
%%   \AxiomC{$\DHANDLE{h_2}{e} = \EXPANDE{s_2}$}
%%   \TrinaryInfC{$\DHANDLE{h}{e}\ \Gamma
%%     \ \CONVERGESE \ 
%%     \SEQUENCE{}\ \UPDATESTATEIN{\Gamma}{\CODE{procedures}}{f}{(\SEQUENCE{x_1...x_n}, \DHANDLE{h_2}{e})}
%%     $}
%% \end{prooftree}}
\end{axiom}
The two cases are trivially exclusive, as $s_1$ cannot be an s-symbol and an
s-list at the same time.
\\
\\
The reason why we did not provide an explicit non-macro expansion for
\CODE{e1:define} in
Definition~\ref{nonmacro-expansion-definition}
should be obvious at this point; since the actual
implementation is in Scheme and its semantics is very clear, we have avoided
writing a uselessly complex expansion into \EPSILONZERO assuming some
global-updating primitive, even if that would have been possible; in particular it would have
been very painful to provide an explicit encoding of \EPSILONZERO
expressions as \EPSILONZERO data; the problem will be dealt with in
Phase~\BOOTSTRAPPHASE{iii} and in \SECTION\ref{expressions-as-an-extensible-sum-type}, where it becomes relevant
for the implementation.
\\
\\
\label{e1define-updates-global-data-structures}The actual Guile definition
of \CODE{e1:define} shows an interesting feature: after performing the binding,
\CODE{e1:define} updates global (Scheme) data structures keeping track of all
the \EPSILONZERO definitions which have been performed, including procedure
bodies.  The need for this will become apparent in Phase~\BOOTSTRAPPHASE{iv}.

%%%%%%%%%%%%%%%%%%%%%%%%%%
\subsubsection{Phase~\BOOTSTRAPPHASE{iii}: build reflective data structures and interpreter in \EPSILONZERO}
\label{bootstrap-phase-3}
The purpose of this long phase is to define the global reflective data
structures holding the program state and then the interpreter.  We start from
the associated library functionality we need, using only \EPSILONZERO in its
s-expression encoding and \CODE{e1:define}.  The task is complicated by the
restrictions of the language, allowing for procedural but not syntactic
abstraction.

Equipped with Definition~\ref{nonmacro-expansion-definition} and Axiom~\ref{definition-forms-bootstrap-axiom}, the reader should be able to easily follow
our running commentary on the main sections of \FILE{bootstrap/scheme/core.e}.
Each section is delimited by a well-visible comment including a full line of
semicolons.

The general low-level ``feel'' of \EPSILONONE becomes apparent right from the
first sections: to create some order in the context of a global flat namespace,
we adopt the convention of having all procedure and global names begin with a
reasonable \TDEF{namespace prefix} delimited by a colon.  Most of the
procedures we define must work also after unexecing, hence they must not rely
on unexecing tags primitives: all the code in \FILE{bootstrap/scheme/core.e}
works on untyped objects ignoring any boxedness tags; and of course the
distinction between booleans, characters or small fixnums is purely
conventional; a whatever \WHATEVER{0} object may represent the number zero,
the false boolean or even a null pointer, according to the context: the machine
representation after unexecing is exactly the same.
\\
\\
To make \EPSILONZERO more convenient to write, we usually define global
procedures to wrap primitives, with their same names; either of the definitions of the
\TDEF{test-for-zero} procedure \CODE{whatever:zero?} and the \TDEF{equality-by-identity}
procedure \CODE{whatever:eq?} in the first section \textit{Utility procedures working on
  any data} is a good example:
\begin{Verbatim}
(e1:define (whatever:zero? a)
  (e0:primitive whatever:zero? a))
\end{Verbatim}
Such definitions only serve to simplify calling syntax, for example allowing the user to write
\CODE{(whatever:zero?\ $s$)} instead of
\CODE{(e0:primitive whatever:zero?\ $s$)}; in terms of procedural ``abstraction power'', they
abstract very little.

We start defining operations over the simplest types: the empty list object is
simply the fixnum \WHATEVER{0}; booleans are represented as physical machines
usually do, using \WHATEVER{0} for false (written \CODE{\#f}) and any other
value for true, including \WHATEVER{1} which we also write as \CODE{\#t}; in other
words, we use \TDEF{generalized booleans}; the procedure
\CODE{boolean:canonicalize} canonicalizes a generalized boolean into either
\WHATEVER{0} or \WHATEVER{1}.
%\CODE{boolean:not} is an ordinary procedure.
As a convention derived from Scheme, a question mark ``\CODE{?}'' at the end of
a procedure name serves to remind the user that the procedure is a
\TDEF{predicate}, which is to say a procedure returning a boolean result.
\\
\\
The section dealing with \textit{Fixnums} contains primitive wrappers for
arithmetic and bitwise operations, plus some very simple definitions such as
\TDEF{minimum, maximum} and a \TDEF{parity test}; the only slightly more
sophisticated procedures are the \TDEF{fixnum exponentiation} (by squaring)
procedure \CODE{fixnum:**}, and the base-10 logarithm.  An annoying repeating
pattern with conditionals is already visible at this point: we often have an
\CODE{e0:if-in} form testing for a boolean condition, using $\{\CODE{\#f}\}$ as
the conditional case set to discriminate between false and any other value:
what we think of as the ``else'' branch is always \textit{the first one}:
\begin{Verbatim}
(e1:define (fixnum:** base exponent)
  (e0:if-in exponent (0)
    (e0:value 1)
    (e0:if-in (fixnum:odd? exponent) (#f)
      (fixnum:square (fixnum:** base (fixnum:half exponent)))
      (fixnum:* base (fixnum:** base (fixnum:1- exponent))))))
\end{Verbatim}
Unfortunately we cannot factor away this ugly pattern before
introducing macros.
\\
\\
\label{buffer-set!-primitive}The \textit{Buffers} section contains more trivial primitive wrappers for
memory-related primitives: a buffer may be \TDEF{created} by
\CODE{buffer:make}, \TDEF{destroyed} by \CODE{buffer:destroy}, \TDEF{read} by
\CODE{buffer:get} or \TDEF{updated} by \CODE{buffer:set!}.  Buffer size is
\textit{only} stored as part of boxedness tags and must not be accessed out of
unexec, which is the reason why we did \textit{not} provide a procedure wrapper to
conveniently extract it, lest it be used by mistake.
\CODE{buffer:get} and \CODE{buffer:set!} have two and three parameters
respectively: we chose to use explicit offsets (0-based, in words) rather than pointer
arithmetics for accessing memory, in order to avoid making assumptions on
memory management systems, which often
%\CHECKATTHEEND{quickly refer the GC chapter and the current implementation solution}
constrain the use of inner pointers.
%% It should be obvious at this point how buffers are represented as pointers
%% to be freely passed around and returned, and that buffer elements may be any untyped
%% objects, which includes other pointers.
%% There is no provision for working with
%% objects \textit{smaller} than a word.
As in Scheme, an exclamation mark ``\CODE{!}'' at the end of a procedure name serves
to conventionally remind the user that the procedure has side effects.
\\
\\
The \textit{Boxedness} section contains some functionality to check whether a word is
a candidate pointer: for example fixnums below a fixed small constant, or
fixnums not divisible by the word size in bytes cannot represent pointers on
any modern byte-addressed machine (\SECTION\ref{reflection--tagging-for-unexecing}).
\\
\\
Having defined buffers, it is very easy to define \textit{Conses}: conses are
simply two-word buffers with handy constructor, accessor and updater
procedures.  Differently from s-conses, conses as defined in this section
do not necessarily hold two s-expressions: they are completely generic, mutable
\TDEF{pairs of untyped objects}:
\begin{Verbatim}
(e1:define (cons:make car cdr)
  (e0:let (result) (buffer:make-uninitialized (e0:value 2))
    (e0:let () (buffer:set! result (e0:value 0) car)
      (e0:let () (buffer:set! result (e0:value 1) cdr)
        result))))
\end{Verbatim}
The nested zero-binding \CODE{e0:let} blocks simulating a statement sequence is
another unfortunate recurring pattern which we are forced to live with until we
introduce macros.
\\
\\
The \textit{Lists} section introduces singly-linked lists made of right-nested
conses, and utility procedures to work with them. We ``define'' a list as
either the empty list \CODE{list:nil}, which is to say \WHATEVER{0}, or a cons whose right
element is another list.  The quotes in the previous sentence are necessary: in keeping with
\EPSILONONE's nature the fact that the right side of a list cons be another list is
a pure convention, never enforced with static or dynamic checks.
Faithful to the motto \textit{``garbage in -- garbage out''}, we simply let the
system fail at run time when a non-pointer, non-cons or a cons whose right side
is not a list is used in place of a list
 --- possibly with a reasonable error
message in the case of \FILE{guile+whatever}, but likely with a crude
\textit{Segmentation Fault} after unexec (\SECTION\ref{more-than-one-runtime}).
%% Faithful to the motto \textit{``garbage in -- garbage out''}, uses a
%% non-pointer, non-cons or a cons whose right side is not a list where a list
%% is expected, we simply let the system fail at run time --- with a reasonable
%% error message in the case of \FILE{guile+whatever}, but likely with a crude
%% \textit{Segmentation Fault} after unexec.

Of course we impose no constraint over the element shape, and lists are not
necessarily homogeneous.
Our utility procedures over lists include the usual operations
for appending, flattening, computing length, selecting by index.  Thanks to
\CODE{e0:call-{\allowbreak}indirect} we could have supported higher-order procedures
(without nonlocals), but we refrained from doing so in this low-level core.
\\
\\
A useful way of employing lists is to make
\TDEF{association lists} or
\TDEF{alists} (pronounced ``ey-lists''), the only
slightly delicate issue in our case being the way of comparing keys: the first,
simplest and most efficient way is ``by identity'': the section \textit{Alists
  with unboxed keys} defines procedure using single-word comparison for keys;
this is always appropriate for unboxed or unique keys
% \CHECKATTHEEND{shall I already hint at symbols? I think not}
 and when the key identity matters, but
is not reliable in general with boxed keys where the same content may be
replicated in more than one buffer.
\\
\\
A \TDEF{vector} is a pointer to a buffer with the first element reserved
to store the payload element number, which is useful in many contexts where the size of
a random-access sequence is not fixed, and we cannot rely on boxedness tags.
%Each element is a word, therefore possibly a pointer, and
The section \textit{Vectors} provides procedures to
lookup and update vectors, obtain their length, and some other utility
operations including append, blit, and conversion to and from list.  The
procedure \CODE{vector:equal-unboxed-elements?} compares two vectors comparing their respective
elements by identity, which is the most common case.  No
bound-checking is performed, and elements are allowed to be heterogeneous.  Vectors as
defined in this section cannot be resized, as re-allocating them would change their
pointer ``identity''.
\\
\\
The next section deals with \textit{Characters and Strings}: at this level
characters are just fixnums and strings are just vectors --- which entails the
somewhat space-inefficient choice of having each character take one entire word
in memory.  Yet string support becomes computationally simple, and the wide
range of each character suffices for supporting all Unicode code points, at a
fixed width.  Apart from some trivial I/O, string procedures are just trivial
``aliases'', or actually wrappers, of vector procedures:
\begin{Verbatim}
(e1:define (string:equal? s1 s2)
  (vector:equal-unboxed-elements? s1 s2))
\end{Verbatim}
%\\
%\\
Having defined string support, we are ready to deal with the second kind of
association lists, in the \textit{SALists} section: an \TDEF{salist}
(pronounced ``ess-ey-list'') has
strings, or other vectors with elements compared by identity, as keys.
\\
\\
We then have a short section about \textit{Boxes}: a \TDEF{box}, similarly to an
ML \CODE{ref}, encapsulates the idea of a mutable memory cell, implemented
as a pointer to a single-word buffer.  Utility procedures include support for
incrementing a mutable counter, in case a box contains a fixnum.  For example
(omitting the other obvious variant \CODE{box:bump-and-get!}):
\begin{Verbatim}
(e1:define (box:get-and-bump! box)
  (e0:let (old-value) (box:get box)
    (e0:let () (box:set! box (fixnum:1+ old-value))
      old-value)))
\end{Verbatim}
Of course \CODE{fixnum:1+} is a successor procedure, and
\CODE{fixnum:1-} is the predecessor\footnote{The names
  ``\CODE{1+}'' and ``\CODE{1-}'' for successor and predecessor come from the
  Lisp tradition, which we suppose inherited the convention from some
  Reverse-Polish stack
  language; to decrement the top stack element in Forth, for
  example, we may push \CODE{1} and then subtract, replacing the two topmost
  elements with the subtraction result.  The Forth code is two words long:
  ``\CODE{1 -}''.  In fact Forth also provides a functionally equivalent
  (and usually more efficient) predefined single word ``\CODE{1-}'',
  factoring away the common pattern.}.
\\
\\
With alists and vectors at our disposal we are ready to implement
\textit{Hashes} having either unboxed objects or strings as keys: we also use a
box to introduce a level of indirection making a hash easy to resize.  A hash
table is therefore a box referring a vector; the first payload element of the
vector (after the vector ``header'' word holding the element number) is reserved to keep
the hash element number, so that the fill factor is easy to compute at any
time; the other vector elements are the hash buckets, implemented as alists or
salists; of course the associated data may have any shape.  As we need the
hash function to be portable with respect to unexec, it respects the
constraint in \SECTION\ref{reflection--unexecing-and-hashes}.  Hashes are our
most complex data structure so far: automatic resizing, and particularly
comparing fill factors with a threshold \textit{using only fixnums} requires
some sophistication, but at around 250 lines our hash table implementation in
\EPSILONZERO
%, long as it may be,
does not end up being overly complicated, despite our choice
of avoiding higher order procedures leading to some code redundancy:
\begin{Verbatim}
(e1:define (unboxed-hash:set! hash key value)
  (e0:let () (e0:if-in (hash:overfull? hash) (#f)
               (e0:bundle)
               (unboxed-hash:enlarge! hash))
    (unboxed-hash:set-without-resizing! hash key value)))
(e1:define (string-hash:set! hash key value)
  (e0:let () (e0:if-in (hash:overfull? hash) (#f)
               (e0:bundle)
               (string-hash:enlarge! hash))
    (string-hash:set-without-resizing! hash key value)))
\end{Verbatim}
The two-way conditional performing a side effect in one branch and returning a
``dummy'' bundle in the other is another code pattern which we can't factor away without macros.
\\
\\
Our first and most important application of hash tables is in the
implementation of \textit{Symbols}.  Since we use them as identifiers,
symbols are central in our design; they must be efficient to compare with
one another, and as keys in associative structures such as the global
(\SECTION\ref{global-environment}) and procedure
(\SECTION\ref{procedure-state-environment}) state environments.

\label{symbol-table}
The \TDEF{symbol table} is a global hash mapping each symbol name, as a
string, into a unique boxed \TDEF{symbol object} with that name; by requiring
that all named symbols be \TDEF{interned} in the table, as all Lisps do, we
obtain that symbol pointers can be safely compared by identity, just like
unboxed objects; interning the same name more than once yields the same symbol
object pointer:
\pagebreak
\begin{Verbatim}
(e1:define (symbol:intern name-as-string)
  (e0:if-in (symbol:interned-symbol-name? name-as-string) (#f)
    (symbol:intern-without-checking! (vector:copy name-as-string))
    (string-hash:get symbol:table name-as-string)))

(e1:define (symbol:intern-without-checking! name-as-string)
  (e0:let (new-symbol) (symbol:make-uninterned)
    (e0:let () (buffer:set! new-symbol (e0:value 0) name-as-string)
      (e0:let () (string-hash:set! symbol:table name-as-string new-symbol)
        new-symbol))))
\end{Verbatim}
\label{uninterned-symbols}
Following the example of several Lisps \cite{common-lisp,emacs-lisp} we also
support \TDEF{uninterned symbols}, which is to say symbol objects with no name
and hence not occurring in the symbol table; an uninterned symbol pointer can be
compared by identity with any other symbol pointer.

The issue of what to store in the symbol object itself appears of little
consequence: it is useful to point the symbol name string within the symbol
object, to have an efficient mean of retrieving it when needed, typically for
printing --- but apart from this, the symbol object seems much more useful for
its \textit{identity} than for its content.  And indeed, were it not for the
problem of \SECTION\ref{reflection--unexecing-and-hashes}, we could safely
use \textit{symbol pointers} like ``unboxed'' hashed keys for
state environments (\SECTION\ref{procedure-state-environment}) such as the
global environment or the procedure table; but a better solution has been known
for decades, described for example in \cite{bibop-in-maclisp} and in
embryonic form already in \cite{lisp-mccarthy}: the idea is to entirely do away
with such global tables whenever possible, and \textit{store the global data associated to each
  symbol within the symbol object} itself; then the symbol object may be seen
as a record, whose fields include the symbol name, the value in the global
environment, the formal names of the symbol interpreted as a procedure, the
procedure body, and so on.  Where a pointer to a symbol is available, accessing
it in a state environment only costs one load or store instruction with
constant offset.  In the interest of extensibility, we also keep one alist
field in each symbol object, to which the user is free to add bindings\footnote{An alist, later called ``property list'', was actually the
  \textit{only} datum globally associated to symbols in \cite{lisp-mccarthy}
  (p.~25),
  also containing a binding for the symbol name.  Having fields at
  fixed offsets in the record object, out of the alist, may be regarded as an
  optimization.}.

As a natural consequence of this design, in \EPSILONONE we may use the same
symbol as a key for different state environments, for example the global
environment and the procedure table, and use the same name for a global
non-procedure and a procedure:
\label{lisp-1}\label{lisp-2}the possibility of using the same name as key for two (or more) distinct state
environments is what distinguishes a so-called \TDEF{Lisp-2}, such as Common
Lisp, from a \TDEF{Lisp-1}, such as Scheme \cite{lisp-1-vs-lisp-2}.

The \textit{Symbols} section in the source code would be straightforward except
for \EPSILONZERO's painful lack of records, which at this point we still have
to simulate with buffers.
\label{state:procedure-set!-definition}
\label{state:procedure-get!-definition}
\begin{Verbatim}
(e1:define (symbol:make-uninterned)
  (e0:let (result) (buffer:make (e0:value 9))
    (e0:let () (buffer:set! result (e0:value 0) (e0:value 0)) ; name
      (e0:let () (buffer:set! result (e0:value 1) (e0:value 0)) ; unbound in global environment
        (e0:let () (buffer:set! result (e0:value 2) (e0:value 127)) ; conventional unbound marker
          (e0:let () (buffer:set! result (e0:value 3) (e0:value 0)) ; empty formal list
            (e0:let () (buffer:set! result (e0:value 4) (e0:value 0)) ; no procedure body
              (e0:let () (buffer:set! result (e0:value 5) (e0:value 0)) ; no macro definition
                (e0:let () (buffer:set! result (e0:value 6) (e0:value 0)) ; no macro procedure
                  (e0:let () (buffer:set! result (e0:value 7) (e0:value 0)) ; no primitive descriptor
                    (e0:let () (buffer:set! result (e0:value 8) alist:nil) ; no extensions
                      result))))))))))

(e1:define (state:global-set! name value)
  (e0:let () (buffer:set! name (e0:value 1) (e0:value 1)) ;; the name is bound as a global
    (buffer:set! name (e0:value 2) value))) ;; value
(e1:define (state:procedure-set! name formals body)
  (e0:let () (buffer:set! name (e0:value 3) formals)
    (buffer:set! name (e0:value 4) body)))
\end{Verbatim}
Only from this point on we can find instances of symbol literals in the source,
such as \CODE{(e0:value foo)}: in the Scheme implementation of
\EPSILONZERO from Phase~\BOOTSTRAPPHASE{ii}, \CODE{e0:value} is a Guile macro
which generates a whatever object at Scheme macroexpansion time, calling the
functionality above for named symbols.  The same holds for string literals in the
\textit{Strings} section above, but the case of symbols is much more remarkable
due to their greater complexity and due to a global structure being involved.

Finally we provide a functionality to automatically generate fresh symbols,
particularly useful for machine-generated code.  Fresh symbols are interned and have a
name starting with a prefix the user is not supposed to use in her own
identifiers, currently ``\CODE{\_}''.  We adopted this solution based on a
convention rather than the alternative of using
uninterned symbols because of the need to extract \textit{all}
symbols from the symbol table (see for example \CODE{state:global-names} in \SECTION\ref{reflective-procedure-names}); moreover interned symbols with a conventional
prefix are easier to print and read, when needed for debugging.  Assuming that
``\CODE{\_}''-prefixed symbols do not occur in user identifiers, generated
symbols may be safely garbage-collected\footnote{Interned symbols are not yet
  garbage collected at the time of writing.  A solution is employing a second symbol
  table for  for ``\CODE{\_}''-prefixed symbols, implemented as a
  \textit{weak hash table} \cite{survey--wilson}.  Globals and procedures
  explicitly named by the user should cannot be safely destroyed in general, as
  they could be referenced in the future, possibly by dynamically-created expressions.}.
\\
\\
The \textit{Expressions} section defines \EPSILONZERO expressions, as per
Definition~\ref{epsilonzero-syntax} plus \CALLINDIRECTNAME
(\SECTION\ref{call-indirect}), as a data structure.  An expression may be
conceptually seen as a sum-of-products in the style of ML, and in practice it
is implemented as a boxed object: the pointed buffer contains an expression
case tag in its first position; the expression handle, present in all cases,
resides in the second element; the other element contents, and the buffer length,
depend on the specific expression case.  Where expressions require homogeneous
sequences of undetermined length (for example \BUNDLENAME items and \IFNAME
conditional cases), by convention we always use lists.

There are operators to build, inspect, update, and \TDEF{explode} (obtain all
components as a bundle) expressions.  Such code could conceptually be written
by hand, but due to its regularity and length we chose to generate it with a
Scheme program; the machine-written \EPSILONZERO code is the first non-comment
line in the section, easy to spot as the lone, huge 14000-character line with
only the minimum required whitespace.

Our expression as managed by the program-generated code has the exact same
memory representation as an equivalent data structure defined (much later: see \SECTION\ref{sum-types}) by
our general sum-of-product definition facility.

The machine-generated expression constructors require to always specify
\textit{all} components, including handles.  As the practice is tedious and we
can easily generate fresh handles (as fixnums), we also provide a set of
hand-written wrappers, named after the symbols identifying the corresponding
expression case in expansion followed by an ``\CODE{*}'' character.
\begin{Verbatim}
(e1:define e0:handle-generator-box
  (box:make-initialized (e0:value 0)))
(e1:define (e0:fresh-handle)
  (box:bump-and-get! e0:handle-generator-box))

(e1:define (e0:variable* name)
  (e0:expression-variable-make (e0:fresh-handle) name))
(e1:define (e0:call* name actuals)
  (e0:expression-call-make (e0:fresh-handle) name actuals))
\end{Verbatim}
For example the \EPSILONZERO expression which \CODE{(e0:call p 57)} expands to may be built by
\CODE{(e0:call* (e0:value p) (list:singleton (e0:value* 57)))},
which is not unreadable after considering how
the literal constant \textit{expression}
which \CODE{57} expands to, being an expression, has a different
representation from the \textit{fixnum} \WHATEVER{57}.

It may be worth to explicitly stress how the \EPSILONZERO implementation of
Phase~\BOOTSTRAPPHASE{ii} does \textit{not} rely
at all
%\textit{at all}
on the \EPSILONZERO expression data structure as defined in this section.
\\
\\
The next section \textit{State: global dynamic state, with reflection}
conceptually implements state environments, but for the most part is actually a
thin wrapper over buffer accessors used on symbol objects.  For example, the
following definitions concern the ``procedure table'', despite it not existing
anywhere as a single data structure:
\begin{Verbatim}
(e1:define (state:procedure? name)
  (state:procedure-get-body name)) ;; #f iff unbound, which is to say return the body
(e1:define (state:procedure-get-formals name)
  (buffer:get name (e0:value 3)))
(e1:define (state:procedure-get-body name)
  (buffer:get name (e0:value 4)))
(e1:define (state:procedure-get-in-dimension name)
  (list:length (state:procedure-get-formals name)))
(e1:define (state:procedure-get name)
  (e0:bundle (state:procedure-get-formals name)
             (state:procedure-get-body name)))

(e1:define (state:procedure-set! name formals body)
  (e0:let () (buffer:set! name (e0:value 3) formals)
    (buffer:set! name (e0:value 4) body)))
(e1:define (state:procedure-unset! name)
  (e0:let () (buffer:set! name (e0:value 3) list:nil)
    (buffer:set! name (e0:value 4) (e0:value 0)))) ;; make the body invalid
\end{Verbatim}
%The global environment support is similar.  
It is obvious at this point that our implementation of \CODE{e1:define}
from Phase~\BOOTSTRAPPHASE{ii} could \textit{not} update \EPSILONONE's state
environments such as the global environment and the procedure table, which we
have just implemented: up to this point anything which has been globally
defined in Phase~\BOOTSTRAPPHASE{iii} has been only set in \textit{Guile's} state
environments.  This problem will persist until Phase~\BOOTSTRAPPHASE{iv}.

\label{applier-procedures}
By examining all the buckets in the symbol table it is easy
to obtain the list of all interned symbols bound to a procedure, or a
primitive.  From this information we can automatically build an \textit{apply
  function} in the style of \cite{defunctionalization}, plus another procedure
in the same spirit for primitives,
without resorting to indirect calls.  The automatically-generated procedures
\CODE{state:apply} and \CODE{state:apply-primitive} (collectively \TDEF{appliers})
each take two arguments, a
symbol naming the object to call, and a parameter list.
The generated applier body consist in a deeply-nested
conditional, comparing the first parameter with each applicable name: when the
matching name is found the applier calls the corresponding procedure or
primitive, returning its results.

Generating appliers is our first example of dynamically-generated code in the
implementation.  Of course such generation heavily relies on dynamically-built
\EPSILONZERO expressions.
\\
\\
Armed with global tables
% and appliers, % No!  I don't think eval actually ever used appliers.
we are finally ready to implement a working
interpreter in the next section, \textit{epsilon0 self-interpreter}.

%% Even in our first-order context, by means of \CODE{e0:call-indirect} or the
%% automatically-generated \CODE{state:apply} one could obtain the traditional
%% \textit{eval/apply} mutual recursion as beautifully described in \cite{wizard},
%% if one strived for such esthetics.

The interpreter code is not overly complex and the main procedure
\CODE{e0:eval} is but a long dispatcher, selecting the appropriate
expression case and tail-calling a helper procedure.  Its second parameter is
the local environment, encoded as an alist:
\begin{Verbatim}
(e1:define (e0:eval e local)
  (e0:if-in (e0:expression-variable? e) (#t)
    (e0:let (h name) (e0:expression-variable-explode e)
      (e0:eval-variable name local))
    (e0:if-in (e0:expression-value? e) (#t)
      (e0:let (h content) (e0:expression-value-explode e)
        (e0:eval-value content))
      (e0:if-in (e0:expression-bundle? e) (#t)
        (e0:let (h items) (e0:expression-bundle-explode e)
          (e0:eval-bundle items local))
        (e0:if-in (e0:expression-primitive? e) (#t)
          (e0:let (h name actuals) (e0:expression-primitive-explode e)
            (e0:eval-primitive name actuals local))
          (e0:if-in (e0:expression-let? e) (#t)
            (e0:let (h bound-variables bound-expression body) (e0:expression-let-explode e)
              (e0:eval-let bound-variables bound-expression body local))
            (e0:if-in (e0:expression-call? e) (#t)
              (e0:let (h name actuals) (e0:expression-call-explode e)
                (e0:eval-call name actuals local))
              (e0:if-in (e0:expression-call-indirect? e) (#t)
                (e0:let (h procedure-expression actuals) (e0:expression-call-indirect-explode e)
                  (e0:eval-call-indirect procedure-expression actuals local))
                (e0:if-in (e0:expression-if-in? e) (#t)
                  (e0:let (h discriminand values then-branch else-branch)
                          (e0:expression-if-in-explode e)
                    (e0:eval-if-in discriminand values then-branch else-branch local))
                  (e0:if-in (e0:expression-fork? e) (#t)
                    (e0:let (h name actuals) (e0:expression-fork-explode e)
                      (e0:eval-fork name actuals local))
                    (e0:if-in (e0:expression-join? e) (#t)
                      (e0:let (h future) (e0:expression-join-explode e)
                        (e0:eval-join future local))
                      (e0:if-in (e0:expression-extension? e) (#t)
                        (e0:let (h name subexpressions) (e0:expression-extension-explode e)
                          (e0:eval-extension name subexpressions local))
                        (e1:error (e0:value "impossible"))))))))))))))
\end{Verbatim}
In order to keep the code understandable despite the deeply-nested
conditionals we chose \textit{not} to assume generalized booleans in
\CODE{e0:eval}, making sure that all the predicates we used only return
\WHATEVER{\#t} or \WHATEVER{\#f}.

Several helper procedures in their turn rely on \CODE{e0:eval-expressions}, which
sequentially evaluates a \textit{list} of expressions which have to return
1-dimension bundles, and returns the result list.

Many interpreter procedures are strongly interdependent and mutually recursive,
which is served quite well by procedural abstraction.  It is very convenient to
define mutually-recursive procedures without concern for the definition order,
so that the programmer does not need to keep a call graph in her mind.
\begin{Verbatim}
(e1:define (e0:eval-expressions expressions local)
  (e0:if-in expressions (0)
    list:nil
    (list:cons (e0:unbundle (e0:eval (list:head expressions) local))
               (e0:eval-expressions (list:tail expressions) local))))
(e1:define (e0:unbundle bundle)
  (e0:if-in (list:null? bundle) (#f)
    (e0:if-in (list:null? (list:tail bundle)) (#f)
      (e1:error (e0:value "e0:unbundle: the bundle has at least two elements"))
      (list:head bundle))
    (e1:error (e0:value "e0:unbundle: empty bundle"))))
\end{Verbatim}
Most helper procedures dealing with specific expression cases end up being simple:
\pagebreak
\begin{Verbatim}
(e1:define (e0:eval-variable name local)
  (list:singleton (e0:if-in (alist:has? local name) (#f)
                    (state:global-get name)
                    (alist:lookup local name))))

(e1:define (e0:eval-value content)
  (list:singleton content))

(e1:define (e0:eval-if-in discriminand values then-branch else-branch local)
  (e0:let (discriminand-value) (e0:unbundle (e0:eval discriminand local))
    (e0:if-in (list:memq discriminand-value values) (#f)
      (e0:eval else-branch local)
      (e0:eval then-branch local))))

(e1:define (e0:eval-bundle items local)
  (e0:eval-expressions items local))
\end{Verbatim}
A possibly striking implementation choice consists in encoding \EPSILONZERO
bundles as lists; this is necessary for examining bundle results, for example by
testing their length --- the fact that bundles are not denotable
(\SECTION\ref{bundles-are-expressible-but-not-denotable}) makes them hard to
deal with directly, in exchange for their potential efficiency in a compiled
implementation.  But ironically in this self-interpreter, where however performance
is not a priority, the need of building lists for bundles entails a high rate of heap
allocation, which is expensive.

The self-interpreter does not rely on explicit stacks and is quite far
from the semantics in \SECTION\ref{epsilonzero-dynamic-semantics-section}; yet,
of course without hope of certifying the implementation \textit{here}, we claim that
we believe it respects the semantics and implementation notes
of \SECTION\ref{epsilonzero-chapter}.
%footnote{The \EPSILONZERO self-interpreter does not currently include the in-dimension relaxation of \SECTION\ref{in-dimension-relaxation-hack}.}
\\
\\
\label{type-table}With the most fundamental addend types at our disposal we are ready to deal with
general support for user-defined types, in the \textit{Type table} section.

We start with building support for tracking the extensible set of ``types''
recognized by the system, such as the empty list, booleans, fixnums, and
conses; since in this context we assume dynamic typing, there is no need for
type parameters: all the subcomponents of an s-expressions are tagged, at any
level.

As types tend to be relatively few in number and this reflective information is
not particularly critical to performance, in this case we preferred a global
table to the alternative of reserving fields in \textit{all} symbol objects.

\label{expander}
\label{expander-procedure}
Information for each type (empty list, boolean, fixnum, cons,
string, ...) is encoded in a descriptor record implemented as a buffer, also
containing a unique tag, other information including a printer procedure, and
once again an alist which the user may employ to add more fields as type-dependent
attributes --- it is
unfortunately too early to define a general-purpose ``extensible record'' data
structure, without any support for syntactic abstraction.
\\
\\
The most interesting fields in type descriptor records is the
\TDEF{expression-expander} procedure name.  An expression-expander specifies
how to turn an object of the given type into an \textit{expression}.  Since of
course we provide procedures to update the type table, the user has the power
to define and change the way addends expand, including, for example the mapping
from a symbol into a variable as discussed in Syntactic
Convention~\ref{lisp-syntax-syntactic-convention}
and \SECTION\ref{lisp-syntax-could-be-made-more-regular}.

All support for macros, following Lisp syntactic conventions\footnote{Common
  Lisp also supports ``symbol macros'' \cite{common-lisp}: some symbols defined
  by the user are macroexpanded like zero-parameter macro calls.  Support for a
  similar feature can be added in \EPSILONONE by changing the \textit{symbol},
  rather than cons, expander.}, will be defined later in procedures called in
its turn by the \textit{cons expander} procedure; our predefined expanders
will implement expansion in a way compatible with
Definition~\ref{nonmacro-expansion-definition}.
\begin{definition}[expression-expander]
Let $\SET{A}_0,...,\SET{A}_{n-1}$ be the addend types of $\SET{S}$.  Then the
\TDEF{expander procedure for $\SET{A}_i$} or \TDEF{$\SET{A}_i$-expander}
 is a procedure of one parameter
returning one result.  The procedure is guaranteed not to fail
if the parameter has type $\SET{A}_i$, in which case the result has type
$\SET{E}$.
\QEDDEFINITION
\end{definition}
Most atomic objects such as the empty list, booleans and fixnums, are
expression-expanded by \CODE{sexpression:literal-expression-expander} into a
literal constant expression, which will finally allow the user to omit explicit
``\CODE{e0:value}''s for non-symbol literals;
\CODE{sexpression:variable-expression-expander} expression-expands a symbol
into a variable expression;
\CODE{sexpression:expression-expression-expander} trivially expression-expands
an expression into itself.  Cutting away some comments:
\begin{Verbatim}
(e1:define (sexpression:literal-expression-expander whatever)
  (e0:value* whatever))
(e1:define (sexpression:variable-expression-expander symbol)
  (e0:variable* symbol))
(e1:define (sexpression:expression-expression-expander expression)
  expression)
\end{Verbatim}
The cons expression expander is not complicated either, and resembles Syntactic
Convention~\ref{lisp-syntax-syntactic-convention} (p.~\pageref{lisp-syntax-syntactic-convention}) in regarding the procedure
call as a ``default case''.  It can be already seen from this code how
\textit{even \EPSILONZERO syntactic forms are implemented as macros}:
\begin{Verbatim}
(e1:define (sexpression:cons-expression-expander cons)
  (e0:let (car-sexpression) (cons:car cons)
    (e0:if-in (sexpression:symbol? car-sexpression) (#f)
      (e1:error (e0:value "cons:expression-expander: the car is not a symbol"))
      (e0:let (car-symbol) (sexpression:eject-symbol car-sexpression)
        (e0:if-in (state:macro? car-symbol) (#f)
          ;; The car is a symbol which is not a macro name:
          (e0:call* car-symbol (e1:macroexpand-sexpressions (cons:cdr cons)))
          (e1:macroexpand-macro-call car-symbol (cons:cdr cons)))))))
\end{Verbatim}
Lines~3-4 show how the specific cons-expander above is not suitable for
higher-order personalities where the operator can be encoded by an s-expression
different from an s-symbol.  Of course the user is free to replace the cons-expander
at a later time.
\\
\\
The \textit{S-expressions} section deals with the implementation of
s-expressions as data structures, and operations over them.  Some of our
procedures defined up to this point already \textit{call} procedures working
over s-expressions such as \CODE{sexpression:{\allowbreak}symbol?}, \CODE{sexpression:eject}
and \CODE{sexpression:inject-cons} in procedure bodies --- although of course
our procedures calling not-yet-defined procedures have never been called themselves, yet.

The specific memory representation of an s-expression object has always been
seen considered important for efficiency in Lisp: all practical Lisps
employ some form of bitwise tagging of unboxed objects, boxed
pointers and/or buffer words, allowing to store in a compact way elements of
the most common addends; such representation techniques are often complex (see
\cite[\SECTION{}Data Representation]{guile} for Guile's solution and
\cite{representing-type-information-in-dynamically-typed-languages--gudeman} as
a useful collection of ``a body of folklore''), but such complexity is
motivated by the need for tagging \textit{all}\footnote{Advanced optimizing
  Lisp compilers such as SBCL \cite{sbcl} may actually avoid run-time tagging
  at some code points, in cases when a type inference analysis succeeds and in
  favorable contexts.  This optimization, however, is not possible for the bulk
  of the code.} data in Lisp.

However the situation of \EPSILONONE is quite different from Lisp:
s-expressions are mostly used for representing user syntax before
macroexpansion, but not necessarily as a data structure after macroexpansion.
Even if an efficient implementation is certainly possible (potentially by
machine generation as in
\cite{representation-analysis}) for the time being we make do with a
quite literal implementation of Definition~\ref{s-expression-definition}: we
represent an s-expression as a pointer to a two-element buffer, whose first cell
holds the type tag while the second holds the representation of the
addend-type content.  Some sample definitions:
\begin{Verbatim}
(e1:define (sexpression:make tag value)
  (cons:make tag value))
(e1:define (sexpression:get-tag sexpression)
  (cons:get-car sexpression))
(e1:define (sexpression:eject sexpression)
  (cons:get-cdr sexpression))
(e1:define (sexpression:has-tag? x tag)
  (whatever:eq? (sexpression:get-tag x) tag))

;; We have generated unique tags when adding entries to the type table
(e1:define (sexpression:null? x)
  (sexpression:has-tag? x sexpression:empty-list-tag))
(e1:define (sexpression:boolean? x)
  (sexpression:has-tag? x sexpression:boolean-tag))
(e1:define (sexpression:cons? x)
  (sexpression:has-tag? x sexpression:cons-tag))

(e1:define (sexpression:inject-fixnum x)
  (sexpression:make sexpression:fixnum-tag x))
(e1:define (sexpression:eject-fixnum x)
  (e0:if-in (sexpression:fixnum? x) (#f)
    (e1:error (e0:value "sexpression:eject-fixnum: not a fixnum"))
    (sexpression:eject x)))

(e1:define sexpression:nil ;; an empty s-list object
  (sexpression:make sexpression:empty-list-tag empty-list:empty-list))
(e1:define (sexpression:car x)
  (e0:if-in (sexpression:cons? x) (#f)
    (e1:error (e0:value "sexpression:car: not a cons"))
    (cons:get-car (sexpression:eject x))))

(e1:define (sexpression:cadr x) (sexpression:car (sexpression:cdr x)))
(e1:define (sexpression:caadr x) (sexpression:car (sexpression:cadr x)))
\end{Verbatim}
In our representation \textit{all} s-expressions are boxed, and even traditionally
``unique'' objects such as the empty s-list or s-booleans may exist in more
than one instance.

We also define alternate versions of some procedures over fixnums and lists
suitable to work on s-fixnums and s-lists, as it will be convenient later in
macros to manipulate s-expressions without explicit injections and ejections:
\begin{Verbatim}
(e1:define (sexpression:1+ x)
  (sexpression:inject-fixnum (fixnum:1+ (sexpression:eject-fixnum x))))

(e1:define (sexpression:reverse x)
  (sexpression:append-reversed x sexpression:nil))
(e1:define (sexpression:append-reversed x y)
  (e0:if-in (sexpression:null? x) (#f)
    (sexpression:append-reversed (sexpression:cdr x)
                                 (sexpression:cons (sexpression:car x) y))
    y))
\end{Verbatim}
%% \\
%% \\
\subsubsection{Macros}
\label{macros}
Still following the bootstrap code in \FILE{core.e}, we are now finally ready
to add support for \textit{Macros}.

The \CODE{e1:macroexpand}\footnote{The name ``\CODE{macroexpand}'' may not be entirely
  appropriate, but has long been traditional in Lisp circles.  Even an
  s-expression containing no macro calls can be successfully ``macroexpanded''.}
procedure, turning an s-expression into a corresponding expression,
is but a trivial dispatcher tail-calling the
appropriate expression-expander; but some expanders, as by the default the one
for cons does, may involve expanding actual macro calls.
\begin{Verbatim}
(e1:define (e1:macroexpand s)
  (e0:let (tag) (sexpression:get-tag s)
    (e0:let (content) (sexpression:eject s)
      (e0:call-indirect (sexpression:type-tag->expression-expander-procedure-name tag)
                        content))))
\end{Verbatim}
The general idea of macros is simple enough\footnote{Our mechanism is in practice not
  unlike the Common Lisp or Emacs Lisp macro systems, despite our explicit
  distinction between expressions and s-expressions.  Common Lisp uses an
  auxiliary procedure ``\CODE{macroexpand-1}'' returning two results: the
  result of expanding \textit{one} call, and a boolean saying whether expansion
  should continue.  In our case we can simply use expression-expanders, which
  in the terminal case will receive an injected expression.}: the user defines
each macro ``in concrete syntax'' as an \textit{s-expression}, often relying
on other macros.  Before a macro call can be expanded the macro body itself
must have been macroexpanded in its turn into an \textit{expression}, which
makes up the body of the associated \TDEF{macro procedure}.  At macro call time,
the macro procedure is called by supplying it with the macro actuals; the macro
procedure result, an \textit{s-expression}, is then macroexpanded in its turn, which
may involve expanding other macro calls.  If the process does not diverge, the
final result will be an \textit{expression}.  Since predefined macros allow to
express all \EPSILONZERO forms our macro system is trivially Turing-complete,
already because of macro procedures.  Of course it is permitted, and useful, for
a macro to return another macro call: this allows to build upon user-defined forms,
``stacking'' syntactic abstractions one onto another.

Macro definition and lookup are easy enough, based as they are is on
symbol objects similarly to the global and procedure state environments:
\begin{Verbatim}
(e1:define (state:macro-set! macro-name macro-body-sexpression)
  (e0:let ()
    ;; If we're re-defining an existing macro, invalidate its previous procedure:
    (e0:if-in (buffer:get macro-name (e0:value 5)) (0)
      (e0:bundle)
      (state:invalidate-macro-procedure-name-cache-of! macro-name))
    (buffer:set! macro-name (e0:value 5) macro-body-sexpression)))
(e1:define (state:macro-get-body macro-name)
  (buffer:get macro-name (e0:value 5)))
(e1:define (state:macro? name)
  (state:macro-get-body name)) ;; 0 iff unbound, which is to say return the body
\end{Verbatim}
The careful reader may have noticed a small difference in
\CODE{state:macro-set!} compared to the analogous code for procedures: no
formal parameters names are provided for macros.  This absence is a conscious
choice of ours, leading to a small simplification: as no nonlocal is ever
visible from a macro body, parameter shadowing is impossible and we can safely
use the same formal name ``\CODE{arguments}'' for \textit{all} macros.  Moreover we can
use \textit{one} formal for \textit{all} the parameters of a macro call by
viewing them as an s-expression, which is to say the s-cdr of the macro call
s-expression --- for example \CODE{(m a (324) 3)} has parameters \CODE{(a (324) 3)}.

Of course we will also add support for friendlier macros with
named formals, later as a syntactic extension.
\\
\\
As a concession to efficiency we \textit{cache} macro procedures, by
generating them at the time of the first expansion of a macro call and then
re-using them.  It is important to specify this point, because in rare cases 
caching may have a observable effect on the result --- that is the case of
macros performing side effects very early, while building the returned
expression.

The corresponding code is surprisingly simple, ignoring the
references to transformations for the time being:
\begin{Verbatim}
(e1:define (state:macro-get-macro-procedure-name macro-name)
  (e0:let (cached-macro-procedure-name-or-zero) (buffer:get macro-name (e0:value 6))
    (e0:if-in cached-macro-procedure-name-or-zero (0)
      (state:macro-get-macro-procedure-name-ignoring-cache macro-name)
      cached-macro-procedure-name-or-zero)))

(e1:define (state:macro-get-macro-procedure-name-ignoring-cache macro-name)
  (e0:let (body-as-sexpression) (state:macro-get-body macro-name)
    (e0:let (untransformed-name) (symbol:fresh)
      (e0:let (untransformed-formals) (list:singleton (e0:value arguments))
        (e0:let (untransformed-body) (e1:macroexpand body-as-sexpression)
          (e0:let () (state:procedure-set! untransformed-name
                                           untransformed-formals
                                           untransformed-body)
            (e0:let (transformed-name transformed-formals transformed-body)
                    (transform:transform-procedure untransformed-name
                                                   untransformed-formals
                                                   untransformed-body)
              (e0:let () (state:procedure-set! transformed-name
                                               transformed-formals
                                               transformed-body)
                (e0:let () (buffer:set! macro-name (e0:value 6) transformed-name)
                  transformed-name)))))))))
\end{Verbatim}
A macro call expansion consists in macroexpanding one call into an
\textit{s-expression} and then tail-calling to a further macroexpansion of the
result, which will hopefully terminate; the usual terminal case is an injected
expression.
\begin{Verbatim}
(e1:define (e1:macroexpand-1-macro-call symbol arguments)
  (e0:let (macro-procedure-name) (state:macro-get-macro-procedure-name symbol)
    (e0:call-indirect macro-procedure-name arguments)))

(e1:define (e1:macroexpand-macro-call symbol arguments)
  (e0:let (sexpression-after-one-expansion) (e1:macroexpand-1-macro-call symbol arguments)
    (e1:macroexpand sexpression-after-one-expansion)))
\end{Verbatim}
Just for completeness we also show the trivial helper, called by the cons
expander, which macroexpands an s-list of s-expressions into a list of expressions,
left-to-right:
\begin{Verbatim}
(e1:define (e1:macroexpand-sexpressions sexpressions)
  (e0:if-in (sexpression:null? sexpressions) (#f)
    (list:cons (e1:macroexpand (sexpression:car sexpressions))
               (e1:macroexpand-sexpressions (sexpression:cdr sexpressions)))
    list:nil))
\end{Verbatim}

We close by developing an illustrative and we hope not too artificial example.
Let us assume to have somehow added an \CODE{e1:trivial-define-macro} form for
globally defining a macro, internally using \CODE{state:macro-set!};
\CODE{e1:trivial-define-macro} has two parameters: the macro name, and the
macro body s-expression.  We use \CODE{e1:trivial-define-macro} to define our
sample macro:
\begin{Verbatim}
(e1:define (sexpression:list3 a b c)
  (sexpression:cons a (sexpression:cons b (sexpression:cons c sexpression:nil))))

(e1:trivial-define-macro silly-square
  ;; We would write `(fixnum:* ,(sexpression:car arguments)
  ;;                           ,(sexpression:car arguments))
  ;; if we already had quasiquoting.
  (sexpression:list3 (sexpression:inject-symbol (e0:value fixnum:*))
                     (sexpression:car arguments)
                     (sexpression:car arguments)))
\end{Verbatim}
The \CODE{silly-square} macro takes at least one parameter (ignoring any one
after the first) and returns an expression multiplying the parameter by
itself; the resulting expression will contain two copies of the macroexpanded
parameter, which therefore will be \textit{evaluated twice}.

For example, \CODE{(silly-square 4 5 6)} would eventually macroexpand to
$\DCALL{h_1'}{\CODE{fixnum:*}}{\DCONSTANT{h_2'}{4}\ \DCONSTANT{h_3'}{4}}$,
for some fresh $h_1',h_2',h_3' \in \SET{H}$.

When calling \CODE{e1:macroexpand} on \CODE{(silly-square 4 5 6)} we
immediately go through the cons expression-expander; assuming that
\CODE{silly-square} is \textit{not} a procedure name, we tail-call
\CODE{e1:macroexpand-macro-call} with two parameters: the symbol
\CODE{silly-square}, and the s-list \CODE{(4 5 6)};
\CODE{e1:macroexpand-macro-call} attempts to expand the first call by using
\CODE{e1:macroexpand-1-macro-call}.  Assuming \CODE{silly-square} has not been
used before, \CODE{state:macro-get-macro-procedure-name} builds its macro
procedure, which requires several expression-expansion calls not involving
macros; \CODE{state:macro-get-macro-procedure-name} then returns the macro procedure
name to \CODE{e1:macroexpand-1-macro-call}, which calls it on \CODE{(4 5 6)};
the result is the s-expression \CODE{(fixnum:* 4 4)}, which is returned by
\CODE{e1:macroexpand-{\allowbreak}1-{\allowbreak}macro-{\allowbreak}call}; so control goes back to
\CODE{e1:macroexpand-{\allowbreak}macro-{\allowbreak}call}, which tail-calls \CODE{e1:macroexpand} on
\CODE{(fixnum:* 4 4)}; by trivial expression-expansions, we finally obtain the
expression
$\DCALL{h_1'}{\CODE{fixnum:*}}{\DCONSTANT{h_2'}{4}\ \DCONSTANT{h_3'}{4}}$.

%%%%%%%%%%%%%%%%%%%%%%%%%%
\subsubsection{Transforms}
\label{transforms}
\label{transformation-section}
\label{transformations-section}
Many mathematical presentations deal with ``transformations'', meant as
code-to-code functions.  Our \TDEF{transform} strategy adopts the same approach,
with the sole significant extension of also permitting side-effecting
procedures.

When building a transform, the user or personality developer has to simply
define ordinary procedures working on code, and then to ``hook'' them to the
system.  There are two reasonable ways of running such \TDEF{transform procedures}:
\begin{itemize}
\label{retroactive-transforms}
\item a procedure can be applied \TDEF{retroactively} to the
current state, adding (and usually replacing) definitions;
\item
or it can be \TDEF{installed}, to be applied automatically in the future to each
toplevel expression or each procedure created from that point on; since
the composition order \textit{is} usually significant, the user can
control where each transform procedure fits in a global list of names.
\end{itemize}
In general we are interested in transforming three different entities:
\textit{expressions}, \textit{procedure bindings}, and \textit{global
  bindings}.  Transform procedures will need to be different for each case,
since the procedure interface cannot be compatible; but in our experience,
``companion'' transform procedures tend to rely on some common helper doing
most of the actual work: for example, closure-converting a procedure binding
involves closure-converting its body, which is an expression
(\SECTION\ref{closure-conversion-transform}).

Global bindings are difficult to work with in practice, since they contain
already-evaluated values with no fixed shape, rather than expressions; we have
not yet found a use for global binding transform procedures, but we include
such support for symmetry reasons.

For generality's sake, we decided to have \textit{binding transform procedures}
also return a transformed name: this may be either the same as the original
untransformed binding, or a new one.  It might be convenient to keep the old
definition around for debugging reasons, for example, but change all the uses
of the old entity to new one, by systematically renaming references.

A transform procedure may have one of three different interfaces:
\begin{itemize}
\item \textit{one parameter -- one result}, to transform an expression;
\item \textit{three parameters -- three results}, to transform a procedure binding: name, formals, body;
\item \textit{two parameters -- two results}, to transform a global binding: name, value.
\end{itemize}
The ultimate purpose of our code-rewriting system is to let the user write in
an expressive high-level language, to be then automatically reduced to
\EPSILONZERO by ``transforming away'' extensions.  The transform procedures
mapping ``syntax into syntax'' therefore will need to support not only syntax,
but \textit{extended syntax} as data.  We will show an elegant solution to this
problem in \SECTION\ref{expressions-as-an-extensible-sum-type}, but we do not
need to be concerned with it now while discussing the code which \textit{invokes}
procedure transforms.

\label{a-transform-may-be-a-good-place-to-run-an-analysis}
As a less obvious consequence of our design, side-effecting transforms
provide for another interesting opportunity: a simple way of performing a
\textit{code analysis} is to implement a trivial transform procedure returning its
parameters unchanged \textit{while recoding data in some global structure},
possibly a global table with handles (\SECTION\ref{handle-introduction}) as
keys.  Transforms actually returning modified code might also store their
parameters somewhere or simply record the relation between transformed an
untransformed code\footnote{A practical problem of the current implementation
  which makes debugging difficult is the lack of a reverse mapping from code
  to its original untransformed form and ultimately to its the s-expression
  concrete syntax and source location.  Solving this problem requires some
  care in writing transform, parsing and expression-expansion procedures, so
  that when an expression is built from others, its origin is somehow
  recorded in a graph.  A linguistic extension to somehow automate this
    tracing process might be appropriate.} in some global structure, available for
later debugging or analysis.

Transforms are also a convenient way to run some \textit{optimizations}
rewriting expressions into more efficient versions.  As a first ``low-hanging
fruit'' we plan to use some \textit{heuristic search} algorithm such as hill-climbing to
search the neighborhood of an expression for operationally-equivalent but
faster versions.
\\
\\
We are finally ready to present some of the code in the \textit{Transforms}
section.  At around 150 lines the code is quite short and also very uniform,
often present in three just slightly different versions because of the three
entities to manage.

The global lists of transform names to be applied in order are simple boxed
global variables:
\begin{Verbatim}
(e1:define transform:expression-transforms (box:make-initialized list:nil))
(e1:define transform:procedure-transforms (box:make-initialized list:nil))
(e1:define transform:global-transforms (box:make-initialized list:nil))

(e1:define (transform:prepend-expression-transform! new-transform-name)
  (box:set! transform:expression-transforms
            (list:cons new-transform-name (box:get transform:expression-transforms))))
(e1:define (transform:append-procedure-transform! new-transform-name)
  (e0:let () (box:set! transform:procedure-transforms
                       (list:append2 (box:get transform:procedure-transforms)
                                     (list:singleton new-transform-name)))
    (state:invalidate-macro-procedure-name-cache!))) ;; All macros have to be re-transformed
\end{Verbatim}
The interaction with macros is interesting as it reminds us that an untransformed
procedure may be incompatible with its transformed version (for example, in a
CPS transform the argument number may change): it is hence important to
invalidate any cached macro procedure, so that new ones are created, \textit{and
subjected to the current transforms}.

Applying transform procedures is trivial; this is the code which gets executed
when a procedure binding is transformed; the other two cases are essentially
identical.
\begin{Verbatim}
(e1:define (transform:transform-procedure name formals body)
  (e0:let (transform-names) (box:get transform:procedure-transforms)
    (transform:apply-procedure-transforms transform-names name formals body)))
(e1:define (transform:apply-procedure-transforms remaining-transforms name formals body)
  (e0:if-in remaining-transforms (0)
    (e0:bundle name formals body)
    (e0:let (transformed-name transformed-formals transformed-body)
            (e0:call-indirect (list:head remaining-transforms) name formals body)
      (transform:apply-procedure-transforms (list:tail remaining-transforms)
                                            transformed-name
                                            transformed-formals
                                            transformed-body))))
\end{Verbatim}
Retroactive transformation is more interesting.  The user will call
\CODE{transform:{\allowbreak}transform-{\allowbreak}retroactively!} to install transform
procedures for global and procedure bindings, also specifying the names of some
objects \textit{not} to transform.
\begin{Verbatim}
(e1:define (transform:transform-retroactively! globals-not-to-transform
                                               value-transform-names
                                               procedures-not-to-transform
                                               procedure-transform-names)
  (e0:let (global-names) (list:without-list (state:global-names) globals-not-to-transform)
    (e0:let (procedure-names)
            (list:without-list (state:procedure-names) procedures-not-to-transform)
      (e0:let (transformed-name-global-list)
              (transform:compute-transformed-globals global-names value-transform-names)
        (e0:let (transformed-name-formal-body-list)
                (transform:compute-transformed-procedures procedure-names
                                                          procedure-transform-names)
          (e0:primitive state:update-globals-and-procedures! transformed-name-global-list
                                                             transformed-name-formal-body-list))))))
\end{Verbatim}
The code works by first \textit{computing} all transformed bindings (the
trivial helpers \CODE{transform:{\allowbreak}compute-{\allowbreak}transformed-{\allowbreak}globals} and
\CODE{transform:{\allowbreak}compute-{\allowbreak}transformed-{\allowbreak}procedures} simply return
a list of transformed bindings) without performing any global update; then, \textit{with
  a single primitive call}, it activates all new bindings.

\label{atomic-global-update-primitive}Why having such a complex primitive written in C?
  And why do we have to compute all the bindings before applying
any?  The answer is that \textit{we need the state environment update to be
  performed atomically}\footnote{Nothing to do with concurrency, in this case.
  Our current code does not even support synchronization primitives other than
  \JOINNAME, so background threads performing imperative operations are not
  used at all.}, again because of the incompatibility introduced by some
transforms.  The alternative of updating global definitions in \EPSILONZERO
would fail when at some point \textit{the updater procedure itself} or its
helpers would be reached by the incompatible \textit{change wave}, and break on a call
from an untransformed procedure to a transformed one, or vice-versa.  For this
reason only, \CODE{state:update-{\allowbreak}globals-{\allowbreak}and-{\allowbreak}procedures!} must be a primitive.
%% \\
%% \\
%% \TODO{This essentially closes our general discussion about transforms.  At
%%   around 150 lines of code including comments, their support is extremely simple.}
%% \FILL
%% \TODO{Say explicitly that transforms do not subsume macros: macros have the
%%   role of building (possibly extended) expressions from user syntax;
%%   transformations globally rewrite extended syntax into core syntax.}
\\
\\
The \textit{REPL} section is the last interesting part of \FILE{core.e}.  Its
helper procedure \CODE{repl:macroexpand-{\allowbreak}transform-{\allowbreak}and-{\allowbreak}execute} can
be given an s-expression to expression-expand, transform and evaluate:
\begin{Verbatim}
(e1:define (repl:macroexpand-transform-and-execute sexpression)
  (e0:let (untransformed-expression) (e1:macroexpand sexpression)
    (e0:let (transformed-expression) (transform:transform-expression untransformed-expression)
      (e0:eval-ee transformed-expression))))
\end{Verbatim}
The REPL itself is very crude, and currently relies on a primitive
\CODE{io:{\allowbreak}read-{\allowbreak}sexpression} \textit{calling Scheme from C} to read a Guile
s-expression and then convert it into our representation.  This lack of a real
frontend written in \EPSILONONE is the last remaining reason why we still
depend on Guile after bootstrap (\SECTION\ref{implementation-status}).
\begin{Verbatim}
(e1:define (repl:repl)
  (e0:let () (string:write "Welcome to the epsilon REPL\n")
    (repl:loop (io:standard-input))))
(e1:define (repl:loop port)
  (e0:let () (string:write "e1>\n")
    (e0:let (next-sexpression) (e0:primitive io:read-sexpression port)
      (e0:let (results) (repl:macroexpand-transform-and-execute next-sexpression)
        (e0:let () (repl:write-results results port)
          (e0:let () (string:write "\n")
            (repl:loop port)))))))
\end{Verbatim}

%% \TODO{insert before: why are macro bodies s-expressions, rather than
%%   expressions?  the problem is expansion time: we don't want to expand too
%%   early}

%%%%%%%%%%%%%%%%%%%%%%%%%%
\subsubsection{An aside: developing, testing, and the ordering of phases}
\label{the-ordering-of-phases}
In this presentation we have chosen to show the \textit{final structure} of our
bootstrapping system as a working body of code, rather than recounting the
\textit{process} of writing it; the two views do not perfectly overlap.

The preceding phase is by far the most problematic in this respect: as any
reader with implementation experience may witness \EPSILONONE with its
interpreter, global data structures, macro and transform systems is a very
strongly recursive system, where each component tends to require all the others
in a loop of circular dependencies apparently very difficult to break.  And
indeed, the preceding third phase was not easy to implement on the machine.

After deciding on the general bootstrapping strategy we wrote a first
approximation of the system, in \EPSILONZERO with with no macros (see next
phase) and no transforms, up to the interpreter included.
Some subsystems, for example the implementation of sum-of-products types for
\EPSILONZERO expressions, were first prototyped in Guile.  Transformations were
added as the very last step, after macros worked reliably and were used to
make \EPSILONONE considerably more friendly.

With the marshalling/unmarshalling support needed for
\textit{unexec} (\SECTION\ref{marshalling}), we followed a route of progressively
reducing the abstraction level: after writing its first version in \EPSILONONE using
several comfortable language extensions, we translated it into \EPSILONZERO, to
make it possible to run it earlier at boot, when extensions are not loaded yet.  The
translated marshalling code is understandable, but some complexity which would
have been a little too daring for \EPSILONZERO still shows up in the code,
particularly in nested conditionals.

Later we rewrote the marshalling and unmarshalling support for a third time in
C, for performance reasons (\SECTION\ref{bootstrapping-optimizations}).
\\
\\
At the beginning we wrote a considerable body of debugging code in Scheme,
including for example for example the procedure \CODE{print-expression} writing
expressions in \EPSILONZERO's syntax of \SECTION\ref{epsilonzero-chapter}
including handles in Unicode subscript digits, or
\CODE{hash-dump-sizes} which has served to test how well our hash functions
distribute; and maybe most importantly \CODE{meta:print-{\allowbreak}procedure-{\allowbreak}definition}
and \CODE{meta:print-{\allowbreak}macro-{\allowbreak}definition}, useful for inspecting the global state
and obtain readable syntax.
Such code is still available in
\FILE{bootstrap/{\allowbreak}scheme/{\allowbreak}conversion.scm}, and still occasionally
useful for debugging:
\begin{Verbatim}
guile> (meta:print-procedure-definition 'cons:make)
Formals: (car cdr)
[let [result] be [call buffer:make 2₇₇₉]₇₈₀ in [let [] be [call buffer:set! result₇₈₁ 0₇₈₂
car₇₈₃]₇₈₄ in [let [] be [call buffer:set! result₇₈₅ 1₇₈₆ cdr₇₈₇]₇₈₈ in result₇₈₉]₇₉₀]₇₉₁]₇₉₂

guile> (meta:print-macro-definition 'e0:call)
(sexpression:inject-expression (e0:call* (sexpression:eject-symbol (sexpression:car arguments))
(e1:macroexpand-sexpressions (sexpression:cdr arguments))))

guile> (e0:value whatever:identity) ;; the symbol dump is painful to read
0x1471040[0x14c41e0[17 119 104 97 116 101 118 101 114 58 105 100 101 110 116 105 116 121]
0 127 0x1409f00[0x14633c0[0x145c900[1 120] 0 127 0 0 0 0 0 0] 0] 0x1462b80[0 10 0x14633c0]
0 0 0 0]
\end{Verbatim}
But as our design of \EPSILONONE changed until its crystallization
into the present form, some of our crude debugging and code-generating tools also broke down
and became unusable, whenever their underlying assumptions failed.  As old
scaffoldings not supporting any more a structure now able to stand by itself,
we abandoned them.
\\
\\
Our bootstrapping code running on top of an inefficient extension to Guile had
low performance, which was unsurprising.  What we didn't expect was that
waiting times were in practice so unbearable even on our fastest
machine\footnote{\textit{optimum} is a Dell Precision T7400 with two quad-core Intel Xeon (EM64T)
chips at 3GHz, 8Gb of RAM, heavily customized debian GNU/Linux ``unstable''.}
 that it necessitated optimizations
already in this phase.
\SECTION\ref{bootstrapping-optimizations} provides some insights.

%%%%%%%%%%%%%%%%%%%%%%%%%%
\subsubsection{Phase~\BOOTSTRAPPHASE{iv}: fill reflective data structures}
\label{bootstrap-phase-4}

Phase~\BOOTSTRAPPHASE{iii} consisted of about 2500 lines in \EPSILONZERO, which
we have executed on top of the \EPSILONZERO implementation of
Phase~\BOOTSTRAPPHASE{ii} based on \CODE{guile+whatever}; in other words, our
global definitions up to this point affected \textit{Guile's state environments},
rather than ours.  This phase consists in using the Guile data structures we
updated at each definition to fill ``reflective data structures'' --- in quotes,
since we are actually speaking of data to be stored as part of symbol objects.

The code
%for this phase
is in \FILE{bootstrap/{\allowbreak}scheme/{\allowbreak}fill-{\allowbreak}reflective-{\allowbreak}structures.scm}.
\\
\\
Our Scheme implementation of \CODE{e1:define} from Phase~\BOOTSTRAPPHASE{ii},
at the end of
\FILE{bootstrap/{\allowbreak}scheme/{\allowbreak}epsilon0-in-scheme.scm},
updates two global Scheme data structures:
\linebreak
\CODE{globals-{\allowbreak}to-{\allowbreak}define}, a list of names
of non-procedure globals which have been defined, and
\CODE{procedures-{\allowbreak}to-{\allowbreak}define}, an alist binding each defined procedure name to
its formals as a Guile list and its body as a Guile s-expression.
The idea is to scan the lists and for each element to copy the corresponding
data into our state environments.

\begin{itemize}
\item
Non-procedures are easy to manage: given a global name as a Guile symbol
we simply have to look it up as a Guile global: the value we find, an
injected whatever, has to be copied into the
appropriate field of the \EPSILONONE symbol corresponding to the Guile
symbol.

\item
Procedures are more involved: for each one \CODE{procedures-to-define} contains
its name as a Guile symbol, its formals as a Guile symbol list, and its body as
a Guile s-expression.  Name and formals are easy enough to translate, but for
our state environment \textit{we need the body as an \EPSILONZERO expression}
encoded in the expression data structure we defined in Phase~\BOOTSTRAPPHASE{iii};
converting each body into an expression is the main problem of this phase.
\end{itemize}
At this point we can better justify our rigidly constrained way of
writing code in Phase~\BOOTSTRAPPHASE{iii}, in which we \textit{only} used \EPSILONZERO
plus \CODE{e1:define}: the s-expression-to-expression translation we need to
perform at this point is \textit{non-macro expansion}.  Since the
translation has to be executed only once without the translation code itself
being part of the output, we can implement it in Scheme rather than in
\EPSILONZERO.  The procedure \CODE{e1:non-{\allowbreak}macro-{\allowbreak}expand}, defined in a
mutually-recursive fashion with its helpers
\CODE{e0:non-{\allowbreak}macro-{\allowbreak}expand-{\allowbreak}sexpressions}, \CODE{e0:non-{\allowbreak}macro-{\allowbreak}expand-{\allowbreak}symbols} and
\CODE{e0:non-{\allowbreak}macro-{\allowbreak}expand-{\allowbreak}values},
follows very closely our Definition~\ref{nonmacro-expansion-definition}.

The code is slightly less readable than the corresponding mathematical
definition just because of explicit representation conversions between Guile's
and \EPSILONONE's data; for example, the procedure
\CODE{whatever->{\allowbreak}guile-{\allowbreak}boolean} converts an untyped
\EPSILONONE object into a Guile dynamically-typed boolean, and
\CODE{guile-sexpression->sexpression} converts a native Guile s-expression into
our own representation, as per the previous phase.  All such conversion
operators, by themselves quite unremarkable, are implemented in Scheme, in
\FILE{bootstrap/{\allowbreak}scheme/{\allowbreak}conversion.scm}.

Our \CODE{e0:non-macro-expand} is also ``unsafe'' and in practice accepts a
superset of valid syntax encodings: we avoided safety checks in the code, for
example ignoring the s-cddr of \CODE{(e0:join x .\ $s$)} instead of verifying
that it really is \CODE{()}.  This expansion unsafety is not a problem in
practice at this point, since the code to be translated has already been well
tested on \EPSILONZERO's implementation of Phase~\BOOTSTRAPPHASE{ii}, using
Guile's interactive REPL (actually \CODE{guile+whatever}'s), and just a little care in
reading untyped data structure dumps (\SECTION\ref{memory-dumps}).
\\
\\
The real work is in the Scheme procedure \CODE{set-metadata!}, a zero-parameter
procedure which consists of two loops, the first scanning the
global binding list and adding the definition to symbols, and the second doing
the same for procedures after s-expression conversion and non-macro expansion.

Even on our \textit{optimum} machine, when
using Guile 1.8, which is faster in this phase, the computation of
Phase~\BOOTSTRAPPHASE{iv} takes about 15 seconds, compared to 0.2
seconds for the previous phases combined; fortunately, unless there are recent changes in
\FILE{core.e}, we can in practice entirely skip this phase by \textit{exec}ing
(\SECTION\ref{unexec}) over the Phase~\BOOTSTRAPPHASE{ii} interpreter.
\\
\\
One problem remains: the globals and procedures we have defined up to this
point \textit{also} remain in Guile's state environments, and this state of
things will persist up until we remove the dependency on Guile.  Re-defining
some procedure directly invoked from Guile, would lead to subtle problems,
making the two definition sets inconsistent.  We will simply avoid to
override any \EPSILONONE definition with an incompatible one.

With the caveat above, and now having global and procedure definitions in place,
we can finally use \CODE{e1:eval}.

%%%%%%%%%%%%%%%%%%%%%%%%%%
\subsection{Unexec}
\label{unexec-the-symbol-table-and-the-main-expression}
At this stage it is finally possible to use unexec, which depended on reflective
structures to dump a program.
Our vague reference in \SECTION\ref{few-in-number-vague-reference} to the
``surprisingly few'' data structures involved should be clear now: a simple way
of obtaining a program is to dump a pair containing:
\begin{itemize}
\item the \textit{symbol table}, holding all global and procedure definitions and from
  which all alive data in memory can be reached (\SECTION\ref{the-nature-of-values});
\item the \textit{main expression}.
\end{itemize}
At \textit{exec} time, it suffices to unmarshal the pair, define the symbol
table, and run the interpreter on the expression.
\\
\\
The \EPSILONZERO implementation of \CODE{unexec:unexec} and \CODE{unexec:exec}
is in \FILE{bootstrap/{\allowbreak}scheme/{\allowbreak}unexec.e}; the same file also contains the
\EPSILONZERO implementation of marshalling and unmarshalling.

%%%%%%%%%%%%%%%%%%%%%%%%%%%%%%%%%%%%%%%
\subsection{Optimizations}
\label{bootstrapping-optimizations}
In a preliminary version of \EPSILONONE, macros were not associated to
procedures to be called, but to expressions to be evaluated.  The current
definition has a cleaner interaction with transforms, but if we ignore
transforms the old solution was perfectly workable as well: instead of passing
parameters to a procedure, we evaluated an expression in some environment, with
the same effect.

With the old solution, macroexpansion returned correct results, but
the system's incredible inefficiency led us to investigate the issue until we
discovered a perverse pattern: the complicated circular nature of the
dependencies between \CODE{e1:{\allowbreak}macroexpand}, expression-expanders,
\CODE{e0:{\allowbreak}eval} and its helpers made it difficult to understand how, indirectly,
\textit{\CODE{e0:eval} was interpreting calls to itself}.
\\
\\
It is easy to see how, if adding one layer of interpretation worsens
performance by some constant factor $k$, we have that a stack of $n$
interpreters has exponential overhead $k^n$; and given that symbolic
interpreters easily cause order-of-magnitude overheads, the slowdown was
evident even for very small values of $n$.

For the first time in our programming experience we discovered that some code
had \textit{unbounded interpretation overhead}.  Despite being now unnecessary
because of the macroexpansion changes, we still find the problem and its solution
quite beautiful, and potentially instructive for others.
\begin{implementationnote}[The Hack]\label{the-hack-implementation-note}
When evaluating a call to \CODE{e0:eval}, the self-interpreter does \textit{not}
evaluate \CODE{e0:eval}'s body, but directly the given expression in the given
local environment.
\QEDIMPLEMENTATIONNOTE
\end{implementationnote}
The idea is simply to recognize as a particular case any call of a procedure named
\CODE{e0:eval}:
\begin{Verbatim}
(e1:define (e0:eval-call name actuals local)
  (e0:if-in (whatever:eq? name (e0:value e0:eval)) (#f)
    (e0:eval-non-eval-call name actuals local)
    (e0:eval-eval-call actuals local)))

(e1:define (e0:eval-eval-call actuals local)
  (e0:let (actual-values) (e0:eval-expressions actuals local)
    (e0:if-in (whatever:eq? (list:length actual-values) (e0:value 2)) (#f)
      (e1:error (e0:value "e0:eval-eval-call: in-dimension mismatch"))
      (e0:let (expression) (list:head actual-values)
        (e0:let (local) (list:head (list:tail actual-values))
          (list:singleton (e0:eval expression local))))))) ; wrap as inner eval would
\end{Verbatim}
Subjectively, it could be said that \textit{The Hack} changed the interpreter
from being \textit{comically} slow to being still very slow, but at least usable.
\\
\\
Despite being an obvious idea, the following implementation aspect deserves
prominence because of its dramatic impact on performance:
\begin{implementationnote}[interpreter in C]\label{interpreter in C}
We re-implemented an \EPSILONZERO interpreter in low-level C, using explicit
stacks and no heap allocation, except implicitly for building whatevers.  The
C implementation is a few hundreds of lines long, and performs runtime dimension
checks.  The interpreter is available from \EPSILONONE as the primitive
\CODE{e0:fast-eval}, and has the same interface of \CODE{e0:eval}.
\QEDIMPLEMENTATIONNOTE
\end{implementationnote}
Replacing the \EPSILONZERO self-interpreter with the C implementation
led to a speedup of around 200 for an exponential-time recursive implementation
of Fibonacci's function:
\begin{Verbatim}
(e1:define (fibo n)
  (e0:if-in n (0 1)
    n
    (fixnum:+ (fibo (fixnum:- n (e0:value 2))))
              (fibo (fixnum:1- n))))
\end{Verbatim}
The high speedup is not surprising, if we consider that the \EPSILONZERO
self-interpreter had to run on top of Guile, itself an interpreter.

The ``Interpreter in C'' strategy subsumes The Hack: being written in a
different language than \EPSILONZERO, the interpreter never accesses its own
body, and interpreter calls in the interpreted code are instead executed by a
primitive, thus avoiding overhead multiplication.

\begin{implementationnote}[\textit{exec}/\textit{unexec} in C]\label{exec/unexec-in-c}
We re-implemented the marshalling and un-marshalling procedures
\CODE{unexec:dump} and \CODE{unexec:undump} by using primitives written in C.
The C implementation consists of about 100 lines, and adopts exactly the same
data structures and algorithm as the corresponding \EPSILONZERO code. \QEDIMPLEMENTATIONNOTE
\end{implementationnote}
Again, the optimization of Implementation Node~\ref{exec/unexec-in-c} has an
order-of-magnitude impact on performance: thanks to it, \textit{exec} ``quick-start'' takes
only a short fraction of a second on \textit{optimum}.
\\
\\
Re-implementing part of the functionality in C was an aid to development and
rapid testing, more than a definite commitment: after a good native compiler is
developed, the need for such optimizations will attenuate.

%% There are no self-evaluating atoms in this first phase: self-evaluating addends
%% are used for predefined \textit{Scheme} s-expressions; in this first phase
%% there is item such as this:
%% ``{\rm $\EXPANDE{s}
%%   =
%%   \EXPANDE{\CODE{(e0:value }s\CODE{)}}$}
%% where $s$ is a self-evaluating atom;''

\subsection{Sample extensions}
\label{sample-extensions}
The file
\FILE{bootstrap/{\allowbreak}scheme/toplevel-{\allowbreak}in-{\allowbreak}scheme.{\allowbreak}scm},
run right after
\FILE{fill-{\allowbreak}reflective-{\allowbreak}structures.{\allowbreak}scm},
defines a few simple Scheme macros to let the user evaluate \EPSILONONE
forms within the Guile REPL: ``\CODE{(e1:toplevel .\ $s$)}'' evaluates
each element of the s-list $s$ as an \EPSILONONE expression, which of
course is macroexpanded and transformed before execution.
``\CODE{(e1:trivial-define-macro $m$ $s$)}'', available both as a Scheme macro and
as an \EPSILONONE macro, defines the macro named $m$ (an s-symbol) as $s$ (a
generic s-expression).

Armed with just this knowledge, the reader should be able to follow quite easily
\FILE{bootstrap/{\allowbreak}scheme/{\allowbreak}epsilon1.{\allowbreak}scm},
which contains around 2000 lines worth of \EPSILONONE extensions.
\\
\\
We think that the power of syntactic abstraction is easy to appreciate now, by
looking at how fast language expressivity improves after each definition,
compared to the development work in phase~\BOOTSTRAPPHASE{iii} during which only
procedural abstraction was available.

This sequence of extensions, quite impressive in its accelerating rhythm, raises
the language from the clumsy beginnings of \EPSILONZERO to a respectable power, with
\textit{sequences}, \textit{multi-way conditionals}, \textit{short-circuit
  logical operators}, Common Lisp-style \textit{destructuring macros},
\textit{variadic procedures}, \textit{tuples}, \textit{records},
\textit{extensible sum-of-product definitions}, \textit{closures},
\textit{imperative loops} and \textit{futures}.

Interestingly, only the three last extensions in the sequence above depend on a
transform; macros alone can already bring the language to a quite high
level.

Most of our syntactic conventions are inspired to Scheme, and the form names
are indeed largely compatible, apart from the ``\CODE{e1:}'' prefix.
\\
\\
The very beginning of
\FILE{bootstrap/{\allowbreak}scheme/{\allowbreak}epsilon1.{\allowbreak}scm}
deals with macros for core \EPSILONZERO forms:
\pagebreak
\begin{Verbatim}
;;; These first crude versions do not perform error-checking, silently
;;; ignoring additional subforms at the end.
(e1:trivial-define-macro e0:variable
  (sexpression:inject-expression
    (e0:variable* (sexpression:eject-symbol (sexpression:car arguments)))))
(e1:trivial-define-macro e0:let
  (sexpression:inject-expression
    (e0:let* (sexpression:eject-symbols (sexpression:car arguments))
             (e1:macroexpand (sexpression:cadr arguments))
             (e1:macroexpand (sexpression:caddr arguments)))))
\end{Verbatim}
The code is simple, but it is questionable whether it really belongs in this
file, rather than in \FILE{core.e}.  The reason why we defined such important
features so late is mostly pragmatic: such macro definitions would have been be
much less comfortable to write using only \CODE{state:macro-set!} without
\CODE{e1:trivial-define-macro}.  For similar reasons we defined quoting and
quasiquoting in this file, rather than in \FILE{core.e}.

The debate about where exactly the \EPSILONONE ``core'' ends and ``extensions''
begin looks futile anyway, and indeed the very notion of personality, possibly
useful for humans to identify a set of features, has no consequence for the
implementation.  The same objection may be raised ``at the other end'', about
the CPS transformation and the reason why we defined in its own source file
instead of in the end of \FILE{epsilon1.scm}.  In the same somewhat arbitrary
fashion, we proclaim that continuations do not belong in \EPSILONONE but are
part of another experimental personality \textit{based on} \EPSILONONE.  First-class
continuations provide another qualitative jump in expressivity, but our
implementation is less mature and quite expensive in terms of performance,
therefore less appropriate as part of the general-purpose ``library'' to build
personalities which \EPSILONONE is meant to be.
\\
\\
In the following we will just add some quick considerations about the main extensions.
%% mostly related to how
%% some extensions are \textit{used}, saying very little about how most of them
%% are implemented.

\subsubsection{Quoting and quasiquoting}
\textit{Quoting} and \textit{quasiquoting}, heavily relying on the
type table (\SECTION\ref{type-table}) so that support for new new types can be
added smoothly, are different from their Lisp homologous: in \EPSILONONE
\textit{a quoted or quasiquoted s-expression yields an expression which will
  build that s-expression when evaluated}; for example \CODE{'1} macroexpands
to a procedure call of \CODE{sexpression:inject-fixnum} which, if evaluated,
will build \textit{the s-expression} (not the unboxed fixnum) \CODE{1}.

Despite this difference, we can easily adapt the standard algorithms for
quasiquoting\footnote{We followed Bawden's updated proposal (different from his older
one in \cite[\SECTION{}B]{quasiquotation-in-lisp}),
as quoted by Kent Dybvig at
\url{http://www.r6rs.org/r6rs-editors/2006-June/001376.html}.  This new
version was eventually adopted in \cite{r6rs}.}, which is convenient since
\textit{nested quasiquoting} is famously tricky to implement correctly.
\\
\\
The non-homoiconicity of \EPSILONONE forces us to think of the difference
between s-expressions, uninjected values and expressions, and costs us
some injection and ejection operators in macros.  The inconvenience in practice
is tolerable, and we consider the advantages of our syntactic extensions well worth
this minor trouble.

%%%%%%%%%%%
\subsubsection{Variadic procedure wrappers}
\label{variadic-procedures}
All practical Lisps permit to define \textit{variadic} procedures, which is to say
procedures with an arbitrary number of optional arguments.  \EPSILONZERO and
\EPSILONONE \textit{do not}, for reasons of efficiency.  Anyway we can recover
the convenience of variadic calls by introducing \textit{variadic macros}, and
using them to wrap procedures.

The following \EPSILONONE definitions extend binary operators with a neutral
element to make them accept any number of arguments:
\begin{Verbatim}
(variadic:define-associative fixnum:+ fixnum:+ 0)
(variadic:define-right-deep fixnum:** fixnum:** 1)
\end{Verbatim}
The macro-call expansions of \CODE{variadic:{\allowbreak}define-{\allowbreak}associative} and
\CODE{variadic:{\allowbreak}define-{\allowbreak}associative} \textit{generate more macro definitions}, in
this case for ``\CODE{fixnum:+}'' and ``\CODE{fixnum:**}'', which for added
convenience are also the names of the corresponding procedures.

After the definition, using the debugging procedure \CODE{meta:macroexpand} we
can examine how variadic calls are always ``eliminated'' at macroexpansion
time, yielding efficient residual code:
\begin{Verbatim}
guile> (meta:macroexpand '(fixnum:+)) ;; no arguments: neutral element as a literal
0₇₂₅₃₁
guile> (meta:macroexpand '(fixnum:+ 7)) ;; one argument: no calls are needed
7₇₂₅₃₂
guile> (meta:macroexpand '(fixnum:+ 1 2)) ;; one sum
[call fixnum:+ 1₇₂₅₃₃ 2₇₂₅₃₄]₇₂₅₃₅
guile> (meta:macroexpand '(fixnum:+ 1 2 3 4)) ;; three sums, left-deep (currently faster)
[call fixnum:+ [call fixnum:+ [call fixnum:+ 1₇₂₅₃₆ 2₇₂₅₃₇]₇₂₅₃₈ 3₇₂₅₃₉]₇₂₅₄₀ 4₇₂₅₄₁]₇₂₅₄₂
guile> (meta:macroexpand '(fixnum:** 2 3 4 5)) ;; three calls, right-deep as requested
[call fixnum:** 2₇₂₆₀₅ [call fixnum:** 3₇₂₆₀₆ [call fixnum:** 4₇₂₆₀₇ 5₇₂₆₀₈]₇₂₆₀₉]₇₂₆₁₀]₇₂₆₁₁
\end{Verbatim}
Since variadic syntax is so convenient, we use it also for of many other macros which are not
procedure wrappers:
\begin{Verbatim}
guile> (meta:macroexpand '(e1:or))
0₇₂₄₆₄
guile> (meta:macroexpand '(e1:or a))
a₇₂₄₆₅
guile> (meta:macroexpand '(e1:or a b c))
[if a₇₂₄₆₆ ∈ {0} then [if b₇₂₄₆₇ ∈ {0} then c₇₂₄₆₈ else 1₇₂₄₆₉]₇₂₄₇₀ else 1₇₂₄₇₁]₇₂₄₇₂
guile> (meta:macroexpand '(e1:and a b c))
[if a₇₂₄₇₃ ∈ {0} then 0₇₂₄₇₄ else [if b₇₂₄₇₅ ∈ {0} then 0₇₂₄₇₆ else c₇₂₄₇₇]₇₂₄₇₈]₇₂₄₇₉
\end{Verbatim}

%%%%%%%%%%%
\subsubsection{Sum-of-product types}
\label{sum-types}
\label{extensible-sum-types}
\label{expressions-as-an-extensible-sum-type}
\TDEF{Sum-of-product} or \TDEF{sum} types are a kind of variant records,
introduced by ML and popular in the functional programming community.
\\
\\
Even in \EPSILONONE's untyped context, it is very convenient to automatically turn a sum
description into procedures for building, accessing and updating data, and for testing
the case of a given object.

As a classic example, let a list be either nil, or the cons of a head and a tail:
\begin{Verbatim}
e1> (sum:define my-list (nil) (cons head tail))
Defining the procedure my-list-nil...
Defining the procedure my-list-nil?...
Defining the procedure my-list-nil-explode...
Defining the procedure my-list-cons...
Defining the procedure my-list-cons-make-uninitialized...
Defining the procedure my-list-cons-explode...
Defining the procedure my-list-cons-get-head...
Defining the procedure my-list-cons-with-head...
Defining the procedure my-list-cons-set-head!...
Defining the procedure my-list-cons-get-tail...
Defining the procedure my-list-cons-with-tail...
Defining the procedure my-list-cons-set-tail!...
Defining the procedure my-list-cons?...
\end{Verbatim}
Our sum type definitions keep into account the number of cases which must be
represented as boxed, and do not generate tag fields unless needed.  We derive
the representation in memory from the sum definition in a way similar to
\cite[\SECTION{4.1}]{compiling-with-continuations}.
\begin{Verbatim}
e1> (my-list-nil) ;; unboxed
0
e1> (my-list-cons 10 (my-list-nil)) ;; just head and tail, no case tag
0x1b984d0[10 0]
e1> (my-list-cons 10 (my-list-nil)) ;; make a *new* cons: different address
0x1be3b70[10 0]
e1> (my-list-cons 10 20) ;; "ill-typed" as a list: the system doesn't care
0x1998350[10 20]
e1> (my-list-cons? (my-list-nil)) ;; is nil a cons?  (No, it's not)
0
e1> (sum:define complex (cartesian real imaginary)
                        (polar angle radius)) ;; two boxed cases: case tag needed
;;; [...]
e1> (complex-cartesian 100 200) ;; first case: tag 0
0x152c0c0[0 100 200]
e1> (complex-polar 100 200) ;; second case: tag 1
0x1502620[1 100 200]
\end{Verbatim}
We can now redefine \EPSILONZERO expressions as an \textit{open} sum-of-products,
openness meaning that more cases can be added later.  This permits more
flexibility, at the cost of a slightly less efficient representation in the
general case:
\begin{Verbatim}
(sum:define-open e0:expression
  (variable handle name)
  (value handle content)
  (bundle handle items)
  (primitive handle name actuals)
  (let handle bound-variables bound-expression body)
  (call handle procedure-name actuals)
  (call-indirect handle procedure-expression actuals)
  (if-in handle discriminand values then-branch else-branch)
  (fork handle procedure-name actuals)
  (join handle future)))
\end{Verbatim}
The representation is compatible with the one used in \CODE{core.e}, and from
now on it will also be possible to add new expression cases, for user-defined
expression forms.

%%%%%%%%%%%
\subsubsection{Closure Conversion}
\label{closure-conversion}
\label{closure-conversion-transform}
The purpose of this extension is adding \textit{statically-scoped},
\textit{higher-order} \textit{anonymous procedures} to \EPSILONONE,
implemented as \textit{closures}.
\\
\\
Anonymous procedures require two\footnote{\CODE{Call-closure} does not
  technically need to be added as a new syntactic case, as it would also be
  definable as a macro; only the \CODE{lambda} case has an expansion which
  depends on its context.  However having both cases representable as
  expressions is useful for the CPS transform, and may be a good idea in case
  we want to add type analyses in the future.} new syntax cases, \CODE{lambda} and \CODE{call-closure}:
\begin{Verbatim}
(sum:extend-open e0:expression
  (lambda handle formals body)
  (call-closure handle closure-expression actuals))
(e1:define (e1:lambda* formals body) ;; make a lambda expression
  (e0:expression-lambda (e0:fresh-handle) formals body))
;;; "Concrete syntax" for lambda, generating the new expression case.  This is a
;;; variadic macro of one or more arguments: body-forms is bound to the argument s-cdr.
(e1:define-macro (e1:lambda formals . body-forms)
  (sexpression:inject-expression
    (e1:lambda* (sexpression:eject-symbols formals)
                (e1:macroexpand `(e1:begin ,@body-forms)))))
\end{Verbatim}
Our closures are flat and minimal \cite[p.~132]{dybvig-phd-thesis}, consisting
of a single buffer holding a procedure name as its first element, followed by
its zero or more nonlocal values; for us, nonlocals are the variables
\textit{locally-bound out of the lambda-expression, occurring free in the
  lambda body, hence in particular not shadowed by the lambda formals}.  The
procedure referred by the closure takes \textit{the closure itself} as a
parameter, followed by the ones explicitly mentioned as formals, and locally
binds the nonlocal names by loading their values off the closure.

For example \CODE{(e0:let (a) 57 (e1:lambda (x) a))} will yield a closure of
two elements: a procedure name (automatically generated, two parameters: the
closure and \CODE{x}), and the only nonlocal value \CODE{57}.  The procedure
body will contain an \CODE{e0:let} block binding the name \CODE{a} to the second
element of the buffer pointed by its closure parameter.
\\
\\
Calling a closure is easy: given a closure and its actuals, we load the first
element referred by the closure and we perform an indirect call to it, passing
as parameters the closure itself followed by the actuals.
\\
\\
\CODE{closure:closure-convert} implements closure-conversion; it needs a
procedure and the set of locally-bound variables.  The procedure is very
simple, based on a multi-way conditional \CODE{e1:cond} which dispatches on the expression case\footnote{\textit{Pattern-matching} over sum types can be implemented with macros, and in fact we did that in a previous prototype: see \SECTION\ref{implementation-status}.}.  \CODE{e1:let*}, not to be confused with \CODE{e0:let*}\footnote{The naming convention is unfortunate in this case, but the sequential-binding name ``\CODE{let*}'', as distinct from parallel-binding ``\CODE{let}'', is convenient and has been conventional in Lisp for decades.}, is a block binding sequentially:
\begin{Verbatim}
(e1:define (closure:closure-convert e bounds)
  (e1:cond ((e0:expression-variable? e)
            (e0:variable* (e0:expression-variable-get-name e)))
           ((e0:expression-bundle? e)
            (e0:bundle* (closure:closure-convert-expressions
                           (e0:expression-bundle-get-items e) bounds)))
           ;; [...] other trivial cases [...]
           ((e0:expression-lambda? e) ;; Interesting case
            (e1:let* ((formals (e0:expression-lambda-get-formals e))
                      (nonlocals (set-as-list:subtraction bounds formals))
                      (old-body (e0:expression-lambda-get-body e))
                      (new-body (closure:closure-convert old-body
                                                         (set-as-list:union bounds formals)))
                      (used-nonlocals (set-as-list:intersection nonlocals
                                                                (e0:free-variables new-body))))
               ;; closure:make* defines a global procedure, then returns an expression
               ;; which builds a closure data structure including the global procedure name
               (closure:make* used-nonlocals
                              (closure:variables* used-nonlocals)
                              formals
                              new-body)))
           ((e0:expression-call-closure? e) ;; The second interesting case
            (e1:let* ((closure-expression (e0:expression-call-closure-get-closure-expression e))
                      (actuals (e0:expression-call-closure-get-actuals e))
                      (transformed-closure-name (symbol:fresh)))
              (e0:let* (list:singleton transformed-closure-name)
                       (closure:closure-convert closure-expression bounds)
                       (e0:call-indirect*
                           (e0:primitive* (e0:value buffer:get)
                                          (list:list (e0:variable* transformed-closure-name)
                                                     (e0:value* 0)))
                           (list:cons (e0:variable* transformed-closure-name)
                                      (closure:closure-convert-expressions actuals bounds))))))
           (else
            (e1:error "unknown extended or invalid expression"))))
(e1:define (closure:closure-convert-expressions es bounds)
  (e1:if (list:null? es)
    list:nil
    (list:cons (closure:closure-convert (list:head es) bounds)
               (closure:closure-convert-expressions (list:tail es) bounds))))
\end{Verbatim}
\CODE{closure:closure-convert} is the basis of our procedure transforms:
\begin{Verbatim}
(e1:define (closure:closure-convert-expression-transform expression)
  (closure:closure-convert expression set-as-list:empty))
(e1:define (closure:closure-convert-procedure-transform name formals body)
  (e0:bundle name
             formals
             (closure:closure-convert body formals)))

(transform:prepend-expression-transform! (e0:value closure:closure-convert-expression-transform))
(transform:prepend-procedure-transform! (e0:value closure:closure-convert-procedure-transform))
\end{Verbatim}
Now that we have installed the transform procedures, we can use closures:
\begin{Verbatim}
e1> (e1:define q (e1:let* ((a 1) (b 2) (c 3))
                   (e1:lambda (x)
                     (fixnum:+ a b c x))))
e1> (e1:call-closure q 4)
10
\end{Verbatim}
It should be remarked that closures are distinct from and incompatible with
\EPSILONZERO procedures.
Should we hide ordinary procedures from the user, and use closures only?

We could: it is possible to introduce (trivial) closures for all existing
procedures, retroactively transform away all procedure calls into closure
calls (and then into indirect calls by closure-conversion) and finally
change the cons-expander to generate a \textit{closure call} rather than a
procedure call as its default case.
This would make \EPSILONONE similar to a ``Lisp-1'' \cite{lisp-1-vs-lisp-2}
by hiding from the user the existence of procedures which are independent from
closures.

Such a move would be perfectly reasonable in many high-level personalities, but
we reject it for \EPSILONONE, for which we want to retain low-level control.

%%%%%%%%%%%
\subsubsection{Futures}
\label{futures}
Our \FORKNAME form in \EPSILONZERO is very inconvenient to use, needing a
procedure which must be given parameters to evaluate in foreground,
rather than just an expression (\SECTION\ref{fork-and-nonlocals},
p.~\pageref{fork-and-nonlocals}).  But closures make it easy to define friendlier futures, by a simple macro:
\begin{Verbatim}
(e1:define (future:fork-procedure thread-name future-closure)
  (e1:call-closure future-closure))

;;; Build a future which will asynchronously call the given closure:
(e1:define (future:asynchronously-call-closure closure)
  (e0:fork future:fork-procedure closure))

(e1:define-macro (e1:future . forms) ;; friendly syntax: any number of forms in sequence
  `(future:asynchronously-call-closure (e1:lambda () ,@forms)))
\end{Verbatim}

%%%%%%%%%%%
\subsubsection{First-class continuations}
\label{cps}
\label{cps-transform}
We implemented first-class continuations with a CPS transform
\cite{rabbit--steele,compiling-with-continuations,orbit--kranz-phd-thesis,transforms--leroy}
on expressions
extended with a \CODE{let/cc} form (``\CODE{CATCH}'' in \cite{scheme}), with
\CODE{call/cc} defined as a macro over \CODE{let/cc}.

Our CPS transform is more tentative than the \EPSILONONE personality, and
currently resides in \FILE{bootstrap/{\allowbreak}scheme/{\allowbreak}cps.{\allowbreak}scm}, and the trivial driver \FILE{bootstrap/{\allowbreak}scheme/{\allowbreak}cps-{\allowbreak}repl.{\allowbreak}scm}.  Implemented in a
very conventional style, it works but yields inefficient code and is
inefficient at transformation time as well: in particular the high number of
local variables generated by CPS stresses closure-conversion and its algorithm
to compute the free variables of an expression, currently quadratic.

The generated code allocates closures at a very high rate; it can be optimized
and some improvements appear easy, but to obtain really efficient code we would
need escape analysis, so that code sure not to escape could be recognized and
transformed differently.  Such global (or just ``incremental'') analyses can be
performed in our model, by having a CPS transform return its provisional
inefficient result but save the original untransformed code, to be reconsidered
later.
\\
\\
\label{in-dimension-relaxation-hack}Bundles have been problematic, since CPS maps our \CODE{e0:let} form, which
ignores excess items (\SECTION\ref{let-blocks-can-ignore-some-values},
p.~\pageref{let-blocks-can-ignore-some-values}), into a procedure call, which
does \textit{not} ignore excess parameters; in order to respect our
\CODE{e0:let} semantics we had to relax some dynamic checks in the \EPSILONZERO
interpreter, and
rely on a behavior which constitutes an error according to the semantics.  It
is not clear whether it would be best to update the semantics to ignore extra
parameters (hence \textit{defining a non-error behavior in more cases}, which
constrains implementations\footnote{We only touched the C version, by trivially removing two
  conditionals, one for \CODE{e0:call} and the other for \CODE{e0:call-indirect}.
  The same change can be easily
  replicated in the \EPSILONZERO self-interpreter.  In such symbolic interpreters
  removing the check is trivial and actually slightly \textit{improves}
  performance: this will not be true in a compiled implementation.}), or to forbid bundles altogether in conjunction
with CPS.
\\
\\
Continuations have been very useful to test and stress our transform system,
since a CPS transform is much less ``well-behaved'' than a closure-conversion
transform: CPS adds one more argument to \textit{every} procedure, making
transformed code fundamentally incompatible with its original version.  When
closures are not used, closure-conversion returns unchanged code, up to
handles; but a (naïve) CPS transform fundamentally changes the expression
structure even where no jump is performed.  Of course the CPS transform needs
to be applied \textit{retroactively} (\SECTION\ref{retroactive-transforms},
p.~\pageref{retroactive-transforms}).
\\
\\
We are not positive about traditional ``full'' continuations being pragmatically
the best foundation to base further extensions on; \textit{delimited
  continuations}
\cite{final-shift-for-call/cc,a-library-of-high-level-control-operators--using-delimited-continuations,delimited-continuations}
seem to provide some advantages, and we have experimented with them in early prototypes;
thanks to our open-ended design we may adopt them in the future.

%%%%%%%%%%%%%%%%%%%%%%%%%%%%%%%
\subsection{Implementation status}
%\subsection{Implementation status}
\label{implementation-status}The implementation is not mature, but it can be played with.  We currently
depend on Guile to parse and print s-expressions, and our current
implementation still lacks a compiler.

Such limitations are temporary and incidental: in the time available we chose
to develop transforms, more innovative and interesting, rather than
implementing well-known algorithms once again.  We do not envisage any
particular difficulty, and development will proceed during the following
months.

Some older prototypes, unmaintained but available at
\url{http://ageinghacker.net/epsilon-thesis-prototypes/},
contain code which could possibly be worth adapting and integrating into the
current implementation:
\begin{itemize}
\item
  an s-expression frontend
  written in OCaml for an older prototype, supporting the grammar of
  Figure~\SECTION\ref{s-expression-grammar-figure}; it works and contains a
  very powerful scanner supporting a variant of Thompson \cite{thompson-construction}
  and Rabin-Scott Constructions \cite{rabin-scott-construction} over
  large character-sets;
\item
  an incomplete compiler including liveness analysis and RTL generation;
\item
  pattern-matching macros working on a different implementation of sum-of-product types;
\item
  a mostly complete CamlP4 printer, intended to automatically translate the OCaml code
  into maintainable \EPSILONONE code.
\end{itemize}
An official part of the GNU project, epsilon is free software, released under
the GNU GPL version 3 or later \cite{gpl}.  Its home page is
\url{http://www.gnu.org/software/epsilon}.

\label{new-repo}
We manage the source code on a public bzr server at
\url{bzr://bzr.savannah.gnu.org/epsilon/trunk} \RED{[2015 note: switched from bzr to
git in late 2013: see \url{https://savannah.gnu.org/git/?group=epsilon}]}, and a public mailing list is
available for discussion.  See \url{https://savannah.gnu.org/projects/epsilon}
for more information.

%%%%%%%%%%%%%%%%%%%%%%%%%%%%%%%%%%%%%%%%%%%%%
\section{Future work}
\label{extension-composition}
Building a large body of extensions raises the issue of controlling their
interaction.  Transform-based extensions in particular, relying as they do on
the enumeration of all expression forms, require knowledge of all the
previously-added expression forms.  No solution to this problem is apparent.
However, without promising a ``silver bullet'' to language extension, we still
maintain our approach of \textit{layered} syntactic forms to be much more
suitable to extensibility than the traditional solution of a large unstructured
collection of language forms.

As an orthogonal problem, our current implementation does not currently keep a
map from expressions to original source locations
(\SECTION\ref{tracking-source-locations--introduction}), which may complicate
debugging.  Ad-hoc solutions involving an s-expression frontend keeping track
of source locations, then to be threaded through macros and transforms up to
the final generated code, seem perfectly feasible, with handles coming in
handy; on the other hand it is desirable to keep extension definitions as
uncluttered as possible, ideally by leaving the ``current'' location
information always implicit at each stage, in a monadic fashion.  A clean
solution to this problem seems well worth investigating.

%%%%%%%%%%%%%%%%%%%%%%%%%%%%%%%%%%%%%%%%%%%%%
\section{Summary}
Lisp is a powerful language, and its homoiconic syntax based on s-expressions
makes it easy to extend with macros.

We adopted a form of Lisp-style s-expressions
% whose addend types are not fixed
as a data structure to represent user syntax, but we keep it distinct from
expressions: our macros map s-expressions into expression objects; then, going
beyond Lisp, expression objects can be manipulated by user-specified
transform procedures, until all syntactic extensions are ``transformed away''
and only \EPSILONZERO forms remain.

We have shown in detail how \EPSILONONE, a low-level \EPSILON personality
useful as a basis to build other extensions, is bootstrapped from \EPSILONZERO
temporarily leaning on Guile.  Our bootstrapping code also constitutes a
complete definition of the macro and transform systems.

We closed by showing some interesting language extensions in \EPSILONONE, as
representative examples of our syntactic abstraction facilities.

\chapter{A parallel BiBOP garbage collector}
\label{gc-chapter}
\label{parallel-chapter}
\label{parallel-features-chapter}
\label{parallelism-chapter}

When a high-level program requiring garbage collection runs in parallel on a
multi-core machine, the memory subsystem easily becomes the bottleneck.  For
this reason we implemented a parallel collector for \EPSILON, actually starting
back when the current incarnation of the language was still taking shape,
testing it on a toy Lisp implementation we originally wrote as a teaching aid.

The collector's performance profile is meant to best match a mostly-functional
personality.  It is relatively easy to interface to C systems, and by design is not
limited to \EPSILON.

We call our system ``\CODE{epsilongc}''.  As for the language name, the initial
``\CODE{e}'' is always written lowercase.

\minitoc
Our parallel collector is non-moving, based on a variant of the BiBOP
organization. Building on the experience of Boehm's work and revisiting some
older ideas in the light of current hardware performance trends, we propose a
design leading to compact data representation and some measurable speedups,
particularly in the context of functional programs.

%% While discussing in detail the performance-critical sections of the
%% implementation we provide an intuitive justification for our choices.,
%% which we then corroborate with some measurements.

This effort results in a clean architecture based on just a few data
structures, which lends itself to experimentation with alternative techniques.
%% %\TODO{
%% %\category{CR-number}{subcategory}{third-level}
%% %\category{D.3.3}{Language Constructs and Features}{Dynamic storage management}
%% \category{D.3.4}{Processors}{Memory management (garbage collection)}
%% \category{D.4.2}{Storage Management}{Garbage collection}
%% %\category{D.1.3}{Programming Techniques}{Concurrent Programming - Parallel programming}
%% \category{D.1.1}{Applicative (Functional) Programming}{}
%% %% \category{D.1.3}{Concurrent Programming}{}

%% \terms{Design, Performance}

%\keywords{Garbage Collection, Parallel Computation, Functional Programming}

%% \begin{figure}[h]%[h]
%% %  \begin{center}
%%   \includegraphics[scale=0.35]{page-structure.eps}
%% %  \end{center}
%%   \caption{\label{page-structure-figure}\TODO{Foo!}}
%% \end{figure}

%%%%%%%%%%%%%%%%%%%%%%%%%%%%%%%%%%%%%%%%%%
\section{Motivation}
In recent years improvements in processor performance have been due more and
more to increased parallelism, while the trend of rising processor clock
frequency has dramatically slowed down. In contrast to what happens with
instruction-level parallelism, the task parallelism offered by modern
multi-cores must be explicitly exploited by the software, if {\em any} speedup
is to be obtained \cite{the-free-lunch-is-over}.

As multi-core architectures support a shared-memory
model\footnote{The architecture shown here does not generalize so well to NUMA
machines, more suitable as they are to a message-passing style where each
task runs in its own addressing space; message-passing is also
interesting, as the same interfaces could scale up to parallel computation
{\em over the network}.\\Moving away from thread parallelism to pure {\em process}
parallelism (one heap per process) would essentially eliminate the problem of
parallel non-distributed garbage collection, but such a revolution appears unlikely.
Other organizations like NUMA machines composed by SMP nodes, or 
machines where the NUMA effect is pronounced only between ``distant'' nodes,
look more realistic and are already being adopted by some current high-class machines
\cite{memory--drepper}. For such a hybrid SMP-in-NUMA model the techniques shown
here apply at the SMP level, just in the same way as they would apply to each
single machine in a cluster of SMPs.}
the
techniques presented here extend from the now ubiquitous desktop multi-core
machines to the older multi-socket SMPs, and to most recent medium-size
parallel machines containing several multi-core CPU dies.

The architecture we illustrate here is also suitable for sequential machine,
but the need for such a software is particularly stringent in a parallel context.
In a sense the rise of the number of CPUs {\em amplifies} the memory wall problem:
the memory bandwidth is a limited hardware resource which all cores have to
share, and raising the parallelism degree inevitably tightens up the
bottleneck, even without any synchronization.

%% \subsection{Motivations}

%% we looked for a free software parallel garbage collector;
%% Boehm's conservative-pointer-finding collector \cite{conservative--boehm} was
%% the obvious choice, and essentially the only one: see \SECTION\ref{related-work}
%% for a short discussion of the available alternatives.
%% \\
%% Despite its ease of use and good speed we felt that 
%% it was possible to obtain better performance and scalability than offered by
%% Boehm's collector with relatively little effort, particularly in the context of
%% a functional language; some tests with \NANOLISP, a simple implementation of Lisp that
%% we initially wrote as a prototype of \EPSILON's runtime, confirmed this intuition.
%% \footnote{Indeed, one reason why our
%%   initial tests with Boehm's collector gave disappointing results was the fact that
%%   initially we had not modified some of its runtime parameters such as the maximum
%%   heap size. With appropriate runtime settings Boehm's collector can perform
%%   significantly better than in its default configuration in the parallel case, as
%%   of version {\em 7.1alpha3-080224}; see \SECTION\ref{boehm-default-settings}
%%   for a comparison.}.

%%%%%%%%%%%%%%%%%%%%%%%%%%%%%%%%%%%%%%%%
\subsection{Boehm's garbage collector}
Boehm's garbage collector \cite{conservative--boehm,fast-multiprocessor-boehm} is the natural point of comparison for our work
because of several design similarities, including the idea of
(partially) conservative pointer finding, and the use of Unix signals to
interrupt mutators\footnote{Notwithstanding the outdated information at
{\small{\url{http://www.hpl.hp.com/}}} {\small{\url{personal/Hans_Boehm/gc/gcdescr.html}}}
Boehm's collector now also employs signals to stop
  mutators on all major platforms except Windows, where Unix signals are not
  supported but an analogous mechanism exists for suspending a thread from another
  thread.
  \cite{how-signals-are-really-used--boehm} mentions GNU/Linux, Solaris, Irix and Tru64.
  The Windows implementation in {\tt win32\_{\allowbreak}threads.c} uses signal-like
  primitives like {\tt SuspendThread()}.}.  For this reason it may be worth to
quickly highlight the main objectives we have set
forth for our implementation, in order to better explain the need for our effort
and to illustrate key similarities and differences.
Our objectives also more or less dictate several design and implementation
choices which we prefer to make explicit from the beginning.

First of all, C is clearly the language providing the best control on
performance for such a low level implementation where each memory access
matters.
A slightly less obvious choice is determined by the typical usage of
{\em parallel} systems, tending to concentrate on bulk processing rather than
interactive applications: for this reason we consider {\em bandwidth}, and not
latency, to be a priority; this choice excludes most incremental schemes
% (but see \SECTION\ref{incremental})
 and favors a {\em stop-the-world} model where
many threads can mutate in parallel or collect in parallel, but without any
time overlap between the two phases --- all of which is similar to Boehm's solution.
Since we are interested in the allocation pattern of functional programs,
consisting in a large number of small objects, it is paramount to make a good use of the limited
space in the primary and secondary caches (henceforth simply \TDEF{L1} and
\TDEF{L2}), by tightly packing objects together: we want to avoid padding
space between heap objects and not to force alignment constraints not specified
by the user.
Anyway, even if functional programs are our first concern, we would like
our collector to be also useful for (human-written) C programs, which
encourages us to adopt a
non-moving strategy like {\em mark-sweep} and to  avoid safe-points and 
use {\em conservative pointer finding {\bf for roots}}; on the other hand there
is no reason why other heap objects should not be traced exactly. The collector
API should be usable by humans, but not necessarily similar to {\tt malloc()} 
--- an important difference with respect to Boehm's collector.

\subsection{High-level design}
Most of our implementation ideas rely on a variant of the 
classic BiBOP strategy \cite{bibop-in-maclisp, dont-stop-the-bibop}
which, despite its simplicity, has been exploited surprisingly little: the only
discussion of an actually implemented similar solution that we have found is in a
little-cited 1993 paper by E. Ulrich Kriegel, \cite{card-based-bibop}.

In the different context of today, we propose BiBOP as a good match for
modern multiprocessor architectures.

We cannot claim novelty for most ideas, some of which are variations of very
old implementation techniques, as it is understandable after fifty
years\footnote{We remark one last time how McCarthy \textit{also} introduced
  garbage collection, in his wonderful \cite{lisp-mccarthy}.} of
research.

Nonetheless, we feel that our organization may have at least some aesthetic
value, in terms of its data structures and C interface.
%\begin{itemize}
%% \item Our user-level interface to the memory management system based on
%%   \TDEF{kinds}, \TDEF{sources} and \TDEF{pumps} may have some aesthetic
%%   value, besides slightly boosting allocation performance.
%\TODO{This is ok, but it should be qualified later by comparing the allocation
%  performance of kindless objects with kinded objects. Put a forward-link here.}

%% \item We define a set of core data structures and primitives which can
%%   be efficiently implemented, providing a clean and natural high-level
%%   architecture for parallel collectors.
%  \TODO{What I want to say: it has a clean understandable {\em architecture},
%    it's not a blob of stuff slammed together like all the collectors whose
%    source I've seen. Since GC implementation is so hard, having a model like mine is
%    useful. Do I convey this idea in the paragraph above?}

%\item

Our main idea is that \textit{the BiBOP scheme is appropriate for reducing
  memory pressure on machines with modern memory hierarchies}; we describe this
point by introducing the concept of \TDEF{data density} which we show to be at
least one reason for the good performance of our implementation.

%% \item If further developments follow our implementation may evolve into a
%%   testbed for comparing the performance of different collection
%%   strategies, in a parallel setting; such a framework may show its utility in
%%   tailoring the best collection algorithm to a particular mutator program, on
%%   the machine at hand.
%%   As far as we know, no implementation is currently available to do this in a
%%   parallel setting and without being tied to a particular language.
%% \end{itemize}

%%%%%%%%%%%%%%%%%%%%%
\subsection{The functional hypothesis}
Functional programs tend to allocate many small objects, the great majority of
which have one of only a few possible ``shapes''; in practice, most heap objects
will be conses, nodes of balanced binary trees, or more generally
components of inductive data structures with fixed size and
layout, often containing some constant attributes which must be frequently
inspected at runtime, such as the tags of our sum types (\SECTION\ref{sum-types}).
Depending on the programming style closures might also be allocated in
quantity; allocating other objects tends to be statistically much less
frequent, hence less critical for performance.
\\
We define the above set of assumptions as the \TDEF{functional hypothesis}: our
system is designed to run most efficiently when such hypothesis is verified,
yet \EPSILONGC  can and does work with any language, and may even be directly
employed for user-written C programs.%%  Anyway the functional hypothesis is the
%% underlying idea dictating many implementation choices, and in particular the
%% BiBOP strategy central to our design.

%%%%%%%%%%%%%%%%%%%%%
\section{The user view: kinds, sources and pumps}
At a very high level, any automatic memory management system serves to provide
\textit{an illusion of infinity}: an unlimited stream of objects created on
demand, each satisfying some specified requirements such as size and alignment.

Objects which are not useful any longer can be simply ignored: there is no need,
in general, for a user interface to the recycling system itself as the whole
point of garbage collection is to make object reusing {\em invisible} to the user,
who just keeps creating more objects as if the memory were unlimited.

The user-level API is built upon three main data structures: the {\em kind},
each instance of which defines one particular set of requirements for a group
of homogeneous objects, the {\em source}, which arranges for the creation of
objects of one specified kind, and the {\em pump}, providing a single
mutator thread with objects from a given source on demand, one object
at a time.

\subsection{Kinds}
%\paragraph{Kinds:}
We define a \TDEF{kind} as the specific representation of a group of
homogeneous heap objects. Each kind is characterized by a given
\TDEF{object size}, \TDEF{object alignment}, a \TDEF{tracer} function specifying
how to mark the pointers contained in an object given its address,
and particular \TDEF{metadata} values:
metadata include\footnote{Even if they currently comprise only tag and pointer,
  more metadata can be easily added in the future if the need arises.} an 
integer \TDEF{tag} and a \TDEF{pointer}, sharing the same values for all the objects
of the same kind. Given a pointer to a heap object, mutators are permitted to
inspect, but not modify, its metadata.

In general a kind should not be confused with a {\em type}: rather than a type
it identifies one {\em case} among the potentially many variants which, together,
make up a type. For example a {\em cons} kind could be defined, but
{\bf not}
a {\em list} kind, which would also comprise the empty list case, having of course
a different representation --- by the way, reasonably unboxed, as in \SECTION\ref{sum-types}.

%% $$ K \triangleq \mathbb{N} \times
%%                 \mathbb{N} \times
%%                 {\tt epsilongc\_{\allowbreak}tracer\_{\allowbreak}t}
%%                 \times \mathbb{Z}%{\tt tag\_{\allowbreak}t}
%%                 \times {\tt void*} $$
%% \TODO{Chr. suggested to be formal here. Does it look good?}

%\TODO{Christophe suggests to hint at what a {\em compiler} has to do to define
%  the needed kinds (like I did at the seminar). Why not.}

The tag could be usefully employed in a dynamically-typed language such as Lisp,
for example in order to test at runtime whether a given object is, effectively, a cons. In
a statically-typed language like ML the tag can encode the constructor
of tagged-sum objects.
The pointer metadatum can be useful to refer any reflection-related data not
fitting in a single integer.
%\TODO{Shall I advocate using the pointer for virtual method tables in
%  object-oriented languages? I'm not sure: kinds may become a lot; is this a
%  problem? If it's not then I should strongly advocate my idea also for
%  object-oriented languages.}

All the needed kinds are typically defined at initialization time, as global
structures shared by all mutator threads.
%\\
%\\
%% Since want our collector to also work \textit{with no boxedness tags}
%% (\SECTION\ref{boxedness-tags}), we can assume that pointers are not implicitly
%% recognizable as such at runtime.

\subsection{Sources}
%\paragraph{Sources:}
From the user's point of view a \TDEF{source} can be seen as a global inexhaustible
source of objects of a given kind.
In the typical case
%\footnote{\TODO{in some cases it makes sense to have more
%than one source with the same kind, but this really belongs to the implementation
%section.}}
the user will define exactly one source per kind at initialization
time, as an object shared by all mutator threads; after initialization mutator
threads will only refer sources to create their pumps.

\subsection{Pumps}
%\paragraph{Pumps:}
A \TDEF{pump} is a {\em thread-local} data structure 
%essentially
%\footnote{\TODO{I shouldn't speak about explicit page release here,
%but that topic is appropriate for a footnote in the next section}}
implementing but one user-level functionality, the creation of an object.

Each mutator thread will create its own pumps referring the shared, global
sources, then use its pumps to obtain new objects. Pumps have to be
explicitly destroyed at thread exit time.
% \TODO{In some cases they can be
%  destroyed even before, but saying also this would really risk to be pedantic...}

%\subsection{What does not fit in the picture: \TDEF{kindless} objects}
%\paragraph{What does not fit in the picture and \TDEF{kindless} objects:}
%\paragraph{\TDEF{Kindless} objects:}
\subsection{Kindless objects}
The strategy outlined above --- creating objects of some kind which has been
defined in advance --- suffices the great majority of
the objects ever created at runtime: for example in Lisp most heap-allocated
objects will be (s-)conses, and Prolog heaps will mostly be made of terms.  We call
\TDEF{kinded} all the objects created as shown above.

Some other heap-allocated objects do not fit so well in the picture as it
is not possible to foresee in advance their exact size: arrays and character strings come to
mind\footnote{Other slightly less obvious cases are {\em procedure activation
records}, which some runtimes of Scheme, Prolog and SML allocate
    on the heap; if the language supports dynamic code generation even {\em code blocks}
(either machine language or bytecode) might be heap-allocated and garbage
collected.
}.
We provide more ``traditional'' allocation primitives for such \TDEF{kindless} objects.

Notice how the kindless object API (see Figure~\ref{api-figure}) provides for less control:
vector elements can be either {\em all} potential pointers, or they can be
guaranteed by the user to include {\em no} pointers. There is not much control on metadata
either: all objects share the same\footnote{The actual values can be specified
  at initialization time, but nonetheless they must be the same for all
  kindless objects; it is typically reasonable to choose some values not
  used for kinds, so that at least kindless objects can be distinguished
  from kinded ones.
% \TODO{Well, large objects offer more control, but I don't want to
%    talk about large objects here. I will insert a forward ref.}
} tag and metadatum pointer; a user
requiring more expressive metadata has to explicitly encode them in the payload.
For reasons of general applicability and performance, \textit{we assume \textbf{not} to have boxedness tags
  available (\SECTION\ref{boxedness-tags}).}

\subsection{Miscellaneous user functionalities:}

Other primitives are provided to initialize and finalize the collector, to
register and unregister roots, to notify the memory system about new threads or
exited threads, to explicitly force a collection, and to temporarily disable
collections and re-enable them. 

As all of this is canonical and not particularly interesting, we will not
further % Yes, further is also an adverb 
pursue such details.

\begin{figure*}[htbp]
%\begin{multicols}{2}
\small{
%\begin{Verbatim}[numbers=none]
\begin{Verbatim}
/* A tracer is a pointer to a function taking a pointer as its parameter and
  returning nothing: */
typedef void (*epsilongc_tracer_t)(epsilongc_word_t);

/* Create a kind: */
epsilongc_kind_t  epsilongc_make_kind(const size_t object_size_in_words,
                                      const epsilongc_unsigned_integer_t
                                      pointers_per_object_in_the_worst_case,
                                      const size_t object_alignment_in_words,
                                      const epsilongc_metadatum_tag_t tag,
                                      const epsilongc_metadatum_pointer_t pointer,
                                      const epsilongc_tracer_t tracer);

/* Create a source from a kind: */
epsilongc_source_t epsilongc_make_source(epsilongc_kind_t k);

/* Initialize a (thread-local) pump from a source: */
void epsilongc_initialize_pump(epsilongc_pump_t pump,
                               epsilongc_source_t source);

/* Finalize a pump before exiting the thread: */
void epsilongc_finalize_pump(epsilongc_pump_t pump);

/* Allocate a kinded object from a thread-local pump: */
epsilongc_word_t epsilongc_allocate_from(epsilongc_pump_t pump);

/* Lookup metadata: */
epsilongc_tag_t epsilongc_object_to_tag(const epsilongc_word_t object);

epsilongc_metadatum_pointer_t
epsilongc_object_to_metadatum_pointer(const epsilongc_word_t object);

epsilongc_integer_t epsilongc_object_to_size_in_words(const epsilongc_word_t object);

/* Allocate kindless objects: */
epsilongc_word_t epsilongc_allocate_words_conservative(const epsilongc_integer_t size_in_words);
epsilongc_word_t epsilongc_allocate_words_leaf(const epsilongc_integer_t size_in_words);
epsilongc_word_t epsilongc_allocate_bytes_conservative(const epsilongc_integer_t size_in_bytes);
epsilongc_word_t epsilongc_allocate_bytes_leaf(const epsilongc_integer_t size_in_bytes);
\end{Verbatim}
}
%\end{multicols}
\caption{\label{api-figure} \EPSILONGC's essential user-level API.}
The source above is directly copied from header files, with only GCC function
attributes (to force inlining and such) removed and comments eliminated.
Despite looking unconventional the interface is not particularly complex, and in fact is
conceived so that performance-critical operations such as
{\tt epsilongc\_{\allowbreak}allocate\_{\allowbreak}from()} and metadata lookup functions can be easily
re-implemented in assembly, to be generated by a compiler as intrinsics.
\end{figure*}

%%%%%%%%%%%%%%%%%%%%%
\section{Implementation}

Despite their visual intuitiveness, the data structures above were designed
primarily for efficiency, and the actual role of each
structure is not apparent to the user: in particular the central data structure,
the \TDEF{page}, is completely hidden.

\label{gc-alignment}
Pointers are essential in the implementation of any language requiring dynamic
memory allocation, and in order to make pointers easier to recognize at runtime
in the absence of boxedness tags
and their dereference more efficient\footnote{On many RISC architectures
  pointers to misaligned objects may not be just a performance concern: some processor
  families such as {\em MIPS} and {\em Sparc} simply raise an exception in response to any
  attempt to dereference a non-word-aligned pointer. Others, such as the {\em
    x86} family, execute the misaligned dereference, but imposing a heavy
  execution time penalty.

  We prefer to simply forbid such pointers for all architectures, which may improve
  performance and helps to avoid the misidentification of many false
  pointers.

  We also assume
  convertibility from integer to pointer and vice versa without loss of
  information: even if not mandated by the C Standard (the type {\tt intptr\_{\allowbreak}t}
  itself is optional in \cite{c99}) such an assumption is in practice true on all architectures.}, we
restrict the set of heap pointers considered valid to {\em word-aligned}
pointers; one word is also the minimum size of a heap object 
representable without space overhead, and all the integers internally used in
the implementation are of type {\tt intptr\_{\allowbreak}t}, so that the size of all
memory structures remains a multiple of a word size.

The description below will proceed {\em from the bottom up}: since many data
structures and operations are usable with different collection strategies
requiring little or no modifications, we illustrate the various
possible operations before our way of combining them, in the spirit of
separating policy from mechanism.

\subsection{Kinded objects}
We represent each kinded object as a buffer of words, with {\em no header}; the
rationale of this choice is discussed in more depth in \SECTION\ref{data-density},
but the main idea is simply to have long packed arrays of objects in memory,
without any padding unless absolutely necessary\footnote{Padding {\em must} to
  introduced sometimes in order to respect the alignment constraints stated
  by the user: for example the user might require a three-word structure to be
  aligned to two or four words; in such cases there is no way to avoid wasting
  some space for each object.}. 

\subsection{BiBOP pages}
All kinded objects are allocated from data structures called
\TDEF{pages}\footnote{There is no {\em a priori} relation between BiBOP pages
  and operating system pages, whose sizes may well be different: BiBOP pages
  will typically be at least a few times larger than operating system
  pages, but still smaller than the L2 cache. In the following we always use the term
  {\em page} to mean ``BiBOP page''.}, similar to Kriegel's ``STSS
cards'' \cite{card-based-bibop}:
whenever a pump returns a pointer to a new object, the resulting address will refer a word
contained in a page.

Each page can only contains objects of {\em one} kind.
For each kind any number of pages, including zero, may exist
at any given time.

All pages have the same size, which must be a power of two; the page size is
also equal to
%the page
its
{\em alignment}: the rightmost $log_2 \texttt{epsilongc\_{\allowbreak}PAGE\_{\allowbreak}SIZE\_{\allowbreak}IN\_{\allowbreak}BYTES}$
bits of a page pointer are always guaranteed to be zero.
\\
A page is divided into \TDEF{page header}, \TDEF{mark array} and \TDEF{object slot
array}.

\paragraph{Page header}
The page header contains a copy of the kind metadata, which of course are valid
for all the objects in the page; the object referred by the metadatum pointer,
if any, is shared by all the pages of the same kind: only the pointer is copied.

Other information contained in the header includes kind-dependant data
such as the object size and effective size, the payload offset, and the number
of object slots in the page. All of this is computed once and for all when
a kind is created, and simply copied at page initialization time.
The address of the first dead slot (see below) is also held in the header.
\\
\\
\label{interior-pointers-are-ok}\label{masking-magic}Since
the header has offset zero within the page, given a pointer to any kinded
object, \textit{even interior}, the address of its page header can be trivially obtained
by {\em bitwise and}ing the pointer and the \TDEF{page mask}, defined as 
the {\em bitwise negation} of $\texttt{epsilongc\_{\allowbreak}PAGE\_{\allowbreak}SIZE\_{\allowbreak}IN\_{\allowbreak}BYTES} - 1$.
This allows mutators to access metadata at runtime with an overhead of two to four
assembly instructions, when needed; on the other hand the negligible space
overhead of storing metadata once per page makes this solution completely
acceptable even for languages which don't make use of them.

%\TODO{list element data}

\paragraph{Mark array}

The mark array is placed right after the header, with no padding: since the
header size is a multiple of the word size, the mark array is guaranteed
to always begin at a word boundary.

The mark array stores liveness information for each object\footnote{It is
  interesting to compare this with Boehm's collector, which stores one element
  per object {\em word}, thus making tracing simpler. We have chosen to slightly
  complicate the mapping from mark array elements to objects instead, to speed
  up the critical operation of page sweeping, and in general trading more
  computation for fewer memory accesses.}
 in the page: since
we currently need only one bit per object, the array could conceptually always be
implemented as a bit vector.

As marking is parallel, mark arrays are concurrently updated by several
threads, which requires some atomic memory accesses (see \SECTION\ref{marking}). On many
machines byte stores are always atomic, and
even when suitable atomic instructions for bitwise operations are provided
working with a {\em byte vector} may be more
efficient\footnote{\cite{fast-multiprocessor-boehm}, written in 2000, compares
  the solutions on several architectures, finding that the optimal solution
  depends on the machine. According to our recent tests, the best strategy
  between bit arrays and byte arrays remains machine-dependent.}. On (hypothetical)
architectures where the compiler did not support the required intrinsics, and
where an atomic byte store were not provided, one could use a {\em word vector}.
The implementation allows the user to choose at configuration time among bit, byte or
word, bit being the default.

%\TODO{one element per object, not per word. This should be stressed: Boehm uses
%  one bit per word, and one byte per word in the parallel case:
%  \cite{fast-multiprocessor-boehm} (section 6.1, page 6, top) says that in the
%  parallel case on 32-bit machines mark bytes take one eighth of the heap
%  size. From this point of view our solution is much superior.}

Alternatively, it is possible to enable \TDEF{out-of-page mark arrays} at
configuration time, so that mark arrays are stored as separate {\tt malloc()}ed
buffers; in this case the mark array area
in a page degenerates to a single pointer, and accessing the mark array from a
page requires one indirection.
\label{cache-access-patterns-with-aligned-pages}
Our original rationale for
implementing this strategy was to avoid some cache conflict misses due to the fact that mark
arrays share the same alignment on all pages.
As benchmarks showed that this is not a problem in practice with
modern multi-way set associative caches, this strategy has not been
pursued further by separating headers from slot arrays.
%\MOVE{Link here when speaking about cache access patterns using aligned pages}

\paragraph{Object slot array}
The \TDEF{object slot array} begins after the end of the mark array, at the first
word with the required alignment. Object slots contain the payload of each
page. At any given time each object slot may be either \TDEF{used} or
\TDEF{unused}: when used it contains an object payload; when unused, its first
word contains a pointer to the next unused object in the same page, or {\tt
NULL} in the case of the last unused slot.

For each page unused slots make up an independent free-list where elements are
always ordered by address.
%: this invariant is only broken in one case,
%shown in XXXsection~\ref{explicit-kinded-deallocation}.
\\\\
In order to avoid mistaking free list pointers in unused objects for pointers in
used objects during conservative pointer finding, free list pointers are stored in
\TDEF{concealed} form by default\footnote{Free list pointer concealing can be disabled at configuration time.}.

Concealing consists in applying some function
$c : \SET{A} \rightarrow \SET{A}^{\complement}$
to a free list pointer; it is important for $c$ to be bijective, as
concealing and then \TDEF{unconcealing} (i.e. applying $c^{-1}$ to) a pointer
must preserve information.

$c$ is trivially implemented as a C macro computing the successor function in
{\tt unsigned} (wrap-around) arithmetic: since
its domain consists of word-aligned pointers, the elements of its image are
guaranteed to be misaligned, hence they cannot be
mistaken for pointers. The cost of applying either $c$ or $c^{-1}$ is one
assembly instruction requiring no memory accesses\footnote{Assuming
  instructions such as either {\tt inc}/{\tt dec} or {\tt add}/{\tt sub}
  with a small {\em immediate} parameter; again, all modern machines satisfy
  this condition.}.

Depending on the kind, some unused space may be present between the
end of the mark array and the beginning of the slot array, and at the end of the
page; in either case these two padding spaces are strictly smaller than the
object effective size.

\paragraph{The global page table}
The global \TDEF{page table} serves to recognize which part of the address space
is being used for the garbage-collected heap; such information is important
for avoiding dereferencing false pointers when doing conservative pointer
finding.
\\
Moreover, the collector needs to be able to recognize whether a heap
pointer refers a kinded object in a page slot array or a large object --- no
particular provision is needed for kindless small objects, but we defer the
justification of this fact to \SECTION\ref{implicit-kinds}. Since we
\label{interior-pointers-are-ok-for-large-objects}
support interior pointers for large objects, it must also be possible to
efficiently map an arbitrary (word-aligned) interior pointer to an initial
pointer.

We call \TDEF{candidate pointer} a word which is suspected to be a
(possibly interior) object pointer at marking time, and \TDEF{candidate page}
the address of the hypothetical page which would contain the object referred by
a candidate pointer. Of course candidate pages have alignment $log_2 \texttt{epsilongc\_{\allowbreak}PAGE\_{\allowbreak}SIZE\_{\allowbreak}IN\_{\allowbreak}BYTES}$.
\\\\
At an abstract level, the table implements a function $f$ mapping a
non-{\tt NULL} candidate page $p$
to an element $s$ of the disjoint sum
$$Sort \triangleq \{kinded\} + \{nonheap\} + LargeObjects$$

%\TODO{Christophe suggests to make a dotted list but I have a space problem!}

If $f : p \mapsto kinded$ then the candidate page $p$ is actually a page;
if instead $f : p \mapsto nonheap$ then $p$ is a pointer referring some object
out of the garbage-collected heap, or a false pointer.
Otherwise $f : p \mapsto l$, where $l$ is the address of the beginning of the
large object containing the word pointed by $p$.
\\\\
Given a value for $p$ stored as a key, a simple encoding allows us to represent
any element of $Sort$ in a single word:
{\tt NULL} represents $nonheap$, $s = p$ stands for $kinded$, and any
other value of $s$ is interpreted as a large object pointer.

The table is implemented as a simple resizable chained hash where the first element of each
bucket is stored within the bucket pointer array itself\footnote{This optimization is the reason why we
  don't include {\tt NULL} in the domain of $f$: we use the value {\tt NULL} as
  a key in a hash table element out of the bucket to mean that the element
  is currently unused.}, as first described in
\cite{my-hash-optimization}; the hash function is modulo.

One essential optimization at mark time consists in {\em not} consulting the
page at all, which would be comparatively expensive, for {\tt NULL} or
misaligned candidate pointers.

It is interesting to notice how all {\em updates} to the global page table occur
at mutation time, when creating or destroying\footnote{See
  \SECTION\ref{page-destruction} for the reason why pages must be
  destroyed at {\em mutation} rather than collection time.} pages and large objects;
unfortunately such updates require critical sections which, short as they are,
may nonetheless limit scalability.
By contrast at collection time the table is only {\em read}, which allows us to
completely avoid critical sections for table access during that stage.%\OR{phase}{stage}.

%\subsection{Page operations}
%Individual pages may be \TDEF{created}, \TDEF{destroyed}, \TDEF{swept} and \TDEF{refurbished}.
\subsubsection{Page creation}
\TDEF{Creating} a page involves allocating space from the C heap,
filling the header fields, initializing the mark and
object slot arrays and registering the page in global structures.

%\begin{lstlisting}
%\EPSILONGC_page_t
%\EPSILONGC_make_page_and_add_it_to_global_structures(\EPSILONGC_source_t source)
%  __attribute__((malloc));
%\end{lstlisting}

Because of the alignment requirements we currently allocate pages with
{\tt posix\_{\allowbreak}memalign()}\footnote{An interesting alternative to explore would involve
  using {\tt mmap()} to allocate a group of pages; for some (non-GNU) implementations
  of {\tt posix\_{\allowbreak}memalign()}, the {\tt mmap()} solution might incur a
  significantly lower space overhead, at the cost of always involving the
  kernel in page creation. Using {\tt mmap()} could in fact make deallocation more portable, as {\tt free()}ing
  buffers allocated with {\tt posix\_{\allowbreak}memalign()} is only permitted on GNU
  systems, as far as we know (\cite{gnu-c-library}, ``Allocating Aligned Memory Blocks'', currently at
  subsection~3.2.2.7).
  
  However the {\tt mmap()} solution has some issues of its own:  
  {\tt mmap()} only guarantees {\tt sysconf(\_{\allowbreak}SC\_{\allowbreak}PAGESIZE)} alignment,
  hence pages could only be reasonably {\tt mmap}ped in large groups, with some
  space overhead at the beginning and the end.
Making {\tt epsilongc\_{\allowbreak}PAGE\_{\allowbreak}SIZE\_{\allowbreak}IN\_{\allowbreak}BYTES} equal to {\tt sysconf(\_{\allowbreak}SC\_{\allowbreak}PAGESIZE)}
  would solve the space overhead problem, but at the
  price of forcing pages to be unacceptably small.
  {\tt unmmap}ping space from the
  middle of a {\tt mmap}ped buffer is supported, but deallocation of
  single pages would still be a problem unless {\tt epsilongc\_{\allowbreak}PAGE\_{\allowbreak}SIZE\_{\allowbreak}IN\_{\allowbreak}BYTES}
  were chosen to be a multiple of {\tt sysconf(\_{\allowbreak}SC\_{\allowbreak}PAGESIZE)}. Re-{\tt mmap}ping 
  a previously {\tt unmmap}ped part of a buffer is {\em typically}
  supported, even if such behavior is not mandated by POSIX.
  In addition we would need some data structure to keep track of which pages in a
  large buffer are {\tt mmap}ped at any given time.
\\Anyway, despite all the complexity, such an idea seems worthy of some exploration.
% Fuck, this formula was so nice, and I don't need it any more!
%$\frac{mcm({\tt sysconf(\_{\allowbreak}SC\_{\allowbreak}PAGESIZE)}, {\tt epsilongc\_{\allowbreak}PAGE\_{\allowbreak}SIZE\_{\allowbreak}IN\_{\allowbreak}BYTES})}{\tt epsilongc\_{\allowbreak}PAGE\_{\allowbreak}SIZE\_{\allowbreak}IN\_{\allowbreak}BYTES}$.
}; as this may involve a kernel call and/or synchronization in the C library,
such operation tends to be both expensive and hard to parallelize.

Filling the header involves little more than copying some fields from the kind
data structure, which is directly referred by the source, and making the free-list
head point to the payload beginning. Nothing of this is performance-critical.

The mark array has to be zeroed at creation, with a {\tt memset()} call.
This should be relatively efficient, just involving some evictions
from L1 --- however having the mark array in L1 at page
creation time does not buy us anything, as mark arrays are only touched
during collection.
If out-of-page mark arrays are enabled then we should add a {\tt malloc()} call
to the cost.

Building the free list involves some memory traffic, as all objects need to be
touched. Unless objects have effective size larger than a cache line the
complete object slot array has to be brought into cache. Even if this phase by itself
is expensive, it
may work like a sort of prefetching: if the page is used soon, all of it
will already be loaded at least in the L2 cache.
\\
We define \TDEF{backward free list building}\footnote{The actual direction of
  free list building, from higher addresses down to lower ones or from lower
  addresses up to higher ones, has no effect on performance as long as it is {\em the opposite}
  of the allocation direction: note in particular how automatic hardware
  prefetching works in either direction on modern processors 
  (\cite{memory--drepper}, section 3.3.2, ``Single Threaded Sequential Access'').}
the strategy of building the free list starting from the {\em last} slot which
will be used for allocation. This solution has locality advantages in case of
large page size, under the assumption that a just-created page will be used
soon for allocating: if the page size is larger than the L1 data cache,
building the free list backwards makes it very likely that the memory
touched first while allocating will be already in L1; the rest of the page
will be still in L2. It is possible to choose between forward and backward free
list building at configuration time.

The final step is registering the page in the page table, which requires a
critical section on the global mutex,
%\TODO{Where should I introduce all the mutexes I use?}
plus a {\tt malloc()} call within the critical section in case of hash collision.
\\
%\MAYBE{\UGLY{All of this} shows that}
All of this makes page creation a relatively expensive and non-scalable operation.
%\TODO{Shall I introduce the idea of minimizing page creation and the need for refurbishing here?}

\subsubsection{Page sweeping}

\TDEF{Sweeping} can be performed on an individual page without need for
synchronization or kernel calls. It simply involves scanning the mark array
and, for each $i$-th element, either clearing the corresponding element if
$array[i]$ is one, or making the $i$-th object slot {\em unused} by re-adding it
to the free-list if $array[i]$ is zero. Since the mark array is examined in order (either forward or
backward, as per the free-list building direction), free list elements
are kept ordered by address in the list. All the words of dead objects other
than the first one are overwritten\footnote{Each word is overwritten with a
  configuration-dependent value impossible to mistake for a pointer: either
  the {\tt 0xdead} constant (which is easy to recognize for humans) if the collector is configured in
  debug mode, or otherwise simply {\tt 0} (which might lead to a slightly more
    efficient implementation on some architectures, possibly saving a
    {\em load immediate} instruction). Overwriting dead slots can also be
  completely disabled at configuration time.}, to 
prevent future false pointers referring the slot to
keep alive the objects which were referred by the now dead slot.

Memory access patterns in sweeping are similar to the ones in mark array
initialization and free-list construction; in particular a just-swept page will
likely remain cached at least in L2 --- and the next lines to be used will be
in L1, if backward free list building is enabled.

\subsubsection{Page refurbishing}

It is possible to re-use an empty page of some kind for objects of another
kind: such operation is called \TDEF{refurbishing}, and involves reconstructing
the header, mark array and free list.

Refurbishing has essentially the same overhead as sweeping, and the cache
effects of the two operations are also comparable: allocations from a
just-refurbished page on the same thread which performed the refurbishing is
efficient as all the page cache lines will still be in L1 and L2.

\subsubsection{Page destruction}

Destroying a page involves its deallocation and removal from the global page
table: such operations are expensive and non-scalable, involving
synchronization and possibly kernel calls.

\subsection{Sources}
From the implementation point of view a source is quite a trivial structure,
serving as repository of pages. Each source simply contains two lists of pages,
the \TDEF{full pages list} and the \TDEF{non-full pages list}, plus a mutex
for synchronizing access to such lists.

\subsection{Pumps}
Pumps are performance-critical structures whose purpose at the implementation
level consists in caching frequently accessed data about the objects to
allocate. Such criticality is evident from the API in Figure~\ref{api-figure},
showing how existing pump data structures are {\em initialized} rather than
dynamically allocated, in an effort to save a pointer indirection at runtime:
pumps are conceived to be declared in programs as {\tt \_{\allowbreak}\_{\allowbreak}thread} variables of
type {\tt struct epsilongc\_{\allowbreak}pump}, rather than as pointers.

At any given moment a pump may conceptually ``contain'' a page reserved to the
allocating thread, or no page; of course at the implementation level such an inclusion is
represented with a page pointer field. Its other relevant field is the current
head of the page free list, again kept in the pump rather than in the contained page in
order to avoid a pointer indirection at allocation time: in fact the free-list
head field of the page is, counter-intuitively, {\em not} updated at each
allocation. The free-list head field of the {\em pump} is set to {\tt NULL} when the
pump contains no page.

\subsubsection{The allocation function}
Despite the allocation being the only user-level operation on a pump, such a
functionality is very performance-critical. Allocating from a given pump
involves unconcealing the free-list field into a temporary variable, if
non-{\tt NULL} dereferencing it, setting the free-list head to the just loaded
value and finally returning the temporary. This shorter and far more common
execution path is carefully optimized and costs about {\em ten assembly
instructions}, with no taken\footnote{It is worth to provide GCC with an
optimization hint via \tt{\_{\allowbreak}\_{\allowbreak}builtin\_{\allowbreak}expect}().} jumps; the other execution
path is taken in case of \TDEF{page change} time, when a page is filled and another
one must be acquired from the relevant pool, or at the first allocation for a
pump with no page: it involves synchronization with the pool mutex and access
to its lists. If no non-full pages are available, a page is taken from a
\TDEF{global empty pages list} (at the cost of one further synchronization) and
refurbished if needed. If no empty pages are available, an heuristic is
employed to decide whether to create a new page, or to trigger a
collection. Page change is also the taken as the occasion for destroying empty
pages, if an heuristic says that there are more than enough: the rationale here
is to avoid destroying pages too frequently, since they might be needed again
and both creation and destruction are expensive.

Repeatedly allocating from a page which was recently swept by the same thread
and which contains many unused slots should be cache-friendly: sweeping works
like a prefetch phase to load the page payload into the L1 or L2 cache, and
even without on-demand sweep the hardware automatic prefetch may be activated
when there is much free space on the page, as consecutive addresses are
generated. Using pumps automatically guarantees that a page is only used for
allocation by one CPU at a time, which avoids cache ping-pong.

\subsection{Kindless and large objects}
\label{implicit-kinds}
The data structures and primitives shown above provide no hints about the
implementation of kindless objects, yet the idea is quite simple. A set of
\TDEF{implicit kinds}, \TDEF{sources} and per-thread \TDEF{pumps}\footnote{Implicit pumps
are created at thread registration and destroyed at thread un-registration
time.}, of user-definable sizes, are automatically defined: in this sense most
kindless objects are just kinded objects ``in disguise'', only slightly less
efficient because of the need for mapping an object size to a pump at runtime,
and because of the possibility of internal fragmentation: not all possible
sizes will be realistically provided, so the allocation of an object of a given
size might be satisfied by using a larger buffer. For each size two kinds are
provided, one with a fully conservative tracer, and another one with a leaf
tracer (called ``atomic'' in the jargon of Boehm's collector).

It is easy to see how the solution above is not completely general, as it
cannot satisfy allocation requests for objects larger than a page or even
just larger than the maximum implicit kind size which has been fixed by the
user. A different mechanism is provided for \TDEF{large objects}, which are simply
allocated one by one with {\tt malloc()} and destroyed with {\tt free()}. Their
implementation is simple-minded and quite inefficient in both space and time,
which given the functional hypothesis should hopefully not be serious. Of
course the user-level API completely hides the difference between
implicitly-kinded and large objects.

\subsection{Garbage collection}

\label{page-destruction}
A collection is initiated by one mutator, which stops all the other mutators with a
signal. This choice has the advantage of allowing a simple user API, but
significantly complicates the collector implementation: any function not
reentrant with respect to signals, notably including {\tt malloc()} and {\tt
free()}, can not be used at collection time: this is the reason why empty pages
have to be destroyed at {\em mutation} rather than collection time.

The collection phase may internally proceed in two different orders according to a
configuration option: if \TDEF{on-demand sweeping} is enabled, as per the
default, the three 
sub-phases are {\em non-deferred sweeping, root marking and marking}, otherwise
they are {\em root marking, marking and sweeping}. In any case it is central to
maintain the invariant according to which a complete heap marking is
followed by a complete sweeping, before the next marking can begin.

On-demand sweeping consists in sweeping a page during mutation at page change
time, {\em just before allocation from it begins}: such a choice is more
cache-friendly than the traditional \TDEF{stop-the-world sweep}, but it may leave some
pages still to be swept when a 
collection begins: the non-deferred sweeping sub-phase, typically very short,
serves to sweep such remaining pages. 
Non-deferred sweeping and stop-the-world sweeping share the exact same implementation.

After collection all mutators are restarted with a second signal.

\paragraph{Root marking}
Root marking is very simple, and currently {\em sequential}. Just like Boehm's collector
in most of its configurations, it uses {\tt setjmp()} for finding register
roots in a portable way.

\paragraph{Marking}
\label{marking}
Given the atomicity of mark array stores parallel marking can easily proceed in
parallel without synchronization, if we accept
the possibility of some (statistically unlikely) duplicate work; our
implementation is quite canonical and closely follows Boehm's one
\cite{mostly-parallel--boehm}, with load balancing in the style of Taura and
Yonezawa \cite{sgc}. It should be noted that the BiBOP organization does not
affect marking in any significant way.

\paragraph{Sweeping}
Parallel sweeping is even simpler, with pages dictating the natural granularity
for the operation of each thread: pages are simply taken from a list, swept and
put back into another list.

%%%%%%%%%%%%%%%%%%%%%
\subsection{Synchronization}
One interesting and possibly original detail involves our locking style: in
order to prevent a collection from starting during a critical section at mutation
time, a global {\em read-write lock} is locked for reading at mutation, before
acquiring the relevant mutex: the collection triggering function, before
sending the signal, locks the same read-write lock {\em for writing}.

\subsection{Data density}
\label{data-density}

The system internally measures object size and alignment in {\em machine
words}, and one word is the minimum size of a kinded object which can be
represented without padding, in absence of alignment constraints specified by
the user; with an alignment greater than one word, it becomes necessary in some
cases to add some padding space right after the object payload; we call
the \TDEF{effective size} of an object the sum of its size and its alignment padding.
%\MAYBE{In the following, when speaking about the ``size'' of an object, I will always
%implicitly refer to its effective size.}

\label{effective-size}
Given a kind $k$ of objects with alignment $a_k$ and size
$s_k$, we define the effective size $e_k$ needed to store each object, and
the corresponding \TDEF{data density} $d_k$, the number of objects
representable per word, as:
$$e_k \triangleq a_k \cdot \left\lceil\frac{s_k}{a_k}\right\rceil \qquad\qquad d_k \triangleq \frac{1}{e_k}$$

The definitions above intentionally disregard all the sources of memory
overhead out of object slot arrays, including mark arrays and
all garbage collector data structures, the rationale being that density is not
meant as a measure of memory occupation, but rather as an index of
{\em the number of objects fitting in a cache line}: 
as mark arrays and other collector data
structures are mostly accessed
%\OR{in non-overlapping time intervals}
{at different times} 
from the objects {\em per se} and reside in different cache lines, optimizing
data density maximizes the amount of useful information stored in the
physically limited cache space at mutation time.

Data density may be reasonably defined in the same way independently from the
garbage collecting strategy, and indeed it is of some interest to compare
the values of $d_k$ in different memory management systems for two kinds
which are widely employed in functional programs, the {\em cons} (two words) and the 
non-empty {\em node} of an Red-Black binary tree of one given color\footnote{The
  example trivially generalizes to AVL trees, the idea being simply that the
  balance-related information can usefully be represented as {\em meta-}data
  rather than data.}
(three words: {\em left},
%\OR
{\em datum} %{\em data}
and {\em right}). Neither kind has
alignment requirements, hence $a_{cons}=a_{node}=1$.

Several systems such as the {\em GNU libc} {\tt malloc()} facility
\cite{gnu-c-library}, all the other allocators derived from Doug Lea's
{\tt malloc()} and --- even more interestingly --- Boehm's collector
\cite{conservative--boehm}, allocate all buffers at double-word-aligned
addresses and may also add some internal status information {\em to each buffer};
%\cite{malloc--lea};
metadata, when needed, must be
represented as part of {\em each} object, adding to $s_k$. Instead many other
systems, including just for example OCaml, do not 
force any alignment but always add one header word per
object\footnote{Some systems add even more than one header word per
  object. Sun's JDK, MMTk \cite{mmtk} and Microsoft's CLR, for example,
  use two words.},
sufficient to include a short tag, which again we consider part of $s_k$.

If metadata are accessed at runtime, as it is the case with dynamically-typed
languages, with Boehm's collector we have $d_{cons}=d_{node}=\frac{1}{4}$.
When metadata are not needed Boehm has optimal density in the cons
case with $d_{cons}=\frac{1}{2}$, but again $d_{node}=\frac{1}{4}$.
In OCaml, with or without metadata, $d_{cons}=\frac{1}{3}$ and
$d_{node}=\frac{1}{4}$.

Independently of the need for metadata at runtime our model allows us to reach
optimal density for both kinds, with $d_{cons}=\frac{1}{2}$ and $d_{node}=\frac{1}{3}$.
\\
\\
The data density of a particular representation seems likely
to play a role in the {\em overall} efficiency of the system, even
ignoring the cost of allocation and collection and considering only object
accesses; anyway further empirical evidence will be needed to confirm this
supposition for real world programs.

%% \SECTION\ref{data-density-benchmark}
%% includes a comparison of the performance of \EPSILONGC\ on \NANOLISP in its default
%% configuration and when using a different data representation with intentionally lower
%% data density.

%%%%%%%%%%%%%%%%%%%%%
\subsection{Closures}
\label{closures}
\label{closure}
Functional programs written in certain styles\footnote{And in particular
when using simple compilers or interpreters: higher-order code can be
simplified with flow analysis.} or CPS-transformed (\SECTION\ref{cps-transform})
create a considerable number of short-lived closures at runtime.
At a first look such a scenario does not seem to
respect the functional hypothesis, as in principle closures can have many
different shapes, depending on the number of non-locals captured in the
environment, and on the fact that each non-local can be a pointer or a
non-pointer.
\\
Even if allocating all closures as kindless objects would work, the overhead of
such a simple-minded solution is in fact easy to avoid.

First of all it should be observed that the great majority of functions
need either zero or one variable in their non-local environment; it may be
worth to add specific kinds for such common cases, and possibly also
for the most performance-critical functions with larger non-local
environments, when it is possible to recognize them with compile-time
heuristics or after profiling.
\\
The number of needed kinds can be reduced by establishing a convention for
ordering non-locals in their environment arrays, according to whether they are
pointers or not: either first all pointers then all non-pointers, or vice-versa.

The idea of \TDEF{normalizing the representation} is a sort of pattern in the
BiBOP scheme, generalizable to many other cases when using statically typed
languages or \EPSILON personalities: there is no reason why two cases of different concrete types,
possibly completely unconnected at a semantic label but with the same 
effective size and number of potential pointer fields, cannot be represented
in such a way to share the same kind.

%%%%%%%%%%%%%%%%%%%%%
\subsection{Lazy and object-oriented personalities}
\label{lazy}
\label{lazy-languages}
Lazy languages require a slightly more sophisticated data representation than
call-by-value languages, as in a realistic implementation it must be possible to
destructively update a still-unevaluated thunk, and replace it with the result
at the end of its computation.
\\
Unsurprisingly, \EPSILONGC\ does not provide any support for changing the kind
of an existing object while maintaining its identity; that could be possible
{\em at collection time} in a moving scheme, but not with
mark-sweep\footnote{In a moving BiBOP collector otherwise similar to
\EPSILONGC\ it might be reasonable to split each kind into a \TDEF{evaluated}
kind, plus a \TDEF{thunk-or-evaluated} one: all the alive evaluated objects of a
thunk-or-evaluated kind would be \TDEF{re-kinded} at collection time. This idea
does not look particularly hard to implement, but keeping the collector both
efficient and language-agnostic might be challenging. Moving-time hooks
definable by the user would solve the problem, at some cost; the overhead could
be reduced by allowing to re-compile the hooks {\em as part of} the collector,
to be called as inline functions, like described for example in \cite{gc-interface-in-evm}.}.

Any standard solution already employed by the collectors for lazy languages
such as Haskell can be adopted: unfortunately some of the cleanness of the
BiBOP model is lost in this case, as data% ({\em not} metadata, for obvious reasons)
needs to be tagged with at least a boolean (two in a concurrent
environment: objects may be \TDEF{thunks}, \TDEF{in flux} or \TDEF{ready})
recording the evaluation state of an object; any unused bit sequence in the
payload or even the mark array entry of the object can do the job.

Accessing possibly still-to-be-evaluated objects will often require a
conditional at runtime, just like in conventional implementations of lazy
languages; after an object is known to be ready, BiBOP metadata can be
accessed just as for eager languages.

Such a solution also necessarily requires some form of synchronization if
the mutator threads are more than one: of course it is always possible to add a
synchronization word in the payload, if needed.
\\
From this point of view the situation is not different for ``managed''
languages such as Java, where {\em each} object contains a header word reserved
for that purpose; yet we believe that not forcing such an expensive
representation for {\em all} objects is preferable in the general case; the
user can always implement some additional logic where needed, out of the
memory management system {\em per se}.

For most runtimes there is no reason for keeping {\em one mutex per object},
and even lazy languages such as Haskell normally employ {\em strictness
analysis} to statically recognize many cases in which laziness is not
needed, and more efficient traditional representations can be safely used.

The work about {\em Prolific Types} \cite{prolific}
%cited in \SECTION\ref{related-work} is
is relevant for object-oriented languages.

\section{Status}
\label{gc-status}
\label{gc-implementation-status}
\EPSILONGC's implementation totals around 5000 lines of heavily commented C
code, quite easy to understand for being such a low-level concurrent
piece of code not sparing C macros, {\tt \#ifdef}s, GCC function attributes and
intrinsics in order to be support {\em Autoconf} options and be as general and
efficient as possible.

In preliminary micro-benchmarks
(\url{http://ageinghacker.net/publications/gc-draft.pdf}) \EPSILONGC
appears to perform better than Boehm's collector; anyway no realistic
parallel workload has been measured, and we feel that our conclusions can only be
tentative in this respect.

Our collector is \textit{not} currently used by \EPSILON: since the current
implementation of \EPSILON still relies on Guile for the s-expression frontend
(\SECTION\ref{implementation-status})
it also shares its memory management system, which has been a custom sequential
mark-sweep garbage collector up to Guile 1.8.x, then replaced with Boehm's collector
in the new series Guile 2.0.x.
We expect to integrate \EPSILONGC into \EPSILON as soon as we drop the
dependency on Guile, or when we write a compiler --- which can be done even
while keeping the interactive system linked with Guile without re-implementing
the frontend, at the cost of not having access to the frontend from compiled code.

Support for mutexes and other imperative synchronization features is trivial to
add to \EPSILON's implementation with C primitives.
\\
\\
Exploiting the BiBOP organization from \EPSILONONE does not appear particularly
problematic.  It will be interesting to test the benefits of the BiBOP strategy
in code strongly based on sum-of-products (\SECTION\ref{sum-types}),
\textit{such as in complex transforms}; the necessary changes in the representation of sum tags do not seem
very involved.

\EPSILONGC will work as it is with \EPSILON, but in the longer term we plan to
turn the current mark-sweep collector into the old generation of a generational
system, where the younger generation is copying; this will be particularly
relevant for the allocation patterns of CPS code, which tends to produce
short-lived objects at a high rate.
Implementing a collector which can be interrupted at any time by signals has
been a fun and instructive challenge, but in the future we plan to seize the
opportunity of coping with a moving collector in the young generation to
introduce safe points.
\EPSILONGC\ also needs a couple of new functionalities, the most urgent of
which are support for {\em finalization} and {\em weak pointers}.
\\
\\
The \CODE{epsilongc} sources have been committed to the main \EPSILON
repository (see
\url{https://savannah.gnu.org/bzr/?group=epsilon})\footnote{\RED{2015 note: the
GNU epsilon repository no longer uses bzr, and is now managed with git: see \SECTION\ref{new-repo}.}}, as an
independent subdirectory with its own build system.

Like the rest of the system it is free software, released under the GNU GPL
version 3 or later \cite{gpl}.

\section{Summary}
We implemented \EPSILONGC, a parallel mark-sweep conservative-pointer-finding
garbage collector for multicore machines.  Conceived for \EPSILON, it is
general enough to be used by other systems as well.

In order to exploit the memory hierarchies of modern machines, we pack data in
a dense way without prefixing every object with a header, segregating objects
by memory representation in a BiBOP organization.

This solution is most appropriate for functional personalities, in which most
objects belong to one of a small set of kinds.

% -*- mode: latex; fill-column: 79; mode: auto-fill; mode: flyspell; buffer-file-coding-system: utf-8 -*-
\UNNUMBEREDCHAPTER{Conclusion}

We formally specified and implemented a practical \textit{extensible
programming language} based on a very small first-order imperative core, plus
powerful syntactic abstraction features: Lisp-style \textit{macros} map user
s-expression syntax into expression data structures; user-specified \textit{transforms}
permit arbitrary code-to-code transformation, with the intent of supporting
extended syntactic features which are gradually ``transformed away'' into core
forms.  This open-ended approach enables research and experimentation.

As examples of the power of our extension mechanisms, we used transforms to
implement
\textit{higher-order lexically-scoped anonymous procedures}
and \textit{first-class continuations}, on top of a
core language only supporting named global procedures.
\\
\\
The language is very expressive and permits \textit{reflection} and
\textit{self-modification}; it is possible to update the global state of the
system by global modifications, possibly up until a state where the program is
``static'', convenient for analysis and compilable with traditional techniques.

We formally developed an analysis for static programs, and proved a \textit{soundness}
property about it with respect to the dynamic semantics.  We argue that such
formal reasoning is only possible thanks to the size and simplicity of the
core language.
\\
\\
The state of the system can be saved and restored with \textit{unexec} and
\textit{exec} facilities based on marshalling.

The language supports asynchronous threads and is suitable for modern multi-core
machines.  We implemented a parallel garbage collector, not yet integrated in
the system, to limit garbage collection bottlenecks.

The implementation is not mature yet, but can be played with.  The bulk of the
system is written in itself, using C for the runtime, and Guile as a temporary
dependency for bootstrapping.
\\
\\
An official part of the GNU project, epsilon is free software, released under
the GNU GPL version 3 or later \cite{gpl}.  Its home page is
\url{http://www.gnu.org/software/epsilon}.

The source code is managed on a public bzr server\footnote{\RED{2015 note: the repository
  switched from bzr to git in late 2013: see
  \SECTION\ref{new-repo}.}}, and a public mailing
list is available for discussion: see
\url{https://savannah.gnu.org/projects/epsilon} for more information.

%% %% %%%%%%%%%%%%%%%%%%%%%%%%%%%%%%%%%%%%%%%%%%%%%%%%%%%%%%%%%%%%%%%%%%%%
%% %% %%%%%%%%%%%%%%%%%%%%%%%%%%%%%%%%%%%%%%%%%%%%%%%% Appendices
%% %% %%%%%%%%%%%%%%%%%%%%%%%%%%%%%%%%%%%%%%%%%%%%%%%%%%%%%%%%%%%%%%%%%%%%
% Appendices (and unnumbered chapters) start here; from now on \chapter
% introduces an appendix, and as far as I know there's no way to undo
% this.  \chapter* continues to work, so it's ok to put bibliography and
% indices here.
\appendix

%% %% %%%%%%%%%%%%%%%%%%%%%%%%%%%%%%%%%%%%%%%%%%%%%%%%%%%%%%%%%%%%%%%%%%%%
%% %% %%%%%%%%%%%%%%%%%%%%%%%%%%%%%%%%%%%%%%%%%%%%%%%% Bibliography
%% %% %%%%%%%%%%%%%%%%%%%%%%%%%%%%%%%%%%%%%%%%%%%%%%%%%%%%%%%%%%%%%%%%%%%%
% -*- mode: latex; fill-column: 79; mode: auto-fill; mode: flyspell; buffer-file-coding-system: utf-8 -*-

% I want all of them to be in the bibliography, even if I won't have a lot to
% say about RRS and R3RS.
\nocite{scheme,rrs,r2rs,r3rs,r4rs,r5rs,r6rs}

%% Here comes the actual bibliography:

\bibliographystyle{plain} % [42]

\nocite{artificial-beings--pitrat}
\nocite{implementation-of-a-reflective-system--pitrat}

\bibliography{thesis}

%% %% %% %%%%%%%%%%%%%%%%%%%%%%%%%%%%%%%%%%%%%%%%%%%%%%%%%%%%%%%%%%%%%%%%%%%%
%% %% %% %%%%%%%%%%%%%%%%%%%%%%%%%%%%%%%%%%%%%%%%%%%%%%%% Indices
%% %% %% %%%%%%%%%%%%%%%%%%%%%%%%%%%%%%%%%%%%%%%%%%%%%%%%%%%%%%%%%%%%%%%%%%%%
%% \INCLUDEWHENCOMPLETE{indices}

%% %% %%%%%%%%%%%%%%%%%%%%%%%%%%%%%%%%%%%%%%%%%%%%%%%%%%%%%%%%%%%%%%%%%%%%
%% %% %%%%%%%%%%%%%%%%%%%%%%%%%%%%%%%%%%%%%%%%%%%%%%%% Scratch reminders
%% %% %%%%%%%%%%%%%%%%%%%%%%%%%%%%%%%%%%%%%%%%%%%%%%%% for myself
%% %% %%%%%%%%%%%%%%%%%%%%%%%%%%%%%%%%%%%%%%%%%%%%%%%%%%%%%%%%%%%%%%%%%%%%
%\INCLUDEWHENCOMPLETE{scratch}
%\INCLUDEUNLESSCOMPLETE{scratch}
%\INCLUDEWHENCOMPLETE{scratch}
%\include{limbo}
%\include{pgf}

%% %%%%%%%%%%%%%%%%%%%%%%%%%%%%%%%%%%%%%%%%%%%%%%%%%%%%%%%%%%%%%%%%%%%%
%% %%%%%%%%%%%%%%%%%%%%%%%%%%%%%%%%%%%%%%%%%%%%%%%% Abstract at the end, etc.
%% %%%%%%%%%%%%%%%%%%%%%%%%%%%%%%%%%%%%%%%%%%%%%%%%%%%%%%%%%%%%%%%%%%%%
\WHENCOMPLETE{% -*- mode: latex; fill-column: 79; mode: flyspell; buffer-file-coding-system: -*- utf-8 -*-

{\cleartoleftpage}
{\thispagestyle{empty} % No header or footer
\BETWEENTINYANDSMALL
%% \setlength{\parskip}{0pt}
%% \setlength{\parsep}{0pt}
%% \setlength{\headsep}{0pt}
%% \setlength{\topskip}{0pt}
%% \setlength{\topmargin}{0pt}
%% \setlength{\topsep}{0pt}
%% \setlength{\partopsep}{0pt}

%\hrule
\hrule
\paragraph*{Titre en français}
\FRENCHTITLE
\\\hrule

\paragraph*{Résumé en français}
% -*- mode: latex; fill-column: 79; mode: auto-fill; mode: flyspell; buffer-file-coding-system: utf-8 -*-
%%%%%%%%%%%%%%%%%%%%%%%%%%%%%%%%%%%%%%%%%%%%%%%%%%%%

\index{résumé}

%\em % I like all the abstract text to be in italics. I can change that if needed
Le \TDEF{réductionnisme} est une technique réaliste de conception et
implantation de vrais langages de programmation, et conduit à des
solutions plus faciles à étendre, expérimenter et analyser.

Nous spécifions formellement et implantons un langage de programmation
extensible, basé sur un \textit{langage-noyau minimaliste impératif du
  premier ordre}, équipé de \textit{mécanismes d'abstraction} forts et
avec des possibilités de \textit{réflexion} et
\textit{auto-modification}.  Le langage peut être étendu à des niveaux
très hauts~: en utilisant des \textit{macros} à la Lisp et des
\textit{transformations de code à code} réécrivant les expressions
étendues en expressions-noyau, nous définissons les clôtures et les
continuations de première classe au dessus du noyau.

Les programmes qui ne s'auto-modifient pas peuvent être analysés
formellement, grâce à la simplicité de la sémantique.  Nous
développons formellement un exemple d'\textit{analyse statique} et
nous prouvons une \textit{propriété de soundness} par apport à la
sémantique dynamique.

Nous développons un \textit{ramasse-miettes parallèle} qui convient aux machines
multi-cœurs, pour permettre l'exécution efficace de programmes parallèles.

\\\hrule

\paragraph*{Titre en anglais}
\TITLE\ --- \SUBTITLE
\\\hrule

\paragraph*{Résumé en anglais}
% -*- mode: latex; fill-column: 79; mode: auto-fill; mode: flyspell; buffer-file-coding-system: utf-8 -*-
%%%%%%%%%%%%%%%%%%%%%%%%%%%%%%%%%%%%%%%%%%%%%%%%%%%%

\index{abstract}
%\em % I like all the abstract text to be in italics. I can change that if needed

\textit{Reductionism} is a viable strategy for designing and implementing practical
programming languages, leading to solutions which are easier to extend,
experiment with and formally analyze.

We formally specify and implement an extensible programming language, based on
a \textit{minimalistic first-order imperative core language} plus strong
\textit{abstraction mechanisms}, \textit{reflection} and
\textit{self-modification} features.  The language can be extended to very high
levels: by using Lisp-style \textit{macros} and code-to-code
\textit{transforms} which automatically rewrite high-level expressions into
core forms, we define closures and first-class continuations on top of the core.

Non-self-modifying programs can be analyzed and formally reasoned upon, thanks
to the language simple semantics.  We formally develop a \textit{static
  analysis} and prove a \textit{soundness property} with respect to the
dynamic semantics.

We develop a \textit{parallel garbage collector} suitable to multi-core
machines to permit efficient execution of parallel programs.

\\\hrule\hrule

\paragraph*{Discipline}
Informatique
\\\hrule

\paragraph*{Mots-clés}
% -*- mode: latex; fill-column: 79; mode: flyspell; buffer-file-coding-system: -*- utf-8 -*-
programmation, langage, extensibilité, macro, transformation, reflection, \textit{bootstrap}, interprétation, compilation, parallélisme, concurrence, ramasse-miettes

\\\hrule

\paragraph*{Intitulé et adresse du laboratoire}
\INVISIBLE{.}
\\
LIPN, UMR 7030 -- CNRS, Institut Galilée, Université Paris 13\\
99, avenue J.-B. Clément\\
93430 Villetaneuse\\
France
}\\
\hrule
}

\end{document}